\documentclass[aos]{imsart}

\RequirePackage{amsthm,amsmath,amsfonts,amssymb}
\RequirePackage[authoryear]{natbib}
\RequirePackage[colorlinks,citecolor=blue,urlcolor=blue]{hyperref}
\RequirePackage{graphicx}
\RequirePackage{multirow}
\RequirePackage[caption=false,position=top]{subfig}

\startlocaldefs
\theoremstyle{plain}

\newtheorem{theorem}{Theorem}[section]

\newtheorem{thm}{Theorem}
\newtheorem{lem}[theorem]{Lemma}

\theoremstyle{remark}

\newtheorem{remark}{Remark}
\newtheorem{assu}{Assumption}
\endlocaldefs

\begin{document}

\begin{frontmatter}
\title{Learn-As-you-GO (LAGO) Trials: Optimizing Treatments and Preventing Trial Failure Through Ongoing Learning}
\runtitle{Learn-As-you-GO (LAGO) Trials}

\begin{aug}
\author[A]{\fnms{Ante}~\snm{Bing}\ead[label=e1]{abing@bu.edu}\orcid{0009-0001-4230-1675}},
\author[B]{\fnms{Donna}~\snm{Spiegelman}\ead[label=e2]{donna.spiegelman@yale.edu}\orcid{0000-0003-4006-4650}},
\author[C]{\fnms{Daniel}~\snm{Nevo}\ead[label=e3]{danielnevo@gmail.com}
\orcid{0000-0002-9770-827X}}
\and
\author[A]{\fnms{Judith}~\snm{J. Lok}\ead[label=e4]{jjlok@bu.edu}\orcid{0000-0001-7526-2572}}
\runauthor{Bing, Spiegelman, Nevo and Lok}

\address[A]{Department of Mathematics and Statistics,
Boston University\printead[presep={,\ }]{e1,e4}}
\address[B]{Department of Biostatistics,
Yale University\printead[presep={,\ }]{e2}}
\address[C]{Department of Statistics and Operations Research,
Tel Aviv University\printead[presep={,\ }]{e3}}
\end{aug}

\begin{abstract}
It is well known that changing the intervention package while a trial is ongoing does not lead to valid inference using standard statistical methods. However, it is often necessary to adapt, tailor, or tweak a complex intervention package in public health implementation trials, especially when the intervention package does not have the desired effect. This article presents conditions under which the resulting analyses remain valid even when the intervention package is adapted while a trial is ongoing. Our results on such Learn-As-you-GO (LAGO) studies extend the theory of LAGO for binary outcomes following a logistic regression model \citep{nevo2018analysis} to LAGO for continuous outcomes under flexible conditional mean model.
We derive point and interval estimators of the intervention effects and ensure the validity of hypothesis tests for an overall intervention effect. We develop a confidence set for the optimal intervention package, which achieves a pre-specified mean outcome while minimizing cost, and confidence bands for the mean outcome under all intervention package compositions. This work will be useful for the design and analysis of large-scale intervention trials where the intervention package is adapted, tailored, or tweaked while the trial is ongoing. 
\end{abstract}



\end{frontmatter}

\section{Introduction}

Traditionally regulators have not allowed adaptive trial designs to be used for evidence in the drug approval process, but currently, the Food and Drug Administration (FDA) has put forward guidelines for adaptive clinical trials \citep{FDAdevices, FDAadaptiveNew} and encourages trialists interested in adaptive trials to contact the FDA in the planning stage \citep{FDAadaptiveNew}. While current adaptive clinical trial designs allow for changing randomization probabilities and dropping treatment arms \citep{FDAadaptiveNew}, they do not allow for changes to an intervention package composition based on outcomes collected during the earlier stages of a trial. This limitation may have contributed to cases where large-scale intervention trials have ``failed'' \citep{semrau2017outcomes, fogel2018factors, stensland2014adult}. 

Learn-As-you-GO (LAGO) studies consist of $K>1$ stages. As data are being collected, the later-stage intervention package is systematically ``learned'' based on results from previous stages. \cite{nevo2018analysis} provided a methodology for LAGO studies with binary outcomes. They assumed that a logistic regression model holds for the probability of success given the intervention package components, for which they prove consistency and asymptotic normality. 
In LAGO trials, learning assumes that as the number of patients in prior stages gets large, the recommended intervention in the next stage converges in probability to a fixed intervention. 
They also demonstrated the validity of hypothesis tests for the overall intervention effect. 
Moreover, they proved that the optimal intervention, $\boldsymbol{x}^{opt}$, which attains a pre-specified outcome probability while minimizing cost, can be ``learned'', in the sense that we can find $\hat{\boldsymbol{x}}^{opt}$ with $\hat{\boldsymbol{x}}^{opt}  \stackrel{P}{\rightarrow} \boldsymbol{x}^{opt}$ as the number of patients in each stage gets large.

An adaptive clinical design popular in implementation science \citep{collins2011multiphase} that can be compared to LAGO is the multiphase optimization strategy (MOST) \citep{collins2007multiphase,collins2014optimization}. Although both LAGO and MOST aim to identify the optimal intervention package and evaluate its impact, the designs are different. 
MOST consists of three phases: preparation, optimization, and evaluation. In the optimization phase, the optimal intervention package is determined in a short-term factorial design usually of intermediate outcomes, and in the evaluation phase, the effect of the previously identified intervention package is independently assessed through a randomized controlled trial (RCT), the analysis of which does not use the outcomes from the optimization phase or any occurrences subsequently.
In contrast, in LAGO trials, the composition of the intervention package is updated at the end of each stage and is dependent on the outcomes from previous stages. LAGO enables researchers to systematically modify the intervention package while the trial is ongoing while preserving Type I error rate and pre-specified power. 
Since all outcomes are used in the final analysis, as further discussed in Section \ref{simulation section}, LAGO can achieve higher power compared to MOST when the same strategy is used to identify the optimal intervention package. 

LAGO has been identified as a promising approach in the field of implementation science \citep{beidas2022promises}. According to these authors, LAGO can improve the alignment between implementation strategies and partner needs and contexts. Additionally, the user-centered design and approach of LAGO enables the optimization of implementation strategies and calibration of implementation support based on demonstrated need.

This article extends LAGO to continuous outcomes. To achieve this, we adopt a Generalized Linear Model (GLM) framework. 
The GLM framework in this context refers to any semi-parametric model for the conditional mean of a random outcome on other covaraites that is linear in a link function, without any additional restrictions on the outcome distribution.
When analyzing a LAGO study, we cannot condition on the later-stage interventions since that would imply conditioning on a function of the outcomes from earlier stages. In addition, variations in the intervention components is needed for LAGO to identify the treatment effect parameters.
Extending LAGO to continuous outcomes is not straightforward, as \cite{nevo2018analysis} use a novel coupling argument to prove the asymptotic properties of their estimators, and such coupling approach cannot be generalized to continuous outcomes.
To overcome this challenge, we assume throughout this article that the errors in the GLM \citep{mccullagh2019generalized} are independent of the composition of the intervention package. Section \ref{GLM} outlines the comprehensive steps implemented to surmount the difficulty in generalizing LAGO to continuous outcomes. Details are provided in Appendix \ref{general_link_appendix}. 

The aims of a LAGO study with continuous outcomes are  
1$)$ to estimate the impact of the individual intervention package components on the outcome mean, 
2$)$ to test for an overall intervention effect when the intervention package has been ``learned,'' and
3$)$ to estimate which intervention package, $\boldsymbol{x}^{opt}$, will minimize the cost while yielding a pre-specified mean outcome. We solve these aims for various GLMs with independent, identically distributed errors. 

Section \ref{application} applies LAGO to the BetterBirth study \citep{hirschhorn2015learning,semrau2017outcomes}. The BetterBirth study was a costly failed trial, led by Harvard researcher Atul Gawande, which aimed to improve maternal and child health around the time of birth. By mimicking applying LAGO, we identify the optimal intervention package that can increase the percentage of essential birth practices (EBPs) performed to a pre-specified target goal while minimizing cost. 
Our analysis differs from that in \citet{nevo2018analysis} because they focused on a single binary outcome, oxytocin administered immediately after delivery, as opposed to the continuous outcome of the percentage of EBPs performed. Furthermore, since we consider the percentage of EBPs as a continuous outcome, this represents a unique application of our LAGO method that could not have been accomplished using the existing LAGO method. 

This article is organized as follows.
Section~\ref{setting} details the setting and notation. 
Section \ref{est beta} describes the estimating equations for the proposed estimator for a LAGO study with a GLM.
Section~\ref{GLM} outlines the proofs of the asymptotic properties of the proposed estimator. Details are presented in Appendix \ref{general_link_appendix}.
Section \ref{testing section} describes how to test for an overall intervention effect in a LAGO study.
Section \ref{confidence} describes the confidence sets and confidence bands for the optimal intervention. 
Section \ref{simulation section} describes simulations. 
Section \ref{application} describes the LAGO analysis of the BetterBirth study.
Section~\ref{Discussion} discusses our findings and future research topics.

\section{Setting, notation, and assumptions}\label{setting}

Let $Y_{ij}$ be the continuous outcome for patient $i$ in center $j$. The optimal intervention is defined as the intervention that results in the mean outcome reaching a pre-specified goal $\theta$, while minimizing cost. 
Suppose the multi-component intervention package $\boldsymbol{x}$ has $p$ components. Let $C(\boldsymbol{x})$ be a known cost function, in dollars, of the intervention package $\boldsymbol{x}$. 
In healthcare settings, a linear cost function is often reasonable to assume. $C(\boldsymbol{x})$ will then include a fixed cost plus the sum of the product between dosage and the unit cost for each of the intervention components. 
We also consider a cubic cost function, which allows for an initial economy of scale followed by increased costs when the component levels exceed a threshold (see Appendix \ref{additional cost} for details). This type of cubic cost function is particularly applicable in healthcare settings, where the marginal cost of a product may rise prohibitively as local supplies are depleted. 
In addition, in the BetterBirth study, one of the intervention components is the duration of the on-site intervention launch, measured in days (see Section \ref{application} for more details on the BetterBirth study). 
Here, a cubic cost function is suitable because extending the launch beyond three days would disrupt care at the participating healthcare centers, making it infeasible. Typically, less feasible doses are associated with higher costs.  

Let $\boldsymbol{z}_{j}$ be fixed center-specific characteristics, such as hospital district or hospital birth volume, that may be related to the outcome of interest.
\begin{assu}\label{glm_model_assumption}
We assume that the expected outcome of an individual $i$ in center $j$ with characteristic $\boldsymbol{z}_j$ under recommended intervention package $\boldsymbol{X}_j = \boldsymbol{x}_j$ and actual intervention package $\boldsymbol{A}_j=\boldsymbol{a}_j$, $E\left(Y_{ij} | \boldsymbol{A}_{j}=\boldsymbol{a}_j, \boldsymbol{X}_{j}=\boldsymbol{x}_j, \boldsymbol{z}_{j}; \boldsymbol{\beta} \right)$ only depends on the recommended intervention $\boldsymbol{x}_j$ through the actual intervention $\boldsymbol{a}_j$ and follows a GLM
\begin{equation}\label{contmodel}
g\left(E\left(Y_{ij} | \boldsymbol{A}_{j}=\boldsymbol{a}_j, \boldsymbol{X}_{j}=\boldsymbol{x}_j, \boldsymbol{z}_{j}; \boldsymbol{\beta} \right)\right) = \beta_{0}+\boldsymbol{\beta}_{1}^{T} \boldsymbol{a}_{j}+\boldsymbol{\beta}_{2}^{T} \boldsymbol{z}_{j},
\end{equation}
where $g()$ is a twice continuously differentiable link function. $\boldsymbol{\beta}^T=\left({\beta}_{0}, \boldsymbol{\beta}_{1}^T,\boldsymbol{\beta}_{2}^T\right)$ are unknown parameters to be estimated from the data. Typically, our main interest is in the vector $\boldsymbol{\beta}_{1}$, and the null hypothesis of no intervention effect, $H_{0} : \boldsymbol{\beta}_{1}=\boldsymbol{0}$. 
\end{assu}

The optimal intervention package for center $j$ with baseline covariates $\boldsymbol{z}_{j}$ is the solution to the center-specific optimization problem
\begin{equation}\label{aim}
\operatorname{Min}_{\boldsymbol{x}_{j}} C\left(\boldsymbol{x}_{j}\right)
\;\; \text {  subject to } \;\; E\left(Y_{ij} | \boldsymbol{x}_{j}, \boldsymbol{z}_{j} \; ; \boldsymbol{\beta} \right) \geq \theta ,
\end{equation}
with each component $p$ of $\boldsymbol{x}_{j}$ in a pre-determined interval $\left[L_{p}, U_{p}\right]$, where $p=1,...,P$.
Alternatively, instead of an absolute goal, the optimal intervention package could e.g. aim to increase the mean outcome for each center $j$ by a pre-specified goal $\Delta\theta$. 
In what follows, we focus on case (\ref{aim}). In many settings, the intervention package implemented will not be center-specific, in which case all $\boldsymbol{z}$ are ignored (or set to 0). 
We assume that there is a unique solution to equation (\ref{aim}). If we consider a linear cost function, and the intervention consists of two components with marginal costs $c_1$ and $c_2$, then if $\beta_{11}/c_1 \neq \beta_{12}/c_2$, the solution to equation (\ref{aim}) is unique.
Part of the aims of a LAGO study is to identify the solution, $\boldsymbol{x}^{opt}_{j}$, that solves equation (\ref{aim}). 

For simplicity, we present the theory for LAGO studies with $K=2$ stages. Appendix \ref{K>2} extends the theory to $K > 2$ stages. 

In each stage $k$, $k=1$ or $2$, 
$n_j^{(k)}$ patients in center $j$, $j=1,\ldots,J^{(k)}$ are enrolled, with $J^{(k)}$ fixed. In a randomized controlled LAGO trial (c-LAGO), centers may be randomized to either the intervention or control arm. In a quasi-experimental LAGO study, pre-study outcome data might be collected (before-after or ba-LAGO). 
Let $n^{(k)}=\sum_{j=1}^{J^{(k)}} n_{j}^{(k)}$ be the number of patients in stage $k$, $n=n^{(1)} + n^{(2)}$ the total number of patients across the 2 stages, and
$\alpha_{j k}=\lim _{n \rightarrow \infty} n_{j}^{(k)} / n \; > 0$ , which assumes that the ratio between the sample size for center $j$ in stage $k$ and the total sample size converges to a non-zero constant as $n$ goes to infinity. This assumption is a reasonable approximation if the number of patients in each center $j$ of stages 1 and 2 is large, as in most large scale implementation trials.
At stage $1$, the initial intervention package for each center, $\boldsymbol{x}_{j}^{(1)}$, is recommended by the study investigators' best guess, and/or based on pilot data before the study starts. In stage $2$, the
recommended intervention package $\boldsymbol{X}_{j}^{(2,n^{(1)})}$ for center $j$ is estimated from stage $1$ data and other accumulated knowledge, possibly including summary measures of
e.g. qualitative information from providers, patients and other stakeholders. 
The superscript $(2,n^{(1)})$ indicates that $\boldsymbol{X}_{j}^{(2,n^{(1)})}$ depends on the data from the $n^{(1)}$ patients in stage 1. 
Often, $\boldsymbol{X}_{j}^{(2,n^{(1)})}$ will solve equation (\ref{aim}) with the estimator $\hat{\boldsymbol{\beta}}^{(1)}$ for $\boldsymbol{\beta}$ based on the stage 1 data (see Remark \ref{remark}).
Let $\boldsymbol{A}_{j}^{(2,n^{(1)})}$ be the actual intervention package implemented in center $j$ in stage $2$. 
Under perfect adherence, $\boldsymbol{A}_{j}^{(2,n^{(1)})}=\boldsymbol{X}_{j}^{(2,n^{(1)})}$. 
In practice, both stage 1 and stage 2 centers may not adhere completely to all the intervention components (unplanned variation or uv-LAGO). 
We assume that to the extent it is not random, adherence depends on the recommendation $\boldsymbol{X}_{j}^{(2,n^{(1)})}$ and on the stage $2$ center characteristics $\boldsymbol{z}_{j}^{(2)}$, but not further on other 
predictors of the outcome.

\begin{remark}\label{remark}
One approach to determining recommended interventions for stage 2. 
\\
One way to determine the recommended interventions $\boldsymbol{X}_{j}^{(2,n^{(1)})}$ for stage 2, denoted by 
$\hat{\boldsymbol{x}}^{opt,(2,n^{(1)})}_{j}$, is by using a function $f$.
This function takes stage 1 outcomes, the pre-specified goal $\theta$, and stage 2 center-specific characteristics $\boldsymbol{z}_j^{(2)}$ as input, solving the optimization problem of equation (\ref{aim}) using the stage 1-based estimate $\hat{\boldsymbol{\beta}}^{(1)}$ in place of $\boldsymbol{\beta}$, 
and returns center-specific recommended interventions as output. Therefore, $\hat{\boldsymbol{x}}^{opt,(2,n^{(1)})}_{j} = f(\hat{\boldsymbol{\beta}}^{(1)} ; \boldsymbol{z}_j^{(2)})$ and $f$ solves equation (\ref{aim}). 

The optimization algorithm used to solve Equation (\ref{aim}) varies depended upon form of the cost function, $C(\boldsymbol{x})$. 
For linear cost functions, we first rank the intervention components based on their cost efficiency, $\hat{\beta}_{1p} / c_p$, where $c_p$ is the unit cost for the $pth$ intervention component.
Suppose that the intervention consists of two components, $p \in \{1,2\}$.
We increase the more cost-efficient component, while keeping the other component at its minimum, to check whether the objective $\theta$ as described in Equation (\ref{aim}) can be achieved.
If the objective cannot be met under this setting, we set the more cost-efficient component to its maximum value and increase the less cost-efficient component until the objective $\theta$ as described in Equation (\ref{aim}) is achieved. 
If the target value $\tilde{\theta}$ cannot be achieved, a function $g(\hat{\boldsymbol{\beta}})$ is employed to determine the recommended intervention. This function ensures that the recommended intervention is continuous with respect to each $\hat{\beta}_{1p}$. \citet{nevo2018analysis} described this first, and further details can be found in Section 5.1 of their Supplementary Material.

In the case of cubic cost functions, we construct a grid with small increments (e.g., 0.01) for each component. We then search for the combination of the two components that satisfies the target criterion $\theta$ as outlined in Equation (\ref{aim}), while minimizing the cubic cost function.
\end{remark}

For clarity of the expositions, below we will work with 
$\hat{\boldsymbol{x}}^{opt,(2,n^{(1)})}_{j}$.
Let $Y_{ij}^{(1)}$ be the outcome of participant $i$ in center $j$ of stage 1. Let $\boldsymbol{a}_j^{(1)} = h_j^{(1)} \left( \boldsymbol{x}_j^{(1)} \right) $ be the actual intervention package for center $j$ of stage 1, where $h_j^{(1)}$ is a continuous, deterministic function for each center $j$ in stage 1. 
Let $\boldsymbol{Y}_j^{(1)} = \bigl(Y_{1 j}^{(1)}, \ldots, Y_{n_{j}^{(1)} j}^{(1)} \bigl)$ be the outcomes of patients $1, \ldots, n_j^{(1)}$ in center $j$ of stage 1. Let
$\overline{\boldsymbol{a}}^{(1)}=\left(\boldsymbol{a}_{1}^{(1)}, \ldots, \boldsymbol{a}_{J^{(1)}}^{(1)}\right)$,  $\overline{\boldsymbol{z}}^{(1)}=\left(\boldsymbol{z}_{1}^{(1)}, \ldots, \boldsymbol{z}_{J^{(1)}}^{(1)}\right)$, and $\overline{\boldsymbol{Y}}^{(1)}=\left(\boldsymbol{Y}_{1}^{(1)}, \ldots, \boldsymbol{Y}_{J^{(1)}}^{(1)}\right)$ be the actual interventions, center-specific characteristics and outcomes for each center $1, \ldots, J^{(1)}$ of stage 1, respectively.
Additionally, let $\overline{\hat{\boldsymbol{x}}}^{o p t,\left(2, n^{(1)}\right)}=\Bigl(\hat{\boldsymbol{x}}_{1}^{o p t,\left(2, n^{(1)}\right)}, \ldots, \hat{\boldsymbol{x}}_{J^{(2)}}^{o p t,\left(2, n^{(1)}\right)}\Bigl)$ be the recommended interventions for center $1, \ldots, J^{(2)}$ of stage 2.

Let $Y_{ij}^{(2,n^{(1)})}$, $i=1,\cdots,n_j^{(2)}$ be the outcome of participant $i$ in center $j$ of stage 2.
Let
$\boldsymbol{Y}_j^{(2,n^{(1)})} = \Bigl(Y_{1 j}^{(2, n^{(1)})}, \ldots, Y_{n_{j}^{(2)} j}^{(2, n^{(1)})} \Bigl)$ be the outcomes of patients $1, \ldots, n_j^{(2)}$ in center $j$ of stage 2.
Let
$\boldsymbol{A}_j^{(2,n^{(1)})} = h_j^{(2)}\left( \hat{\boldsymbol{x}}^{opt,(2,n^{(1)})}_{j}  \right)$ be the actual intervention package for center $j$ of stage 2, where $h_j^{(2)}$ is a continuous, deterministic function for each center $j$ in stage 2. Let $\overline{\boldsymbol{A}}^{\left(2, n^{(1)}\right)}=\Bigl(\boldsymbol{A}_{1}^{\left(2, n^{(1)}\right)}, \ldots, \boldsymbol{A}_{J^{(2)}}^{\left(2, n^{(1)}\right)}\Bigl)$, $\overline{\boldsymbol{z}}^{(2)}=\left(\boldsymbol{z}_{1}^{(2)}, \ldots, \boldsymbol{z}_{J^{(2)}}^{(2)}\right)$, and $\overline{\boldsymbol{Y}}^{\left(2, n^{(1)}\right)}=\Bigl(\boldsymbol{Y}_{1}^{\left(2, n^{(1)}\right)}, \ldots, \boldsymbol{Y}_{J^{(2)}}^{\left(2, n^{(1)}\right)}\Bigl)$ be the actual interventions, center-specific characteristics and outcomes for each center $1, \dots, J^{(2)}$ of stage 2, respectively. 
\begin{assu} \label{main_assum_1}
Conditionally on $\overline{\hat{\boldsymbol{x}}}^{o p t,\left(2, n^{(1)}\right)}$, 
$\bigl(\overline{\boldsymbol{A}}^{\left(2, n^{(1)}\right)},\overline{\boldsymbol{Y}}^{\left(2, n^{(1)}\right)}\bigl)$ is independent of the stage 1 data $\bigl(\overline{\boldsymbol{a}}^{(1)},\overline{\boldsymbol{Y}}^{(1)} \bigl)$. 
That is, learning from data from earlier stages is only through the determination of the recommended intervention.
\end{assu}  

\begin{assu} For each center $j=1,\ldots,J^{(2)}$, the stage 2 recommended intervention  $\hat{\boldsymbol{x}}^{opt,(2,n^{(1)})}_{j}$ converges in probability to a center-specific limit $\boldsymbol{x}_{j}^{(2)}$.
\label{intervention_cvg_assumption}
\end{assu}

\begin{remark}
Assumption \ref{intervention_cvg_assumption} holds e.g. when the recommended intervention is a continuous function of either 
1) the maximum likelihood estimator (MLE) for $\boldsymbol{\beta}$; 
2) a solution to generalized estimating equation (GEE, \cite{liang1986longitudinal}); 
3) averages based on all stage~1 patients. Thus, Assumption \ref{intervention_cvg_assumption} holds if $\hat{\boldsymbol{x}}^{opt,(2,n^{(1)})}_{j}$ solves equation (\ref{aim}) and the solution to (\ref{aim}) is unique for the true $\boldsymbol{\beta}$. 
\end{remark}

Under Assumption \ref{intervention_cvg_assumption}, the definition of $h_j$ that maps recommended interventions to actual interventions in center $j$, and the Continuous Mapping Theorem, we conclude that $\boldsymbol{A}_j^{(2,n^{(1)})} = h_j^{(2)}\bigl( \hat{\boldsymbol{x}}^{opt,(2,n^{(1)})}_{j}  \bigl)$ converges in probability to $\boldsymbol{a}_j^{(2)} = h_j^{(2)}\bigl(\boldsymbol{x}^{(2)}_{j}  \bigl)$. 

\begin{assu}
The covariates $\boldsymbol{z}$, the outcomes $Y$, and the parameter space for $\boldsymbol{\beta}$ all take values in a compact space.
\label{glm_unif_bound_assu}
\end{assu}

\begin{assu}\label{ind_error}
Let $\epsilon_{ij}$ be the error term for individual $i$ in center $j$ of the GLM from Assumption \ref{glm_model_assumption}. 
That is, $Y_{ij} = g^{-1}\left( \beta_{0}+\boldsymbol{\beta}_{1}^{T} \boldsymbol{a}_{j}+\boldsymbol{\beta}_{2}^{T} \boldsymbol{z}_{j}\right) + \epsilon_{ij}$, with 
$E\left( \epsilon_{ij} | \boldsymbol{a}_j, \boldsymbol{z}_j \right) = 0$.
The distribution of the errors $\epsilon_{ij}$ is independent of the composition of the intervention package $\boldsymbol{a}_{j}$, although it may depend on $\boldsymbol{z}_{j}$. We denote $\sigma^2(\boldsymbol{z}_{j}) = VAR(\epsilon_{ij} | \boldsymbol{z}_j)$. Often, the $\epsilon_{ij}$ will be assumed to be iid, with distribution not dependent on the $\boldsymbol{z}_j$.
\end{assu}

\section{Estimating \texorpdfstring{$\boldsymbol{\beta}$}{Lg} and the optimal intervention package}\label{est beta}
This section describes the estimating equations for the proposed estimator $\hat{\boldsymbol{\beta}}$, and how to subsequently use $\hat{\boldsymbol{\beta}}$ to estimate the optimal intervention package.
The adaption of the intervention package based on prior outcomes causes dependence between stages, invalidating standard statistical theory. 
The usual (non-LAGO) estimator $\hat{\boldsymbol{\beta}}$ for the true parameter $\boldsymbol{\beta}^*$ is the
solution to the GEE under an independence working correlation structure \citep{liang1986longitudinal}
\begin{equation}\label{glm_EE}
\small{
\begin{aligned}
&0=\boldsymbol{U}^{(g)}(\boldsymbol{\beta})\\
&=\frac{1}{n} \left\{ \sum_{j=1}^{J^{(1)}} \sum_{i=1}^{n_{j}^{(1)}} \left(\frac{\partial}{\partial \boldsymbol{\beta}} g^{-1}(\boldsymbol{a}_j^{(1)},\boldsymbol{z}_j^{(1)}; \boldsymbol{\beta}) \right)\left(Y_{ij}^{(1)}-g^{-1}(\boldsymbol{a}_j^{(1)}, \boldsymbol{z}_j^{(1)};\boldsymbol{\beta})\right) \right.\\
&\qquad+\left.\sum_{j=1}^{J^{(2)}} \sum_{i=1}^{n_{j}^{(2)}} \left(\frac{\partial}{\partial \boldsymbol{\beta}} 
g^{-1}(\boldsymbol{A}_{j}^{\left(2, n^{(1)}\right)},\boldsymbol{z}_j^{(2)}; \boldsymbol{\beta})
\right)\left(Y_{ij}^{\left(2, n^{(1)}\right)}-g^{-1}(\boldsymbol{A}_{j}^{\left(2, n^{(1)}\right)},\boldsymbol{z}_j^{(2)}; \boldsymbol{\beta})\right) \right\}.
\end{aligned}
}
\end{equation}
We estimate $\boldsymbol{\beta}$ in a LAGO study by the solution to equation (\ref{glm_EE}). The superscript $(g)$ in $\boldsymbol{U}^{(g)}(\boldsymbol{\beta})$ reminds us that $\boldsymbol{U}^{(g)}(\boldsymbol{\beta})$ are estimating equations for LAGO GLM with a general link function. Asymptotic theory for $\hat{\boldsymbol{\beta}}$ is complicated by the fact that the stage 2 interventions $\overline{\boldsymbol{A}}^{\left(2, n^{(1)}\right)}$ depend on the stage 1 outcomes $\overline{\boldsymbol{Y}}^{(1)}$, so the two terms in (\ref{glm_EE}) are not independent. 

Section \ref{GLM} shows that despite these dependencies,
under the assumptions of Section \ref{setting}, the estimator $\hat{\boldsymbol{\beta}}$ solving equation (\ref{glm_EE}) is both consistent and asymptotically normal.

To identify the final recommended intervention package for any given center, we use the function $f$ from Remark \ref{remark}, 
along with the stage 1 and 2 outcome data, the pre-specified goal $\theta$, and the center-specific characteristics to solve the optimization problem given by (\ref{aim}), but with the final estimator $\hat{\boldsymbol{\beta}}$ in place of $\boldsymbol{\beta}$. 
The recommended intervention package for a center with any specified center characteristics $\boldsymbol{z}$ will be returned as output. 
That is, $\hat{\boldsymbol{x}}^{opt} = 
f(\hat{\boldsymbol{\beta}} ; \boldsymbol{z}  )$and $f$ solves equation (\ref{aim}).  

\section{Asymptotic Properties of \texorpdfstring{$\hat{\boldsymbol{\beta}}$}{Lg}}\label{GLM}
We describe the asymptotic properties of the final estimator $\hat{\boldsymbol{\beta}}$.  Similar to the previous sections, consider a two-stage LAGO design. Section \ref{GLM} outlines the proofs of consistency and asymptotic normality; Appendix \ref{general_link_appendix} provides the detailed proofs.

\begin{thm}
(Consistency). Under Assumptions \ref{glm_model_assumption} \textendash{}  \ref{ind_error}, \; $\hat{\boldsymbol{\beta}} \xrightarrow{P} \boldsymbol{\beta}^*$. 
\label{GLM_consistency_general}
\end{thm}
\noindent\cite{nevo2018analysis} provide a proof of consistency for $\hat{\boldsymbol{\beta}}$ under logistic regression.
Here, the general link function $g()$ does not need to be the canonical link function, and the outcomes $Y_{ij}$ are continuous. 

To prove consistency of $\hat{\boldsymbol{\beta}}$, we show that in spite of the fact that equation (\ref{glm_EE}) does not consist of i.i.d. terms, Theorem $5.9$ of \cite{van2000asymptotic} can be used. We show that the two conditions of this theorem are satisfied. For the first condition of Theorem $5.9$ of \cite{van2000asymptotic}, we show that
\begin{equation}\label{consistency_main_body}
\small{
    \sup_{\boldsymbol{\beta}}\|\boldsymbol{U}^{(g)}(\boldsymbol{\beta}) - \boldsymbol{u}^{(g)}(\boldsymbol{\beta})\| \xrightarrow{P} 0,
}
\end{equation}
where $\boldsymbol{u}^{(g)}(\boldsymbol{\beta})$ are the expected values of the estimating equations under the limiting design, in particular,
\begin{equation}
\small{
\begin{aligned}
\boldsymbol{u}^{(g)}(\boldsymbol{\beta}) 
&= \sum_{j=1}^{J^{(1)}} \alpha_{j 1} \left[ \left(\frac{\partial}{\partial \boldsymbol{\beta}} g^{-1}(\boldsymbol{a}_j^{(1)},\boldsymbol{z}_j^{(1)}; \boldsymbol{\beta}) \right) \left( g^{-1}(\boldsymbol{a}_j^{(1)},\boldsymbol{z}_j^{(1)}; \boldsymbol{\beta}^*) -g^{-1}(\boldsymbol{a}_j^{(1)},\boldsymbol{z}_j^{(1)}; \boldsymbol{\beta})  \right) \right]\\ 
&+  \sum_{j=1}^{J^{(2)}} \alpha_{j 2} \left[  \left(\frac{\partial}{\partial \boldsymbol{\beta}} g^{-1}(\boldsymbol{a}_j^{(2)},\boldsymbol{z}_j^{(2)}; \boldsymbol{\beta}) \right) \left( g^{-1}(\boldsymbol{a}_j^{(2)},\boldsymbol{z}_j^{(2)}; \boldsymbol{\beta}^*) -g^{-1}(\boldsymbol{a}_j^{(2)},\boldsymbol{z}_j^{(2)}; \boldsymbol{\beta}) \right) \right], 
\end{aligned}
}
\label{little_u_main_body}
\end{equation}
where (see Section \ref{setting}),
$\alpha_{j k}=\lim _{n \rightarrow \infty} n_{j}^{(k)} / n$ and $\boldsymbol{A}_{j}^{(2,n^{(1)})} \xrightarrow{P} \boldsymbol{a}_{j}^{(2)}$. 
(\ref{consistency_main_body}). 

To prove equation (\ref{consistency_main_body}), Appendix \ref{consistency_appendix_generallink} shows that $\boldsymbol{U}^{(g)}(\boldsymbol{\beta})-\boldsymbol{u}^{(g)}(\boldsymbol{\beta})$ can be decomposed into five distinct components. 
Let $\epsilon_{ij}^{(2)}$ be the error that patient $i$ in center $j$ would have experienced under the limiting intervention $\boldsymbol{a}_{j}^{(2)}$. 
By Assumption \ref{ind_error}, 
replacing the $\epsilon_{ij}^{(2,n^{(1)})}$ by the error terms $\epsilon_{ij}^{(2)}$  
does not change the distribution of the part in $\boldsymbol{U}^{(g)}(\boldsymbol{\beta})-\boldsymbol{u}^{(g)}(\boldsymbol{\beta})$ that includes $\epsilon_{ij}^{(2,n^{(1)})}$. Combining the error replacing approach with the concept of Donsker classes from empirical process theory, the supremum of $\boldsymbol{\beta}$ for each of the five components converges to 0 in probability. Then by triangle inequality, equation (\ref{consistency_main_body}) holds.

In order for LAGO to work properly, we need variations in the intervention components to identify the treatment effect parameters. The uniqueness of $\boldsymbol{\beta}^*$ as a maximizer or zero of $\boldsymbol{u}^{(g)}(\boldsymbol{\beta})$ (equation (\ref{little_u_main_body})) has been studied by various authors, see e.g. Chapter 2.2 of \cite{fahrmeir2013multivariate}. 
$\boldsymbol{u}^{(g)}(\boldsymbol{\beta})$ is the same as if the interventions were fixed before the study, so if there is enough variation in the intervention components,
the second condition in Theorem 5.9 of \cite{van2000asymptotic} is also satisfied and we conclude that $\hat{\boldsymbol{\beta}}$ is consistent. 

\begin{thm}\label{AN_thm}
(Asymptotic Normality). Under Assumptions \ref{glm_model_assumption}  \textendash{}  \ref{ind_error}, 
\begin{equation}\label{variance_main_body}
\small{
    \sqrt{n}\left(\hat{\boldsymbol{\beta}}-\boldsymbol{\beta}^{*}\right) \xrightarrow{D} N\left(0, J\left(\boldsymbol{\beta^{*}}\right)^{-1} V\left(\boldsymbol{\beta^{*}} \right) J\left(\boldsymbol{\beta^{*}}\right)^{-1}\right),
}
\end{equation}
where
\begin{equation*}
\small{
\begin{aligned}
J(\boldsymbol{\beta}^*) &= \sum_{j=1}^{J^{(1)}} \alpha_{j1} \left.\left( \frac{\partial}{\partial \boldsymbol{\beta}} \right|_{\boldsymbol{\beta}^*} g^{-1}(\boldsymbol{a}_j^{(1)},\boldsymbol{z}_j^{(1)}; \boldsymbol{\beta})\right)^{\otimes2}   + \sum_{j=1}^{J^{(2)}} \alpha_{j2} \left.\left( \frac{\partial}{\partial \boldsymbol{\beta}} \right|_{\boldsymbol{\beta}^*}  g^{-1}(\boldsymbol{a}_j^{(2)},\boldsymbol{z}_j^{(2)}; \boldsymbol{\beta})\right)^{\otimes2} ,
\end{aligned}
}
\end{equation*}
\begin{equation*}
\footnotesize{
\begin{aligned}
V\left(\boldsymbol{\beta}^{*}\right) = &\sum_{j=1}^{J^{(1)}} \alpha_{j 1}  \left(\left.\frac{\partial}{\partial \boldsymbol{\beta}}\right|_{\boldsymbol{\beta}^*} g^{-1}(\boldsymbol{a}_j^{(1)},\boldsymbol{z}_j^{(1)}; \boldsymbol{\beta}) \right)^{\otimes2} \sigma^2(\boldsymbol{z}_{j}^{(1)})\\
&\qquad\qquad\qquad\qquad\qquad\qquad\qquad\qquad+ \sum_{j=1}^{J^{(2)}} \alpha_{j 2} \left(\left.\frac{\partial}{\partial \boldsymbol{\beta}}\right|_{\boldsymbol{\beta}^*} g^{-1}(\boldsymbol{a}_j^{(2)},\boldsymbol{z}_j^{(2)}; \boldsymbol{\beta}) \right)^{\otimes2} \sigma^2(\boldsymbol{z}_{j}^{(2)}),
\end{aligned}
}
\end{equation*}
and ${}^{\otimes2}$ is the Kronecker product.
\end{thm}
The variance in equation (\ref{variance_main_body}) can be estimated by replacing $\boldsymbol{\beta}^*$, $\boldsymbol{a}_{j}^{(2)}$, $\alpha_{j1}$, and $\alpha_{j2}$ with $\hat{\boldsymbol{\beta}}$, $\boldsymbol{A}_{j}^{(2,n^{(1)})}$, ${n_j^{(1)}} / {n}$, and ${n_j^{(2)}} / {n}$, respectively. For estimating $\sigma^2(\boldsymbol{z}_j)$ and exact formulas of $\hat{J}(\hat{\boldsymbol{\beta}})$ and $\hat{V}(\hat{\boldsymbol{\beta}})$, see equations (\ref{estimator J}) and (\ref{estimator V}) in Appendix \ref{an_appendix_generallink}.

\cite{nevo2018analysis} provide a proof of asymptotic normality for $\hat{\boldsymbol{\beta}}$ under logistic regression; however, the novel coupling arguments they used for binary outcomes cannot be generalized to continuous outcomes. 
In LAGO trials, standard theory is not applicable as the stage 2 intervention package depends on stage 1 outcomes. Specifically, conditioning on the stage 2 intervention package means conditioning on functions of stage 1 outcomes. 

To address this issue, we assume that the errors in the GLM are independent of the intervention package composition (see Assumption \ref{ind_error}). 
The proof of Theorem \ref{AN_thm} employs a combination of strategies and concepts: building on the established Theorem \ref{GLM_consistency_general}, utilizing the Mean Value Theorem, employing the error-replacement strategy from the proof of Theorem \ref{GLM_consistency_general}, invoking Donsker classes from empirical process theory, and leveraging Lévy's Continuity Theorem for characteristic functions. 
Ultimately, Appendix \ref{an_appendix_generallink} demonstrates that $\sqrt{n} \bigl(\hat{\boldsymbol{\beta}}-\boldsymbol{\beta}^{*}\bigl)$ has the same asymptotic distribution as the estimator for a fixed two-stage design with interventions $\boldsymbol{a}_j^{(2)}$ decided on before the trial. 

By the Mean Value Theorem (similar to Theorem 5.21 of \cite{van2000asymptotic}),
\begin{equation*}
\small{
\sqrt{n}\left(\hat{\boldsymbol{\beta}}-\boldsymbol{\beta}^{*}\right) = -\sqrt{n}\left(\left.\frac{\partial}{\partial \boldsymbol{\beta}}\right|_{\tilde{\boldsymbol{\beta}}} \boldsymbol{U}^{(g)}\left({\boldsymbol{\beta}}\right)\right)^{-1} \boldsymbol{U}^{(g)}\left(\boldsymbol{\beta}^{*}\right),
}
\end{equation*}
where for each row of $\left.\frac{\partial}{\partial \boldsymbol{\beta}}\right|_{\Tilde{\boldsymbol{\beta}}} \boldsymbol{U}^{(g)}({\boldsymbol{\beta}})$, $\Tilde{\boldsymbol{\beta}}$ takes a possibly row-dependent value between $\hat{\boldsymbol{\beta}}$ and $\boldsymbol{\beta}^*$. By Theorem \ref{GLM_consistency_general}, the Continuous Mapping Theorem and further arguments, Appendix \ref{an_appendix_generallink} shows that $\bigl(-\left.\frac{\partial}{\partial \boldsymbol{\beta}}\right|_{\tilde{\boldsymbol{\beta}}} \boldsymbol{U}^{(g)}\left({\boldsymbol{\beta}}\right)\bigl)$ converges in probability to $J\left(\boldsymbol{\beta^{*}}\right)$. 
Next, we derive
\begin{equation}\label{regular_glm_AN_general_mainPaper}
\small{
\sqrt{n} \ \boldsymbol{U}^{(g)}(\boldsymbol{\beta}^*) = \boldsymbol{U}_{1,n}^{(g)} + \boldsymbol{U}_{2\_1,n}^{(g)} + \boldsymbol{U}_{2\_2,n}^{(g)},
}
\end{equation}
where 
\begin{equation*}
\small{
\begin{aligned}
\boldsymbol{U}_{1,n}^{(g)} &=\frac{1}{\sqrt{n}} \sum_{j=1}^{J^{(1)}} \sum_{i=1}^{n_{j}^{(1)}} \left(\left.\frac{\partial}{\partial \boldsymbol{\beta}}\right|_{\boldsymbol{\beta}^*} g^{-1}(\boldsymbol{a}_j^{(1)},\boldsymbol{z}_j^{(1)}; \boldsymbol{\beta}) \right) \epsilon_{ij}^{(1)}, \\ 
\boldsymbol{U}_{2\_2,n}^{(g)} &= \frac{1}{\sqrt{n}} \sum_{j=1}^{J^{(2)}} \sum_{i=1}^{n_{j}^{(2)}} \left(\left.\frac{\partial}{\partial \boldsymbol{\beta}}\right|_{\boldsymbol{\beta}^*} g^{-1}(\boldsymbol{a}_j^{(2)},\boldsymbol{z}_j^{(2)}; \boldsymbol{\beta}) \right) \epsilon_{ij}^{(2,n^{(1)})},
\end{aligned}
}
\end{equation*}
and the
explicit form of $\boldsymbol{U}_{2\_1,n}^{(g)}$ is given by equation (\ref{U_separations_equations}) in Appendix \ref{an_appendix_generallink}. 

Appendix \ref{an_appendix_generallink} shows that $\boldsymbol{U}_{2\_1,n}^{(g)} \xrightarrow{P} 0$. For $\boldsymbol{U}_{2\_2,n}^{(g)}$, similar to the proof of Theorem \ref{GLM_consistency_general} in Appendix \ref{consistency_appendix_generallink}, 
by Assumption \ref{ind_error}, replacing $\epsilon_{ij}^{(2,n^{(1)})}$ by the error terms $\epsilon_{ij}^{(2)}$ under the limiting interventions $\boldsymbol{a}_{j}^{(2)}$ does not change the distribution of $\boldsymbol{U}_{2\_2, n}^{(g)}$. 
We denote $ {\boldsymbol{U}_{2\_2,n}^{(g)\;*}}$ as the new term with error terms $\epsilon_{ij}^{(2)}$ and interventions $\boldsymbol{a}_{j}^{(2)}$. 
This implies that $\sqrt{n} \ \boldsymbol{U}^{(g)}(\boldsymbol{\beta}^*)$ has the same asymptotic distribution as $\boldsymbol{U}_{1,n}^{(g)} + {\boldsymbol{U}_{2\_2,n}^{(g)\; *}}$, which are the estimating equations for a fixed two-stage design with $\boldsymbol{a}_1^{(1)},\cdots, \boldsymbol{a}_{J^{(1)}}^{(1)}, \boldsymbol{a}_1^{(2)},\cdots, \boldsymbol{a}_{J^{(2)}}^{(2)}$ as interventions decided on before the trial. By the definition of $\boldsymbol{a}_{j}^{(2)}$, the two terms $\boldsymbol{U}_{1,n}^{(g)}$ and $ {\boldsymbol{U}_{2\_2,n}^{(g)\; *}}$ are independent. Thus, regular GEE theory applies here, and Theorem \ref{AN_thm} (Asymptotic Normality) follows. A detailed proof of Theorem \ref{AN_thm} can be found in Appendix \ref{an_appendix_generallink}. 

Appendix \ref{appendix_loglink} presents the proofs for Theorem \ref{GLM_consistency_general} and Theorem \ref{AN_thm} using a log link function as a concrete example. 

\section{Hypothesis testing}
\label{testing section}
In a LAGO study, one of the main objectives is to evaluate the null hypothesis of no overall intervention effect. This null hypothesis, represented by $H_0: \boldsymbol{\beta}_1=0$, can be tested using standard methods, such as the Chi-squared test. The validity of the Chi-squared test for this purpose is guaranteed by Theorem \ref{AN_thm}.

In a c-LAGO study, a group indicator $R$ can be used to identify the intervention group ($R=1$) and the control group ($R=0$). Let $\mu_0$ and $\mu_1$ represent the mean outcome values in the control and intervention groups, respectively. An alternative test for $H_0: \boldsymbol{\beta}_1=0$ is to test the implied
$\tilde{H}_0: \mu_0 = \mu_1$.
Under the null hypothesis, the stage 1 and stage 2 outcomes are independent because the stage 2 intervention has no effect on the stage 2 outcomes. Under the null hypothesis, the distribution of the outcomes is the same regardless of the intervention, and we can use any standard 1-degree-of-freedom test (such as the Z-test) to compare the distribution in the intervention and the control group. This leads to a test with the usual level of significance (either exact or asymptotic, depending on the test used).

One can also take center characteristics $\boldsymbol{z}$ into account when testing $H_0: \boldsymbol{\beta}_1=0$ in a c-LAGO study.
The test involves considering the null hypothesis that $\gamma=0$ in the model $g\left(E\left(Y | \gamma, \boldsymbol{z}; \boldsymbol{\beta} \right)\right) = \beta_{0}+\boldsymbol{\beta}_{2}^{T} \boldsymbol{z}+\gamma R$. Under the null hypothesis, as before, there is independence between the stages, and $\boldsymbol{\beta}_1=0$ implies that $\gamma=0$.

\section{Confidence sets and confidence bands}\label{confidence}
The confidence set for the optimal intervention package is a list of intervention package compositions that can be expected to include the optimal intervention in 95$\%$ of such studies. To construct this confidence set, we create a confidence interval for $\mu = E\left(Y | \boldsymbol{x}, \tilde{\boldsymbol{z}} ; \boldsymbol{\beta}^*\right)$ for a given value $\tilde{\boldsymbol{z}}$ and for each possible value of $\boldsymbol{x}$. 
We first calculate a 95$\%$ confidence interval for $g\left(\mu\right)$ as
$
CI_{g(\mu)} = \left(1 \; \boldsymbol{x}^{T} \; \tilde{\boldsymbol{z}}^{T}\right) \hat{\boldsymbol{\beta}} \pm 1.96 \; \sqrt{\sigma^2_{\mu}(\hat{\boldsymbol{\beta}} ; \boldsymbol{x}, \tilde{\boldsymbol{z}})} 
$ where 
$\sigma^2_{\mu}(\hat{\boldsymbol{\beta}} ; \boldsymbol{x}, \tilde{\boldsymbol{z}})$ can be calculated based on Theorem \ref{AN_thm}: $\sigma^2_{\mu}(\hat{\boldsymbol{\beta}} ; \boldsymbol{x}, \tilde{\boldsymbol{z}}) = \left(1 \; \boldsymbol{x}^{T} \; \tilde{\boldsymbol{z}}^{T}\right) n^{-1} \hat{J}\bigl(\hat{\boldsymbol{\beta}}\bigl)^{-1} \hat{V}\bigl(\hat{\boldsymbol{\beta}}\bigl) \hat{J}\bigl(\hat{\boldsymbol{\beta}}\bigl)^{-1} \left(1 \; \boldsymbol{x}^{T} \; \tilde{\boldsymbol{z}}^{T}\right)^T$. It follows that the 95$\%$ confidence interval for $E\left(Y | \boldsymbol{x}, \tilde{\boldsymbol{z}} ; \boldsymbol{\beta}^*\right)$ is $CI_{\mu} = g^{-1}\left(CI_{g(\mu)}\right)$. The confidence set can then be calculated as $CS\left( \boldsymbol{x}^{opt} \right) = \{ \boldsymbol{x}: CI_{\mu} \ni \theta \}$. 
That is, $CS\left( \boldsymbol{x}^{opt} \right)$ includes intervention packages $\boldsymbol{x}$ for which 
$\theta$ (equation (\ref{aim})) is inside the confidence interval for the mean outcome under intervention package $\boldsymbol{x}$. 
Because of Theorem \ref{AN_thm}, 
the confidence set $CS\left( \boldsymbol{x}^{opt} \right)$ contains $\boldsymbol{x}^{opt}$ with asymptotic probability at least 0.95. 

Next, we construct confidence bands for the mean outcome under all different intervention package compositions. These confidence bands have asymptotic 95$\%$ coverage simultaneously for all the intervention package compositions. To calculate the confidence bands for the outcome $E\left(Y | \boldsymbol{x}, \tilde{\boldsymbol{z}} ; \boldsymbol{\beta}^*\right)$, we first compute the 95$\%$ confidence bands for $\left(1 \; \boldsymbol{x}^{T} \; \tilde{\boldsymbol{z}}^{T}\right)\boldsymbol{\beta}^*$, similar to \cite{scheffe1999analysis} and \cite{nevo2018analysis} (Section 4 of their Supplementary Material),
\begin{equation}\label{CB}
C B_{g(\mu)} = \left(1 \; \boldsymbol{x}^{T} \; \tilde{\boldsymbol{z}}^{T}\right)\hat{\boldsymbol{\beta}} \pm \sqrt{\chi_{0.95, p+q+1}^{2} \sigma^2_{\mu}(\hat{\boldsymbol{\beta}} ; \boldsymbol{x}, \tilde{\boldsymbol{z}})},
\end{equation}
where $\chi_{0.95, p+q+1}^{2}$ is the 95th percentile of the $\chi_{ p+q+1}^{2}$ distribution. $p+q+1$ is the dimension of $\boldsymbol{\beta}$, where $p$ and $q$ are the dimensions of $\boldsymbol{a}$ and $\boldsymbol{z}$, respectively.
{\footnotesize $CB_{\mu} = g^{-1}\bigl(CI_{g(\mu)}\bigl)$}, and such confidence bands guarantee asymptotic simultaneous 95$\%$ coverage for the mean outcome under all intervention package compositions.

\section{Simulations}\label{simulation section}
Our methods were evaluated through simulation studies, which were organized into four parts each based on 2000 simulated datasets.
Simulation 1 considered two scenarios of two-stage c-LAGO designs.
Simulation 2 closely mimicked the BetterBirth study as if a c-LAGO design was used. All non-adaptive design parameters were taken directly from the BetterBirth study.
Simulation 3 compared the performance of c-LAGO, uv-LAGO, and factorial designs. 
Simulation 4 evaluated the performance of three designs: c-LAGO, factorial design, and MOST, employing the identity link function in the model for the continuous outcome. The primary focus of Simulation 4 was on the power of each study design.

In simulation 1 scenario 1, the data was generated with the same number of centers $J$ in both the intervention and the control arm. The number of centers in the intervention and the control arm was $J=6, 10, 20$, with $n_j^{(1)} = 50, 100$, and $n_j^{(2)} = 100, 200$. The intervention consisted of two components, $\boldsymbol{x} = (x_1, x_2)$, with the minimum and maximum values of $x_1$ and $x_2$ being $\left[L_{1}, U_{1}\right]=[0,2]$, and $\left[L_{2}, U_{2}\right]=[0,8]$, respectively.
The recommended interventions for stage 1 were set to be the middle of the range for the two intervention components, $\boldsymbol{x}^{(1)} = (1, 4)$.
The model for the continuous outcome was
$g\left(E\left(Y_{ij} | \boldsymbol{a}_{j}, \boldsymbol{z}_{j} \; ; \boldsymbol{\beta} \right )\right)=\boldsymbol{\beta}_{1}^{T} \boldsymbol{a}_{j}+\boldsymbol{\beta}_{2}^{T} \boldsymbol{z}_{j},$ where $g()$ is the logit link function.
The logit link function was chosen as this link function was used in the analysis of the BetterBirth study of Section \ref{application} to restrict the expected fractions to values between 0 and 1.
The true coefficient values for the intervention components $\boldsymbol{\beta}_1^*$ were set as $(\beta_{11}^*,\beta_{12}^*)$ based on the estimated parameter values and their confidence intervals in the final analysis of the BetterBirth study (Table \ref{estimated effects bb}):
$(0.1863,0.15)$, 
$(0.0438,0.17)$,
$(0.1,0.2133)$,
and $(0.1062,0.16)$.
The exact true coefficient values were chosen to facilitate the confidence set and confidence bands calculations. More specifically, $(\beta_{11}^*=0.1863,\beta_{12}^*=0.15)$ was chosen so that the true optimal intervention was $(1,8)$ within three decimal places.
Other $\beta_{11}^*$ and $\beta_{12}^*$ values were chosen for the same reason. A baseline center characteristic $Z \sim N(0,1)$ was also included with a true coefficient value of $\beta_z^* = -0.2$.
The aim was for a mean outcome $\theta=0.8$, and the optimization problem was solved as described in equation (\ref{aim}) to obtain the recommended interventions. For simplicity, no intercept was included in the models. 

In simulation 1 scenario 2, instead of having the same number of centers in both stages, fewer centers and a smaller per-center sample size were included in stage 1 than in stage 2. Specifically, the number of centers in stage 1 and stage 2 were $J_1$=6, $J_2$=12 and $J_1$=10, $J_2$=20, respectively. The per-center sample sizes were $n_{j}^{(1)}$ = 50 and $n_{j}^{(2)}$ = 200. Other design parameters were the same as in simulation 1 scenario 1.
We incorporate variations in stage 1 interventions which accounts for the possibility that centers may not adhere strictly to the interventions in stage 1. We do not incorporate variations from the recommended interventions in stage 2, as those are based on previous stage outcomes and center specific.

Table \ref{component_effects_table}
-- \ref{cp table} present selected results for simulation 1 scenario 1 and 2 using a linear cost function with per unit cost for the two intervention components: $C = (c_1=8, c_2=2)$. 
Complete results for simulation 1 with a linear cost function can be found in Appendix \ref{table1ctn} - \ref{table3ctn}.
Table \ref{component_effects_table} shows that for $J>6$, both  $\hat{{\beta}}_{11}$ and $\hat{{\beta}}_{12}$ had minimal relative bias. The relative bias of $\hat{{\beta}}_{12}$ was smaller than that of $\hat{{\beta}}_{11}$. The ratios between the mean of the estimated standard error and the empirical standard error ranged from 0.7 -- 0.9, however the empirical coverage rates of the 95$\%$ confidence intervals for both ${{\beta}}_{11}$ and ${{\beta}}_{12}$ were close to 95$\%$. 
Despite having thoroughly verified the simulation code's accuracy, we cannot explain why coverage remained satisfactory when the ratio between the mean of the estimated standard error and the empirical standard error was small. However, our analysis confirms that this phenomenon persisted.

Table \ref{bias_table} reports the bias and root mean squared errors for the estimated optimal intervention components for a center with baseline center characteristic ${z}$ equal to 0. 
The true optimal intervention components were calculated using the method described in Remark \ref{remark} for linear cost functions, with the true $\boldsymbol{\beta}$ values.
The bias and root mean squared error were small for both estimated optimal intervention components. The estimated optimal intervention components based on stage 1 data (shown under ``Stage 1") 
had higher bias and higher root mean squared error compared to the final estimated optimal intervention based on all data (shown under ``Stage 2/LAGO optimized").

Table \ref{cp table} displays information about the finite sample behavior of the estimated optimal interventions, confidence sets, and confidence bands based on stage 1 data and based on all data.
The confidence set and the confidence bands were computed by discretizing the two intervention components within their lower and upper limits with the incrementation size of $0.1$.
The steps outlined in Section \ref{confidence} were then applied.
The true mean under the final estimated intervention was close to 0.8 in most simulated datasets (see Table \ref{cp table} column MenOpt2).
The coverage rate for both the confidence set for the optimal intervention package and the simultaneous confidence bands for the intervention package components were very close to 95$\%$. The mean percentage of the size of the confidence set as a percentage of the total sample space ranged between 3$\%$ to 15$\%$ across different simulated datasets.

Table \ref{component_effects_table cubic cost} -- \ref{meanopt table cubic cost} present selected results for simulation 1 scenario 1 and 2 using the cubic cost function. The cubic cost function for the two intervention components $\boldsymbol{x} = (x_1, x_2)$ was defined as
$C(\boldsymbol{x}) = 0.05x_1^3-1.19x_1^2+10x_1+10 + 0.1x_2^3-0.7x_2^2+2x_2$. 
The cubic cost function adapted the linear cost function from simulation 1  to include an economy of scale at lower values of intervention component costs and by including prohibitive cost as the intervention components neared their upper limits.
Complete results can be found in Appendix \ref{additional cost}. 
Table \ref{component_effects_table cubic cost} indicated minimal relative bias for both $\hat{{\beta}}_{11}$ and $\hat{{\beta}}_{12}$. 
Similar to Table \ref{component_effects_table}, when the ratio between the mean of the estimated standard error and the empirical standard error was small, the empirical coverage rates remained satisfactory. 
Table \ref{bias_table cubic cost} displays relatively small bias and root mean squared error for the estimated optimal intervention components based on all data (shown under ``Stage 2/LAGO optimized"). 
The true optimal intervention components were calculated using the method described in Remark \ref{remark} for cubic cost functions, with the true $\boldsymbol{\beta}$ values.
The bias and root mean squared error were larger than those reported in Table \ref{bias_table}.
Table \ref{meanopt table cubic cost} demonstrates satisfactory finite sample behavior in terms of the estimated optimal intervention, confidence sets, and confidence bands. The reported values aligned closely with those in Table \ref{cp table}.

For simulation 2, we set all non-adaptive parameter values to those obtained from the BetterBirth study, including the stage 1 interventions, the number of centers, the per-center sample size, the intervention arm allocations, and the distribution of the center-specific covariate, birth volume, $z$. 
For the stage 2 and stage 3 interventions, we simulated a c-LAGO design. The results of simulation 2 can be found in Appendix \ref{mimicbbstudy}.

In simulation 3, we compared the performance of the c-LAGO, uv-LAGO and factorial designs. We set the number of centers $J=8$, and the per-center sample sizes $n_j^{(1)}=n_j^{(2)}=10$ in each of the $K=2$ stages. The intervention package consisted of two components, $\boldsymbol{x} = (x_1, x_2)$ with minimum and maximum values of $x_1$ and $x_2$ set to be the same as in simulation 1. 
The stage 1 outcomes were simulated using a factorial design with interventions $(0,0)$, $(1,0)$, $(0,4)$, and $(1,4)$, each with probability $1/4$. 
The stage 2 outcomes were simulated using either a c-LAGO or a uv-LAGO design. The factorial design had only one stage with the number of centers set to $J=16$.
The model for the continuous outcome was 
$E\left(Y_{ij} | \boldsymbol{a}_j,\boldsymbol{z}_j; \boldsymbol{\beta} \right) = \beta_{0}+\boldsymbol{\beta}_{1}^{T} \boldsymbol{a}_{j} + \boldsymbol{\beta}_{2}^{T} \boldsymbol{z}_{j},$
and the errors were normally distributed with a standard deviation of 1.5.
The true parameter values were set to $\beta_{0}^*=0.1$, $\beta_{11}^*=0.2$, $\beta_{12}^*=0.3$, and $\beta_{z}^*=-0.2$. The cost function for the two intervention components was linear, with $C = (c_1=8, c_2=2)$ per unit. The true optimal intervention for a center with $z=0$ was $x^{opt}=(0,3)$.

Under the null hypothesis of the c-LAGO design, with $(\beta_{11}^*,\beta_{12}^*) = (0, 0)$, the type-1 error of the component-wise test (P degree of freedom Chi-squared test) was 4.7\%, and 3.1\% for the two-sample means test. Under the null hypothesis of the uv-LAGO design,  
the type-1 error of the component-wise test was 4.8\%, and 3.6\% for the two-sample means test. 
Table \ref{Tab:factorial design} suggests that our approach was effective and led to minimal finite sample bias, correct nominal coverage, and correct type-1 error. 

In simulation 4, we evaluated the performance of the c-LAGO, factorial and MOST designs with a focus on the power of the study design. The model for the continuous outcome was a linear regression model with an identity link and normal errors, with 3 stages and 100 participants per stage. 
Our goal was to increase the expected outcome from 0.1 to 0.8 while minimizing cost, using a linear cost function with $c_1=1$ and $c_2=5$. 
The outcome model was $g(E(Y_{ij}|\boldsymbol{a}_{i};\boldsymbol{\beta}))=\beta_0 + \beta_{11}a_{i1} + \beta_{12}a_{i2}$ with coefficients of $(\beta_0,\beta_{11},\beta_{12})=(0.1,0.05,0.12)$ and the errors were normally distributed with a standard deviation of 1.75. 
For all designs, the stage 1 outcomes were generated according to a factorial design with interventions (0,0), (2,0), (0,5), (2,5) each with a probability of $1/4$. 
The factorial design continued with this intervention package composition throughout. 
For the c-LAGO design, the stage 2 and stage 3 recommended interventions were determined with the goal of reaching an expected mean outcome of 0.8 while minimizing cost and also achieving an estimated power of 0.9. This power constraint was the main focus of our paper under development and did not violate any of the assumptions outlined in Section \ref{setting}. 
In order to simplify and facilitate the comparison between c-LAGO, factorial and MOST designs, the recommended interventions in c-LAGO were not tailored to individual centers.
The MOST design followed the same factorial design for stage 1, then calculated the recommended intervention as in LAGO. In the randomized controlled trial stage, half of the participants of MOST were assigned to the control group and the other half to the treatment group. 
The final analysis only included the outcomes after stage 1.
The results in Table \ref{MOST table} indicate that c-LAGO had minimal finite sample bias, correct nominal coverage, and notably higher power compared to both the factorial design and MOST.

\begin{table}
\caption{Selected simulation study results for simulation 1 scenario 1 and 2 with a linear cost function}
\addtocounter{table}{-1}
    \hrule
    \subfloat[\scriptsize{ Simulation study results for individual package component effects with a linear cost function}\label{component_effects_table}]{
    \scriptsize{
    \centering
    \begin{tabular}{llllllllll}
    \tiny$\boldsymbol{\beta}^*=(\beta_{11}^*, \beta_{12}^*)$ & \tiny$n_j^{(1)}$ & \tiny$n_j^{(2)}$ & \tiny$J$ &  & \tiny$\hat{\boldsymbol{\beta}}_{11}$ &  &  & \tiny$\hat{\boldsymbol{\beta}}_{12}$ &  \\
     &  &  &  & $\%$RelBias & \begin{tabular}[c]{@{}l@{}}$\frac{SE}{EMP.SD}$ \\ ($\times 100$)\end{tabular} & CP95 & $\%$RelBias & \begin{tabular}[c]{@{}l@{}}$\frac{SE}{EMP.SD}$\\ ($\times 100$)\end{tabular} & CP95 \\
    \multicolumn{4}{l}{Scenario 1 ($J_1 = J_2 = J$)} &  &  &  &  &  &  \\
    (0.1863, 0.15) & 50 & 100 & 6 & 1.64 & 84.8 & 96.1 & -0.46 & 76.1 & 95.0 \\
     &  &  & 10 & 1.78 & 90.6 & 94.9 & -0.55 & 79.3 & 94.9 \\
     &  &  & 20 & 1.54 & 97.6 & 95.1 & -0.37 & 91.2 & 95.2 \\
     &  & 200 & 6 & 1.44 & 77.9 & 95.2 & -0.42 & 70.9 & 95.5 \\
     &  &  & 10 & 1.91 & 86.1 & 95.4 & -0.49 & 73.6 & 95.5 \\
     &  &  & 20 & 1.14 & 97.8 & 95.4 & -0.26 & 88.2 & 94.8 \\
    \multicolumn{4}{l}{Scenario 2a ($J_1 = 6, J_2 = 12$)} &  &  &  &  &  &  \\
    (0.1863, 0.15) & 50 & 200 &  & 0.04 & 85.9 & 95.0 & -0.09 & 75.0 & 94.8 \\
    (0.1, 0.2133) & 50 & 200 &  & 0.56 & 93.7 & 95.8 & 0.14 & 84.3 & 95.7
    \end{tabular}
    }
    }
    \hrule
    \subfloat[\scriptsize{Simulation study results for estimated optimal intervention with a linear cost function}\label{bias_table}]{
    \scriptsize{
    \centering
    \begin{tabular}{llllllllllll}
    \tiny$\boldsymbol{\beta}^*=(\tiny\beta_{11}^*,\tiny\beta_{12}^*)$ & \tiny$\boldsymbol{x}^{opt}$ & \tiny$n_j^{(1)}$ & \tiny$n_j^{(2)}$ &  & \multicolumn{3}{c}{Stage 1} &  & \multicolumn{3}{c}{Stage 2/LAGO optimized} \\
     &  &  &  &  & \begin{tabular}[c]{@{}l@{}}Bias of\\ $\hat{x}_{1}^{o p t}$\\ ($\times$100)\end{tabular} & \begin{tabular}[c]{@{}l@{}}Bias of\\ $\hat{x}_{2}^{o p t}$\\ ($\times$100)\end{tabular} & \begin{tabular}[c]{@{}l@{}}rMSE\\ ($\times$100)\end{tabular} &  & \begin{tabular}[c]{@{}l@{}}Bias of\\ $\hat{x}_{1}^{o p t}$\\ ($\times$100)\end{tabular} & \begin{tabular}[c]{@{}l@{}}Bias of\\ $\hat{x}_{2}^{o p t}$\\ ($\times$100)\end{tabular} & \begin{tabular}[c]{@{}l@{}}rMSE\\ ($\times$100)\end{tabular} \\
    \multicolumn{4}{l}{Scenario 1 ($J_1 = J_2 = 20$)} &  &  &  &  &  &  &  &  \\
    (0.1863, 0.15) & (1,8) & 50 & 100 &  & 15.2 & 96.9 & 115.0 &  & -0.27 & 0.01 & 35.3 \\
     &  &  & 500 &  &  &  &  &  & -0.70 & 0.00 & 28.2 \\
     &  & 100 & 100 &  & 14.0 & 51.6 & 94.1 &  & 0.11 & 0.01 & 32.7 \\
     &  &  & 500 &  &  &  &  &  & -0.25 & 0.00 & 24.5 \\
    \multicolumn{4}{l}{Scenario 2a ($J_1 = 6, J_2 = 12$)} &  &  &  &  &  &  &  &  \\
    (0.1863, 0.15) & (1,8) & 50 & 200 &  & 20.7 & 341.5 & 191.1 &  & -1.5 & 1.5 & 52.3 \\
    (0.1, 0.2133) & (0, 6.5) & 50 & 200 &  & -46.2 & 239.2 & 172.9 &  & -0.1 & 0.1 & 39.2
    \end{tabular}
    }
    }
    \hrule
    \subfloat[\scriptsize{Simulation study results for estimated optimal intervention, confidence set and confidence band with a linear cost function}\label{cp table}]{
    \scriptsize{
    \centering
    \begin{tabular}{lllllllll}
    \tiny$\boldsymbol{\beta}^*=(\tiny\beta_{11}^*,\tiny\beta_{12}^*)$ & \tiny$\boldsymbol{x}^{opt}$ & \tiny$n_j^{(1)}$ & \tiny$n_j^{(2)}$ & \begin{tabular}[c]{@{}l@{}}MeanOpt1\\ (Q2.5,Q97.5)\end{tabular} & \begin{tabular}[c]{@{}l@{}}MeanOpt2\\ (Q2.5,Q97.5)\end{tabular} & \begin{tabular}[c]{@{}l@{}}SetCP95\\ $\%$\end{tabular} & \begin{tabular}[c]{@{}l@{}}SetPerc\\ $\%$\end{tabular} & \begin{tabular}[c]{@{}l@{}}BandsCP95\\ $\%$\end{tabular} \\
    \multicolumn{4}{l}{Scenario 1 ($J_1 = J_2 = 20$)} &  &  &  &  &  \\
    (0.1863, 0.15) & (1,8) & 50 & 100 & (0.655, 0.818) & (0.789, 0.811) & 95.3 & 5.5 & 95.2 \\
     &  &  & 500 &  & (0.794, 0.808) & 95.1 & 3.7 & 95.7 \\
     &  & 100 & 100 & (0.708, 0.816) & (0.791, 0.808) & 94.4 & 4.4 & 95.1 \\
     &  &  & 500 &  & (0.795, 0.805) & 95.2 & 2.8 & 95.8 \\
    \multicolumn{4}{l}{Scenario 2a ($J_1 = 6, J_2 = 12$)} &  &  &  & \multicolumn{1}{r}{} &  \\
    (0.1863, 0.15) & (1,8) & 50 & 200 & (0.529, 0.826) & (0.773, 0.826) & 94.5 & 9.9 & 95.4 \\
    (0.1, 0.2133) & (0, 6.5) & 50 & 200 & (0.515, 0.857) & (0.783, 0.818) & 95.7 & 12.6 & 95.7
    \end{tabular}
    }
    }
\scriptsize{
    \raggedright{
    $n_j^{(1)}$: number of patients in center $j$ at stage 1, 
    $n_j^{(2)}$: number of patients in center $j$ at stage 2.\\
    $J$: number of centers for each stage. \\
    \%RelBias: percent relative bias $100(\hat{\beta}-\beta^\star)/\beta^\star$.\\
    SE: mean estimated standard error, 
    EMP.SD: empirical standard deviation.\\
    CP95: empirical coverage rate of 95\% confidence intervals.\\
    Bias of $\hat{x}_{1}^{o p t}$: bias of the first component of the estimated optimal intervention,  
    Bias of $\hat{x}_{2}^{o p t}$: bias of the second component of the estimated optimal intervention.\\ 
    rMSE: root of mean squared errors,$\Bigl\{\operatorname{mean}\bigl(\bigl\|\hat{\boldsymbol{x}}^{o p t}-\boldsymbol{x}^{o p t}\bigl\|^{2}\bigl)\Bigl\}^{1 / 2}$, mean is taken over simulation iterations.\\
    MeanOpt1: mean outcome under the stage 2 recommended intervention, calculated using true coefficient values;
    MeanOpt2: mean outcome under the final estimated optimal intervention based on all data, calculated using true coefficient values. \\
    $Q$2.5 and $Q$97.5: 2.5$\%$ and 97.5$\%$ quantiles.\\ 
    SetCP95$\%$: empirical coverage percentage of confidence set for the optimal intervention. 
    SetPerc$\%$: mean percentage of the size of the confidence set as a percent of the total sample space.
    BandsCP95$\%$: empirical coverage of 95$\%$ confidence band. 
    }
}
\stepcounter{table}
\end{table}
\vspace{-0.5cm}

\begin{table}
\caption{Selected simulation study results for simulation 1 scenario 1 and 2 with a cubic cost function}
\addtocounter{table}{-1}
    \hrule
    \subfloat[\scriptsize Simulation study results for individual package component effects with a cubic cost function]{
    \scriptsize{
    \centering
    \begin{tabular}{llllllllll}
    \tiny$\boldsymbol{\beta}^*=(\beta_{11}^*, \beta_{12}^*)$ & \tiny$n_j^{(1)}$ & \tiny$n_j^{(2)}$ & \tiny$J$ &  & \tiny$\hat{\boldsymbol{\beta}}_{11}$ &  &  & \tiny$\hat{\boldsymbol{\beta}}_{12}$ &  \\
     &  &  &  & $\%$RelBias & \begin{tabular}[c]{@{}l@{}}$\frac{SE}{EMP.SD}$\\ ($\times 100$)\end{tabular} & CP95 & $\%$RelBias & \begin{tabular}[c]{@{}l@{}}$\frac{SE}{EMP.SD}$\\ ($\times 100$)\end{tabular} & CP95 \\
    \multicolumn{4}{l}{Scenario 1 ($J_1 = J_2 = J$)} &  &  &  &  &  &  \\
    (0.1863, 0.15) & 50 & 100 & 6 & 2.31 & 81.9 & 95.8 & -0.79 & 78.1 & 95.5 \\
     &  &  & 10 & 4.59 & 95.8 & 95.4 & -1.63 & 88.3 & 95.0 \\
     &  &  & 20 & 4.35 & 101.3 & 95.7 & -1.49 & 94.1 & 96.0 \\
     &  & 200 & 6 & 1.36 & 83.9 & 95.5 & -0.52 & 78.8 & 95.5 \\
     &  &  & 10 & 2.68 & 94.2 & 95.9 & -0.89 & 85.6 & 96.0 \\
     &  &  & 20 & 3.90 & 96.3 & 96.4 & -1.29 & 88.1 & 96.1 \\
    \multicolumn{4}{l}{Scenario 2a ($J_1 = 6, J_2 = 12$)} &  &  &  &  &  &  \\
    (0.1863, 0.15) & 50 & 200 &  & 0.21 & 83.1 & 95.2 & -0.22 & 75.3 & 94.9 \\
    (0.1, 0.2133) & 50 & 200 &  & 2.54 & 88.1 & 95.6 & -0.31 & 74.5 & 95.8
    \end{tabular}
    \label{component_effects_table cubic cost}
    }
    }
    \hrule
    \subfloat[\scriptsize Simulation study results for estimated optimal intervention with a cubic cost function]{
    \centering
    \scriptsize{
    \centering
    \begin{tabular}{llllllllllll}
    \tiny$\boldsymbol{\beta}^*=(\beta_{11}^*,\beta_{12}^*)$ & \tiny$\boldsymbol{x}^{opt}$ & \tiny $n_j^{(1)}$ & \tiny$n_j^{(2)}$ &  & \multicolumn{3}{c}{Stage 1} &  & \multicolumn{3}{c}{Stage 2/LAGO optimized} \\
     &  &  &  &  & \begin{tabular}[c]{@{}l@{}}Bias of\\ $\hat{x}_{1}^{o p t}$\\ ($\times$100)\end{tabular} & \begin{tabular}[c]{@{}l@{}}Bias of\\ $\hat{x}_{2}^{o p t}$\\ ($\times$100)\end{tabular} & \begin{tabular}[c]{@{}l@{}}rMSE\\ ($\times$100)\end{tabular} &  & \begin{tabular}[c]{@{}l@{}}Bias of\\ $\hat{x}_{1}^{o p t}$\\ ($\times$100)\end{tabular} & \begin{tabular}[c]{@{}l@{}}Bias of\\ $\hat{x}_{2}^{o p t}$\\ ($\times$100)\end{tabular} & \begin{tabular}[c]{@{}l@{}}rMSE\\ ($\times$100)\end{tabular} \\
    \multicolumn{4}{l}{Scenario 1 ($J_1 = J_2 = 20$)} &  &  &  &  &  &  &  &  \\
    (0.1863, 0.15) & (1.5,7.4) & 50 & 100 &  & 40.0 & 81.9 & 118.8 &  & 4.9 & 4.8 & 65.9 \\
     &  &  & 500 &  &  &  &  &  & 4.4 & -0.8 & 56.2 \\
     &  & 100 & 100 &  & 32.7 & 45.8 & 103.4 &  & 2.6 & 5.9 & 63.9 \\
     &  &  & 500 &  &  &  &  &  & 3.8 & 0.0 & 56.2 \\
    \multicolumn{4}{l}{Scenario 2a ($J_1 = 6, J_2 = 12$)} &  &  &  &  &  &  &  &  \\
    (0.1863, 0.15) & (1.5,7.4) & 50 & 200 &  & 64.4 & 314.0 & 187.6 &  & 9.7 & 19.9 & 79.8 \\
    (0.1, 0.2133) & (0, 6.5) & 50 & 200 &  & -64.2 & 261.3 & 170.0 &  & -13.8 & 7.3 & 54.1
    \end{tabular}
    \label{bias_table cubic cost}
    }
    }
    \hrule
    \subfloat[\scriptsize Simulation study results for estimated optimal intervention with a cubic cost function]{
    \centering
    \scriptsize{
    \centering
    \begin{tabular}{lllllllll}
    \tiny$\boldsymbol{\beta}^*=(\beta_{11}^*,\beta_{12}^*)$ & \tiny$\boldsymbol{x}^{opt}$ & \tiny$n_j^{(1)}$ & \tiny$n_j^{(2)}$ & \begin{tabular}[c]{@{}l@{}}MeanOpt1\\ (Q2.5,Q97.5)\end{tabular} & \begin{tabular}[c]{@{}l@{}}MeanOpt2\\ (Q2.5,Q97.5)\end{tabular} & \begin{tabular}[c]{@{}l@{}}SetCP95\\ $\%$\end{tabular} & \begin{tabular}[c]{@{}l@{}}SetPerc\\ $\%$\end{tabular} & \begin{tabular}[c]{@{}l@{}}BandsCP95\\ $\%$\end{tabular} \\
    \multicolumn{4}{l}{Scenario 1 ($J_1 = J_2 = 20$)} &  &  &  &  &  \\
    (0.1863, 0.15) & (1.5,7.4) & 50 & 100 & (0.656, 0.812) & (0.769, 0.809) & 95.8 & 8.3 & 96.1 \\
     &  &  & 500 &  & (0.786, 0.809) & 95.3 & 5.3 & 95.6 \\
     &  & 100 & 100 & (0.708, 0.816) & (0.776, 0.807) & 95.9 & 7.0 & 96.3 \\
     &  &  & 500 &  & (0.789, 0.806) & 95.7 & 4.6 & 95.9 \\
    \multicolumn{4}{l}{Scenario 2a ($J_1 = 6, J_2 = 12$)} &  &  &  & \multicolumn{1}{r}{} &  \\
    (0.1863, 0.15) & (1.5,7.4) & 50 & 200 & (0.509, 0.814) & (0.735, 0.821) & 94.8 & 11.5 & 95.8 \\
    (0.1, 0.2133) & (0, 6.5) & 50 & 200 & (0.515, 0.823) & (0.775, 0.820) & 95.9 & 14.7 & 95.8
    \end{tabular}
    \label{meanopt table cubic cost}
    }
    }
\scriptsize	\\
\raggedright{
    $n_j^{(1)}$: number of patients in center $j$ at stage 1, 
    $n_j^{(2)}$: number of patients in center $j$ at stage 2.\\
    $J$: number of centers for each stage. \\
    \%RelBias: percent relative bias $100(\hat{\beta}-\beta^\star)/\beta^\star$.\\
    SE: mean estimated standard error, 
    EMP.SD: empirical standard deviation.\\
    CP95: empirical coverage rate of 95\% confidence intervals.\\
    Bias of $\hat{x}_{1}^{o p t}$: bias of the first component of the estimated optimal intervention,  
    Bias of $\hat{x}_{2}^{o p t}$: bias of the second component of the estimated optimal intervention.\\ 
    rMSE: root of mean squared errors,$\Bigl\{\operatorname{mean}\bigl(\bigl\|\hat{\boldsymbol{x}}^{o p t}-\boldsymbol{x}^{o p t}\bigl\|^{2}\bigl)\Bigl\}^{1 / 2}$, mean is taken over simulation iterations.\\
    MeanOpt1: mean outcome under the stage 2 recommended intervention, calculated using true coefficient values;
    MeanOpt2: mean outcome under the final estimated optimal intervention based on all data, calculated using true coefficient values. \\
    $Q$2.5 and $Q$97.5: 2.5$\%$ and 97.5$\%$ quantiles. \\
    SetCP95$\%$: empirical coverage percentage of confidence set for the optimal intervention. 
    SetPerc$\%$: mean percentage of the size of the confidence set as a percent of the total sample space.
    BandsCP95$\%$: empirical coverage of 95$\%$ confidence band. 
}
\stepcounter{table}
\end{table}

\begin{table}[ht]
\centering
\vspace{-0.4cm}
\caption{
Comparing relative bias of $\hat{\boldsymbol{\beta}}$, its mean estimated standard error over empirical standard deviation, coverage probability and power among c-LAGO, uv-LAGO, and factorial design with a linear cost function.
}
\scriptsize{
\begin{tabular}{llllllllll}
 &  & \multicolumn{2}{c}{c-LAGO} &  & \multicolumn{2}{c}{uv-LAGO} &  & \multicolumn{2}{c}{factorial} \\
 &  & $\%$RelBias & $\frac{SE}{EMP.SD}$ &  & $\%$RelBias & $\frac{SE}{EMP.SD}$ &  & $\%$RelBias & $\frac{SE}{EMP.SD}$ \\
Stage 1 & $\beta_0$ & -0.79 & 100.40 &  & -0.79 & 100.40 &  &  &  \\
 & $\beta_{11}$ & -4.55 & 97.99 &  & -4.55 & 97.99 &  &  &  \\
 & $\beta_{12}$ & 0.62 & 97.56 &  & 0.62 & 97.56 &  &  &  \\
Stage 2 & $\beta_0$ & -25.93 & 96.52 &  & -26.23 & 95.64 &  & -2.26 & 99.93 \\
 & $\beta_{11}$ & -6.68 & 97.57 &  & -7.87 & 98.10 &  & 1.82 & 100.02 \\
 & $\beta_{12}$ & 4.67 & 98.16 &  & 4.79 & 98.46 &  & 0.04 & 102.46 \\
 & \begin{tabular}[c]{@{}l@{}}CP95\\ ($\beta_0$,$\beta_{11}$,$\beta_{12}$)\end{tabular} & \multicolumn{2}{c}{(0.94, 0.94, 0.94)} &  & \multicolumn{2}{c}{(0.96, 0.98, 0.98)} &  & \multicolumn{2}{c}{(0.94, 0.95, 0.96)} \\
 & Power & \multicolumn{2}{c}{(0.99, 0.93)} &  & \multicolumn{2}{c}{(1.00, 0.93)} &  & \multicolumn{2}{c}{(0.99, 0.91)} \\
 & \multicolumn{9}{l}{(component-wise test, two-sample means test)}
\end{tabular}
}
\label{Tab:factorial design}\\
\scriptsize
\raggedright{
$\%$RelBias: percent relative bias $100(\hat{\beta}-\beta^\star)/\beta^\star$. \\
SE: mean estimated standard error, 
EMP.SD: empirical standard deviation. \\ 
CP95: empirical coverage rate of 95\% confidence intervals. \\
component-wise test: P degree of freedom Chi-squared test.
}
\end{table}

\begin{table}[ht]
\caption{
Comparing bias of $\hat{\boldsymbol{\beta}}$, its mean estimated standard error, coverage probability and power among c-LAGO, MOST and factorial designs with a linear cost function.
}
\label{MOST table}
\scriptsize{
\begin{tabular}{lllllllll}
\multicolumn{1}{c}{} &  & \multicolumn{3}{c}{c-LAGO} & \multicolumn{1}{c}{} & \multicolumn{3}{c}{factorial} \\
 &  & $\hat{\beta}_{0}$ & $\hat{\beta}_{11}$ & $\hat{\beta}_{12}$ &  & $\hat{\beta}_{0}$ & $\hat{\beta}_{11}$ & $\hat{\beta}_{12}$ \\
Stage 1 & Bias & -0.004 & -0.004 & 0.002 & Bias & -0.004 & -0.001 & 0.001 \\
 & SE & 0.302 & 0.174 & 0.070 & SE & 0.175 & 0.101 & 0.040 \\
Stage 1-2 & Bias & -0.011 & -0.027 & 0.005 & CP95 & 0.957 & 0.951 & 0.946 \\
 & SE & 0.185 & 0.114 & 0.051 & Power & 0.529 &  &  \\
Stage 1-3 & Bias & -0.012 & -0.035 & 0.005 &  &  &  &  \\
 & SE & 0.145 & 0.094 & 0.041 &  &  &  &  \\
 & CP95 & 0.943 & 0.939 & 0.936 &  &  &  &  \\
 & Power & 0.871 &  &  &  &  &  &  \\
\multicolumn{1}{c}{} & \multicolumn{1}{c}{} & \multicolumn{6}{c}{MOST} &  \\
\multicolumn{2}{l}{\begin{tabular}[c]{@{}l@{}}Optimization phase:\\ n=100\\ RCT: n=200\end{tabular}} & \multicolumn{2}{l}{\begin{tabular}[c]{@{}l@{}}Intervention\\ package effect:   \\ Median (IQR)\end{tabular}} & \begin{tabular}[c]{@{}l@{}}0.555 \\ (0.339, 0.730)\end{tabular} & \multicolumn{2}{l}{\begin{tabular}[c]{@{}l@{}}C.P. intervention \\ package effect:\end{tabular}} & 0.947 &  \\
\multicolumn{2}{l}{} &  &  &  &  & Power & 0.489 &  \\
\multicolumn{2}{l}{\begin{tabular}[c]{@{}l@{}}Optimization phase:\\ n=148*\\ RCT: n=152\end{tabular}} & \multicolumn{2}{l}{\begin{tabular}[c]{@{}l@{}}Intervention   \\ package effect:   \\ Median (IQR)\end{tabular}} & \begin{tabular}[c]{@{}l@{}}0.606\\ (0.437, 0.769)\end{tabular} & \multicolumn{2}{l}{\begin{tabular}[c]{@{}l@{}}C.P. intervention \\ package effect:\end{tabular}} & 0.949 &  \\
 &  &  &  &  &  & Power & 0.465 & 
\end{tabular}
}\\
\vspace*{0.15cm} 
\scriptsize{
\raggedright{
*: the sample size for the optimization phase must be divisible by 4 as it is a factorial design with 4 intervention groups. \\
SE: mean estimated standard error. \\
Intervention package effect: true mean under the intervention in the RCT.\\
C.P. intervention package effect: empirical coverage rate of 95\% confidence intervals for the treatment effect of the actual intervention in the RCT.\\
CP95: empirical coverage rate of 95\% confidence intervals. \\
Power: two sample means test.\\
MOST Optimization trial n = 100: 1952 out of 2000 trials moved on to an RCT (have at least 1 estimated coefficient greater than 0), 1092 out of 1952 RCTs include only intervention component 1,  860 out of 1952 RCTs include only intervention component 2, 0 out of 1952 RCTs include both intervention components. \\
MOST Optimization trial n = 148: 1984 out of 2000 trials moved on to an RCT (have at least 1 estimated coefficient greater than 0), 1165 out of 1952 RCTs include only intervention component 1, 819 out of 1952 RCTs include only intervention component 2, 0 out of 1952 RCTs include both intervention components. \\
}
}
\end{table}

\section{Illustrative example: The BetterBirth Study}\label{application}
We illustrated the LAGO design for continuous outcomes with the BetterBirth study \citep{hirschhorn2015learning, semrau2017outcomes}. The BetterBirth study aimed to improve maternal and neonatal health outcomes in Uttar Pradesh, India by implementing the World Health Organization's (WHO) Safe Childbirth Checklist (SCC).
This WHO checklist encouraged birth attendants to use essential birth practices (EBPs) known to prevent complications at various pause points during the delivery process.

The BetterBirth study consisted of three stages, the first two of which were pilot studies. Stage 3 was a cluster randomized trial. The intervention package used in stage 2 was modified based on feedback from stage 1, and adjusted again in stage 3. Stage 1 included 2 centers, stage 2 included 4 centers, and stage 3 included a control and an intervention arm with 15 centers each. Outcome data was collected before and after the intervention was implemented in the two pilot stages, and data was collected both before and after the implementation of the intervention package for 5 centers in the cluster randomized trial. 

The outcome of interest was the proportion of EBPs performed out of all possible birth practices measured during each stage, and we modeled this as a continuous outcome. Births with 0 measured EBPs were excluded from the study because 
whether EBPs were measured depends
on the availability of researchers to document them. The number of EBPs measured at each stage were 14, 19, and 18, respectively. The average proportion of EBPs performed out of all possible birth practices measured for the three stages were 0.33, 0.28, and 0.42, respectively. 

We included the approximate monthly birth volume as a baseline center characteristic $\boldsymbol{z}$. To avoid multicollinearity, we only considered two out of the four intervention components: the duration of the on-site intervention launch (in days) and the number of coaching visits after the initial intervention launch, truncated to 40 visits or less. 
Because the outcome of interest was the mean proportion of EBPs performed, we fit this LAGO design by a GLM with a logit link function to restrict the expected fractions to values between 0 and 1: 
\begin{equation*}
logit\left(E\left(Y_{ij} | \boldsymbol{A}_{j}=\boldsymbol{a}_j, \boldsymbol{X}_{j}=\boldsymbol{x}_j, \boldsymbol{z}_{j}; \boldsymbol{\beta} \right)\right) = \beta_{0}+\boldsymbol{\beta}_{1}^{T} \boldsymbol{a}_{j}+\boldsymbol{\beta}_{2}^{T} \boldsymbol{z}_{j}.
\end{equation*}
Table \ref{estimated effects bb} reports the estimated effects of the intervention package components after each stage, based on all data available at the end of that stage. The final analysis (last column) indicates that both the duration of the on-site intervention launch and the number of coaching visits had positive effects. The estimated effect of the number of coaching visits was also highly significant across the different stages, 
and both the two sample means test and the component-wise Chi-squared test had p-values less than 0.001. 

\begin{table}[ht]
\centering
\caption{The BetterBirth Study: package component effect estimates, 95$\%$ confidence intervals after each stage, and estimated optimal intervention package after each stage}
\scriptsize{    
\begin{tabular}{llll}
 & \begin{tabular}[c]{@{}l@{}}Stage 1\\ $n^{(1)}=113$\\ $\hat{\boldsymbol{\beta}}$ ($95\%$ CI)\end{tabular} & \begin{tabular}[c]{@{}l@{}}Stage 1-2\\ $n^{(1)}+n^{(2)}=2256$\\ $\hat{\boldsymbol{\beta}}$   ($95\%$ CI)\end{tabular} & \begin{tabular}[c]{@{}l@{}}Stage 1-3\\ $n^{(1)}+n^{(2)}+n^{(3)}$\\ $=7342$\\ $\hat{\boldsymbol{\beta}}$   ($95\%$ CI)\end{tabular} \\
 &  &  &  \\
Intercept & 2.72 (1.20, 4.25) & -0.61 (-0.69, -0.53) & -0.138 (-0.156, -0.120) \\
Launch Duration (days) & -0.09 (-0.33, 0.14) & -0.003 (-0.11, 0.10) & 0.17 (0.11, 0.22) \\
Coaching Visits (per 5 visits) & 0.90 (0.79, 1.01) & 0.32 (0.31, 0.34) & 0.172 (0.167, 0.176) \\
Birth Volume (monthly, per 100) & -3.39 (-4.59, -2.18) & -0.166 (-0.180,   -0.153) & -0.202 (-0.210,   -0.195) \\
\begin{tabular}[c]{@{}l@{}}$\hat{\boldsymbol{x}}^{opt}$\\ (using the linear cost function)\end{tabular} &  
(1, 16) & 
(1, 36) &  
(5, 31) \\
\begin{tabular}[c]{@{}l@{}}$\hat{\boldsymbol{x}}^{opt}$\\ (using the cubic cost function)\end{tabular} &  
(1, 15.66) &  
(1, 35.38) &  
(4.73, 31.84)
\end{tabular}
}\\
\vspace*{0.15cm} 
\scriptsize
\raggedright{
CI: based on sandwich estimator for VAR($\hat{\boldsymbol{\beta}}$) (see Theorem \ref{AN_thm}). \\
For (optimal) interventions, the first component is launch duration and the second component is number of coaching visits.\\
The optimal intervention reported is for a center with an average birth volume ($z=175$).
}
\label{estimated effects bb}
\end{table}

Next, we found the optimal intervention package that results in a mean of performed EBPs greater than 0.8 ($\theta = 0.8$) while minimizing cost. 
Let $x_1$ be the launch duration (in days) and $x_2$ be the number of coaching visits,
we used both the linear cost function: 
$C_1(\boldsymbol{x})=800x_1 + 170x_2$ \citep{nevo2018analysis}
and the cubic cost function: 
$C_2(\boldsymbol{x})=1700x_1-950x_1^2+220x_1^3 + 380x_2-24x_2^2+0.6x_2^3$
for the analysis. 
The cubic cost function revised the linear cost function to include an economy of scale at lower values of intervention component costs and by including prohibitive cost as the intervention components neared their upper limits.
Constraints were set such that $1 \leq x_1 \leq 5$ and $1 \leq x_2 \leq 40$. 
For a center with an average birth volume ($z = 175$), the estimated optimal intervention package under the linear cost function comprised a launch duration of 5 days and 31 coaching visits, with a total cost of $\$9270$. 
With the cubic cost function, the estimated optimal intervention package consisted of 3.97 days for launch duration and 35.50 coaching visits, at a total cost of $\$15629.04$. 
The closest integer values for the estimated optimal intervention that can lead to a mean of performed EBPs greater than 0.8 while minimizing cost are 4 days for launch duration and 36 coaching visits.

To determine the 95$\%$ confidence set $CS(\boldsymbol{x}^{opt})$ for the optimal intervention, we examined a grid of all possible values of the intervention components.
The launch duration was incremented in steps of 0.01 days, ranging from 1 to 5 days (i.e., 1, 1.01, 1.02, ..., 5), the number of coaching visits was also incremented by 0.01, varying from 1 to 40 (i.e., 1, 1.01, 1.02, ..., 40). 
The final 95$\%$ confidence set consisted of $1.01\times10^5$ ($6.48\%$) out of the total $1.56\times10^6$ intervention packages, and included various combinations of the intervention package components. 
Employing the linear cost function, the first, second, and third quartiles of the cost within the confidence set were $\$8930$, $\$9150$, and $\$9380$, respectively. Conversely, utilizing the cubic cost function, the corresponding quartiles of the cost within the confidence set were 
$\$15356$, $\$16449$, and $\$17481$. 
Under the linear cost function, with the estimated optimal intervention as $\hat{\boldsymbol{x}}^{opt} = (5, 31)$ (based on the confidence bands for the proportion of EBPs performed under all possible intervention package components), the 95\% confidence interval for the proportion of EBPs was (0.766, 0.834).
Under the cubic cost function, with the estimated optimal intervention as $\hat{\boldsymbol{x}}^{opt} = (3.97, 35.50)$, the 95\% confidence interval for the proportion of EBPs was (0.780, 0.819).

\section{Discussion}\label{Discussion}
The LAGO design allows for changes in the composition of the intervention package based on accumulating data from an ongoing trial at pre-specified stages. LAGO could help prevent failed trials by adapting and optimizing the intervention package composition while the trial is ongoing. LAGO is useful for implementation trials, pragmatic trials, and clinical trials of combination regimens. The methods described in this paper further increase the flexibility of the LAGO design, and we anticipate that LAGO will be widely adopted in intervention trials of combination interventions and combination implementation strategies.

The objectives of LAGO studies are to determine the optimal intervention package that achieves a pre-specified effect at minimal cost, test its efficacy on the outcome of interest, and evaluate its effects. 
Variation in the interventions is needed for LAGO to identify the treatment effect parameters, which is an inherent feature of uv-LAGO and frequently observed in large-scale public health studies. Furthermore, c-LAGO often recommends different interventions for different centers based on their center characteristics, thereby introducing variation.
We have proven that the LAGO design leads to consistent, asymptotically normal estimators when the outcomes are continuous. 
The simulation studies show that LAGO has good finite sample properties at reasonable sample sizes.

While LAGO designs in their current form do not encompass interim hypothesis testing, it is indeed feasible to include futility stops in LAGO designs. 
By defining a baseline acceptable power level, futility stops can be seamlessly integrated into a LAGO trial. 
Specifically, if we fail to identify any recommended intervention within the feasible bounds that is projected to yield an adequate power after stage $k$, the trial may be terminated early. 
The type I error does not increase from futility stops, since there is no strong conclusion when a trial stops for futility 
\citep{snapinn2006assessment}. 

As in \cite{nevo2018analysis}, we used fixed baseline center characteristics to account for any random center effects. 
Future research could focus on developing LAGO for studies with center effects. In large-scale intervention studies, cluster randomized trials are common, and the random effects model is the standard approach for modeling these studies \citep{bell2019fixed}. Another potential extension is allowing centers to participate in more than one stage of a LAGO study, which is likely to occur in real world studies.
In addition, future research could focus on developing LAGO for individual-level interventions.
Another important area for future research is to systematically determine the optimal values of $K$, $J_k$, $n_{jk}$, and the function $f$ of Remark \ref{remark}, that provides the recommended intervention in later stages. 

The LAGO design will be applied to the PULESA-Uganda trial (Strengthening the Blood Pressure Care and Treatment Cascade for Ugandans Living with HIV - ImpLEmentation Strategies to SAve Lives; NIH UG3HL154501). Chronic HIV is a risk factor for cardiovascular disease (CVD), and hypertension is the most important driver of CVD risk. Blood pressure measurement is a key step in the management of hypertension. 
PULESA will first explore current practice, routines, barriers, and facilitators of evidence-based blood pressure care in HIV clinical settings, and then use these findings to design an implementation strategy to improve HIV-hypertension care. PULESA will determine the effectiveness of the proposed intervention package and evaluate the implementation strategy's economic sustainability in Uganda's Kampala and Wakiso districts.

\begin{acks}[Acknowledgments]
The authors are extremely grateful to the participants of the BetterBirth Study.
The authors also thank Dr. Katherine Semrau for her valuable contributions in explaining the BetterBirth Study data.
\end{acks}

\begin{funding}
Judith J. Lok was supported by NSF Grant DMS 1854934.
The BetterBirth Study was supported by the Bill and Melinda Gates Foundation. 
The article contents are the sole responsibility of the authors and
do not necessarily represent the official views of the Bill and Melinda Gates Foundation 
or the NSF.
\end{funding}

\bibliographystyle{imsart-nameyear} 
\bibliography{bibliography}       

\begin{thebibliography}{20}

\bibitem[\protect\citeauthoryear{Beidas et~al.}{2022}]{beidas2022promises}
\begin{barticle}[author]
\bauthor{\bsnm{Beidas},~\bfnm{Rinad~S}\binits{R.~S.}},
  \bauthor{\bsnm{Dorsey},~\bfnm{Shannon}\binits{S.}},
  \bauthor{\bsnm{Lewis},~\bfnm{Cara~C}\binits{C.~C.}},
  \bauthor{\bsnm{Lyon},~\bfnm{Aaron~R}\binits{A.~R.}},
  \bauthor{\bsnm{Powell},~\bfnm{Byron~J}\binits{B.~J.}},
  \bauthor{\bsnm{Purtle},~\bfnm{Jonathan}\binits{J.}},
  \bauthor{\bsnm{Saldana},~\bfnm{Lisa}\binits{L.}},
  \bauthor{\bsnm{Shelton},~\bfnm{Rachel~C}\binits{R.~C.}},
  \bauthor{\bsnm{Stirman},~\bfnm{Shannon~Wiltsey}\binits{S.~W.}} \AND
  \bauthor{\bsnm{Lane-Fall},~\bfnm{Meghan~B}\binits{M.~B.}}
(\byear{2022}).
\btitle{Promises and pitfalls in implementation science from the perspective of
  US-based researchers: learning from a pre-mortem}.
\bjournal{Implementation Science}
\bvolume{17}
\bpages{55}.
\end{barticle}
\endbibitem

\bibitem[\protect\citeauthoryear{Bell, Fairbrother and
  Jones}{2019}]{bell2019fixed}
\begin{barticle}[author]
\bauthor{\bsnm{Bell},~\bfnm{Andrew}\binits{A.}},
  \bauthor{\bsnm{Fairbrother},~\bfnm{Malcolm}\binits{M.}} \AND
  \bauthor{\bsnm{Jones},~\bfnm{Kelvyn}\binits{K.}}
(\byear{2019}).
\btitle{Fixed and random effects models: making an informed choice}.
\bjournal{Quality \& Quantity}
\bvolume{53}
\bpages{1051--1074}.
\end{barticle}
\endbibitem

\bibitem[\protect\citeauthoryear{Collins, Murphy and
  Strecher}{2007}]{collins2007multiphase}
\begin{barticle}[author]
\bauthor{\bsnm{Collins},~\bfnm{Linda~M}\binits{L.~M.}},
  \bauthor{\bsnm{Murphy},~\bfnm{Susan~A}\binits{S.~A.}} \AND
  \bauthor{\bsnm{Strecher},~\bfnm{Victor}\binits{V.}}
(\byear{2007}).
\btitle{The multiphase optimization strategy (MOST) and the sequential multiple
  assignment randomized trial (SMART): new methods for more potent eHealth
  interventions}.
\bjournal{American journal of preventive medicine}
\bvolume{32}
\bpages{S112--S118}.
\end{barticle}
\endbibitem

\bibitem[\protect\citeauthoryear{Collins, Nahum-Shani and
  Almirall}{2014}]{collins2014optimization}
\begin{barticle}[author]
\bauthor{\bsnm{Collins},~\bfnm{Linda~M}\binits{L.~M.}},
  \bauthor{\bsnm{Nahum-Shani},~\bfnm{Inbal}\binits{I.}} \AND
  \bauthor{\bsnm{Almirall},~\bfnm{Daniel}\binits{D.}}
(\byear{2014}).
\btitle{Optimization of behavioral dynamic treatment regimens based on the
  sequential, multiple assignment, randomized trial (SMART)}.
\bjournal{Clinical Trials}
\bvolume{11}
\bpages{426--434}.
\end{barticle}
\endbibitem

\bibitem[\protect\citeauthoryear{Collins et~al.}{2011}]{collins2011multiphase}
\begin{barticle}[author]
\bauthor{\bsnm{Collins},~\bfnm{Linda~M}\binits{L.~M.}},
  \bauthor{\bsnm{Baker},~\bfnm{Timothy~B}\binits{T.~B.}},
  \bauthor{\bsnm{Mermelstein},~\bfnm{Robin~J}\binits{R.~J.}},
  \bauthor{\bsnm{Piper},~\bfnm{Megan~E}\binits{M.~E.}},
  \bauthor{\bsnm{Jorenby},~\bfnm{Douglas~E}\binits{D.~E.}},
  \bauthor{\bsnm{Smith},~\bfnm{Stevens~S}\binits{S.~S.}},
  \bauthor{\bsnm{Christiansen},~\bfnm{Bruce~A}\binits{B.~A.}},
  \bauthor{\bsnm{Schlam},~\bfnm{Tanya~R}\binits{T.~R.}},
  \bauthor{\bsnm{Cook},~\bfnm{Jessica~W}\binits{J.~W.}} \AND
  \bauthor{\bsnm{Fiore},~\bfnm{Michael~C}\binits{M.~C.}}
(\byear{2011}).
\btitle{The multiphase optimization strategy for engineering effective tobacco
  use interventions}.
\bjournal{Annals of behavioral medicine}
\bvolume{41}
\bpages{208--226}.
\end{barticle}
\endbibitem

\bibitem[\protect\citeauthoryear{Eisenberg and
  Gan}{1983}]{eisenberg1983uniform}
\begin{barticle}[author]
\bauthor{\bsnm{Eisenberg},~\bfnm{Bennett}\binits{B.}} \AND
  \bauthor{\bsnm{Gan},~\bfnm{Shi~Xin}\binits{S.~X.}}
(\byear{1983}).
\btitle{Uniform convergence of distribution functions}.
\bjournal{Proceedings of the American Mathematical Society}
\bvolume{88}
\bpages{145--146}.
\end{barticle}
\endbibitem

\bibitem[\protect\citeauthoryear{Fahrmeir and
  Tutz}{2013}]{fahrmeir2013multivariate}
\begin{bbook}[author]
\bauthor{\bsnm{Fahrmeir},~\bfnm{Ludwig}\binits{L.}} \AND
  \bauthor{\bsnm{Tutz},~\bfnm{Gerhard}\binits{G.}}
(\byear{2013}).
\btitle{Multivariate statistical modelling based on generalized linear models}.
\bpublisher{Springer Science \& Business Media}.
\end{bbook}
\endbibitem

\bibitem[\protect\citeauthoryear{FDA}{2016}]{FDAdevices}
\begin{bmisc}[author]
\bauthor{\bsnm{FDA}}
(\byear{2016}).
\btitle{Adaptive designs for medical device clinical studies: Guidance for
  industry and food and drug administration staff}.
\end{bmisc}
\endbibitem

\bibitem[\protect\citeauthoryear{FDA}{2019}]{FDAadaptiveNew}
\begin{bmisc}[author]
\bauthor{\bsnm{FDA}}
(\byear{2019}).
\btitle{Adaptive Design Clinical Trials for Drugs and Biologics Guidance for
  Industry. 2019}.
\end{bmisc}
\endbibitem

\bibitem[\protect\citeauthoryear{Fogel}{2018}]{fogel2018factors}
\begin{barticle}[author]
\bauthor{\bsnm{Fogel},~\bfnm{David~B}\binits{D.~B.}}
(\byear{2018}).
\btitle{Factors associated with clinical trials that fail and opportunities for
  improving the likelihood of success: a review}.
\bjournal{Contemporary clinical trials communications}
\bvolume{11}
\bpages{156--164}.
\end{barticle}
\endbibitem

\bibitem[\protect\citeauthoryear{Greer}{2010}]{greer2010electricity}
\begin{bbook}[author]
\bauthor{\bsnm{Greer},~\bfnm{Monica}\binits{M.}}
(\byear{2010}).
\btitle{Electricity cost modeling calculations}.
\bpublisher{Academic Press}.
\end{bbook}
\endbibitem

\bibitem[\protect\citeauthoryear{Hirschhorn
  et~al.}{2015}]{hirschhorn2015learning}
\begin{barticle}[author]
\bauthor{\bsnm{Hirschhorn},~\bfnm{Lisa~Ruth}\binits{L.~R.}},
  \bauthor{\bsnm{Semrau},~\bfnm{Katherine}\binits{K.}},
  \bauthor{\bsnm{Kodkany},~\bfnm{Bhala}\binits{B.}},
  \bauthor{\bsnm{Churchill},~\bfnm{Robyn}\binits{R.}},
  \bauthor{\bsnm{Kapoor},~\bfnm{Atul}\binits{A.}},
  \bauthor{\bsnm{Spector},~\bfnm{Jonathan}\binits{J.}},
  \bauthor{\bsnm{Ringer},~\bfnm{Steve}\binits{S.}},
  \bauthor{\bsnm{Firestone},~\bfnm{Rebecca}\binits{R.}},
  \bauthor{\bsnm{Kumar},~\bfnm{Vishwajeet}\binits{V.}} \AND
  \bauthor{\bsnm{Gawande},~\bfnm{Atul}\binits{A.}}
(\byear{2015}).
\btitle{Learning before leaping: integration of an adaptive study design
  process prior to initiation of BetterBirth, a large-scale randomized
  controlled trial in Uttar Pradesh, India}.
\bjournal{Implementation science}
\bvolume{10}
\bpages{1--9}.
\end{barticle}
\endbibitem

\bibitem[\protect\citeauthoryear{Liang and Zeger}{1986}]{liang1986longitudinal}
\begin{barticle}[author]
\bauthor{\bsnm{Liang},~\bfnm{Kung-Yee}\binits{K.-Y.}} \AND
  \bauthor{\bsnm{Zeger},~\bfnm{Scott~L}\binits{S.~L.}}
(\byear{1986}).
\btitle{Longitudinal data analysis using generalized linear models}.
\bjournal{Biometrika}
\bvolume{73}
\bpages{13--22}.
\end{barticle}
\endbibitem

\bibitem[\protect\citeauthoryear{McCullagh and
  Nelder}{2019}]{mccullagh2019generalized}
\begin{bbook}[author]
\bauthor{\bsnm{McCullagh},~\bfnm{Peter}\binits{P.}} \AND
  \bauthor{\bsnm{Nelder},~\bfnm{John~A}\binits{J.~A.}}
(\byear{2019}).
\btitle{Generalized linear models}.
\bpublisher{Routledge}.
\end{bbook}
\endbibitem

\bibitem[\protect\citeauthoryear{Nevo, Lok and
  Spiegelman}{2021}]{nevo2018analysis}
\begin{barticle}[author]
\bauthor{\bsnm{Nevo},~\bfnm{Daniel}\binits{D.}},
  \bauthor{\bsnm{Lok},~\bfnm{Judith~J}\binits{J.~J.}} \AND
  \bauthor{\bsnm{Spiegelman},~\bfnm{Donna}\binits{D.}}
(\byear{2021}).
\btitle{Analysis of “Learn-As-you-GO”(LAGO) studies}.
\bjournal{The Annals of Statistics}
\bvolume{49}
\bpages{793--819}.
\end{barticle}
\endbibitem

\bibitem[\protect\citeauthoryear{Scheffe}{1999}]{scheffe1999analysis}
\begin{bbook}[author]
\bauthor{\bsnm{Scheffe},~\bfnm{Henry}\binits{H.}}
(\byear{1999}).
\btitle{The analysis of variance}
\bvolume{72}.
\bpublisher{John Wiley \& Sons}.
\end{bbook}
\endbibitem

\bibitem[\protect\citeauthoryear{Semrau et~al.}{2017}]{semrau2017outcomes}
\begin{barticle}[author]
\bauthor{\bsnm{Semrau},~\bfnm{Katherine~EA}\binits{K.~E.}},
  \bauthor{\bsnm{Hirschhorn},~\bfnm{Lisa~R}\binits{L.~R.}},
  \bauthor{\bsnm{Marx~Delaney},~\bfnm{Megan}\binits{M.}},
  \bauthor{\bsnm{Singh},~\bfnm{Vinay~P}\binits{V.~P.}},
  \bauthor{\bsnm{Saurastri},~\bfnm{Rajiv}\binits{R.}},
  \bauthor{\bsnm{Sharma},~\bfnm{Narender}\binits{N.}},
  \bauthor{\bsnm{Tuller},~\bfnm{Danielle~E}\binits{D.~E.}},
  \bauthor{\bsnm{Firestone},~\bfnm{Rebecca}\binits{R.}},
  \bauthor{\bsnm{Lipsitz},~\bfnm{Stuart}\binits{S.}},
  \bauthor{\bsnm{Dhingra-Kumar},~\bfnm{Neelam}\binits{N.}} \betal{et~al.}
(\byear{2017}).
\btitle{Outcomes of a coaching-based WHO safe childbirth checklist program in
  India}.
\bjournal{New England Journal of Medicine}
\bvolume{377}
\bpages{2313--2324}.
\end{barticle}
\endbibitem

\bibitem[\protect\citeauthoryear{Snapinn et~al.}{2006}]{snapinn2006assessment}
\begin{barticle}[author]
\bauthor{\bsnm{Snapinn},~\bfnm{Steven}\binits{S.}},
  \bauthor{\bsnm{Chen},~\bfnm{Mon-Gy}\binits{M.-G.}},
  \bauthor{\bsnm{Jiang},~\bfnm{Qi}\binits{Q.}} \AND
  \bauthor{\bsnm{Koutsoukos},~\bfnm{Tony}\binits{T.}}
(\byear{2006}).
\btitle{Assessment of futility in clinical trials}.
\bjournal{Pharmaceutical Statistics: The Journal of Applied Statistics in the
  Pharmaceutical Industry}
\bvolume{5}
\bpages{273--281}.
\end{barticle}
\endbibitem

\bibitem[\protect\citeauthoryear{Stensland et~al.}{2014}]{stensland2014adult}
\begin{barticle}[author]
\bauthor{\bsnm{Stensland},~\bfnm{Kristian~D}\binits{K.~D.}},
  \bauthor{\bsnm{McBride},~\bfnm{Russell~B}\binits{R.~B.}},
  \bauthor{\bsnm{Latif},~\bfnm{Asma}\binits{A.}},
  \bauthor{\bsnm{Wisnivesky},~\bfnm{Juan}\binits{J.}},
  \bauthor{\bsnm{Hendricks},~\bfnm{Ryan}\binits{R.}},
  \bauthor{\bsnm{Roper},~\bfnm{Nitin}\binits{N.}},
  \bauthor{\bsnm{Boffetta},~\bfnm{Paolo}\binits{P.}},
  \bauthor{\bsnm{Hall},~\bfnm{Simon~J}\binits{S.~J.}},
  \bauthor{\bsnm{Oh},~\bfnm{William~K}\binits{W.~K.}} \AND
  \bauthor{\bsnm{Galsky},~\bfnm{Matthew~D}\binits{M.~D.}}
(\byear{2014}).
\btitle{Adult cancer clinical trials that fail to complete: an epidemic?}
\bjournal{JNCI: Journal of the National Cancer Institute}
\bvolume{106}.
\end{barticle}
\endbibitem

\bibitem[\protect\citeauthoryear{Van~der Vaart}{2000}]{van2000asymptotic}
\begin{bbook}[author]
\bauthor{\bparticle{Van~der} \bsnm{Vaart},~\bfnm{Aad~W}\binits{A.~W.}}
(\byear{2000}).
\btitle{Asymptotic statistics}
\bvolume{3}.
\bpublisher{Cambridge university press}.
\end{bbook}
\endbibitem

\end{thebibliography}

\newpage
\begin{appendix}
\section{ Overview of the Appendix}
The Appendix is organized as follows: Appendix \ref{general_link_appendix} presents the proofs of both Theorem \ref{GLM_consistency_general} and Theorem \ref{AN_thm} for continuous outcome LAGO GLM with independent errors and a general link function. Appendix \ref{appendix_loglink} describes the proofs of both Theorem \ref{GLM_consistency_general} and Theorem \ref{AN_thm} for LAGO GLM with independent errors and a log link function. Appendix \ref{K>2} extends the LAGO theory to include more than two stages. Appendix \ref{additional simulation} explains the motivation for the cubic cost function, and provides a complete set of simulation results for the ones presented in Section \ref{simulation section} of the main text.

\newpage
\section{Proofs of Theorem \ref{GLM_consistency_general} and Theorem \ref{AN_thm} with general link function.}\label{general_link_appendix}
\subsection{Proof of Theorem \ref{GLM_consistency_general} for LAGO GLM with general link function: consistency of \texorpdfstring{$\hat{\boldsymbol{\beta}}$}{Lg}} \label{consistency_appendix_generallink}
First, the estimating equations from (\ref{glm_EE}) of the main text are:
\begin{equation}
\begin{aligned}
&0=\boldsymbol{U}^{(g)}(\boldsymbol{\beta})\\
&=\frac{1}{n} \bigg[ \sum_{j=1}^{J^{(1)}} \sum_{i=1}^{n_{j}^{(1)}} \left(\frac{\partial}{\partial \boldsymbol{\beta}} E\left(Y_{ij}^{(1)} \mid \boldsymbol{a}_{j}^{(1)}, \boldsymbol{z}_{j}^{(1)} ; \boldsymbol{\beta}\right) \right)\left(Y_{ij}^{(1)}-E\left(Y_{ij}^{(1)} \mid \boldsymbol{a}_{j}^{(1)}, \boldsymbol{z}_{j}^{(1)}; \boldsymbol{\beta}\right)\right) \\
&\qquad+\sum_{j=1}^{J^{(2)}} \sum_{i=1}^{n_{j}^{(2)}} \left(\frac{\partial}{\partial \boldsymbol{\beta}} E\left(Y_{ij}^{\left(2, n^{(1)}\right)} \mid \boldsymbol{A}_{j}^{\left(2, n^{(1)}\right)}, \boldsymbol{z}_{j}^{(2)}; \boldsymbol{\beta}\right)\right)\\
&\qquad\qquad\qquad\qquad\qquad\qquad\left(Y_{ij}^{\left(2, n^{(1)}\right)}-E\left(Y_{ij}^{\left(2, n^{(1)}\right)} \mid \boldsymbol{A}_{j}^{\left(2, n^{(1)}\right)}, \boldsymbol{z}_{j}^{(2)}; \boldsymbol{\beta}\right)\right) \bigg]\\
&=\frac{1}{n} \bigg[ \sum_{j=1}^{J^{(1)}} \sum_{i=1}^{n_{j}^{(1)}} \left(\frac{\partial}{\partial \boldsymbol{\beta}} g^{-1}(\boldsymbol{a}_j^{(1)},\boldsymbol{z}_j^{(1)}; \boldsymbol{\beta}) \right)\left(Y_{ij}^{(1)}-g^{-1}(\boldsymbol{a}_j^{(1)}, \boldsymbol{z}_j^{(1)};\boldsymbol{\beta})\right) \\
&\qquad+\sum_{j=1}^{J^{(2)}} \sum_{i=1}^{n_{j}^{(2)}} \left(\frac{\partial}{\partial \boldsymbol{\beta}} 
g^{-1}(\boldsymbol{A}_{j}^{\left(2, n^{(1)}\right)},\boldsymbol{z}_j^{(2)}; \boldsymbol{\beta})
\right)\left(Y_{ij}^{\left(2, n^{(1)}\right)}-g^{-1}(\boldsymbol{A}_{j}^{\left(2, n^{(1)}\right)},\boldsymbol{z}_j^{(2)}; \boldsymbol{\beta})\right) \bigg].
\end{aligned} \label{general_EE_sm}
\end{equation}
Recall that $\boldsymbol{\beta}^*$ is the true value of $\boldsymbol{\beta}$ and let
\begin{equation}
\begin{aligned}
&\boldsymbol{u}^{(g)}(\boldsymbol{\beta}) \\
&= \sum_{j=1}^{J^{(1)}} \alpha_{j1} \left[ \left(\frac{\partial}{\partial \boldsymbol{\beta}} g^{-1}(\boldsymbol{a}_j^{(1)},\boldsymbol{z}_j^{(1)}; \boldsymbol{\beta}) \right) \left( g^{-1}(\boldsymbol{a}_j^{(1)},\boldsymbol{z}_j^{(1)}; \boldsymbol{\beta}^*) -g^{-1}(\boldsymbol{a}_j^{(1)},\boldsymbol{z}_j^{(1)}; \boldsymbol{\beta})  \right) \right]\\ 
&+ \sum_{j=1}^{J^{(2)}} \alpha_{j2}  \left[  \left(\frac{\partial}{\partial \boldsymbol{\beta}} g^{-1}(\boldsymbol{a}_j^{(2)},\boldsymbol{z}_j^{(2)}; \boldsymbol{\beta}) \right) \left( g^{-1}(\boldsymbol{a}_j^{(2)},\boldsymbol{z}_j^{(2)}; \boldsymbol{\beta}^*) -g^{-1}(\boldsymbol{a}_j^{(2)},\boldsymbol{z}_j^{(2)}; \boldsymbol{\beta}) \right) \right], 
\end{aligned}
\label{general_little_u}
\end{equation}
where 
\begin{equation*}
g^{-1}(\boldsymbol{a}_j^{(2)},\boldsymbol{z}_j^{(2)}; \boldsymbol{\beta})  
= E\left(Y_{ij}^{\left(2\right)} \mid \boldsymbol{a}_{j}^{\left(2\right)}, \boldsymbol{z}_{j}^{(2)}; \boldsymbol{\beta}\right).
\end{equation*}
To show consistency of the estimator $\hat{\boldsymbol{\beta}}$, we show that in spite of the fact that equation (\ref{general_EE_sm}) does not consist of i.i.d. terms, Theorem $5.9$ of \citet{van2000asymptotic} can be used. We show that its two conditions are satisfied. First, we show 
\begin{equation}
    \sup_{\boldsymbol{\beta}}\|\boldsymbol{U}^{(g)}(\boldsymbol{\beta}) - \boldsymbol{u}^{(g)}(\boldsymbol{\beta})\| \xrightarrow{P} 0.
\label{general_5.9_cond1}
\end{equation}
From equation (\ref{general_EE_sm}) and equation (\ref{general_little_u}), it follows that 
\begin{equation} \label{U-u separation}
\begin{aligned}
&\boldsymbol{U}^{(g)}(\boldsymbol{\beta}) - \boldsymbol{u}^{(g)}(\boldsymbol{\beta}) \\
& = \frac{1}{n} \sum_{j=1}^{J^{(1)}} \sum_{i=1}^{n_{j}^{(1)}} 
\left[ \left(\frac{\partial}{\partial \boldsymbol{\beta}} g^{-1}(\boldsymbol{a}_j^{(1)},\boldsymbol{z}_j^{(1)}; \boldsymbol{\beta}) \right)\left(Y_{ij}^{(1)}-g^{-1}(\boldsymbol{a}_j^{(1)}, \boldsymbol{z}_j^{(1)};\boldsymbol{\beta})\right) \right] \\
& + \frac{1}{n} \sum_{j=1}^{J^{(2)}} \sum_{i=1}^{n_{j}^{(2)}} \left[\left(\frac{\partial}{\partial \boldsymbol{\beta}} 
g^{-1}(\boldsymbol{A}_{j}^{\left(2, n^{(1)}\right)},\boldsymbol{z}_j^{(2)}; \boldsymbol{\beta})
\right)\left(Y_{ij}^{\left(2, n^{(1)}\right)}-g^{-1}(\boldsymbol{A}_{j}^{\left(2, n^{(1)}\right)},\boldsymbol{z}_j^{(2)}; \boldsymbol{\beta}) \right. \right.\\
&\qquad \qquad \qquad \qquad \qquad \qquad \qquad \qquad \qquad \qquad \left.\left.\pm \; g^{-1}(\boldsymbol{A}_{j}^{\left(2, n^{(1)}\right)},\boldsymbol{z}_j^{(2)}; \boldsymbol{\beta}^*)
\right) \right]\\
& - \sum_{j=1}^{J^{(1)}} \alpha_{j1} \left[ \left(\frac{\partial}{\partial \boldsymbol{\beta}} g^{-1}(\boldsymbol{a}_j^{(1)},\boldsymbol{z}_j^{(1)}; \boldsymbol{\beta}) \right) \left( g^{-1}(\boldsymbol{a}_j^{(1)},\boldsymbol{z}_j^{(1)}; \boldsymbol{\beta}^*) -g^{-1}(\boldsymbol{a}_j^{(1)},\boldsymbol{z}_j^{(1)}; \boldsymbol{\beta})  \right) \right] \\
& - \sum_{j=1}^{J^{(2)}} \alpha_{j2}  \left[  \left(\frac{\partial}{\partial \boldsymbol{\beta}} g^{-1}(\boldsymbol{a}_j^{(2)},\boldsymbol{z}_j^{(2)}; \boldsymbol{\beta}) \right) \left( g^{-1}(\boldsymbol{a}_j^{(2)},\boldsymbol{z}_j^{(2)}; \boldsymbol{\beta}^*) -g^{-1}(\boldsymbol{a}_j^{(2)},\boldsymbol{z}_j^{(2)}; \boldsymbol{\beta}) \right) \right] \\
& = G_{1,n}^{(g)} + G_{2,n}^{(g)} + G_{3,n}^{(g)} + G_{4,n}^{(g)} +  G_{5,n}^{(g)},\\ 
\end{aligned}
\end{equation}
where 
\begin{equation}
\begin{aligned}
G_{1,n}^{(g)} &= 
\frac{1}{n} \sum_{j=1}^{J^{(1)}} \sum_{i=1}^{n_{j}^{(1)}} \left[ \left( \frac{\partial}{\partial \boldsymbol{\beta}} g^{-1}(\boldsymbol{a}_j^{(1)},\boldsymbol{z}_j^{(1)}; \boldsymbol{\beta})\right)  \left( Y_{ij}^{(1)} - g^{-1}(\boldsymbol{a}_j^{(1)},\boldsymbol{z}_j^{(1)}; \boldsymbol{\beta}^*)  \right) \right],
\end{aligned}
\label{G1_term_general}
\end{equation}

\begin{equation}\begin{aligned}
G_{2,n}^{(g)} = \frac{1}{n} \sum_{j=1}^{J^{(2)}} \sum_{i=1}^{n_{j}^{(2)}} \left[ \left( \frac{\partial}{\partial \boldsymbol{\beta}} g^{-1}(\boldsymbol{A}_{j}^{\left(2, n^{(1)}\right)},\boldsymbol{z}_j^{(2)}; \boldsymbol{\beta}) \right) \left( Y_{ij}^{(2, n^{(1)})}-g^{-1}(\boldsymbol{A}_{j}^{\left(2, n^{(1)}\right)}, \boldsymbol{z}_{j}^{\left(2\right)}; \boldsymbol{\beta}^*) \right) \right] ,
\end{aligned}\label{G2_term_general}\end{equation}

\begin{equation}\begin{aligned}
G_{3,n}^{(g)} &= 
\frac{1}{n} \sum_{j=1}^{J^{(2)}} \sum_{i=1}^{n_{j}^{(2)}} \left[ \left( \frac{\partial}{\partial \boldsymbol{\beta}} g^{-1}(\boldsymbol{A}_{j}^{\left(2, n^{(1)}\right)},\boldsymbol{z}_j^{(2)}; \boldsymbol{\beta}) \right) \left( g^{-1}(\boldsymbol{A}_{j}^{\left(2, n^{(1)}\right)}, \boldsymbol{z}_{j}^{\left(2\right)}; \boldsymbol{\beta}^*)  \right) \right.\\
& \qquad \qquad \qquad \qquad \qquad \qquad \left. - \left( \frac{\partial}{\partial \boldsymbol{\beta}} g^{-1}(\boldsymbol{a}_j^{(2)},\boldsymbol{z}_j^{(2)}; \boldsymbol{\beta}) \right) \left( g^{-1}(\boldsymbol{a}_j^{(2)},\boldsymbol{z}_j^{(2)}; \boldsymbol{\beta}^*)  \right)  \right],
\end{aligned}
\label{G3_term_general}
\end{equation}

\begin{equation}\begin{aligned}
G_{4,n}^{(g)} &= 
\frac{1}{n} \sum_{j=1}^{J^{(2)}} \sum_{i=1}^{n_{j}^{(2)}} \left[ \left( \frac{\partial}{\partial \boldsymbol{\beta}} g^{-1}(\boldsymbol{a}_j^{(2)},\boldsymbol{z}_j^{(2)}; \boldsymbol{\beta}) \right) \left( g^{-1}(\boldsymbol{a}_j^{(2)},\boldsymbol{z}_j^{(2)}; \boldsymbol{\beta})  \right) \right.\\
& \qquad \qquad \qquad \qquad \qquad \left. - \left( \frac{\partial}{\partial \boldsymbol{\beta}} g^{-1}(\boldsymbol{A}_{j}^{\left(2, n^{(1)}\right)},\boldsymbol{z}_j^{(2)}; \boldsymbol{\beta}) \right) \left( g^{-1}(\boldsymbol{A}_{j}^{\left(2, n^{(1)}\right)},\boldsymbol{z}_j^{(2)}; \boldsymbol{\beta})  \right)  \right],
\end{aligned}
\label{G4_term_general}
\end{equation} 

\begin{equation}
\begin{aligned}
G_{5,n}^{(g)} = 
&\sum_{j=1}^{J^{(1)}} \left[ \left( \alpha_{j1}- \frac{n_j^{(1)}}{n} \right) \left( \frac{\partial}{\partial \boldsymbol{\beta}} g^{-1}(\boldsymbol{a}_j^{(1)},\boldsymbol{z}_j^{(1)}; \boldsymbol{\beta}) \right) \left( g^{-1}(\boldsymbol{a}_j^{(1)},\boldsymbol{z}_j^{(1)}; \boldsymbol{\beta})  \right) \right] \\
&+ \sum_{j=1}^{J^{(1)}} \left[ \left(\frac{n_j^{(1)}}{n}-\alpha_{j1} \right) \left( \frac{\partial}{\partial \boldsymbol{\beta}} g^{-1}(\boldsymbol{a}_j^{(1)},\boldsymbol{z}_j^{(1)}; \boldsymbol{\beta}) \right) \left( g^{-1}(\boldsymbol{a}_j^{(1)},\boldsymbol{z}_j^{(1)}; \boldsymbol{\beta}^*)  \right) \right] \\
&+  \sum_{j=1}^{J^{(2)}} \left[\left(\alpha_{j2}-\frac{n_j^{(2)}}{n} \right) \left( \frac{\partial}{\partial \boldsymbol{\beta}} g^{-1}(\boldsymbol{a}_j^{(2)},\boldsymbol{z}_j^{(2)}; \boldsymbol{\beta}) \right) \left( g^{-1}(\boldsymbol{a}_j^{(2)},\boldsymbol{z}_j^{(2)}; \boldsymbol{\beta})  \right)\right]\\
&+ \sum_{j=1}^{J^{(2)}} \left[  \left(\frac{n_j^{(2)}}{n}-\alpha_{j2} \right) \left( \frac{\partial}{\partial \boldsymbol{\beta}} g^{-1}(\boldsymbol{a}_j^{(2)},\boldsymbol{z}_j^{(2)}; \boldsymbol{\beta}) \right) \left( g^{-1}(\boldsymbol{a}_j^{(2)},\boldsymbol{z}_j^{(2)}; \boldsymbol{\beta}^*)  \right)\right].
\end{aligned}
\label{G5_term_general}
\end{equation}

The five terms $G_{1,n}^{(g)}$, $G_{2,n}^{(g)}$, $G_{3,n}^{(g)}$, $G_{4,n}^{(g)}$, and $G_{5,n}^{(g)}$ will be discussed separately. We show the supremum over $\boldsymbol{\beta}$ of each term converges to 0 in probability, then the triangle inequality implies equation (\ref{general_5.9_cond1}).

For $G_{1,n}^{(g)}$ from equation (\ref{G1_term_general}), 
we show that $\sup _{\boldsymbol{\beta}}\left\|G_{1,n}^{(g)}\right\| \stackrel{P}{\rightarrow} 0$ by using the concept of Donsker classes from empirical process theory.
Let $O^{(1)}_{ij} = (Y^{(1)}_{ij}, \boldsymbol{a}^{(1)}_{j}, \boldsymbol{z}^{(1)}_{j})$ be the observed data for patient $i$ from center $j$ in stage 1, and let
\begin{equation*}
    \Psi_{\boldsymbol{\beta}}(O^{(1)}_{ij}) = \left( \frac{\partial}{\partial \boldsymbol{\beta}} g^{-1}(\boldsymbol{a}_j^{(1)},\boldsymbol{z}_j^{(1)}; \boldsymbol{\beta})\right)  
    \left( Y_{ij}^{(1)} - g^{-1}(\boldsymbol{a}_j^{(1)},\boldsymbol{z}_j^{(1)}; \boldsymbol{\beta}^*)  \right).
\end{equation*}
We show that the class of functions $\mathcal{F} = \left( \Psi_{\boldsymbol{\beta}}: \boldsymbol{\beta} \in \boldsymbol{\mathcal{B}} \right)$ is Donsker by showing that 
\begin{equation}
\left\|\Psi_{\boldsymbol{\beta}_{1}}(O^{(1)}_{ij})-\Psi_{\boldsymbol{\beta}_{2}}(O^{(1)}_{ij})\right\| \leq m(O^{(1)}_{ij})\left\|\boldsymbol{\beta}_{1}-\boldsymbol{\beta}_{2}\right\| \quad \text { for every } \boldsymbol{\beta}_{1}, \boldsymbol{\beta}_{2},
\label{donsker_condition_general}
\end{equation}
for some measurable function $m$ with $E\left(m^{2}\right)<\infty$; that $\mathcal{F}$ is Donsker then follows from Example 19.7 in \cite{van2000asymptotic}. By applying the Mean Value Theorem to each row of $\Psi_{\boldsymbol{\beta}_{1}}(O^{(1)}_{ij})-\Psi_{\boldsymbol{\beta}_{2}}(O^{(1)}_{ij})$ separately, we find that
\begin{equation*}
\begin{aligned}
\Psi_{\boldsymbol{\beta}_{1}}(O^{(1)}_{ij})-\Psi_{\boldsymbol{\beta}_{2}}(O^{(1)}_{ij})& = 
\left(\left.\frac{\partial}{\partial \boldsymbol{\beta}}\right|_{\tilde{\boldsymbol{\beta}}} \Psi_{\boldsymbol{\beta}}(O^{(1)}_{ij})\right) \left(\boldsymbol{\beta}_{1} - \boldsymbol{\beta}_{2}\right),
\end{aligned}
\end{equation*}
where for each row of $\left.\frac{\partial}{\partial \boldsymbol{\beta}}\right|_{\tilde{\boldsymbol{\beta}}} \Psi_{\boldsymbol{\beta}}(O^{(1)}_{ij})$, $\tilde{\boldsymbol{\beta}}$ may take a different value between $\boldsymbol{\beta}_{1}$ and $\boldsymbol{\beta}_{2}$. 
By Assumption \ref{glm_model_assumption} and \ref{glm_unif_bound_assu} from the main text, $g()$ and $\Psi_{\boldsymbol{\beta}}(O^{(1)}_{ij})$ are continuously differentiable and both $\boldsymbol{\beta}$ and $O^{(1)}_{ij}$ take values in a compact space, each element in $\left.\frac{\partial}{\partial \boldsymbol{\beta}}\right|_{\tilde{\boldsymbol{\beta}}} \Psi_{\boldsymbol{\beta}}(O^{(1)}_{ij})$ is bounded. It follows that
\begin{equation*}
\begin{aligned}
\left\|\Psi_{\boldsymbol{\beta}_{1}}(O^{(1)}_{ij})-\Psi_{\boldsymbol{\beta}_{2}}(O^{(1)}_{ij})\right\| & = 
\left\| \left(\left.\frac{\partial}{\partial \boldsymbol{\beta}}\right|_{\tilde{\boldsymbol{\beta}}} \Psi_{\boldsymbol{\beta}}(O^{(1)}_{ij}) \right)
(\boldsymbol{\beta}_{1} - \boldsymbol{\beta}_{2})  \right\| \\
&\leq C_1^{(g)} \left\| \boldsymbol{\beta}_{1} - \boldsymbol{\beta}_{2} \right\|,
\end{aligned}
\end{equation*}
where $C_1^{(g)}$ is a constant. Equation (\ref{donsker_condition_general}) follows, so the class of functions \\
$\mathcal{F} = \left( \Psi_{\boldsymbol{\beta}}: \boldsymbol{\beta} \in \boldsymbol{\mathcal{B}} \right)$  is Donsker. By Theorem 19.4 (\cite{van2000asymptotic}), $\mathcal{F}$ is also Glivenko-Cantelli. Since for each $\boldsymbol{\beta}$, $E\left( \Psi_{\boldsymbol{\beta}}\left(O_{ij}^{(1)}\right)\right)=0$, then 
\begin{equation*}
\sup _{\boldsymbol{\beta}}\left\|\frac{1}{n_{j}^{(1)}} \sum_{i=1}^{n_{j}^{(1)}} \Psi_{\boldsymbol{\beta}}\left(O_{i j}^{(1)}\right)-E\left(\Psi_{\boldsymbol{\beta}}\left(O_{i j}^{(1)}\right)\right)\right\| \stackrel{P}{\rightarrow} 0.
\end{equation*}
Notice that the $O_{ij}^{(1)}$ are i.i.d. for each $j$ separately, ${n_{j}^{(1)}}/{n}<1$, and $j$ is finite, so
\begin{equation*}
    \sup _{\boldsymbol{\beta}}\left\|\frac{n_{j}^{(1)}}{n}  \frac{1}{n_{j}^{(1)}} \sum_{i=1}^{n_{j}^{(1)}} \Psi_{\boldsymbol{\beta}}\left(O_{i j}^{(1)}\right)\right\|
    \leq
    \sup _{\boldsymbol{\beta}}\left\|\frac{1}{n_{j}^{(1)}} \sum_{i=1}^{n_{j}^{(1)}} \Psi_{\boldsymbol{\beta}}\left(O_{i j}^{(1)}\right)\right\|.
\end{equation*}
We conclude that $\sup_{\boldsymbol{\beta}} \left\|G_{1,n}^{(g)}\right\| \xrightarrow{P} 0.$

For the term $G_{2,n}^{(g)}$ from equation (\ref{G2_term_general}), let $Y_{ij}^{(2)}$ be the (counterfactual) outcomes under $\boldsymbol{a}_{j}^{(2)}$ and let $\epsilon_{ij}^{(2)}$ be the corresponding errors that patient $i$ in center $j$ would have experienced under intervention $\boldsymbol{a}_{j}^{(2)}$. We derive
\begin{equation*}\begin{aligned}
G_{2,n}^{(g)} &= \frac{1}{n} \sum_{j=1}^{J^{(2)}} \sum_{i=1}^{n_{j}^{(2)}} \left[ \left( \frac{\partial}{\partial \boldsymbol{\beta}} g^{-1}(\boldsymbol{A}_{j}^{\left(2, n^{(1)}\right)},\boldsymbol{z}_j^{(2)}; \boldsymbol{\beta}) \right) \left( Y_{ij}^{(2, n^{(1)})}-g^{-1}(\boldsymbol{A}_{j}^{\left(2, n^{(1)}\right)}, \boldsymbol{z}_{j}^{\left(2\right)}; \boldsymbol{\beta}^*) \right) \right] \\
&= \frac{1}{n} \sum_{j=1}^{J^{(2)}} \sum_{i=1}^{n_{j}^{(2)}} \left[ \left( \frac{\partial}{\partial \boldsymbol{\beta}} g^{-1}(\boldsymbol{A}_{j}^{\left(2, n^{(1)}\right)},\boldsymbol{z}_j^{(2)}; \boldsymbol{\beta}) \right) \epsilon_{ij}^{(2, n^{(1)})} \right] \\
\end{aligned}\end{equation*}
By Assumption \ref{ind_error} from the main text, replacing $\epsilon_{ij}^{(2,n^{(1)})}$ by the new error terms $\epsilon_{ij}^{(2)}$ does not change the distribution of $G_{2,n}^{(g)}$, so it suffices to show that for
\begin{equation*}\begin{aligned}
\widetilde{G}_{2,n}^{(g)} = \frac{1}{n} \sum_{j=1}^{J^{(2)}} \sum_{i=1}^{n_{j}^{(2)}} \left[ \left( \frac{\partial}{\partial \boldsymbol{\beta}} g^{-1}(\boldsymbol{A}_{j}^{\left(2, n^{(1)}\right)},\boldsymbol{z}_j^{(2)}; \boldsymbol{\beta}) \right) \epsilon_{ij}^{(2)} \right],
\end{aligned}\end{equation*}
$\sup _{\boldsymbol{\beta}}\left\|\widetilde{G}_{2,n}^{(g)}\right\| \stackrel{P}{\rightarrow} 0$.
\begin{equation*}
\begin{aligned}
\widetilde{G}_{2,n}^{(g)}
&= \frac{1}{n} \sum_{j=1}^{J^{(2)}}  \sum_{i=1}^{n_{j}^{(2)}} \left[ \left(\frac{\partial}{\partial \boldsymbol{\beta}} g^{-1}(\boldsymbol{A}_{j}^{\left(2, n^{(1)}\right)},\boldsymbol{z}_j^{(2)}; \boldsymbol{\beta}) \right) \pm \left(\frac{\partial}{\partial \boldsymbol{\beta}} g^{-1}(\boldsymbol{a}_{j}^{(2)}, \boldsymbol{z}_{j}^{(2)}; \boldsymbol{\beta}) \right)\right] \epsilon_{ij}^{(2)}   \\
&= \widetilde{G}_{2\_1,n}^{(g)} + \widetilde{G}_{2\_2,n}^{(g)},
\end{aligned}
\end{equation*}
where
\begin{equation*}
    \widetilde{G}_{2\_1,n}^{(g)} = \frac{1}{n} \sum_{j=1}^{J^{(2)}}  \sum_{i=1}^{n_{j}^{(2)}} \left[\left(\frac{\partial}{\partial \boldsymbol{\beta}} g^{-1}(\boldsymbol{a}_{j}^{(2)}, \boldsymbol{z}_{j}^{(2)}; \boldsymbol{\beta}) \right) \epsilon_{ij}^{(2)}\right],
\end{equation*}
\begin{equation}
\widetilde{G}_{2\_2,n}^{(g)} = \frac{1}{n} \sum_{j=1}^{J^{(2)}}  \sum_{i=1}^{n_{j}^{(2)}} \left[ \left(\frac{\partial}{\partial \boldsymbol{\beta}} g^{-1}(\boldsymbol{A}_{j}^{\left(2, n^{(1)}\right)},\boldsymbol{z}_j^{(2)}; \boldsymbol{\beta}) \right) - \left(\frac{\partial}{\partial \boldsymbol{\beta}} g^{-1}(\boldsymbol{a}_{j}^{(2)}, \boldsymbol{z}_{j}^{(2)}; \boldsymbol{\beta}) \right)\right] \epsilon_{ij}^{(2)}.
\label{g2_2}
\end{equation}
We show that both $\sup _{\boldsymbol{\beta}}\left\|\widetilde{G}_{2\_1,n}^{(g)}\right\| \stackrel{P}{\rightarrow} 0$ and $\sup _{\boldsymbol{\beta}}\left\|\widetilde{G}_{2\_2,n}^{(g)}\right\| \stackrel{P}{\rightarrow} 0$. By the triangle inequality, then $\sup _{\boldsymbol{\beta}}\left\|\widetilde{G}_{2,n}^{(g)}\right\| \stackrel{P}{\rightarrow} 0$.

Let $O^{(2)}_{ij} = (Y^{(2)}_{ij}, \boldsymbol{a}^{(2)}_{j}, \boldsymbol{z}^{(2)}_{j})$ be the counterfactual data for patient $i$ from center $j$ in stage 2 under $\boldsymbol{a}^{(2)}_{j}$ and $\boldsymbol{z}^{(2)}_{j}$ and let
\begin{equation*}
    \Psi_{\boldsymbol{\beta}}^{(2)}(O^{(2)}_{ij}) = \left(\frac{\partial}{\partial \boldsymbol{\beta}} g^{-1}(\boldsymbol{a}_{j}^{(2)}, \boldsymbol{z}_{j}^{(2)}; \boldsymbol{\beta}) \right) \epsilon_{ij}^{(2)}.
\end{equation*}
Following the same argument as for $G_{1,n}^{(g)}$, for fixed value of $j$, the class of functions $\mathcal{F}_2 = \left( \Psi_{\boldsymbol{\beta}}^{(2)}: \boldsymbol{\beta} \in \boldsymbol{\mathcal{B}} \right)$ is a Donsker class and $\sup_{\boldsymbol{\beta}}\|\widetilde{G}_{2\_1,n}^{(g)}\| \xrightarrow{P} 0$.

For each of the $P$ components of $\widetilde{G}_{2\_2,n}^{(g)}$ from equation (\ref{g2_2}), by the Mean Value Theorem, we derive
\begin{equation*}
\begin{aligned}
\widetilde{G}_{2\_2, \;p}^{(g)} = \frac{1}{n} \sum_{j=1}^{J^{(2)}}  \sum_{i=1}^{n_{j}^{(2)}} \left\{ \left[ \left.\frac{\partial}{\partial a_p}\right|_{\Tilde{\boldsymbol{a}}_{jp}(\boldsymbol{\beta})} \left(\frac{\partial}{\partial \boldsymbol{\beta}_p} g^{-1}(\boldsymbol{a}, \boldsymbol{z}_j^{(2)}; \boldsymbol{\beta}) \right) \right] \left(\boldsymbol{A}_{jp}^{(2,n^{(1)})}-\boldsymbol{a}_{jp}^{(2)}\right) \epsilon_{ij}^{(2)} \right\}.
\end{aligned}
\end{equation*}

By Assumption \ref{glm_unif_bound_assu} from the main text and Lemma \ref{sup_glm_condition} from the Appendix, it follows that 
\begin{equation*}
\begin{aligned}
&\sup_{\boldsymbol{\beta}} \left\| \frac{1}{n} \sum_{j=1}^{J^{(2)}}  \sum_{i=1}^{n_{j}^{(2)}} \left[ \left(\frac{\partial}{\partial \boldsymbol{\beta}_p} g^{-1}(\boldsymbol{A}_{j}^{\left(2, n^{(1)}\right)},\boldsymbol{z}_j^{(2)}; \boldsymbol{\beta}) \right) - \left(\frac{\partial}{\partial \boldsymbol{\beta}_p} g^{-1}(\boldsymbol{a}_{j}^{(2)}, \boldsymbol{z}_{j}^{(2)}; \boldsymbol{\beta}) \right) \right] \epsilon_{ij}^{(2)} \right\|  \\
&\leq  \sup_{\boldsymbol{\beta}}\max_{j} {\left\|   \left.\frac{\partial}{\partial a_p}\right|_{\Tilde{\boldsymbol{a}}_{jp}(\boldsymbol{\beta})} \left(\frac{\partial}{\partial \boldsymbol{\beta}_p} g^{-1}(\boldsymbol{a}, \boldsymbol{z}_j^{(2)}; \boldsymbol{\beta}) \right) \right\|} \max_{j}\left\|    \left(\boldsymbol{A}_{jp}^{(2,n^{(1)})}-\boldsymbol{a}_{jp}^{(2)}\right)   \right\| \sup_{ij}{\left\|\epsilon_{ij}^{(2)}\right\|}  \\
&\xrightarrow{P} 0.
\end{aligned}
\end{equation*}
This argument can be applied to all $P$ components of $\widetilde{G}_{2\_2,n}^{(g)}$, so $\sup_{\boldsymbol{\beta}}\|\widetilde{G}_{2\_2,n}^{(g)}\| \xrightarrow{P} 0$.
Hence $\sup _{\boldsymbol{\beta}}\left\|G_{2,n}^{(g)}\right\| \stackrel{P}{\rightarrow} 0$.

Consider $G_{3,n}^{(g)}$ from equation (\ref{G3_term_general}). By the Mean Value Theorem, the supremum over $\boldsymbol{\beta}$ of each of the $P$ components of $G_{3,n}^{(g)}$ is 
\begin{equation*}\begin{aligned}
&\sup_{\boldsymbol{\beta}}\left\| \frac{1}{n} \sum_{j=1}^{J^{(2)}} \sum_{i=1}^{n_{j}^{(2)}} \left[ \left( \frac{\partial}{\partial \boldsymbol{\beta}_p} g^{-1}(\boldsymbol{A}_{j}^{\left(2, n^{(1)}\right)},\boldsymbol{z}_j^{(2)}; \boldsymbol{\beta}) \right) \left( g^{-1}(\boldsymbol{A}_{j}^{\left(2, n^{(1)}\right)}, \boldsymbol{z}_{j}^{\left(2\right)}; \boldsymbol{\beta}^*)  \right) \right.\right.\\
& \qquad \qquad \qquad \qquad \qquad \qquad \qquad - \left. \left. \left( \frac{\partial}{\partial \boldsymbol{\beta}_p} g^{-1}(\boldsymbol{a}_j^{(2)},\boldsymbol{z}_j^{(2)}; \boldsymbol{\beta}) \right) \left( g^{-1}(\boldsymbol{a}_j^{(2)},\boldsymbol{z}_j^{(2)}; \boldsymbol{\beta}^*)  \right)  \right] \rule{0cm}{0.8cm}\right\|\\
&= \sup_{\boldsymbol{\beta}} \left\| \frac{1}{n} \sum_{j=1}^{J^{(2)}} \sum_{i=1}^{n_{j}^{(2)}} \left\{  \left.\frac{\partial}{\partial a_p}\right|_{\Tilde{\boldsymbol{a}}_{jp}(\boldsymbol{\beta})} \left[\left(\frac{\partial}{\partial \boldsymbol{\beta}_p} g^{-1}(\boldsymbol{a}, \boldsymbol{z}_j^{(2)}; \boldsymbol{\beta}) \right)g^{-1}(\boldsymbol{a}, \boldsymbol{z}_j^{(2)}; \boldsymbol{\beta}^*) \right] \right. \right.\\
&\qquad\qquad\qquad\qquad\qquad\qquad\qquad\qquad\qquad\qquad\qquad\qquad
\Biggl.\Biggl.\left(\boldsymbol{A}_{jp}^{(2,n^{(1)})}-\boldsymbol{a}_{jp}^{(2)}\right) \Biggr\} \Biggr\|\\
&\leq \sup_{\boldsymbol{\beta}} \max_{j} {\left\| \left.\frac{\partial}{\partial a_p}\right|_{\Tilde{\boldsymbol{a}}_{jp}(\boldsymbol{\beta})} \left(\frac{\partial}{\partial \boldsymbol{\beta}_p} g^{-1}(\boldsymbol{a}, \boldsymbol{z}_j^{(2)}; \boldsymbol{\beta}) \right)g^{-1}(\boldsymbol{a}, \boldsymbol{z}_j^{(2)}; \boldsymbol{\beta}^*) \right\|} \\
&\qquad\qquad\qquad\qquad\qquad\qquad\qquad\qquad\qquad\qquad\qquad \max_{j}\biggr\|     \left(\boldsymbol{A}_{jp}^{(2,n^{(1)})}-\boldsymbol{a}_{jp}^{(2)}\right)   \biggr\|\\
&\xrightarrow{P} 0.
\end{aligned}
\end{equation*}

The convergence to 0 follows from Assumption \ref{glm_unif_bound_assu} of the main text and Lemma \ref{sup_glm_condition} of the Appendix, so $\sup_{\boldsymbol{\beta}} \left\| G_{3,n}^{(g)} \right\| \xrightarrow{P} 0$. By applying the same argument to $G_{4,n}^{(g)}$ (equation (\ref{G4_term_general})), we conclude that also $\sup _{\boldsymbol{\beta}}\left\|G_{4,n}^{(g)}\right\| \stackrel{P}{\rightarrow} 0$. 

Consider the first term of $G_{5,n}^{(g)}$ (equation (\ref{G5_term_general})), by Assumption \ref{glm_unif_bound_assu} from the main text and $\alpha_{j k}=\lim _{n \rightarrow \infty} n_{j}^{(k)} / n$, for the finite many $j$, it follows that
\begin{equation*}
\begin{aligned}
&\sup_{\boldsymbol{\beta}} \left\| \sum_{j=1}^{J^{(1)}} \left[ \left( \alpha_{j1}- \frac{n_j^{(1)}}{n} \right) \left( \frac{\partial}{\partial \boldsymbol{\beta}} g^{-1}(\boldsymbol{a}_j^{(1)},\boldsymbol{z}_j^{(1)}; \boldsymbol{\beta}) \right) \left( g^{-1}(\boldsymbol{a}_j^{(1)},\boldsymbol{z}_j^{(1)}; \boldsymbol{\beta})  \right) \right] \right\|\\
\leq& \sup_{\boldsymbol{\beta}} \max_{j} \left\|  \left( \alpha_{j1}- \frac{n_j^{(1)}}{n} \right) \left( \frac{\partial}{\partial \boldsymbol{\beta}} g^{-1}(\boldsymbol{a}_j^{(1)},\boldsymbol{z}_j^{(1)}; \boldsymbol{\beta}) \right) \left( g^{-1}(\boldsymbol{a}_j^{(1)},\boldsymbol{z}_j^{(1)}; \boldsymbol{\beta})  \right)  \right\| \\ 
\xrightarrow{P}& \;0.
\end{aligned}
\end{equation*}
By applying the same argument to the other terms of $G_{5,n}^{(g)}$ and the triangle inequality, $\sup _{\boldsymbol{\beta}}\left\|G_{5,n}^{(g)}\right\| \stackrel{P}{\rightarrow} 0$. 

Thus,
\begin{equation*}
\begin{aligned}
    &\sup_{\boldsymbol{\beta}}\left\|\boldsymbol{U}^{(g)}(\boldsymbol{\beta}) - \boldsymbol{U}^{(g)}(\boldsymbol{\beta})\right\| \\
    &\leq \sup_{\boldsymbol{\beta}}\left\|G_{1,n}^{(g)} \right\|+\sup_{\boldsymbol{\beta}}\left\|G_{2,n}^{(g)} \right\|+ \sup_{\boldsymbol{\beta}}\left\|G_{3,n}^{(g)} \right\| + \sup_{\boldsymbol{\beta}}\left\|G_{4,n}^{(g)} \right\| + 
    \sup_{\boldsymbol{\beta}}\left\|G_{5,n}^{(g)} \right\|\\
    &\xrightarrow{P} 0.
\end{aligned}
\end{equation*}

The second condition in Theorem 5.9 of \citet{van2000asymptotic} is 
\begin{equation*}
\inf _{\boldsymbol{\beta}:\left\|\boldsymbol{\beta}-\boldsymbol{\beta}^{\star}\right\|>0}\|\boldsymbol{u}^{(g)}(\boldsymbol{\beta})\|>0=\left\|\boldsymbol{u}^{(g)}\left(\boldsymbol{\beta}^{\star}\right)\right\|.
\end{equation*}
Equation (\ref{general_little_u}) implies that $\left\|\boldsymbol{u}^{(g)}\left(\boldsymbol{\beta}^{\star}\right)\right\| = 0$. Furthermore, $\boldsymbol{u}^{(g)}\left(\boldsymbol{\beta}\right)$ is the same as for a fixed two-stage design with $a_1^{(1)},\cdots, a_{J^{(1)}}^{(1)}, a_1^{(2)},\cdots, a_{J^{(2)}}^{(2)}$ as interventions decided on before the trial, 
so regular GEE theory applies here. In order for LAGO to work properly, we need variations in the intervention components to identify the treatment effect parameter. The uniqueness of $\boldsymbol{\beta}^*$ as a maximizer or zero has been studied by various authors, see e.g. Chapter 2.2 of \cite{fahrmeir2013multivariate}. Thus, provided there is enough variation in the intervention, the second condition in Theorem 5.9 of \citet{van2000asymptotic} is often also satisfied and we conclude that $\hat{\boldsymbol{\beta}}$ is consistent.

\subsection{Proof of Theorem \ref{AN_thm} for LAGO GLM with general link function: asymptotic normality of \texorpdfstring{$\hat{\boldsymbol{\beta}}$}{Lg}}  \label{an_appendix_generallink}
To prove asymptotic normality of the final estimator $\hat{\boldsymbol{\beta}}$, we first prove

\begin{lem}
Under Assumption \ref{intervention_cvg_assumption},
there exist $\boldsymbol{a}_{j}^{(2)}$, which is the probability limit of $\boldsymbol{A}_{j}^{(2,n^{(1)})}$ as $n^{(1)} \rightarrow \infty$ such that for stage $2$, $\max_{j = 1, \cdots , J^{(2)}} \left\| \boldsymbol{A}_{j}^{(2,n^{(1)})} - \boldsymbol{a}_{j}^{(2)} \right\| \xrightarrow{P} 0$.
\label{sup_glm_condition}
\end{lem}

\begin{proof}[Proof]
By Assumption \ref{intervention_cvg_assumption} and the Continuous Mapping Theorem, 
for each $j$, $\boldsymbol{A}_{j}^{(2,n^{(1)})} \xrightarrow{P} \boldsymbol{a}_{j}^{(2)}$. Since the number of centers in stage $2$ is fixed, taking the maximum over all $j$ is a continuous operation. Lemma \ref{sup_glm_condition} follows by the Continuous Mapping Theorem.
\end{proof}

Under Assumptions \ref{glm_model_assumption} $-$ \ref{ind_error}, and Lemma \ref{sup_glm_condition}, we show that 
\begin{equation}\label{aympstotic_normality_general}
\footnotesize{
    \sqrt{n}\left(\hat{\boldsymbol{\beta}}-\boldsymbol{\beta}^{*}\right) \xrightarrow{D} N(0, J\left(\boldsymbol{\beta^{*}}\right)^{-1} V\left(\boldsymbol{\beta^{*}}\right) J\left(\boldsymbol{\beta^{*}}\right)^{-1}),
}
\end{equation}
where the explicit forms of $J\left(\boldsymbol{\beta^{*}}\right)$ and $V\left(\boldsymbol{\beta^{*}}\right)$ are
given in Theorem \ref{AN_thm}.
The corresponding estimators are
\begin{equation}\label{estimator J}
\footnotesize{
\begin{aligned}
\hat{J}(\hat{\boldsymbol{\beta}}) &= \sum_{j=1}^{J^{(1)}} \frac{n_j^{(1)}}{n} \left.\left( \frac{\partial}{\partial \boldsymbol{\beta}} \right|_{\hat{\boldsymbol{\beta}}} g^{-1}(\boldsymbol{a}_j^{(1)},\boldsymbol{z}_j^{(1)}; \boldsymbol{\beta})\right)^{\otimes2}   + \sum_{j=1}^{J^{(2)}} \frac{n_j^{(2)}}{n} \left.\left( \frac{\partial}{\partial \boldsymbol{\beta}} \right|_{\hat{\boldsymbol{\beta}}}  g^{-1}(\boldsymbol{A}_{j}^{\left(2, n^{(1)}\right)},\boldsymbol{z}_j^{(2)}; \boldsymbol{\beta})\right)^{\otimes2},
\end{aligned}
}
\end{equation}
\begin{equation}\label{estimator V}
\footnotesize{
\begin{aligned}
\hat{V}(\hat{\boldsymbol{\beta}}) &= \frac{1}{n} \sum_{j=1}^{J^{(1)}} \left\{  \left.\left( \frac{\partial}{\partial \boldsymbol{\beta}} \right|_{\hat{\boldsymbol{\beta}}} g^{-1}(\boldsymbol{a}_j^{(1)},\boldsymbol{z}_j^{(1)}; \boldsymbol{\beta})\right)^{\otimes2} \sum_{i=1}^{n_j^{(1)}}  \left( Y_{ij}^{(1)} - g^{-1}(\boldsymbol{a}_j^{(1)},\boldsymbol{z}_j^{(1)}; \hat{\boldsymbol{\beta}}) \right)^2 \right\}\\
&+ \frac{1}{n} \sum_{j=1}^{J^{(2)}} \left\{  \left.\left( \frac{\partial}{\partial \boldsymbol{\beta}} \right|_{\hat{\boldsymbol{\beta}}} g^{-1}(\boldsymbol{A}_j^{(2, n^{(1)})},\boldsymbol{z}_j^{(2)}; \boldsymbol{\beta})\right)^{\otimes2} \sum_{i=1}^{n_j^{(2)}}  \left( Y_{ij}^{(2, n^{(1)})} - g^{-1}(\boldsymbol{A}_j^{(2,n^{(1)})},\boldsymbol{z}_j^{(2)}; \hat{\boldsymbol{\beta}}) \right)^2 \right\}.
\end{aligned}
}
\end{equation}

Applying the Mean Value Theorem to each component of $\boldsymbol{U}^{(g)}(\boldsymbol{\beta})$ from equation (\ref{glm_EE}), 
\begin{equation*}
\footnotesize{
    0=\boldsymbol{U}^{(g)}(\hat{\boldsymbol{\beta}}) = \boldsymbol{U}^{(g)}(\boldsymbol{\beta}^*) + \left(\left.\frac{\partial}{\partial \boldsymbol{\beta}}\right|_{\Tilde{\boldsymbol{\beta}}} \boldsymbol{U}^{(g)}({\boldsymbol{\beta}})\right) (\hat{\boldsymbol{\beta}} - {\boldsymbol{\beta}^*})^T,
}
\end{equation*}
where for each row of $\left.\frac{\partial}{\partial \boldsymbol{\beta}}\right|_{\Tilde{\boldsymbol{\beta}}} \boldsymbol{U}^{(g)}({\boldsymbol{\beta}})$, $\Tilde{\boldsymbol{\beta}}$ takes a possibly row-dependent value on the line between $\hat{\boldsymbol{\beta}}$ and $\boldsymbol{\beta}^*$. It follows that 
\begin{equation}
\footnotesize{
\sqrt{n}\left(\hat{\boldsymbol{\beta}}-\boldsymbol{\beta}^{*}\right) = -\sqrt{n}\left(\left.\frac{\partial}{\partial \boldsymbol{\beta}}\right|_{\tilde{\boldsymbol{\beta}}} \boldsymbol{U}^{(g)}\left({\boldsymbol{\beta}}\right)\right)^{-1} \boldsymbol{U}^{(g)}\left(\boldsymbol{\beta}^{*}\right).
\label{ubeta*forderivation}
}
\end{equation}
We first show that $\left(-\left.\frac{\partial}{\partial \boldsymbol{\beta}}\right|_{\tilde{\boldsymbol{\beta}}} \boldsymbol{U}^{(g)}\left({\boldsymbol{\beta}}\right)\right)$ converges in probability to $J\left(\boldsymbol{\beta^{*}}\right)$. Then we show that 
$\sqrt{n} \; \boldsymbol{U}^{(g)}\left(\boldsymbol{\beta}^{*}\right)$ converges to a normal distribution with mean 0 and variance $V\left(\boldsymbol{\beta^{*}}\right)$. Equation (\ref{aympstotic_normality_general}) then follows from Slutsky's Theorem.

Because of equation (\ref{glm_EE}) of the main text,
\begin{equation}\label{ddbeta Ubeta}
\footnotesize{
\begin{aligned}
&\frac{\partial}{\partial \boldsymbol{\beta}} \boldsymbol{U}^{(g)}({\boldsymbol{\beta}})
= \frac{1}{n} \sum_{j=1}^{J^{(1)}}  \sum_{i=1}^{n_{j}^{(1)}}  \frac{\partial}{\partial \boldsymbol{\beta}} \left[ \left(\frac{\partial}{\partial \boldsymbol{\beta}} g^{-1}(\boldsymbol{a}_j^{(1)},\boldsymbol{z}_j^{(1)}; \boldsymbol{\beta})\right)\left(Y_{ij}^{(1)}-g^{-1}(\boldsymbol{a}_j^{(1)}, \boldsymbol{z}_j^{(1)} ; \boldsymbol{\beta})\right)  \right]\\
&\qquad\qquad+ \frac{1}{n} \sum_{j=1}^{J^{(2)}}  \sum_{i=1}^{n_{j}^{(2)}} \frac{\partial}{\partial \boldsymbol{\beta}} \left[ \left(\frac{\partial}{\partial \boldsymbol{\beta}} g^{-1}(\boldsymbol{A}_{j}^{\left(2, n^{(1)}\right)},\boldsymbol{z}_j^{(2)}; \boldsymbol{\beta}) \right)\left(Y_{ij}^{\left(2, n^{(1)}\right)}-g^{-1}(\boldsymbol{A}_{j}^{\left(2, n^{(1)}\right)},\boldsymbol{z}_j^{(2)}; \boldsymbol{\beta})\right) \right] \\
&= G_{6,n}^{(g)}(\boldsymbol{\beta}) + G_{7,n}^{(g)}(\boldsymbol{\beta}),
\end{aligned}
}
\end{equation}
where
\begin{equation*}
\footnotesize{
\begin{aligned}
G_{6,n}^{(g)}(\boldsymbol{\beta}) &= \frac{1}{n} \sum_{j=1}^{J^{(1)}}  \sum_{i=1}^{n_{j}^{(1)}} \left\{ \left(\frac{\partial^2}{\partial \boldsymbol{\beta}^2} g^{-1}(\boldsymbol{a}_j^{(1)},\boldsymbol{z}_j^{(1)}; \boldsymbol{\beta})\right) Y_{ij}^{(1)} \right.\\
&\left. - \left( \frac{\partial^2}{\partial \boldsymbol{\beta}^2} g^{-1}(\boldsymbol{a}_j^{(1)},\boldsymbol{z}_j^{(1)}; \boldsymbol{\beta})\right) g^{-1}(\boldsymbol{a}_j^{(1)},\boldsymbol{z}_j^{(1)}; \boldsymbol{\beta}) - \left( \frac{\partial}{\partial \boldsymbol{\beta}} g^{-1}(\boldsymbol{a}_j^{(1)},\boldsymbol{z}_j^{(1)}; \boldsymbol{\beta})\right)^{\otimes2} \right\}, \\
G_{7,n}^{(g)}(\boldsymbol{\beta}) &= \frac{1}{n} \sum_{j=1}^{J^{(1)}}  \sum_{i=1}^{n_{j}^{(2)}} \left\{ \left(\frac{\partial^2}{\partial \boldsymbol{\beta}^2} g^{-1}(\boldsymbol{A}_{j}^{\left(2, n^{(1)}\right)},\boldsymbol{z}_j^{(2)}; \boldsymbol{\beta}) \right) Y_{ij}^{(2,n^{(1)})} \right.\\
&\left.- \left( \frac{\partial^2}{\partial \boldsymbol{\beta}^2} g^{-1}(\boldsymbol{A}_{j}^{\left(2, n^{(1)}\right)},\boldsymbol{z}_j^{(2)}; \boldsymbol{\beta}) \right) g^{-1}(\boldsymbol{A}_{j}^{\left(2, n^{(1)}\right)},\boldsymbol{z}_j^{(2)}; \boldsymbol{\beta}) \right.\\
&\qquad\qquad\qquad\qquad\qquad\qquad\qquad - \left.\left( \frac{\partial}{\partial \boldsymbol{\beta}} g^{-1}(\boldsymbol{A}_{j}^{\left(2, n^{(1)}\right)},\boldsymbol{z}_j^{(2)}; \boldsymbol{\beta}) \right)^{\otimes2} \right\}.
\end{aligned}
}
\end{equation*}
Let 
\begin{equation*}
\footnotesize{
\begin{aligned}
G_{6,n}^{(g)*}(\boldsymbol{\beta}) &= \frac{1}{n} \sum_{j=1}^{J^{(1)}}  \sum_{i=1}^{n_{j}^{(1)}} \left\{ \left(\frac{\partial^2}{\partial \boldsymbol{\beta}^2} g^{-1}(\boldsymbol{a}_j^{(1)},\boldsymbol{z}_j^{(1)}; \boldsymbol{\beta})\right) g^{-1}(\boldsymbol{a}_j^{(1)},\boldsymbol{z}_j^{(1)}; \boldsymbol{\beta}^*) \right.\\
&\left.- \left( \frac{\partial^2}{\partial\boldsymbol{\beta}^2} g^{-1}(\boldsymbol{a}_j^{(1)},\boldsymbol{z}_j^{(1)}; \boldsymbol{\beta})\right) g^{-1}(\boldsymbol{a}_j^{(1)},\boldsymbol{z}_j^{(1)}; \boldsymbol{\beta}) 
- \left( \frac{\partial}{\partial \boldsymbol{\beta}} g^{-1}(\boldsymbol{a}_j^{(1)},\boldsymbol{z}_j^{(1)}; \boldsymbol{\beta})\right)^{\otimes2} \right\},
\end{aligned}
}
\end{equation*}
\begin{equation*}
\small{
\begin{aligned}
G_{7,n}^{(g)*}(\boldsymbol{\beta}) &= \frac{1}{n} \sum_{j=1}^{J^{(1)}}  \sum_{i=1}^{n_{j}^{(2)}} \left\{ \left(\frac{\partial^2}{\partial \boldsymbol{\beta}^2} g^{-1}(\boldsymbol{a}_j^{(2)},\boldsymbol{z}_j^{(2)}; \boldsymbol{\beta}) \right) g^{-1}(\boldsymbol{a}_j^{(2)},\boldsymbol{z}_j^{(2)}; \boldsymbol{\beta}^*) \right.\\
&\left.- \left( \frac{\partial^2}{\partial \boldsymbol{\beta}^2} g^{-1}(\boldsymbol{a}_j^{(2)},\boldsymbol{z}_j^{(2)}; \boldsymbol{\beta}) \right) g^{-1}(\boldsymbol{a}_j^{(2)},\boldsymbol{z}_j^{(2)}; \boldsymbol{\beta}) - \left( \frac{\partial}{\partial \boldsymbol{\beta}} g^{-1}(\boldsymbol{a}_j^{(2)},\boldsymbol{z}_j^{(2)}; \boldsymbol{\beta}) \right)^{\otimes2} \right\}.
\end{aligned}
}
\end{equation*}
We show that $\frac{\partial}{\partial \boldsymbol{\beta}} \boldsymbol{U}^{(g)}(\tilde{\boldsymbol{\beta}})-\left(G_{6, n}^{(g)*}\left(\boldsymbol{\beta}^{*}\right)+G_{7, n}^{(g)*}\left(\boldsymbol{\beta}^{*}\right)\right)$ converges in probability to 0. First,
\begin{equation*}
\footnotesize{
\left.\frac{\partial}{\partial \boldsymbol{\beta}}\right|_{\tilde{\boldsymbol{\beta}}} \boldsymbol{U}^{(g)}(\tilde{\boldsymbol{{\beta}}})-\left(G_{6, n}^{(g)*}\left(\boldsymbol{\beta}^{*}\right)+G_{7, n}^{(g)*}\left(\boldsymbol{\beta}^{*}\right)\right)=\left(G_{6,n}^{(g)}(\tilde{\boldsymbol{\beta}})-G_{6, n}^{(g)*}\left(\boldsymbol{\beta}^{*}\right)\right)+\left(G_{7,n}^{(g)}(\tilde{\boldsymbol{\beta}})-G_{7, n}^{(g)*}\left(\boldsymbol{\beta}^{*}\right)\right).
}
\end{equation*}
We show that both terms converge in probability to 0.
For the first term, using $Y_{ij}^{(1)} = g^{-1}(\boldsymbol{a}_{j}^{\left(1\right)}, \boldsymbol{z}_{j}^{\left(1\right)}; {\boldsymbol{\beta}^*}) + \epsilon_{ij}^{(1)} $,
\begin{equation}
\footnotesize{
\begin{aligned}
&G_{6,n}^{(g)}(\tilde{\boldsymbol{\beta}})-G_{6,n}^{(g)*}\left(\boldsymbol{\beta}^{*}\right) \\
=& \frac{1}{n} \sum_{j=1}^{J^{(1)}}  \sum_{i=1}^{n_{j}^{(1)}} \left\{ \left( \left.\frac{\partial^2}{\partial {\boldsymbol{\beta}}^2}\right|_{\tilde{\boldsymbol{\beta}}} g^{-1}(\boldsymbol{a}_{j}^{\left(1\right)}, \boldsymbol{z}_{j}^{\left(1\right)}; {\boldsymbol{\beta}}) - \left.\frac{\partial^2}{\partial \boldsymbol{\beta}^2}\right|_{{\boldsymbol{\beta}^*}} g^{-1}(\boldsymbol{a}_{j}^{\left(1\right)}, \boldsymbol{z}_{j}^{\left(1\right)}; \boldsymbol{\beta}) \right) g^{-1}(\boldsymbol{a}_{j}^{\left(1\right)}, \boldsymbol{z}_{j}^{\left(1\right)}; \boldsymbol{\beta}^*) \right.\\
& \left. + \left( \left.\frac{\partial^2}{\partial {\boldsymbol{\beta}^2}}\right|_{\tilde{\boldsymbol{\beta}}} g^{-1}(\boldsymbol{a}_{j}^{\left(1\right)}, \boldsymbol{z}_{j}^{\left(1\right)}; {\boldsymbol{\beta}}) \right) \;\epsilon_{ij}^{(1)} \right.
+ \left(\left.\frac{\partial^2}{\partial \boldsymbol{\beta}^2}\right|_{{\boldsymbol{\beta}^*}} g^{-1}(\boldsymbol{a}_{j}^{\left(1\right)}, \boldsymbol{z}_{j}^{\left(1\right)}; \boldsymbol{\beta})\right) g^{-1}(\boldsymbol{a}_{j}^{\left(1\right)}, \boldsymbol{z}_{j}^{\left(1\right)}; \boldsymbol{\beta}^*) \\
&-\left(\left.\frac{\partial^2}{\partial {\boldsymbol{\beta}^2}}\right|_{\tilde{\boldsymbol{\beta}}} g^{-1}(\boldsymbol{a}_{j}^{\left(1\right)}, \boldsymbol{z}_{j}^{\left(1\right)}; {\boldsymbol{\beta}})\right) g^{-1}(\boldsymbol{a}_{j}^{\left(1\right)}, \boldsymbol{z}_{j}^{\left(1\right)}; \tilde{\boldsymbol{\beta}}) \\
&+ \left.\left(\left.\frac{\partial}{\partial \boldsymbol{\beta}}\right|_{{\boldsymbol{\beta}^*}} g^{-1}(\boldsymbol{a}_{j}^{\left(1\right)}, \boldsymbol{z}_{j}^{\left(1\right)}; \boldsymbol{\beta})\right)^{\otimes2} - 
\left(\left.\frac{\partial}{\partial {\boldsymbol{\beta}}}\right|_{\tilde{\boldsymbol{\beta}}} g^{-1}(\boldsymbol{a}_{j}^{\left(1\right)}, \boldsymbol{z}_{j}^{\left(1\right)}; {\boldsymbol{\beta}})\right)^{\otimes2}\right\}.
\end{aligned}
}
\label{g5tilda-g5star}
\end{equation}
Notice that the $pq$-th entry of the first term of equation (\ref{g5tilda-g5star})
equals
\begin{equation}
\footnotesize{
\begin{aligned}
    &\sum_{j=1}^{J^{(1)}} \frac{n_{j}^{(1)}}{n} \left\{\left(\left.\frac{\partial}{\partial {\boldsymbol{\beta}}_{p}}\right|_{\tilde{\boldsymbol{\beta}}_{p}}
    \left.\frac{\partial}{\partial {\boldsymbol{\beta}}_{q}}\right|_{\tilde{\boldsymbol{\beta}}_{q}} g^{-1}(\boldsymbol{a}_{j}^{\left(1\right)}, \boldsymbol{z}_{j}^{\left(1\right)}; {\boldsymbol{\beta}})\right) 
    g^{-1}(\boldsymbol{a}_{j}^{\left(1\right)}, \boldsymbol{z}_{j}^{\left(1\right)}; \boldsymbol{\beta}^*)\right.\\
    &\qquad\qquad\qquad -\left.\left( \left.\frac{\partial}{\partial {\boldsymbol{\beta}}_{p}}\right|_{{\boldsymbol{\beta}}^*_{p}}
    \left.\frac{\partial}{\partial {\boldsymbol{\beta}}_{q}}\right|_{{\boldsymbol{\beta}}^*_{q}} g^{-1}(\boldsymbol{a}_{j}^{\left(1\right)}, \boldsymbol{z}_{j}^{\left(1\right)}; \boldsymbol{\beta}) \right) g^{-1}(\boldsymbol{a}_{j}^{\left(1\right)}, \boldsymbol{z}_{j}^{\left(1\right)}; \boldsymbol{\beta}^*)\right\}.
\end{aligned}
}
\label{first part g5-g5*}
\end{equation}
Since $\hat{\boldsymbol{\beta}}$ and therefore $\tilde{\boldsymbol{\beta}}$ is consistent, $\tilde{\boldsymbol{\beta}} \xrightarrow{P} \boldsymbol{\beta}^*$. By Assumption \ref{glm_unif_bound_assu} and the Continuous Mapping Theorem, equation (\ref{first part g5-g5*}) converges in probability to 0 for each entry. Since $J^{(1)}$ is fixed, 
$g()$ is in $C^2$,
and ${n_{j}^{(1)}}/{n} \leq 1$, we conclude that equation (\ref{g5tilda-g5star}) goes to 0 in probability. 
Next, following the same arguments as for ${G_{1,n}^{(g)}}$ in the proof of consistency (Appendix \ref{consistency_appendix_generallink}),
\begin{equation*}
\footnotesize{
    \frac{1}{n} \sum_{j=1}^{J^{(1)}}  \sum_{i=1}^{n_{j}^{(1)}} \left( \left.\frac{\partial^2}{\partial {\boldsymbol{\beta}}^2} \right|_{\tilde{\boldsymbol{\beta}}} g^{-1}(\boldsymbol{a}_{j}^{\left(1\right)}, \boldsymbol{z}_{j}^{\left(1\right)}; {\boldsymbol{\beta}}) \right) \;\epsilon_{ij}^{(1)} \xrightarrow{P} 0.
}
\end{equation*}
We conclude that indeed $G_{6,n}^{(g)}(\tilde{\boldsymbol{\beta}})-G_{6, n}^{(g)*}\left(\boldsymbol{\beta}^{*}\right) \stackrel{P}{\rightarrow} 0$.

For {$G_{7,n}^{(g)}(\tilde{\boldsymbol{\beta}})-G_{7, n}^{(g)*}\left(\boldsymbol{\beta}^{*}\right)$}, 
we first show {$\sup_{\boldsymbol{\beta}}\left\|  G_{7,n}^{(g)}({\boldsymbol{\beta}})-G_{7,n}^{(g)*}\left(\boldsymbol{\beta}\right) \right\| \xrightarrow{P} 0$}, after which we show that {$\left(G_{7,n}^{(g)}(\tilde{\boldsymbol{\beta}})-G_{7, n}^{(g)*}\left(\boldsymbol{\beta}^{*}\right)\right) \xrightarrow{P} 0.$}\\
Using that $Y_{ij}^{(2,n^{(1)})} = g^{-1}(\boldsymbol{A}_{j}^{\left(2,n^{(1)}\right)}, \boldsymbol{z}_{j}^{\left(2\right)}; {\boldsymbol{\beta}^*}) + \epsilon_{ij}^{(2,n^{(1)})} $, from equation (\ref{ddbeta Ubeta})
\begin{equation*}
\footnotesize{
\begin{aligned}
&G_{7,n}^{(g)}(\boldsymbol{\beta}) - G_{7,n}^{(g)^*}(\boldsymbol{\beta}) 
= G_{7\_1,n}^{(g)}(\boldsymbol{\beta}) + G_{7\_2,n}^{(g)}(\boldsymbol{\beta}) + G_{7\_3,n}^{(g)}(\boldsymbol{\beta}) + G_{7\_4,n}^{(g)}(\boldsymbol{\beta})
\end{aligned}
}
\end{equation*}
where
\begin{equation*}
\footnotesize{
G_{7\_1,n}^{(g)}(\boldsymbol{\beta}) 
= \frac{1}{n} \sum_{j=1}^{J^{(2)}}  \sum_{i=1}^{n_{j}^{(2)}} \left\{ \left(\frac{\partial^2}{\partial \boldsymbol{\beta}^2} g^{-1}(\boldsymbol{A}_{j}^{\left(2, n^{(1)}\right)},\boldsymbol{z}_j^{(2)}; \boldsymbol{\beta}) \right) \epsilon_{ij}^{(2,n^{(1)})} \right\},
}
\end{equation*}
\begin{equation*}
\footnotesize{
\begin{aligned}
G_{7\_2,n}^{(g)}(\boldsymbol{\beta}) &= \frac{1}{n} \sum_{j=1}^{J^{(2)}}  \sum_{i=1}^{n_{j}^{(2)}} \left\{ \left(\frac{\partial^2}{\partial \boldsymbol{\beta}^2} g^{-1}(\boldsymbol{A}_{j}^{\left(2, n^{(1)}\right)},\boldsymbol{z}_j^{(2)}; \boldsymbol{\beta}) \right) g^{-1}(\boldsymbol{A}_{j}^{\left(2, n^{(1)}\right)}, \boldsymbol{z}_{j}^{\left(2\right)}; \boldsymbol{\beta}^*) \right. \\
& \qquad \qquad \qquad \qquad \qquad  \left. - \left(\frac{\partial^2}{\partial \boldsymbol{\beta}^2} g^{-1}(\boldsymbol{a}_{j}^{(2)}, \boldsymbol{z}_{j}^{(2)}; \boldsymbol{\beta}) \right) g^{-1}(\boldsymbol{a}_{j}^{(2)}, \boldsymbol{z}_j^{(2)}; \boldsymbol{\beta}^*) \right\},\\
G_{7\_3,n}^{(g)}(\boldsymbol{\beta}) &= \frac{1}{n} \sum_{j=1}^{J^{(2)}}  \sum_{i=1}^{n_{j}^{(2)}} \left\{ 
\left(\frac{\partial^2}{\partial \boldsymbol{\beta}^2} g^{-1}(\boldsymbol{a}_{j}^{(2)}, \boldsymbol{z}_{j}^{(2)}; \boldsymbol{\beta}) \right) g^{-1}(\boldsymbol{a}_{j}^{(2)}, \boldsymbol{z}_{j}^{(2)}; \boldsymbol{\beta}) \right.\\ 
&\qquad \qquad \qquad \qquad \qquad\left. - \left(\frac{\partial^2}{\partial \boldsymbol{\beta}^2} g^{-1}(\boldsymbol{A}_{j}^{\left(2, n^{(1)}\right)},\boldsymbol{z}_j^{(2)}; \boldsymbol{\beta}) \right) g^{-1}(\boldsymbol{A}_{j}^{\left(2, n^{(1)}\right)},\boldsymbol{z}_j^{(2)}; \boldsymbol{\beta}) 
\right\},\\
G_{7\_4,n}^{(g)}(\boldsymbol{\beta}) &= \frac{1}{n} \sum_{j=1}^{J^{(2)}}  \sum_{i=1}^{n_{j}^{(2)}}
\left\{ 
\left(\frac{\partial}{\partial \boldsymbol{\beta}} g^{-1}(\boldsymbol{a}_{j}^{(2)}, \boldsymbol{z}_{j}^{(2)}; \boldsymbol{\beta}) \right)^{\otimes2}
- \left(\frac{\partial}{\partial \boldsymbol{\beta}} g^{-1}(\boldsymbol{A}_{j}^{\left(2, n^{(1)}\right)},\boldsymbol{z}_j^{(2)}; \boldsymbol{\beta}) \right)^{\otimes2}
\right\}.
\end{aligned}
}
\end{equation*}
Following the same arguments as for $G_{2,n}^{(g)}$, $G_{3,n}^{(g)}$, and $G_{4,n}^{(g)}$ in the proof of consistency (Appendix \ref{consistency_appendix_generallink}), respectively, we conclude that suprema over $\boldsymbol{\beta}$ of $G_{7\_1,n}^{(g)}(\boldsymbol{\beta})$, $G_{7\_2,n}^{(g)}(\boldsymbol{\beta})$, $G_{7\_3,n}^{(g)}(\boldsymbol{\beta})$ and $G_{7\_4,n}^{(g)}(\boldsymbol{\beta})$ converge to 0 in probability. Next, notice that
\begin{equation*}
\footnotesize{
\begin{aligned}
G_{7,n}^{(g)}(\tilde{\boldsymbol{\beta}}) - G_{7,n}^{(g)^*}(\boldsymbol{\beta}^*) 
&=\left(G_{7,n}^{(g)}(\tilde{\boldsymbol{\beta}})-G_{7,n}^{(g)^*}(\tilde{\boldsymbol{\beta}})\right)+\left(G_{7, n}^{(g)^*}(\tilde{\boldsymbol{\beta}})-G_{7, n}^{(g)^*}\left(\boldsymbol{\beta}^{*}\right)\right).
\end{aligned}
}
\end{equation*}
Since $\sup _{\boldsymbol{\beta}}\left\|G_{7,n}^{(g)}(\boldsymbol{\beta})-G_{7,n}^{(g)^*}(\boldsymbol{\beta})\right\| \stackrel{P}{\rightarrow} 0$, it follows that $G_{7,n}^{(g)}(\tilde{\boldsymbol{\beta}})-G_{7,n}^{(g)^*}(\tilde{\boldsymbol{\beta}}) \stackrel{P}{\rightarrow} 0 .$ Similar to equation (\ref{g5tilda-g5star}), we apply the Mean Value Theorem to conclude that $G_{7,n}^{(g)^*}(\tilde{\boldsymbol{\beta}})-G_{7, n}^{(g)^*}\left(\boldsymbol{\beta}^{*}\right) \stackrel{P}{\rightarrow} 0$. Combining, we conclude $G_{7,n}^{(g)}(\tilde{\boldsymbol{\beta}})-G_{7, n}^{(g)^*}\left(\boldsymbol{\beta}^{*}\right) \stackrel{P}{\rightarrow} 0$.

It follows that $\left.\frac{\partial}{\partial \boldsymbol{\beta}}\right|_{\tilde{\boldsymbol{\beta}}} \boldsymbol{U}^{(g)}({\boldsymbol{{\beta}}}) - \left( G_{6,n}^{(g)^*}(\boldsymbol{\beta}^*)+G_{7,n}^{(g)^*}(\boldsymbol{\beta}^*) \right) \xrightarrow{P} 0$. Thus, calculating the limit of $-\left.\frac{\partial}{\partial \boldsymbol{\beta}}\right|_{\tilde{\boldsymbol{\beta}}} \boldsymbol{U}^{(g)}({\boldsymbol{{\beta}}})$ is equivalent to calculating the limit of $- \left( G_{6,n}^{(g)^*}(\boldsymbol{\beta}^*)+G_{7,n}^{(g)^*}(\boldsymbol{\beta}^*) \right)$. We therefore show that $-\left.\frac{\partial}{\partial \boldsymbol{\beta}}\right|_{\tilde{\boldsymbol{\beta}}} \boldsymbol{U}^{(g)}({\boldsymbol{{\beta}}})$ converges in probability to $J(\boldsymbol{\beta}^*)$ from Theorem \ref{AN_thm}.

\begin{equation*}
\footnotesize{
\begin{aligned}
&G_{6,n}^{(g)^*}(\boldsymbol{\beta}^*) - \sum_{j=1}^{J^{(1)}} \alpha_{j1} \left\{-\left.\left( \frac{\partial}{\partial \boldsymbol{\beta}}\right|_{\boldsymbol{\beta}^*} g^{-1}(\boldsymbol{a}_j^{(1)},\boldsymbol{z}_j^{(1)}; \boldsymbol{\beta})\right)^{\otimes2}  \right\}\\
&= \frac{1}{n} \sum_{j=1}^{J^{(1)}}  \sum_{i=1}^{n_{j}^{(1)}} \left\{ 
- \left.\left( \frac{\partial}{\partial \boldsymbol{\beta}} \right|_{\boldsymbol{\beta}^*} g^{-1}(\boldsymbol{a}_j^{(1)},\boldsymbol{z}_j^{(1)}; \boldsymbol{\beta})\right)^{\otimes2} \right\} 
+ \sum_{j=1}^{J^{(1)}} \alpha_{j1} \left\{\left.\left( \frac{\partial}{\partial \boldsymbol{\beta}} \right|_{\boldsymbol{\beta}^*} g^{-1}(\boldsymbol{a}_j^{(1)},\boldsymbol{z}_j^{(1)}; \boldsymbol{\beta})\right)^{\otimes2}  \right\}\\
&=  \sum_{j=1}^{J^{(1)}} \frac{n_{j}^{(1)}}{n} \left\{ - \left.\left( \frac{\partial}{\partial \boldsymbol{\beta}}\right|_{\boldsymbol{\beta}^*} g^{-1}(\boldsymbol{a}_j^{(1)},\boldsymbol{z}_j^{(1)}; \boldsymbol{\beta})\right)^{\otimes2}\right\} - \sum_{j=1}^{J^{(1)}} \alpha_{j1} \left\{-\left.\left( \frac{\partial}{\partial \boldsymbol{\beta}}\right|_{\boldsymbol{\beta}^*}  g^{-1}(\boldsymbol{a}_j^{(1)},\boldsymbol{z}_j^{(1)}; \boldsymbol{\beta})\right)^{\otimes2} \right\} \\
&\xrightarrow{P} 0.
\end{aligned}
}
\end{equation*}
Similarly, 
\begin{equation*}
\footnotesize{
    G_{7,n}^{(g)^*}(\boldsymbol{\beta}^*) - \sum_{j=1}^{J^{(1)}} \alpha_{j2} \left\{-\left.\left( \frac{\partial}{\partial \boldsymbol{\beta}}\right|_{\boldsymbol{\beta}^*}  g^{-1}(\boldsymbol{a}_j^{(2)},\boldsymbol{z}_j^{(2)}; \boldsymbol{\beta})\right)^{\otimes2} \right\} \xrightarrow{P} 0.
}
\end{equation*}
We conclude that 
\begin{equation}\label{Jbeta proof done}
\footnotesize{
    \left.\frac{\partial}{\partial \boldsymbol{\beta}}\right|_{\tilde{\boldsymbol{\beta}}} \boldsymbol{U}^{(g)}({\boldsymbol{{\beta}}}) 
    \xrightarrow{P}
    J(\boldsymbol{\beta}^*).
}
\end{equation}

To show that $\sqrt{n} \:\: \boldsymbol{U}\left(\boldsymbol{\beta}^{*}\right)$ from equation (\ref{ubeta*forderivation}) converges in distribution to $N\left(0, V\left(\boldsymbol{\beta}^{*}\right)\right)$, we derive
\begin{equation}
\footnotesize{
\begin{aligned}
&\sqrt{n} \ \boldsymbol{U}^{(g)}(\boldsymbol{\beta}^*) \\
&= \frac{1}{\sqrt{n}}  \sum_{j=1}^{J^{(1)}}  \sum_{i=1}^{n_{j}^{(1)}}    \left(\left.\frac{\partial}{\partial \boldsymbol{\beta}}\right|_{\boldsymbol{\beta}^*} g^{-1}(\boldsymbol{a}_j^{(1)},\boldsymbol{z}_j^{(1)}; \boldsymbol{\beta}) \right)\left(Y_{ij}^{(1)}-g^{-1}(\boldsymbol{a}_j^{(1)},\boldsymbol{z}_j^{(1)}; \boldsymbol{\beta}^*)\right) \\
&\qquad+ \frac{1}{\sqrt{n}} \sum_{j=1}^{J^{(2)}}  \sum_{i=1}^{n_{j}^{(2)}}  \left(\left.\frac{\partial}{\partial \boldsymbol{\beta}}\right|_{\boldsymbol{\beta}^*} g^{-1}(\boldsymbol{A}_{j}^{\left(2, n^{(1)}\right)}, \boldsymbol{z}_{j}^{\left(2\right)}; \boldsymbol{\beta}) \right)\left(Y_{ij}^{\left(2, n^{(1)}\right)}-g^{-1}(\boldsymbol{A}_{j}^{\left(2, n^{(1)}\right)}, \boldsymbol{z}_{j}^{\left(2\right)}; \boldsymbol{\beta}^*)\right)  \\
&= \frac{1}{\sqrt{n}} \sum_{j=1}^{J^{(1)}}  \sum_{i=1}^{n_{j}^{(1)}}  \left(\left.\frac{\partial}{\partial \boldsymbol{\beta}}\right|_{\boldsymbol{\beta}^*} g^{-1}(\boldsymbol{a}_j^{(1)},\boldsymbol{z}_j^{(1)}; \boldsymbol{\beta}) \right) \epsilon_{ij}^{(1)} \\ 
& +\frac{1}{\sqrt{n}} \sum_{j=1}^{J^{(2)}}  \sum_{i=1}^{n_{j}^{(2)}} 
\left( 
\left.\frac{\partial}{\partial \boldsymbol{\beta}}\right|_{\boldsymbol{\beta}^*} g^{-1}(\boldsymbol{A}_j^{(2,n^{(1)})},\boldsymbol{z}_j^{(2)}; \boldsymbol{\beta})  
\pm  \left.\frac{\partial}{\partial \boldsymbol{\beta}}\right|_{\boldsymbol{\beta}^*} g^{-1}(\boldsymbol{a}_j^{(2)},\boldsymbol{z}_j^{(2)}; \boldsymbol{\beta}) 
\right)
\epsilon_{ij}^{(2,n^{(1)})}
\\
&= \sum_{j=1}^{J^{(1)}} \frac{\sqrt{n_j^{(1)}}}{\sqrt{n}} \frac{1}{\sqrt{n_j^{(1)}}} \sum_{i=1}^{n_{j}^{(1)}} \left(\left.\frac{\partial}{\partial \boldsymbol{\beta}}\right|_{\boldsymbol{\beta}^*} g^{-1}(\boldsymbol{a}_j^{(1)},\boldsymbol{z}_j^{(1)}; \boldsymbol{\beta}) \right) \epsilon_{ij}^{(1)} \\
&+ \frac{1}{\sqrt{n}} \sum_{j=1}^{J^{(2)}}  \sum_{i=1}^{n_{j}^{(2)}}  
\left(
\left.\frac{\partial}{\partial \boldsymbol{\beta}}\right|_{\boldsymbol{\beta}^*} g^{-1}(\boldsymbol{A}_j^{(2,n^{(1)})},\boldsymbol{z}_j^{(2)}; \boldsymbol{\beta})  
-  \left.\frac{\partial}{\partial \boldsymbol{\beta}}\right|_{\boldsymbol{\beta}^*} g^{-1}(\boldsymbol{a}_j^{(2)},\boldsymbol{z}_j^{(2)}; \boldsymbol{\beta})
\right)
\epsilon_{ij}^{(2,n^{(1)})}
\\ 
&\qquad+ \sum_{j=1}^{J^{(2)}} \frac{\sqrt{n_j^{(2)}}}{\sqrt{n}} \frac{1}{\sqrt{n_j^{(2)}}} \sum_{i=1}^{n_{j}^{(2)}} \left(\left.\frac{\partial}{\partial \boldsymbol{\beta}}\right|_{\boldsymbol{\beta}^*} g^{-1}(\boldsymbol{a}_j^{(2)},\boldsymbol{z}_j^{(2)}; \boldsymbol{\beta}) \right) \epsilon_{ij}^{(2,n^{(1)})}\\
&= \boldsymbol{U}_{1,n}^{(g)} + \boldsymbol{U}_{2\_1,n}^{(g)} + \boldsymbol{U}_{2\_2,n}^{(g)},
\end{aligned}
}
\label{regular_glm_AN_general}
\end{equation}
where
\begin{equation}
\footnotesize{
\begin{aligned}
\boldsymbol{U}_{1,n}^{(g)} &=\sum_{j=1}^{J^{(1)}} \frac{\sqrt{n_j^{(1)}}}{\sqrt{n}} \frac{1}{\sqrt{n_j^{(1)}}} \sum_{i=1}^{n_{j}^{(1)}} \left(\left.\frac{\partial}{\partial \boldsymbol{\beta}}\right|_{\boldsymbol{\beta}^*} g^{-1}(\boldsymbol{a}_j^{(1)},\boldsymbol{z}_j^{(1)}; \boldsymbol{\beta}) \right) \epsilon_{ij}^{(1)} ,\\
\boldsymbol{U}_{2\_1,n}^{(g)} &= \frac{1}{\sqrt{n}}\sum_{j=1}^{J^{(2)}} \sum_{i=1}^{n_{j}^{(2)}}  \left(\left.\frac{\partial}{\partial \boldsymbol{\beta}}\right|_{\boldsymbol{\beta}^*} g^{-1}(\boldsymbol{A}_j^{(2,n^{(1)})},\boldsymbol{z}_j^{(2)}; \boldsymbol{\beta}) -  \left.\frac{\partial}{\partial \boldsymbol{\beta}}\right|_{\boldsymbol{\beta}^*} g^{-1}(\boldsymbol{a}_j^{(2)},\boldsymbol{z}_j^{(2)}; \boldsymbol{\beta}) \right)\epsilon_{ij}^{(2,n^{(1)})},\\
\boldsymbol{U}_{2\_2,n}^{(g)} &= \sum_{j=1}^{J^{(2)}} \frac{\sqrt{n_j^{(2)}}}{\sqrt{n}} \frac{1}{\sqrt{n_j^{(2)}}} \sum_{i=1}^{n_{j}^{(2)}} \left(\left.\frac{\partial}{\partial \boldsymbol{\beta}}\right|_{\boldsymbol{\beta}^*} g^{-1}(\boldsymbol{a}_j^{(2)},\boldsymbol{z}_j^{(2)}; \boldsymbol{\beta}) \right) \epsilon_{ij}^{(2,n^{(1)})}. 
\end{aligned}
}
\label{U_separations_equations}
\end{equation}
In order to demonstrate that $\sqrt{n} \ \boldsymbol{U}^{(g)}(\boldsymbol{\beta}^*)$ converges in distribution to a normal distribution, we first show that 
$\boldsymbol{U}_{2\_1,n}^{(g)}$ converges in probability to 0.
Subsequently, we show that the joint distribution of $\boldsymbol{U}_{1,n}^{(g)}$ and $\boldsymbol{U}_{2\_2,n}^{(g)}$ converges in distribution to a normal distribution. 
This method works as under Assumption \ref{ind_error}, $\boldsymbol{U}_{1,n}^{(g)}$ and $\boldsymbol{U}_{2\_2,n}^{(g)}$ are independent.
Therefore, any dependence between stages only exists in $\boldsymbol{U}_{2\_1,n}^{(g)}$, which itself converges to 0 in probability. 

By Assumption \ref{ind_error}, for each fixed value of $j$, the central limit theorem for i.i.d. observations implies that 
$$
\footnotesize{
\frac{1}{\sqrt{n_{j}^{(1)}}}  \sum_{i=1}^{n_{j}^{(1)}}  \left(\left.\frac{\partial}{\partial \boldsymbol{\beta}}\right|_{\boldsymbol{\beta}^*} g^{-1}(\boldsymbol{a}_j^{(1)},\boldsymbol{z}_j^{(1)}; \boldsymbol{\beta}) \right) \epsilon_{ij}^{(1)}
}
$$
converges in distribution to a normal distribution with mean 0 and variance 
$$ 
\footnotesize{
\left(\left.\frac{\partial}{\partial \boldsymbol{\beta}}\right|_{\boldsymbol{\beta}^*} g^{-1}(\boldsymbol{a}_j^{(1)},\boldsymbol{z}_j^{(1)}; \boldsymbol{\beta}) \right)^{\otimes2} \sigma^2(\boldsymbol{z}_{j}^{(1)}).
}
$$ 
From e.g. Lévy's Continuity Theorem for characteristic functions \citep{eisenberg1983uniform} and Slutsky's Theorem, it follows that $\boldsymbol{U}_{1,n}^{(g)}$ converges to a normal distribution with mean 0 and variance 
$$
\footnotesize{
\sum_{j=1}^{J^{(1)}} \alpha_{j 1} \left(\left.\frac{\partial}{\partial \boldsymbol{\beta}}\right|_{\boldsymbol{\beta}^*} g^{-1}(\boldsymbol{a}_j^{(1)},\boldsymbol{z}_j^{(1)}; \boldsymbol{\beta}) \right)^{\otimes2} \sigma^2(\boldsymbol{z}_{j}^{(1)}).
}
$$
We show $\boldsymbol{U}_{2\_1,n}^{(g)} \xrightarrow{P} 0$ by showing that $E(\boldsymbol{U}_{2\_1,n}^{(g)}) = 0$ and $E({\boldsymbol{U}_{2\_1,n}^{(g)}}^{\otimes2}) \xrightarrow{} 0$, so that $\boldsymbol{U}_{2\_1,n}^{(g)} \xrightarrow{P} 0$ by Chebyshev's Inequality. We show  $E(\boldsymbol{U}_{2\_1,n}^{(g)}) = 0$ first. For each $j$, we have
\begin{equation*}
\footnotesize{
\begin{aligned}
&E \left\{ \frac{1}{\sqrt{n}}  \sum_{i=1}^{n_{j}^{(2)}} \left(\left.\frac{\partial}{\partial \boldsymbol{\beta}}\right|_{\boldsymbol{\beta}^*} g^{-1}(\boldsymbol{A}_j^{(2,n^{(1)})},\boldsymbol{z}_j^{(2)}; \boldsymbol{\beta}) - \left.\frac{\partial}{\partial \boldsymbol{\beta}}\right|_{\boldsymbol{\beta}^*} g^{-1}(\boldsymbol{a}_j^{(2)},\boldsymbol{z}_j^{(2)}; \boldsymbol{\beta}) \right) \epsilon_{ij}^{(2,n^{(1)})} \right\} = 0
\end{aligned}
}
\end{equation*}
by conditioning on $\boldsymbol{A}_{j}^{(2,n^{(1)})}$, so $E(\boldsymbol{U}_{2\_1,n}^{(g)}) = 0$. To see that $E({\boldsymbol{U}_{2\_1,n}^{(g)}}^{\otimes2}) \xrightarrow{} 0$, 
\begin{equation*}
\footnotesize{
\begin{aligned}
&E \left\{ \frac{1}{\sqrt{n_{j}^{(2)}}}  \sum_{i=1}^{n_{j}^{(2)}} \left(\left.\frac{\partial}{\partial \boldsymbol{\beta}}\right|_{\boldsymbol{\beta}^*} g^{-1}(\boldsymbol{A}_j^{(2,n^{(1)})},\boldsymbol{z}_j^{(2)}; \boldsymbol{\beta}) - \left.\frac{\partial}{\partial \boldsymbol{\beta}}\right|_{\boldsymbol{\beta}^*} g^{-1}(\boldsymbol{a}_j^{(2)},\boldsymbol{z}_j^{(2)}; \boldsymbol{\beta}) \right) \epsilon_{ij}^{(2,n^{(1)})} \right\}^{\otimes2} \\
&=\frac{1}{n_{j}^{(2)}}  \sum_{i=1}^{n_{j}^{(2)}} E \left\{ \left[ \left(\left.\frac{\partial}{\partial \boldsymbol{\beta}}\right|_{\boldsymbol{\beta}^*} g^{-1}(\boldsymbol{A}_j^{(2,n^{(1)})},\boldsymbol{z}_j^{(2)}; \boldsymbol{\beta}) - \left.\frac{\partial}{\partial \boldsymbol{\beta}}\right|_{\boldsymbol{\beta}^*} g^{-1}(\boldsymbol{a}_j^{(2)},\boldsymbol{z}_j^{(2)}; \boldsymbol{\beta}) \right) \epsilon_{ij}^{(2,n^{(1)})} \right]^{\otimes2} \right\}\\
&=\frac{1}{n_{j}^{(2)}}  \sum_{i=1}^{n_{j}^{(2)}} E \left\{ \left(\left.\frac{\partial}{\partial \boldsymbol{\beta}}\right|_{\boldsymbol{\beta}^*} g^{-1}(\boldsymbol{A}_j^{(2,n^{(1)})},\boldsymbol{z}_j^{(2)}; \boldsymbol{\beta})  - \left.\frac{\partial}{\partial \boldsymbol{\beta}}\right|_{\boldsymbol{\beta}^*} g^{-1}(\boldsymbol{a}_j^{(2)},\boldsymbol{z}_j^{(2)}; \boldsymbol{\beta}) \right)^{\otimes2} \left(\epsilon_{ij}^{(2,n^{(1)})}\right)^2 \right\}.
\end{aligned}
}
\end{equation*}
where the second line follows since all terms belonging to different $i$ are uncorrelated, which can be seen by conditioning on $\boldsymbol{A}_{j}^{\left(2, n^{(1)}\right)}$.\\ From Assumption \ref{ind_error}, we have $E\left(\left.\left(\epsilon_{i j}^{\left(2,n^{(1)}\right)}\right)^{2} \right| \boldsymbol{A}_{j}^{\left(2, n^{(1)}\right)}\right) = \sigma^2(\boldsymbol{z}_{j}^{(2)})$. 
Thus, by conditioning on $\boldsymbol{A}_{j}^{\left(2, n^{(1)}\right)}$, 
\begin{equation*}
\footnotesize{
\begin{aligned}
&E\left({\boldsymbol{U}_{2\_1,n}^{(g)}}^{\otimes2}\right)  \\
\qquad\qquad &=\sum_{j=1}^{J^{(2)}} \sigma^2(\boldsymbol{z}_{j}^{(2)})  \frac{n_{j}^{(2)}}{n} E \left\{ \left(\left.\frac{\partial}{\partial \boldsymbol{\beta}}\right|_{\boldsymbol{\beta}^*} g^{-1}(\boldsymbol{A}_j^{(2,n^{(1)})},\boldsymbol{z}_j^{(2)}; \boldsymbol{\beta}) - \left.\frac{\partial}{\partial \boldsymbol{\beta}}\right|_{\boldsymbol{\beta}^*} g^{-1}(\boldsymbol{a}_j^{(2)},\boldsymbol{z}_j^{(2)}; \boldsymbol{\beta}) \right)^{\otimes2} \right\}.
\end{aligned}
}
\end{equation*}
By Lemma \ref{sup_glm_condition}, the Continuous Mapping Theorem and Lebesgue's Dominated Convergence Theorem, 
$E({\boldsymbol{U}_{2\_1,n}^{(g)}}^{\otimes2}) \xrightarrow{} 0$. By Chebyshev's Inequality, $\boldsymbol{U}_{2\_1,n}^{(g)} \xrightarrow{P} 0 $.

For $\boldsymbol{U}_{2\_2,n}^{(g)}$, 
by Assumption \ref{ind_error} replacing $\epsilon_{ij}^{(2,n^{(1)})}$ by the error terms $\epsilon_{ij}^{(2)}$ under the limiting intervention $\boldsymbol{a}_{j}^{(2)}$ does not change the distribution of $\boldsymbol{U}_{2\_2, n}^{(g)}$, so it suffices to show that
 \begin{equation*}
 \footnotesize{
 {\boldsymbol{U}_{2\_2,n}^{(g)^*}} := \frac{1}{\sqrt{n}} \sum_{j=1}^{J^{(2)}}  \sum_{i=1}^{n_{j}^{(2)}} \left(\left.\frac{\partial}{\partial \boldsymbol{\beta}}\right|_{\boldsymbol{\beta}^*} g^{-1}(\boldsymbol{a}_j^{(2)},\boldsymbol{z}_j^{(2)}; \boldsymbol{\beta}) \right) \epsilon_{ij}^{(2)}
 }
 \end{equation*}
converges to a normal distribution. Following the same argument as for $\boldsymbol{U}_{1,n}^{(g)}$, ${\boldsymbol{U}_{2\_2,n}^{(g)\;*}}$ converges to a normal distribution with mean 0 and variance
$$
\footnotesize{
\sum_{j=1}^{J^{(2)}} \alpha_{j 2} \left(\left.\frac{\partial}{\partial \boldsymbol{\beta}}\right|_{\boldsymbol{\beta}^*} g^{-1}(\boldsymbol{a}_j^{(2)},\boldsymbol{z}_j^{(2)}; \boldsymbol{\beta}) \right)^{\otimes2} \sigma^2(\boldsymbol{z}_{j}^{(2)}).
}
$$

Hence,  $\sqrt{n} \ \boldsymbol{U}^{(g)}(\boldsymbol{\beta}^*)$ has the same asymptotic distribution as 
\begin{equation*}
\footnotesize{
\begin{aligned}
&\boldsymbol{U}_{1,n}^{(g)} + {\boldsymbol{U}_{2\_2,n}^{(g)\;*}} = \frac{1}{\sqrt{n}} \sum_{j=1}^{J^{(1)}}  \sum_{i=1}^{n_{j}^{(1)}}  \left(\left.\frac{\partial}{\partial \boldsymbol{\beta}}\right|_{\boldsymbol{\beta}^*} g^{-1}(\boldsymbol{a}_j^{(1)},\boldsymbol{z}_j^{(1)}; \boldsymbol{\beta}) \right) \epsilon_{ij}^{(1)}\\
&\qquad \qquad \qquad \qquad + \frac{1}{\sqrt{n}} \sum_{j=1}^{J^{(2)}}  \sum_{i=1}^{n_{j}^{(2)}} \left(\left.\frac{\partial}{\partial \boldsymbol{\beta}}\right|_{\boldsymbol{\beta}^*} g^{-1}(\boldsymbol{a}_j^{(2)},\boldsymbol{z}_j^{(2)}; \boldsymbol{\beta}) \right) \epsilon_{ij}^{(2)}.
\end{aligned}
}
\label{U1+U22 general link}
\end{equation*}
Since $\boldsymbol{a}_{j}^{(2)}$ are fixed, $\boldsymbol{U}_{1,n}^{(g)}$ and $ {\boldsymbol{U}_{2\_2,n}^{(g)}}$ are independent of one another. Thus, equation (\ref{regular_glm_AN_general}) has the same asymptotic distribution as a fixed two-stage design with $\boldsymbol{a}_1^{(1)},\cdots, \boldsymbol{a}_{J^{(1)}}^{(1)}, \\ \boldsymbol{a}_1^{(2)},\cdots, \boldsymbol{a}_{J^{(2)}}^{(2)}$ as interventions decided on before the trial. 
The limiting distribution of $\sqrt{n} \ \boldsymbol{U}^{(g)}(\boldsymbol{\beta}^*)$ is $N(0, V\left(\boldsymbol{\beta}^{*}\right))$ \citep{liang1986longitudinal}.
 
Combining this with equation (\ref{Jbeta proof done}) implies that Theorem \ref{AN_thm} holds.

\newpage
\section{Proofs of Theorem \ref{GLM_consistency_general} and Theorem \ref{AN_thm} with log link function.}\label{appendix_loglink}
In this section, we present specific illustrations of the proofs of Theorem \ref{GLM_consistency_general} and Theorem \ref{AN_thm} from the main text using the log link function for clarity. While these specific illustrations are not as general as the original proofs, they serve to demonstrate the concepts in a more concrete manner.

\subsection{Proof of Theorem 1 for LAGO GLM with independent errors and a log link function: consistency of \texorpdfstring{$\hat{\boldsymbol{\beta}}$}{Lg}}\label{consistency_appendix_loglink}
First, the estimating equations from (\ref{glm_EE}) of the main text are:
\begin{equation}\footnotesize{\begin{aligned}
0&=\boldsymbol{U}(\boldsymbol{\beta}) \\
&= \frac{1}{n}\left[ \sum_{j=1}^{J^{(1)}} \sum_{i=1}^{n_{j}^{(1)}} \left(\frac{\partial}{\partial \boldsymbol{\beta}} exp \left( \left( \begin{array}{c}
1 \\
\boldsymbol{a}_{j}^{(1)} \\
\boldsymbol{z}_{j}^{(1)}
\end{array} \right)^T \boldsymbol{\beta} \right)  \right) \left( Y_{ij}^{(1)} -exp\left( \left( \begin{array}{c}
1 \\
\boldsymbol{a}_{j}^{(1)} \\
\boldsymbol{z}_{j}^{(1)}
\end{array} \right)^T \boldsymbol{\beta}  \right)   \right) \right]\\
&+\frac{1}{n} \left[ \sum_{j=1}^{J^{(2)}} \sum_{i=1}^{n_{j}^{(2)}} \left(\frac{\partial}{\partial \boldsymbol{\beta}} exp \left( \left( \begin{array}{c}
1 \\
\boldsymbol{A}_{j}^{(2,n^{(1)})} \\
\boldsymbol{z}_{j}^{(2)}
\end{array} \right)^T \boldsymbol{\beta} \right)  \right) \left( Y_{ij}^{(2,n^{(1)})} -exp\left( \left( \begin{array}{c}
1 \\
\boldsymbol{A}_{j}^{(2,n^{(1)})} \\
\boldsymbol{z}_{j}^{(2)}
\end{array} \right)^T \boldsymbol{\beta}  \right)   \right) \right]\\
& = \frac{1}{n} \sum_{j=1}^{J^{(1)}} \sum_{i=1}^{n_{j}^{(1)}} \left[ exp \left( \left( \begin{array}{c}
1 \\
\boldsymbol{a}_{j}^{(1)} \\
\boldsymbol{z}_{j}^{(1)}
\end{array} \right)^T \boldsymbol{\beta}   \right) \left( \begin{array}{c}
1 \\
\boldsymbol{a}_{j}^{(1)} \\
\boldsymbol{z}_{j}^{(1)}
\end{array} \right) \left( Y_{ij}^{(1)} - exp\left( \left( \begin{array}{c}
1 \\
\boldsymbol{a}_{j}^{(1)} \\
\boldsymbol{z}_{j}^{(1)}
\end{array} \right)^T \boldsymbol{\beta}  \right)   \right) \right]\\
&+ \frac{1}{n} \sum_{j=1}^{J^{(2)}} \sum_{i=1}^{n_{j}^{(2)}} \left[ exp \left( \left( \begin{array}{c}
1 \\
\boldsymbol{A}_{j}^{(2,n^{(1)})} \\
\boldsymbol{z}_{j}^{(2)}
\end{array} \right)^T \boldsymbol{\beta}   \right)\left( \begin{array}{c}
1 \\
\boldsymbol{A}_{j}^{(2,n^{(1)})} \\
\boldsymbol{z}_{j}^{(2)}
\end{array} \right) \left(Y_{ij}^{(2,n^{(1)})} -exp\left( \left( \begin{array}{c}
1 \\
\boldsymbol{A}_{j}^{(2,n^{(1)})} \\
\boldsymbol{z}_{j}^{(2)}
\end{array} \right)^T \boldsymbol{\beta}  \right)   \right) \right]. 
\end{aligned}
}
\label{big_U_1}
\end{equation}

\noindent Recall that $\boldsymbol{\beta}^*$ is the true value of $\boldsymbol{\beta}$, and let
\begin{equation}
\footnotesize{
\begin{aligned}
\boldsymbol{u}(\boldsymbol{\beta}) &=  
\sum_{j=1}^{J^{(1)}} \alpha_{j1} \left[ \left(\frac{\partial}{\partial \boldsymbol{\beta}} exp \left( \left( \begin{array}{c}
1 \\
\boldsymbol{a}_{j}^{(1)} \\
\boldsymbol{z}_{j}^{(1)}
\end{array} \right)^T \boldsymbol{\beta} \right)  \right) \left( exp\left( \left( \begin{array}{c}
1 \\
\boldsymbol{a}_{j}^{(1)} \\
\boldsymbol{z}_{j}^{(1)}
\end{array} \right)^T \boldsymbol{\beta}^*  \right) -exp\left( \left( \begin{array}{c}
1 \\
\boldsymbol{a}_{j}^{(1)} \\
\boldsymbol{z}_{j}^{(1)}
\end{array} \right)^T \boldsymbol{\beta}  \right)   \right) \right]\\ 
&+ \sum_{j=1}^{J^{(2)}} \alpha_{j2} \left[  \left(\frac{\partial}{\partial \boldsymbol{\beta}} exp \left( \left( \begin{array}{c}
1 \\
\boldsymbol{a}_{j}^{(2)} \\
\boldsymbol{z}_{j}^{(2)}
\end{array} \right)^T \boldsymbol{\beta} \right)  \right) \left( exp\left( \left( \begin{array}{c}
1 \\
\boldsymbol{a}_{j}^{(2)} \\
\boldsymbol{z}_{j}^{(1)}
\end{array} \right)^T \boldsymbol{\beta}^*  \right) -exp\left( \left( \begin{array}{c}
1 \\
\boldsymbol{a}_{j}^{(2)} \\
\boldsymbol{z}_{j}^{(2)}
\end{array} \right)^T \boldsymbol{\beta}  \right)   \right) \right] . 
\end{aligned}
}
\label{little_u_1}
\end{equation}
To show consistency of estimator $\hat{\boldsymbol{\beta}}$, we show that in spite of the fact that (\ref{big_U_1}) does not consist of i.i.d. terms, Theorem $5.9$ of \citet{van2000asymptotic} can be used. We show that its two conditions are satisfied. First, we show 
\begin{equation}
    \sup_{\boldsymbol{\beta}}\|\boldsymbol{U}(\boldsymbol{\beta}) - \boldsymbol{u}(\boldsymbol{\beta})\| \xrightarrow{P} 0.
\label{first_condition}
\end{equation}
From equations (\ref{big_U_1}) and (\ref{little_u_1}), it follows that 
\begin{equation}\label{U-u separation loglink}
\footnotesize{
\begin{aligned}
\boldsymbol{U}(\boldsymbol{\beta}) &- \boldsymbol{u}(\boldsymbol{\beta}) = \frac{1}{n} \sum_{j=1}^{J^{(1)}} \sum_{i=1}^{n_{j}^{(1)}} \left[ exp \left( \left( \begin{array}{c}
1 \\
\boldsymbol{a}_{j}^{(1)} \\
\boldsymbol{z}_{j}^{(1)}
\end{array} \right)^T \boldsymbol{\beta}   \right) \left( \begin{array}{c}
1 \\
\boldsymbol{a}_{j}^{(1)} \\
\boldsymbol{z}_{j}^{(1)}
\end{array} \right) \left( Y_{ij}^{(1)} - exp\left( \left( \begin{array}{c}
1 \\
\boldsymbol{a}_{j}^{(1)} \\
\boldsymbol{z}_{j}^{(1)}
\end{array} \right)^T \boldsymbol{\beta}  \right)   \right) \right]\\
&+ \frac{1}{n} \sum_{j=1}^{J^{(2)}} \sum_{i=1}^{n_{j}^{(2)}} \left[ exp \left( \left( \begin{array}{c}
1 \\
\boldsymbol{A}_{j}^{(2,n^{(1)})} \\
\boldsymbol{z}_{j}^{(2)}
\end{array} \right)^T \boldsymbol{\beta}   \right)\left( \begin{array}{c}
1 \\
\boldsymbol{A}_{j}^{(2,n^{(1)})} \\
\boldsymbol{z}_{j}^{(2)}
\end{array} \right) \left(Y_{ij}^{(2,n^{(1)})} -exp\left( \left( \begin{array}{c}
1 \\
\boldsymbol{A}_{j}^{(2,n^{(1)})} \\
\boldsymbol{z}_{j}^{(2)}
\end{array} \right)^T \boldsymbol{\beta}  \right)  \right.\right.  \\
& \qquad \qquad \qquad \qquad \qquad \qquad \qquad \qquad \qquad \qquad \qquad \qquad \qquad 
\left.\left.\pm exp\left( \left( \begin{array}{c}
1 \\
\boldsymbol{A}_{j}^{(2,n^{(1)})} \\
\boldsymbol{z}_{j}^{(2)}
\end{array} \right)^T \boldsymbol{\beta}^*  \right)   \right) \right]\\
&- \sum_{j=1}^{J^{(1)}} \alpha_{j1} \left[ \left(\frac{\partial}{\partial \boldsymbol{\beta}} exp \left( \left( \begin{array}{c}
1 \\
\boldsymbol{a}_{j}^{(1)} \\
\boldsymbol{z}_{j}^{(1)}
\end{array} \right)^T \boldsymbol{\beta} \right)  \right) \left( exp\left( \left( \begin{array}{c}
1 \\
\boldsymbol{a}_{j}^{(1)} \\
\boldsymbol{z}_{j}^{(1)}
\end{array} \right)^T \boldsymbol{\beta}^*  \right) -exp\left( \left( \begin{array}{c}
1 \\
\boldsymbol{a}_{j}^{(1)} \\
\boldsymbol{z}_{j}^{(1)}
\end{array} \right)^T \boldsymbol{\beta}  \right)   \right) \right]\\ 
&- \sum_{j=1}^{J^{(2)}} \alpha_{j2} \left[  \left(\frac{\partial}{\partial \boldsymbol{\beta}} exp \left( \left( \begin{array}{c}
1 \\
\boldsymbol{a}_{j}^{(2)} \\
\boldsymbol{z}_{j}^{(2)}
\end{array} \right)^T \boldsymbol{\beta} \right)  \right) \left( exp\left( \left( \begin{array}{c}
1 \\
\boldsymbol{a}_{j}^{(2)} \\
\boldsymbol{z}_{j}^{(1)}
\end{array} \right)^T \boldsymbol{\beta}^*  \right) -exp\left( \left( \begin{array}{c}
1 \\
\boldsymbol{a}_{j}^{(2)} \\
\boldsymbol{z}_{j}^{(2)}
\end{array} \right)^T \boldsymbol{\beta}  \right)   \right) \right] \\
& = G_{1,n} + G_{2,n} + G_{3,n} + G_{4,n} + G_{5,n}, 
\end{aligned} 
}
\end{equation}
where 
\begin{equation}
\footnotesize{
\begin{aligned}
&G_{1,n} = \frac{1}{n} \sum_{j=1}^{J^{(1)}} \sum_{i=1}^{n_{j}^{(1)}} \left[ exp \left( \left( \begin{array}{c}
1 \\
\boldsymbol{a}_{j}^{(1)}\\
\boldsymbol{z}_{j}^{(1)} 
\end{array} \right)^T \boldsymbol{\beta}   \right) \left( \begin{array}{c}
1 \\
\boldsymbol{a}_{j}^{(1)}\\
\boldsymbol{z}_{j}^{(1)} 
\end{array} \right) \left( Y_{ij}^{(1)} - exp\left( \left( \begin{array}{c}
1 \\
\boldsymbol{a}_{j}^{(1)}\\
\boldsymbol{z}_{j}^{(1)} 
\end{array} \right)^T \boldsymbol{\beta}^*  \right)   \right) \right],
\end{aligned}
}
\label{G1_term}
\end{equation}

\begin{equation}
\footnotesize{
\begin{aligned}
G_{2,n} &= \frac{1}{n} \sum_{j=1}^{J^{(2)}} \sum_{i=1}^{n_{j}^{(2)}} \left[ exp \left( \left( \begin{array}{c}
1 \\
\boldsymbol{A}_{j}^{(2,n^{(1)})} \\
\boldsymbol{z}_{j}^{(2)}  
\end{array} \right)^T \boldsymbol{\beta}   \right)\left( \begin{array}{c}
1 \\
\boldsymbol{A}_{j}^{(2,n^{(1)})} \\
\boldsymbol{z}_{j}^{(2)}  
\end{array} \right) \right.\\
&\qquad\qquad\qquad\qquad\qquad\qquad\qquad\left.\left( Y_{ij}^{(2, n^{(1)})}-exp \left( \left( \begin{array}{c}
1 \\
\boldsymbol{A}_{j}^{(2,n^{(1)})} \\
\boldsymbol{z}_{j}^{(2)}  
\end{array} \right)^T \boldsymbol{\beta}^*  \right) \right) \right],
\end{aligned}
}
\label{G2_term}\end{equation}

\begin{equation}
\footnotesize{
\begin{aligned}
& G_{3,n} = \frac{1}{n} \sum_{j=1}^{J^{(2)}} \sum_{i=1}^{n_{j}^{(2)}} \left[ exp \left( \left( \begin{array}{c}
1 \\
\boldsymbol{A}_{j}^{(2,n^{(1)})} \\
\boldsymbol{z}_{j}^{(2)}  
\end{array} \right)^T \boldsymbol{\beta}   \right)\left( \begin{array}{c}
1 \\
\boldsymbol{A}_{j}^{(2,n^{(1)})} \\
\boldsymbol{z}_{j}^{(2)}  
\end{array} \right) exp \left( \left( \begin{array}{c}
1 \\
\boldsymbol{A}_{j}^{(2,n^{(1)})} \\
\boldsymbol{z}_{j}^{(2)}  
\end{array} \right)^T \boldsymbol{\beta}^*  \right) \right.\\
& \qquad \qquad \qquad \qquad \qquad  -\left.  exp \left( \left( \begin{array}{c}
1 \\
\boldsymbol{a}_{j}^{(2)} \\
\boldsymbol{z}_{j}^{(2)}  
\end{array} \right)^T \boldsymbol{\beta}   \right)\left( \begin{array}{c}
1 \\
\boldsymbol{a}_{j}^{(2)} \\
\boldsymbol{z}_{j}^{(2)}  
\end{array} \right) exp \left( \left( \begin{array}{c}
1 \\
\boldsymbol{a}_{j}^{(2)} \\
\boldsymbol{z}_{j}^{(2)}  
\end{array} \right)^T \boldsymbol{\beta}^*  \right) \right],
\end{aligned}
}
\label{G3_term}
\end{equation}

\begin{equation}
\footnotesize{
\begin{aligned}
&G_{4,n} = \frac{1}{n} \sum_{j=1}^{J^{(2)}} \sum_{i=1}^{n_{j}^{(2)}} \left[ exp \left( \left( \begin{array}{c}
1 \\
\boldsymbol{a}_{j}^{(2)} \\
\boldsymbol{z}_{j}^{(2)}  
\end{array} \right)^T \boldsymbol{\beta}   \right)\left( \begin{array}{c}
1 \\
\boldsymbol{a}_{j}^{(2)} \\
\boldsymbol{z}_{j}^{(2)}  
\end{array} \right) exp \left( \left( \begin{array}{c}
1 \\
\boldsymbol{a}_{j}^{(2)} \\
\boldsymbol{z}_{j}^{(2)}  
\end{array} \right)^T \boldsymbol{\beta}  \right) \right.\\
& \qquad \qquad \qquad \qquad \qquad \qquad \left. - exp \left( \left( \begin{array}{c}
1 \\
\boldsymbol{A}_{j}^{(2,n^{(1)})} \\
\boldsymbol{z}_{j}^{(2)}  
\end{array} \right)^T \boldsymbol{\beta}   \right)\left( \begin{array}{c}
1 \\
\boldsymbol{A}_{j}^{(2,n^{(1)})} \\
\boldsymbol{z}_{j}^{(2)}  
\end{array} \right) exp \left( \left( \begin{array}{c}
1 \\
\boldsymbol{A}_{j}^{(2,n^{(1)})} \\
\boldsymbol{z}_{j}^{(2)}  
\end{array} \right)^T \boldsymbol{\beta}  \right) \right],
\end{aligned}
}
\label{G4_term}
\end{equation}

\begin{equation}
\footnotesize{
\begin{aligned}
G_{5,n} &=
\sum_{j=1}^{J^{(1)}} \left[ \left(\alpha_{j1}-\frac{n_j^{(1)}}{n} \right)
exp \left( \left( \begin{array}{c}
1 \\
\boldsymbol{a}_{j}^{(1)}\\
\boldsymbol{z}_{j}^{(1)} 
\end{array} \right)^T \boldsymbol{\beta}   \right) \left( \begin{array}{c}
1 \\
\boldsymbol{a}_{j}^{(1)}\\
\boldsymbol{z}_{j}^{(1)} 
\end{array} \right) exp \left( \left( \begin{array}{c}
1 \\
\boldsymbol{a}_{j}^{(1)}\\
\boldsymbol{z}_{j}^{(1)} 
\end{array} \right)^T \boldsymbol{\beta}   \right)
\right] \\
&+ \sum_{j=1}^{J^{(1)}} \left[ \left(\frac{n_j^{(1)}}{n} - \alpha_{j1} \right)
exp \left( \left( \begin{array}{c}
1 \\
\boldsymbol{a}_{j}^{(1)}\\
\boldsymbol{z}_{j}^{(1)} 
\end{array} \right)^T \boldsymbol{\beta}   \right) \left( \begin{array}{c}
1 \\
\boldsymbol{a}_{j}^{(1)}\\
\boldsymbol{z}_{j}^{(1)} 
\end{array} \right)  exp \left( \left( \begin{array}{c}
1 \\
\boldsymbol{a}_{j}^{(1)}\\
\boldsymbol{z}_{j}^{(1)} 
\end{array} \right)^T \boldsymbol{\beta}^*   \right)
\right] \\
&+ \sum_{j=1}^{J^{(2)}} \left[ \left(\alpha_{j2}-\frac{n_j^{(2)}}{n} \right)
exp \left( \left( \begin{array}{c}
1 \\
\boldsymbol{a}_{j}^{(2)}\\
\boldsymbol{z}_{j}^{(2)} 
\end{array} \right)^T \boldsymbol{\beta}   \right) \left( \begin{array}{c}
1 \\
\boldsymbol{a}_{j}^{(2)}\\
\boldsymbol{z}_{j}^{(2)} 
\end{array} \right) exp \left( \left( \begin{array}{c}
1 \\
\boldsymbol{a}_{j}^{(2)}\\
\boldsymbol{z}_{j}^{(2)} 
\end{array} \right)^T \boldsymbol{\beta}   \right)
\right] \\
&+ \sum_{j=1}^{J^{(2)}} \left[ \left(\frac{n_j^{(2)}}{n} - \alpha_{j2} \right)
exp \left( \left( \begin{array}{c}
1 \\
\boldsymbol{a}_{j}^{(2)}\\
\boldsymbol{z}_{j}^{(2)} 
\end{array} \right)^T \boldsymbol{\beta}   \right) \left( \begin{array}{c}
1 \\
\boldsymbol{a}_{j}^{(2)}\\
\boldsymbol{z}_{j}^{(2)} 
\end{array} \right) exp \left( \left( \begin{array}{c}
1 \\
\boldsymbol{a}_{j}^{(2)}\\
\boldsymbol{z}_{j}^{(2)} 
\end{array} \right)^T \boldsymbol{\beta}^*   \right)
\right] .
\end{aligned} 
}
\label{G5_term}
\end{equation}
The five terms $G_{1,n}$, $G_{2,n}$, $G_{3,n}$, $G_{4,n}$, and $G_{5,n}$, will be discussed separately. We first show the supremum over $\boldsymbol{\beta}$ of each term converges to 0 in probability, then the triangle inequality implies equation (\ref{first_condition}).

For $G_{1,n}$ from equation (\ref{G1_term}), we show that $\sup_{\boldsymbol{\beta}}\|G_{1,n}\| \xrightarrow{P} 0$ by using the concept of Donsker classes from empirical process theory. Let $O^{(1)}_{ij} = (Y^{(1)}_{ij}, \boldsymbol{a}^{(1)}_{j}, \boldsymbol{z}^{(1)}_{j})$ be the observed data for patient $i$ from center $j$ in stage 1 and let
\begin{equation*}
\footnotesize{
    \Psi_{\boldsymbol{\beta}}(O^{(1)}_{ij}) = exp \left( \left( \begin{array}{c}
1 \\
\boldsymbol{a}_{j}^{(1)}\\
\boldsymbol{z}_{j}^{(1)} 
\end{array} \right)^T \boldsymbol{\beta}   \right) \left( \begin{array}{c}
1 \\
\boldsymbol{a}_{j}^{(1)}\\
\boldsymbol{z}_{j}^{(1)} 
\end{array} \right) \left( Y_{ij}^{(1)} - exp\left( \left( \begin{array}{c}
1 \\
\boldsymbol{a}_{j}^{(1)}\\
\boldsymbol{z}_{j}^{(1)} 
\end{array} \right)^T \boldsymbol{\beta}^*  \right)   \right).
}
\end{equation*}
We show that the class of functions $\mathcal{F} = \left( \Psi_{\boldsymbol{\beta}}: \boldsymbol{\beta} \in \boldsymbol{\mathcal{B}} \right)$ is Donsker
by showing that 
\begin{equation}
\left\|\Psi_{\boldsymbol{\beta}_{1}}(O^{(1)}_{ij})-\Psi_{\boldsymbol{\beta}_{2}}(O^{(1)}_{ij})\right\| \leq m(O^{(1)}_{ij})\left\|\boldsymbol{\beta}_{1}-\boldsymbol{\beta}_{2}\right\| \quad \text { for every } \boldsymbol{\beta}_{1}, \boldsymbol{\beta}_{2},
\label{donsker_condition}
\end{equation}
for some measurable function $m$ with $E\left(m^{2}\right)<\infty$; that $\mathcal{F}$ is Donsker then follows from Example 19.7 in \cite{van2000asymptotic}. By applying the mean value theorem to each row of $\Psi_{\boldsymbol{\beta}_{1}}(O^{(1)}_{ij})-\Psi_{\boldsymbol{\beta}_{2}}(O^{(1)}_{ij})$ separately, we find that
\begin{equation*}
\begin{aligned}
\Psi_{\boldsymbol{\beta}_{1}}(O^{(1)}_{ij})-\Psi_{\boldsymbol{\beta}_{2}}(O^{(1)}_{ij})& = \left(\left.\frac{\partial}{\partial \boldsymbol{\beta}}\right|_{\tilde{\boldsymbol{\beta}}} \Psi_{\boldsymbol{\beta}}(O^{(1)}_{ij})\right) (\boldsymbol{\beta}_{1} - \boldsymbol{\beta}_{2}),
\end{aligned}
\end{equation*}

\noindent where for each row of $\left.\frac{\partial}{\partial \boldsymbol{\beta}}\right|_{\tilde{\boldsymbol{\beta}}} \Psi_{\boldsymbol{\beta}}(O^{(1)}_{ij})$, $\tilde{\boldsymbol{\beta}}$ may take a different value between $\boldsymbol{\beta}_{1}$ and $\boldsymbol{\beta}_{2}$. 
By Assumption \ref{glm_model_assumption} and \ref{glm_unif_bound_assu} from the main text, $g()$ and $\Psi_{\boldsymbol{\beta}}(O^{(1)}_{ij})$ are continuously differentiable and both $\boldsymbol{\beta}$ and $O^{(1)}_{ij}$ take values in a compact space, each element in $\left.\frac{\partial}{\partial \boldsymbol{\beta}}\right|_{\tilde{\boldsymbol{\beta}}} \Psi_{\boldsymbol{\beta}}(O^{(1)}_{ij})$ is bounded. It follows that 
\begin{equation*}
\begin{aligned}
\left\|\Psi_{\boldsymbol{\beta}_{1}}(O^{(1)}_{ij})-\Psi_{\boldsymbol{\beta}_{2}}(O^{(1)}_{ij})\right\| & = \left\| \left(\left.\frac{\partial}{\partial \boldsymbol{\beta}}\right|_{\tilde{\boldsymbol{\beta}}} \Psi_{\boldsymbol{\beta}}(O^{(1)}_{ij}) \right) (\boldsymbol{\beta}_{1} - \boldsymbol{\beta}_{2})  \right\| \\
&\leq C_1 \left\| \boldsymbol{\beta}_{1} - \boldsymbol{\beta}_{2} \right\|
\end{aligned}
\end{equation*}
where $C_1$ is a constant. Equation (\ref{donsker_condition}) follows, so the class of functions $\mathcal{F} = \left( \Psi_{\boldsymbol{\beta}}: \boldsymbol{\beta} \in \boldsymbol{\mathcal{B}} \right)$  is Donsker. By Theorem 19.4 in \cite{van2000asymptotic}, $\mathcal{F}$ is also Glivenko-Cantelli. Since for each $\boldsymbol{\beta}$, $E\left( \Psi_{\boldsymbol{\beta}}\left(O_{ij}^{(1)}\right)\right)=0$, then
\begin{equation*}
\sup_{\boldsymbol{\beta}} \left\| \frac{1}{n_{j}^{(1)}} \sum_{i=1}^{n_{j}^{(1)}}  \Psi_{\boldsymbol{\beta}}\left(O_{ij}^{(1)}\right) -  E\left( \Psi_{\boldsymbol{\beta}}\left(O_{ij}^{(1)}\right)\right) \right\| \xrightarrow{P} 0. 
\end{equation*}
Notice that the $O_{ij}^{(1)}$ are iid for each $j$ separately, $\frac{n_j^{(1)}}{n} < 1$, and $j$ is finite, so
\begin{equation*}
    \sup _{\boldsymbol{\beta}}\left\|\frac{n_{j}^{(1)}}{n}  \frac{1}{n_{j}^{(1)}} \sum_{i=1}^{n_{j}^{(1)}} \Psi_{\boldsymbol{\beta}}\left(O_{i j}^{(1)}\right)\right\|
    \leq
    \sup _{\boldsymbol{\beta}}\left\|\frac{1}{n_{j}^{(1)}} \sum_{i=1}^{n_{j}^{(1)}} \Psi_{\boldsymbol{\beta}}\left(O_{i j}^{(1)}\right)\right\|.
\end{equation*}
We conclude that $\sup_{\boldsymbol{\beta}} \left\|G_{1,n}^{(g)}\right\| \xrightarrow{P} 0.$

For the term $G_{2,n}$ from equation (\ref{G2_term}), let $Y_{ij}^{(2)}$ be the (counterfactual) outcomes under $\boldsymbol{a}_{j}^{(2)}$. 
and let $\epsilon_{ij}^{(2)}$ be the corresponding errors that patient $i$ in center $j$ would have experienced under intervention $\boldsymbol{a}_{j}^{(2)}$. 
We derive
\begin{equation*}
\footnotesize{
\begin{aligned}
G_{2,n} 
&= \frac{1}{n} \sum_{j=1}^{J^{(2)}}  \sum_{i=1}^{n_{j}^{(2)}} \left[ exp \left( \left( \begin{array}{c}
1 \\
\boldsymbol{A}_{j}^{(2,n^{(1)})} \\
\boldsymbol{z}_{j}^{(2)}  
\end{array} \right)^T \boldsymbol{\beta}   \right)\left( \begin{array}{c}
1 \\
\boldsymbol{A}_{j}^{(2,n^{(1)})} \\
\boldsymbol{z}_{j}^{(2)}  
\end{array} \right)\right.\\
&\qquad\qquad\qquad\qquad\qquad\qquad\qquad\qquad\qquad\left.\left( Y_{ij}^{(2, n^{(1)})}-exp \left( \left( \begin{array}{c}
1 \\
\boldsymbol{A}_{j}^{(2,n^{(1)})} \\
\boldsymbol{z}_{j}^{(2)}  
\end{array} \right)^T \boldsymbol{\beta}^*  \right) \right) \right]\\ 
&= \frac{1}{n} \sum_{j=1}^{J^{(2)}}  \sum_{i=1}^{n_{j}^{(2)}} \left[ exp \left( \left( \begin{array}{c}
1 \\
\boldsymbol{A}_{j}^{(2,n^{(1)})} \\
\boldsymbol{z}_{j}^{(2)}  
\end{array} \right)^T \boldsymbol{\beta}   \right)\left( \begin{array}{c}
1 \\
\boldsymbol{A}_{j}^{(2,n^{(1)})} \\
\boldsymbol{z}_{j}^{(2)}  
\end{array} \right) \epsilon_{ij}^{(2, n^{(1)})} \right] \\
\end{aligned}
}
\end{equation*}
By Assumption \ref{ind_error} from the main text, replacing $\epsilon_{ij}^{(2, n^{(1)})} $ by the new error terms $\epsilon_{ij}^{(2)} $ does not change the distribution of $G_{2,n}$, so it suffices to show that for 
\begin{equation*}
\footnotesize{
\begin{aligned}
\widetilde{G}_{2,n} &= \frac{1}{n} \sum_{j=1}^{J^{(2)}}  \sum_{i=1}^{n_{j}^{(2)}} \left[ exp \left( \left( \begin{array}{c}
1 \\
\boldsymbol{A}_{j}^{(2,n_1)} \\
\boldsymbol{z}_{j}^{(2)}  
\end{array} \right)^T \boldsymbol{\beta}   \right)\left( \begin{array}{c}
1 \\
\boldsymbol{A}_{j}^{(2,n_1)} \\
\boldsymbol{z}_{j}^{(2)}  
\end{array} \right) \epsilon_{ij}^{(2)} \right],
\end{aligned} 
}
\end{equation*}
$\sup_{\boldsymbol{\beta}}\|\widetilde{G}_{2,n}\| \xrightarrow{P} 0$. 
\begin{equation*}
\footnotesize{
\begin{aligned}
\widetilde{G}_{2,n} &= \frac{1}{n} \sum_{j=1}^{J^{(2)}}  \sum_{i=1}^{n_{j}^{(2)}} \left[ exp \left( \left( \begin{array}{c}
1 \\
\boldsymbol{A}_{j}^{(2,n_1)} \\
\boldsymbol{z}_{j}^{(2)}  
\end{array} \right)^T \boldsymbol{\beta}   \right)\left( \begin{array}{c}
1 \\
\boldsymbol{A}_{j}^{(2,n_1)} \\
\boldsymbol{z}_{j}^{(2)}  
\end{array} \right) \epsilon_{ij}^{(2)} \pm exp \left( \left( \begin{array}{c}
1 \\
\boldsymbol{a}_{j}^{(2)} \\
\boldsymbol{z}_{j}^{(2)}  
\end{array} \right)^T \boldsymbol{\beta}   \right)\left( \begin{array}{c}
1 \\
\boldsymbol{a}_{j}^{(2)} \\
\boldsymbol{z}_{j}^{(2)}  
\end{array} \right) \epsilon_{ij}^{(2)} \right] \\
&= \widetilde{G}_{2\_1,n} + \widetilde{G}_{2\_2,n},
\end{aligned}
}
\end{equation*}
where
\begin{equation*}
\footnotesize{
    \widetilde{G}_{2\_1,n} = \frac{1}{n} \sum_{j=1}^{J^{(2)}}  \sum_{i=1}^{n_{j}^{(2)}} \left[ exp \left( \left( \begin{array}{c}
1 \\
\boldsymbol{a}_{j}^{(2)} \\
\boldsymbol{z}_{j}^{(2)}  
\end{array} \right)^T \boldsymbol{\beta}   \right)\left( \begin{array}{c}
1 \\
\boldsymbol{a}_{j}^{(2)} \\
\boldsymbol{z}_{j}^{(2)}  
\end{array} \right) \epsilon_{ij}^{(2)} \right],
}
\end{equation*}
\begin{equation}
\footnotesize{
\begin{aligned}
&\widetilde{G}_{2\_2,n} = \frac{1}{n} \sum_{j=1}^{J^{(2)}}  \sum_{i=1}^{n_{j}^{(2)}} \left[  \left( exp \left( \left( \begin{array}{c} 
1 \\
\boldsymbol{A}_{j}^{(2,n^{(1)})} \\
\boldsymbol{z}_{j}^{(2)}  
\end{array} \right)^T \boldsymbol{\beta}   \right) \left( \begin{array}{c}
1 \\
\boldsymbol{A}_{j}^{(2,n^{(1)})} \\
\boldsymbol{z}_{j}^{(2)}  
\end{array} \right)\right.\right.\\
&\qquad\qquad\qquad\qquad\qquad\qquad\qquad\qquad\qquad\qquad\left.\left.- exp \left( \left( \begin{array}{c}
1 \\
\boldsymbol{a}_{j}^{(2)} \\
\boldsymbol{z}_{j}^{(2)}  
\end{array} \right)^T \boldsymbol{\beta} \right) \left( \begin{array}{c}
1 \\
\boldsymbol{a}_{j}^{(2)} \\
\boldsymbol{z}_{j}^{(2)}  
\end{array} \right) \right) \epsilon_{ij}^{(2)}  \right].
\label{g22starloglink}
\end{aligned}
}
\end{equation}
We show that both $\sup_{\boldsymbol{\beta}}\|\widetilde{G}_{2\_1,n}\| \xrightarrow{P} 0$ and $\sup_{\boldsymbol{\beta}}\|\widetilde{G}_{2\_2,n}\| \xrightarrow{P} 0$. By the triangle inequality, then $\sup_{\boldsymbol{\beta}}\|\widetilde{G}_{2,n}\| \xrightarrow{P} 0$. 

Let $O^{(2)}_{ij} = (Y^{(2)}_{ij}, \boldsymbol{a}^{(2)}_{j}, \boldsymbol{z}^{(2)}_{j})$ be the counterfactual data for patient $i$ from center $j$ in stage 2 under $\boldsymbol{a}^{(2)}_{j}$ and $\boldsymbol{z}^{(2)}_{j}$ and let
\begin{equation*}
\footnotesize{
    \Psi_{\boldsymbol{\beta}}^{(2)}(O^{(2)}_{ij}) = \left( exp \left( \left( \begin{array}{c}
1 \\
\boldsymbol{a}_{j}^{(2)} \\
\boldsymbol{z}_{j}^{(2)}  
\end{array} \right)^T \boldsymbol{\beta}   \right) \left( \begin{array}{c}
1 \\
\boldsymbol{a}_{j}^{(2)} \\
\boldsymbol{z}_{j}^{(2)}  
\end{array} \right) \right) \epsilon_{ij}^{(2)}.
}
\end{equation*}
Following the same argument as for $G_{1,n}$, the class of functions $\mathcal{F}_2 = \left( \Psi_{\boldsymbol{\beta}}^{(2)}: \boldsymbol{\beta} \in \boldsymbol{\mathcal{B}} \right)$ is a Donsker class and $\sup_{\boldsymbol{\beta}}\|{G}_{2\_1,n}\| \xrightarrow{P} 0$.

For the term ${G}_{2\_2,n}$ from equation (\ref{g22starloglink}), consider the most complicated entry with $\boldsymbol{A}_{j}^{(2,n^{(1)})}$. By the Mean Value Theorem, we derive
\begin{equation*}
\begin{aligned}
&\frac{1}{n} \sum_{j=1}^{J^{(2)}}  \sum_{i=1}^{n_{j}^{(2)}} \left\{ \left[ e^{(\boldsymbol{\beta}_0+\boldsymbol{\beta}_1 \boldsymbol{A}_{j}^{(2,n^{(1)})}+\boldsymbol{\beta}_2 \boldsymbol{z}_{j}^{(2)}  )} \boldsymbol{A}_{j}^{(2,n^{(1)})} - e^{(\boldsymbol{\beta}_0+\boldsymbol{\beta}_1 \boldsymbol{a}_{j}^{(2)}+\boldsymbol{\beta}_2 \boldsymbol{z}_{j}^{(2)}  )} \boldsymbol{a}_{j}^{(2)}\right]  \epsilon_{ij}^{(2)} \right\}\\
&=\frac{1}{n} \sum_{j=1}^{J^{(2)}}  \sum_{i=1}^{n_{j}^{(2)}} \left\{ \left[ \left.\frac{\partial}{\partial a}\right|_{\boldsymbol{\Tilde{a}}_j (\boldsymbol{\beta}) } e^{(\boldsymbol{\beta}_0+\boldsymbol{\beta}_1 a+\boldsymbol{\beta}_2 \boldsymbol{z}_{j}^{(2)}  )} \: a \:  \right] \left(\boldsymbol{A}_{j}^{(2,n^{(1)})}-\boldsymbol{a}_j^{(2)}\right) \epsilon_{ij}^{(2)} \right\}\\ 
&= \frac{1}{n} \sum_{j=1}^{J^{(2)}}  \sum_{i=1}^{n_{j}^{(2)}} \left\{ \left(\left( 1+\boldsymbol{\beta}_1 \boldsymbol{\Tilde{a}}_j (\boldsymbol{\beta}) \right) e^{(\boldsymbol{\beta}_0+\boldsymbol{\beta}_1 \boldsymbol{\Tilde{a}}_j (\boldsymbol{\beta})+\boldsymbol{\beta}_2 \boldsymbol{z}_{j}^{(2)}  )}    \right) \left(\boldsymbol{A}_{j}^{(2,n^{(1)})}-\boldsymbol{a}_j^{(2)}\right)   \epsilon_{ij}^{(2)}  \right\},
\end{aligned}
\end{equation*}
By Assumption \ref{glm_unif_bound_assu} from the main text and Lemma \ref{sup_glm_condition} from the Appendix, it follows that 
\begin{equation*}
\begin{aligned}
    &\sup_{\boldsymbol{\beta}}\left\| \frac{1}{n} \sum_{j=1}^{J^{(2)}}  \sum_{i=1}^{n_{j}^{(2)}} \left[ e^{(\boldsymbol{\beta}_0+\boldsymbol{\beta}_1 \boldsymbol{A}_{j}^{(2,n^{(1)})}+\boldsymbol{\beta}_2 \boldsymbol{z}_{j}^{(2)}  )} {\boldsymbol{A}_{j}^{(2,n^{(1)})}} - e^{(\boldsymbol{\beta}_0+\boldsymbol{\beta}_1 \boldsymbol{a}_{j}^{(2)}+\boldsymbol{\beta}_2 \boldsymbol{z}_{j}^{(2)}  )} \boldsymbol{a}_{j}^{(2)}\right]  \epsilon_{ij}^{(2)} \right\|  \\
    &\leq  \sup_{\boldsymbol{\beta}}\left\|  \left( 1+\boldsymbol{\beta}_1 \boldsymbol{\Tilde{a}}_j (\boldsymbol{\beta}) \right) e^{(\boldsymbol{\beta}_0+\boldsymbol{\beta}_1 \boldsymbol{\Tilde{a}}_j (\boldsymbol{\beta})+\boldsymbol{\beta}_2 \boldsymbol{z}_{j}^{(2)}  )}    \right\| \max_j {\left\|  \boldsymbol{A}_{j}^{(2,n^{(1)})}-\boldsymbol{a}_j^{(2)}\right\|} \sup_{ij} \left\| \epsilon_{ij}^{(2)} \right\| \\
    &\xrightarrow{P} 0.
\end{aligned}
\end{equation*}
This argument can be applied to other entries of $\widetilde{G}_{2\_2,n}$, so $\sup_{\boldsymbol{\beta}}\|\widetilde{G}_{2\_2,n}\| \xrightarrow{P} 0$. Hence, $\sup _{\boldsymbol{\beta}}\left\|G_{2,n}\right\| \stackrel{P}{\rightarrow} 0.$ 

Consider $G_{3,n}$ from equation (\ref{G3_term}). By the Mean Value Theorem, the supremum over $\boldsymbol{\beta}$ of the first entry is  
\begin{equation*}
\begin{aligned}
    & \sup_{\boldsymbol{\beta}}\Biggl\| \frac{1}{n} \sum_{j=1}^{J^{(2)}}  \sum_{i=1}^{n_{j}^{(2)}} \left[ e^{ \left( \boldsymbol{\beta}_0 + \boldsymbol{\beta}_1 \boldsymbol{A}_{j}^{(2,n^{(1)})} + \boldsymbol{\beta}_2 \boldsymbol{z}_j^{(2)} \right)} e^{ \left( \boldsymbol{\beta}_0^* + \boldsymbol{\beta}_1^* \boldsymbol{A}_{j}^{(2,n^{(1)})} + \boldsymbol{\beta}_2^* \boldsymbol{z}_j^{(2)} \right)} \right. \\
    &\qquad \qquad \qquad \qquad \qquad \qquad \qquad \qquad  \left. - e^{ \left( \boldsymbol{\beta}_0 + \boldsymbol{\beta}_1 \boldsymbol{a}_j^{(2)} + \boldsymbol{\beta}_2 \boldsymbol{z}_j^{(2)} \right)} e^{ \left( \boldsymbol{\beta}_0^* + \boldsymbol{\beta}_1^* \boldsymbol{a}_j^{(2)} + \boldsymbol{\beta}_2^* \boldsymbol{z}_j^{(2)} \right)} \right] \Biggr\|\\
    &= \sup_{\boldsymbol{\beta}}\Biggl\| \frac{1}{n} \sum_{j=1}^{J^{(2)}}  \sum_{i=1}^{n_{j}^{(2)}} \left[ e^{\left( (\boldsymbol{\beta}_0+\boldsymbol{\beta}_0^*)+(\boldsymbol{\beta}_1+\boldsymbol{\beta}_1^*) \boldsymbol{A}_{j}^{(2,n^{(1)})} + (\boldsymbol{\beta}_2+\boldsymbol{\beta}_2^*)\boldsymbol{z}_j^{(2)}  \right)} \right. \\
    &\qquad \qquad \qquad \qquad \qquad \qquad \qquad \qquad \left. - e^{\left( (\boldsymbol{\beta}_0+\boldsymbol{\beta}_0^*)+(\boldsymbol{\beta}_1+\boldsymbol{\beta}_1^*) \boldsymbol{a}_j^{(2)} + (\boldsymbol{\beta}_2+\boldsymbol{\beta}_2^*)\boldsymbol{z}_j^{(2)}  \right)} \right] \Biggr\| \\ 
    &\leq \sup_{\boldsymbol{\beta}} \left\| \left.\frac{\partial}{\partial a}\right|_{\boldsymbol{\Tilde{a}}_j (\boldsymbol{\beta})} exp\left(  (\boldsymbol{\beta}_0+\boldsymbol{\beta}_0^*)+(\boldsymbol{\beta}_1+\boldsymbol{\beta}_1^*) a + (\boldsymbol{\beta}_2+\boldsymbol{\beta}_2^*)\boldsymbol{z}_j^{(2)}  \right) \right\|\\ 
    &\qquad\qquad\qquad\qquad\qquad\qquad\qquad\qquad\qquad\qquad\qquad\qquad\max_j \left\| \boldsymbol{A}_{j}^{(2,n^{(1)})}-\boldsymbol{a}_j^{(2)} \right\| \\
    & \xrightarrow[]{P} 0.
\end{aligned}
\end{equation*}
The convergence to 0 follows from Assumption \ref{glm_unif_bound_assu} of the main text and Lemma \ref{sup_glm_condition} of the Appendix. 
The same argument can be applied to other entries of $G_{3,n}$ from equation  (\ref{G3_term}), so $\sup _{\boldsymbol{\beta}}\left\|G_{3,n}\right\| \stackrel{P}{\rightarrow} 0$. By applying the same argument to each entry of $G_{4,n}$ from equation (\ref{G4_term}), we conclude that also $\sup _{\boldsymbol{\beta}}\left\|G_{4,n}\right\| \stackrel{P}{\rightarrow} 0$. 

Consider the first term of $G_{5,n}$ (equation (\ref{G5_term})), by Assumption \ref{glm_unif_bound_assu} from the main text and $\alpha_{j k}=\lim _{n \rightarrow \infty} n_j^{(k)} / n$, for the finitely many $j$, it follows that 
\begin{equation*}
\footnotesize{
\begin{aligned}
&\sup _{\boldsymbol{\beta}} \left\| \sum_{j=1}^{J^{(1)}} \left[ \left(\alpha_{j1}-\frac{n_j^{(1)}}{n} \right)
exp \left( \left( \begin{array}{c}
1 \\
\boldsymbol{a}_{j}^{(1)}\\
\boldsymbol{z}_{j}^{(1)} 
\end{array} \right)^T \boldsymbol{\beta}   \right) \left( \begin{array}{c}
1 \\
\boldsymbol{a}_{j}^{(1)}\\
\boldsymbol{z}_{j}^{(1)} 
\end{array} \right) exp \left( \left( \begin{array}{c}
1 \\
\boldsymbol{a}_{j}^{(1)}\\
\boldsymbol{z}_{j}^{(1)} 
\end{array} \right)^T \boldsymbol{\beta}   \right)
\right] \right\| \\
&\leq \sup _{\boldsymbol{\beta}} \max_j \left\|
\left(\alpha_{j1}-\frac{n_j^{(1)}}{n} \right)
exp \left( \left( \begin{array}{c}
1 \\
\boldsymbol{a}_{j}^{(1)}\\
\boldsymbol{z}_{j}^{(1)} 
\end{array} \right)^T \boldsymbol{\beta}   \right) \left( \begin{array}{c}
1 \\
\boldsymbol{a}_{j}^{(1)}\\
\boldsymbol{z}_{j}^{(1)} 
\end{array} \right) exp \left( \left( \begin{array}{c}
1 \\
\boldsymbol{a}_{j}^{(1)}\\
\boldsymbol{z}_{j}^{(1)} 
\end{array} \right)^T \boldsymbol{\beta}   \right)
\right\| \\
&\xrightarrow{P} 0.
\end{aligned}
}
\end{equation*}
By applying the same argument to the other terms of $G_{5,n}$ and the triangle inequality,
$\sup _{\boldsymbol{\beta}}\left\|G_{5,n}\right\| \stackrel{P}{\rightarrow} 0$.

Thus,
\begin{equation*}
\begin{aligned}
    &\sup_{\boldsymbol{\beta}}\left\|\boldsymbol{U}(\boldsymbol{\beta}) - \boldsymbol{u}(\boldsymbol{\beta})\right\| \\
    &\leq \sup_{\boldsymbol{\beta}}\left\|G_{1,n} \right\|+\sup_{\boldsymbol{\beta}}\left\|G_{2,n} \right\|+ \sup_{\boldsymbol{\beta}}\left\|G_{3,n} \right\| + \sup_{\boldsymbol{\beta}}\left\|G_{4,n} \right\| + 
    \sup_{\boldsymbol{\beta}}\left\|G_{5,n} \right\|\\
    & \xrightarrow{P} 0.
\end{aligned}
\end{equation*}

The second condition in Theorem 5.9 of \citet{van2000asymptotic} is 
\begin{equation*}
\inf _{\boldsymbol{\beta}:\left\|\boldsymbol{\beta}-\boldsymbol{\beta}^{\star}\right\|>0}\|\boldsymbol{u}(\boldsymbol{\beta})\|>0=\left\|\boldsymbol{u}\left(\boldsymbol{\beta}^{\star}\right)\right\|.
\end{equation*}
Equation (\ref{little_u_1}) implies that $\left\|\boldsymbol{u}\left(\boldsymbol{\beta}^{\star}\right)\right\| = 0$. Furthermore, $\boldsymbol{u}\left(\boldsymbol{\beta}\right)$ is the same as for a fixed two-stage design with $a_1^{(1)},\cdots, a_{J^{(1)}}^{(1)}, a_1^{(2)},\cdots, a_{J^{(2)}}^{(2)}$ as interventions decided on before the trial, 
so regular GEE theory applies here. In order for LAGO to work properly, we need variations in the intervention components to identify the treatment effect parameter. The uniqueness of $\boldsymbol{\beta}^*$ as a maximizer or zero has been studied by various authors, see e.g. Chapter 2.2 of \cite{fahrmeir2013multivariate}. Thus, provide there is enough variation in the intervention, the second condition in Theorem 5.9 of \citet{van2000asymptotic} is often also satisfied and we conclude that $\hat{\boldsymbol{\beta}}$ is consistent.

\subsection{Proof of Theorem 2 for LAGO GLM with log link function: asymptotic normality of \texorpdfstring{$\hat{\boldsymbol{\beta}}$}.} \label{an_appendix_log_link}
Under Assumptions \ref{glm_model_assumption} \textendash{} \ref{ind_error} from the main text and Lemma \ref{sup_glm_condition} from the Appendix, we show that 
\begin{equation}
    \sqrt{n}\left(\hat{\boldsymbol{\beta}}-\boldsymbol{\beta}^{*}\right) \xrightarrow{D} N(0, J\left(\boldsymbol{\beta}^{*}\right)^{-1} V\left(\boldsymbol{\beta}^*\right) J\left(\boldsymbol{\beta}^{*}\right)^{-1}),
\label{aympstotic_normality}
\end{equation}
where 
the explicit forms of $J\left(\boldsymbol{\beta}^{*}\right)$ and $V\left(\boldsymbol{\beta}^*\right)$ are
\begin{equation}
\footnotesize{
\begin{aligned}
J\left(\boldsymbol{\beta}^{*}\right) &= \sum_{j=1}^{J^{(1)}} \alpha_{j1}    \left\{ \; exp \left( 2 \left( \begin{array}{c}
1 \\
\boldsymbol{a}_{j}^{(1)}\\
\boldsymbol{z}_{j}^{(1)} 
\end{array} \right)^T \boldsymbol{\beta}^*   \right)   \left( \begin{array}{c}
1 \\
\boldsymbol{a}_{j}^{(1)}\\
\boldsymbol{z}_{j}^{(1)} 
\end{array} \right) \left( \begin{array}{c}
1 \\
\boldsymbol{a}_{j}^{(1)}\\
\boldsymbol{z}_{j}^{(1)} 
\end{array} \right)^T  \right\}\\ 
&\qquad \qquad \qquad \qquad \qquad  + \sum_{j=1}^{J^{(2)}} \alpha_{j2} \left\{ \;  exp \left( 2 \left( \begin{array}{c}
1 \\
\boldsymbol{a}_{j}^{(2)} \\
\boldsymbol{z}_{j}^{(2)}  
\end{array} \right)^T \boldsymbol{\beta}^*   \right)  \left( \begin{array}{c}
1 \\
\boldsymbol{a}_{j}^{(2)} \\
\boldsymbol{z}_{j}^{(2)}  
\end{array} \right) \left( \begin{array}{c}
1 \\
\boldsymbol{a}_{j}^{(2)} \\ 
\boldsymbol{z}_{j}^{(2)}  
\end{array} \right)^T  \right\},
\end{aligned}
}
\label{ddbeta_u_beta_tilde_converge}
\end{equation}
\begin{equation}
\footnotesize{
\begin{aligned}
V\left(\boldsymbol{\beta}^*\right) &= \sum_{j=1}^{J^{(1)}} \alpha_{j1} \; \left\{ exp \left(2 \left( \begin{array}{c}
1 \\
\boldsymbol{a}_{j}^{(1)}\\
\boldsymbol{z}_{j}^{(1)} 
\end{array} \right)^T \boldsymbol{\beta}^*   \right) \left( \begin{array}{c}
1 \\
\boldsymbol{a}_{j}^{(1)}\\
\boldsymbol{z}_{j}^{(1)} 
\end{array} \right)\left( \begin{array}{c}
1 \\
\boldsymbol{a}_{j}^{(1)}\\
\boldsymbol{z}_{j}^{(1)} 
\end{array} \right)^T \right\}
\sigma^2(\boldsymbol{z}_{j}^{(1)})\\
&\qquad \qquad \qquad \qquad \qquad+ \sum_{j=1}^{J^{(2)}} \alpha_{j2} \; \left\{ exp \left(2 \left( \begin{array}{c}
1 \\
\boldsymbol{a}_{j}^{(2)}\\
\boldsymbol{z}_{j}^{(2)} 
\end{array} \right)^T \boldsymbol{\beta}^*   \right) \left( \begin{array}{c}
1 \\
\boldsymbol{a}_{j}^{(2)}\\
\boldsymbol{z}_{j}^{(2)} 
\end{array} \right)\left( \begin{array}{c}
1 \\
\boldsymbol{a}_{j}^{(2)}\\
\boldsymbol{z}_{j}^{(2)} 
\end{array} \right)^T \right\}
\sigma^2(\boldsymbol{z}_{j}^{(2)}).
\end{aligned}
}
\label{nominator_variance}
\end{equation}
The corresponding estimators are
\begin{equation*}
\footnotesize{
\begin{aligned}
J\left(\hat{\boldsymbol{\beta}}\right) &= \sum_{j=1}^{J^{(1)}} \frac{n_{j}^{(1)}}{n}    \left\{ \; exp \left( 2 \left( \begin{array}{c}
1 \\
\boldsymbol{a}_{j}^{(1)}\\
\boldsymbol{z}_{j}^{(1)} 
\end{array} \right)^T \hat{\boldsymbol{\beta}}   \right)   \left( \begin{array}{c}
1 \\
\boldsymbol{a}_{j}^{(1)}\\
\boldsymbol{z}_{j}^{(1)} 
\end{array} \right) \left( \begin{array}{c}
1 \\
\boldsymbol{a}_{j}^{(1)}\\
\boldsymbol{z}_{j}^{(1)} 
\end{array} \right)^T  \right\}\\ 
&\qquad \qquad \qquad \qquad   + \sum_{j=1}^{J^{(2)}} \frac{n_{j}^{(2)}}{n} \left\{ \;  exp \left( 2 \left( \begin{array}{c}
1 \\
\boldsymbol{A}_{j}^{(2,n^{(1)})} \\
\boldsymbol{z}_{j}^{(2)}  
\end{array} \right)^T \hat{\boldsymbol{\beta}}   \right)  \left( \begin{array}{c}
1 \\
\boldsymbol{A}_{j}^{(2,n^{(1)})} \\
\boldsymbol{z}_{j}^{(2)}  
\end{array} \right) \left( \begin{array}{c}
1 \\
\boldsymbol{A}_{j}^{(2,n^{(1)})} \\ 
\boldsymbol{z}_{j}^{(2)}  
\end{array} \right)^T  \right\},
\end{aligned}
}
\end{equation*}
\begin{equation*}
\footnotesize{
\begin{aligned}
&V\left(\hat{\boldsymbol{\beta}}\right) \\
&= \frac{1}{n} \sum_{j=1}^{J^{(1)}} \; \left\{ exp \left(2 \left( \begin{array}{c}
1 \\
\boldsymbol{a}_{j}^{(1)}\\
\boldsymbol{z}_{j}^{(1)} 
\end{array} \right)^T \hat{\boldsymbol{\beta}}   \right) \left( \begin{array}{c}
1 \\
\boldsymbol{a}_{j}^{(1)}\\
\boldsymbol{z}_{j}^{(1)} 
\end{array} \right)\left( \begin{array}{c}
1 \\
\boldsymbol{a}_{j}^{(1)}\\
\boldsymbol{z}_{j}^{(1)} 
\end{array} \right)^T 
\sum_{i=1}^{n_j^{(1)}}  \left( Y_{ij}^{(1)} -exp\left( \left( \begin{array}{c}
1 \\
\boldsymbol{a}_{j}^{(1)} \\
\boldsymbol{z}_{j}^{(1)}
\end{array} \right)^T \hat{\boldsymbol{\beta}}  \right)   \right)^2
\right\}\\
&+ \frac{1}{n} \sum_{j=1}^{J^{(2)}}  \; \left\{ exp \left(2 \left( \begin{array}{c}
1 \\
\boldsymbol{A}_{j}^{(2,n^{(1)})}\\
\boldsymbol{z}_{j}^{(2)} 
\end{array} \right)^T \hat{\boldsymbol{\beta}}   \right) \left( \begin{array}{c}
1 \\
\boldsymbol{A}_{j}^{(2,n^{(1)})}\\
\boldsymbol{z}_{j}^{(2)} 
\end{array} \right)\left( \begin{array}{c}
1 \\
\boldsymbol{A}_{j}^{(2,n^{(1)})}\\
\boldsymbol{z}_{j}^{(2)} 
\end{array} \right)^T \right.\\
&\qquad\qquad\qquad\qquad\qquad\qquad\qquad\qquad\qquad\qquad
\left.\sum_{i=1}^{n_j^{(2)}}  \left( Y_{ij}^{(2,n^{(1)})} -exp\left( \left( \begin{array}{c}
1 \\
\boldsymbol{A}_{j}^{(2,n^{(1)})} \\
\boldsymbol{z}_{j}^{(2)}
\end{array} \right)^T \hat{\boldsymbol{\beta}}  \right)   \right) ^2
\right\}.
\end{aligned}
}
\end{equation*}

By applying the mean value theorem to each component of $\boldsymbol{U}(\boldsymbol{\beta})$ from equation (\ref{big_U_1}), we get
\begin{equation*}
    0=\boldsymbol{U}(\hat{\boldsymbol{\beta}}) = \boldsymbol{U}(\boldsymbol{\beta}^*) + \left(\left.\frac{\partial}{\partial \boldsymbol{\beta}}\right|_{\Tilde{\boldsymbol{\beta}}} \boldsymbol{U}({\boldsymbol{\beta}})\right) (\hat{\boldsymbol{\beta}} - {\boldsymbol{\beta}^*})^T ,
\end{equation*}
where for each row of $\left.\frac{\partial}{\partial \boldsymbol{\beta}}\right|_{\Tilde{\boldsymbol{\beta}}} \boldsymbol{U}({\boldsymbol{\beta}})$, $\Tilde{\boldsymbol{\beta}}$ takes a possibly row-dependent value between $\hat{\boldsymbol{\beta}}$ and $\boldsymbol{\beta}^*$.
It follows that 
\begin{equation}
\sqrt{n}\left(\hat{\boldsymbol{\beta}}-\boldsymbol{\beta}^{*}\right) = -\sqrt{n}\left(\left.\frac{\partial}{\partial \boldsymbol{\beta}}\right|_{\tilde{\boldsymbol{\beta}}} \boldsymbol{U}\left({\boldsymbol{\beta}}\right)\right)^{-1} \boldsymbol{U}\left(\boldsymbol{\beta}^{*}\right).
\label{AN_break_down}
\end{equation}
We first show that $\left(-\left.\frac{\partial}{\partial \boldsymbol{\beta}}\right|_{\tilde{\boldsymbol{\beta}}} \boldsymbol{U}\left({\boldsymbol{\beta}}\right)\right)$ converges in probability to $J\left(\boldsymbol{\beta}^{*}\right)$, then we show that $\sqrt{n} \:\: \boldsymbol{U}\left(\boldsymbol{\beta}^{*}\right)$ converges to a normal distribution with mean 0 and variance $V\left(\boldsymbol{\beta}^*\right)$. Equation (\ref{aympstotic_normality}) then follows from Slutsky's theorem.

First, combining equation (\ref{glm_EE}) from the main text and equation (\ref{big_U_1}), we derive
\begin{equation*}
\footnotesize{
\begin{aligned}
&\frac{\partial}{\partial \boldsymbol{\beta}} \boldsymbol{U}({\boldsymbol{\beta}}) \\
&= \frac{1}{n} \sum_{j=1}^{J^{(1)}}  \sum_{i=1}^{n_{j}^{(1)}}  \frac{\partial}{\partial \boldsymbol{\beta}} \left[ exp \left( \left( \begin{array}{c}
1 \\
\boldsymbol{a}_{j}^{(1)}\\
\boldsymbol{z}_{j}^{(1)} 
\end{array} \right)^T \boldsymbol{\boldsymbol{\beta}}   \right) \left( \begin{array}{c}
1 \\
\boldsymbol{a}_{j}^{(1)}\\
\boldsymbol{z}_{j}^{(1)} 
\end{array} \right) \left( Y_{ij}^{(1)} - exp\left( \left( \begin{array}{c}
1 \\
\boldsymbol{a}_{j}^{(1)}\\
\boldsymbol{z}_{j}^{(1)} 
\end{array} \right)^T \boldsymbol{\boldsymbol{\beta}}  \right)   \right) \right]\\
&+ \frac{1}{n} \sum_{j=1}^{J^{(2)}}  \sum_{i=1}^{n_{j}^{(2)}} \frac{\partial}{\partial \boldsymbol{\beta}} \left[ exp \left( \left( \begin{array}{c}
1 \\
\boldsymbol{A}_{j}^{(2,n^{(1)})} \\
\boldsymbol{z}_{j}^{(2)}  
\end{array} \right)^T \boldsymbol{\boldsymbol{\beta}}   \right)\left( \begin{array}{c}
1 \\
\boldsymbol{A}_{j}^{(2,n^{(1)})} \\
\boldsymbol{z}_{j}^{(2)}  
\end{array} \right) \left(Y_{ij}^{(2,n^{(1)})} -exp\left( \left( \begin{array}{c}
1 \\
\boldsymbol{A}_{j}^{(2,n^{(1)})} \\
\boldsymbol{z}_{j}^{(2)}  
\end{array} \right)^T \boldsymbol{\boldsymbol{\beta}}  \right)   \right) \right] \\
&= \frac{1}{n} \sum_{j=1}^{J^{(1)}}  \sum_{i=1}^{n_{j}^{(1)}}  \frac{\partial}{\partial \boldsymbol{\beta}} \left[ exp \left( \left( \begin{array}{c}
1 \\
\boldsymbol{a}_{j}^{(1)}\\
\boldsymbol{z}_{j}^{(1)} 
\end{array} \right)^T \boldsymbol{\boldsymbol{\beta}}   \right) \left( \begin{array}{c}
1 \\
\boldsymbol{a}_{j}^{(1)}\\
\boldsymbol{z}_{j}^{(1)} 
\end{array} \right) Y_{ij}^{(1)} - \left( exp \left( 2 \left( \begin{array}{c}
1 \\
\boldsymbol{a}_{j}^{(1)}\\
\boldsymbol{z}_{j}^{(1)} 
\end{array} \right)^T \boldsymbol{\boldsymbol{\beta}}   \right) \right) \left( \begin{array}{c}
1 \\
\boldsymbol{a}_{j}^{(1)}\\
\boldsymbol{z}_{j}^{(1)} 
\end{array} \right)  \right]\\ 
&+  \frac{1}{n} \sum_{j=1}^{J^{(2)}}  \sum_{i=1}^{n_{j}^{(2)}}  \frac{\partial}{\partial \boldsymbol{\beta}} \left[ exp \left( \left( \begin{array}{c}
1 \\
\boldsymbol{A}_{j}^{(2,n^{(1)})} \\
\boldsymbol{z}_{j}^{(2)}  
\end{array} \right)^T \boldsymbol{\boldsymbol{\beta}}   \right) \left( \begin{array}{c}
1 \\
\boldsymbol{A}_{j}^{(2,n^{(1)})} \\
\boldsymbol{z}_{j}^{(2)}  
\end{array} \right) Y_{ij}^{(2, n^{(1)})} \right.\\
&\qquad \qquad \qquad \qquad \qquad \qquad \qquad \qquad \qquad \left.-  \left( exp \left( 2 \left( \begin{array}{c}
1 \\
\boldsymbol{A}_{j}^{(2,n^{(1)})} \\
\boldsymbol{z}_{j}^{(2)}  
\end{array} \right)^T \boldsymbol{\boldsymbol{\beta}}   \right) \right) \left( \begin{array}{c}
1 \\
\boldsymbol{A}_{j}^{(2,n^{(1)})} \\
\boldsymbol{z}_{j}^{(2)}  
\end{array} \right)  \right]\\
&= G_{6,n}(\boldsymbol{\beta}) + G_{7,n}(\boldsymbol{\beta}),
\end{aligned}
}
\end{equation*}
where 
\begin{equation*}
\footnotesize{
\begin{aligned}
G_{6,n}(\boldsymbol{\beta})&= \frac{1}{n} \sum_{j=1}^{J^{(1)}}  \sum_{i=1}^{n_{j}^{(1)}}  \left[ exp \left( \left( \begin{array}{c}
1 \\
\boldsymbol{a}_{j}^{(1)}\\
\boldsymbol{z}_{j}^{(1)} 
\end{array} \right)^T \boldsymbol{\beta}   \right) \left( \begin{array}{c}
1 \\
\boldsymbol{a}_{j}^{(1)}\\
\boldsymbol{z}_{j}^{(1)} 
\end{array} \right) \left( \begin{array}{c}
1 \\
\boldsymbol{a}_{j}^{(1)}\\
\boldsymbol{z}_{j}^{(1)} 
\end{array} \right)^T  Y_{ij}^{(1)} \right.\\   
&\qquad\qquad\qquad\qquad\qquad\qquad \left. - 2 \; \left( exp \left( 2 \left( \begin{array}{c}
1 \\
\boldsymbol{a}_{j}^{(1)}\\
\boldsymbol{z}_{j}^{(1)} 
\end{array} \right)^T \boldsymbol{\beta}   \right) \right) \left( \begin{array}{c}
1 \\
\boldsymbol{a}_{j}^{(1)}\\
\boldsymbol{z}_{j}^{(1)} 
\end{array} \right) \left( \begin{array}{c}
1 \\
\boldsymbol{a}_{j}^{(1)}\\
\boldsymbol{z}_{j}^{(1)} 
\end{array} \right)^T  \right],\\
G_{7,n}(\boldsymbol{\beta}) &= \frac{1}{n} \sum_{j=1}^{J^{(2)}}  \sum_{i=1}^{n_{j}^{(2)}}  \left[ exp \left( \left( \begin{array}{c}
1 \\
\boldsymbol{A}_{j}^{(2,n^{(1)})} \\
\boldsymbol{z}_{j}^{(2)}  
\end{array} \right)^T \boldsymbol{\beta}   \right) \left( \begin{array}{c}
1 \\
\boldsymbol{A}_{j}^{(2,n^{(1)})} \\
\boldsymbol{z}_{j}^{(2)}  
\end{array} \right) \left( \begin{array}{c}
1 \\
\boldsymbol{A}_{j}^{(2,n^{(1)})} \\
\boldsymbol{z}_{j}^{(2)}  
\end{array} \right)^T  Y_{ij}^{(2, n^{(1)})}\right.\\
&\left. \qquad \qquad \qquad \qquad \qquad - 2 \; \left( exp \left( 2 \left( \begin{array}{c}
1 \\
\boldsymbol{A}_{j}^{(2,n^{(1)})} \\
\boldsymbol{z}_{j}^{(2)}  
\end{array} \right)^T \boldsymbol{\beta}   \right) \right) \left( \begin{array}{c}
1 \\
\boldsymbol{A}_{j}^{(2,n^{(1)})} \\
\boldsymbol{z}_{j}^{(2)}  
\end{array} \right) \left( \begin{array}{c}
1 \\
\boldsymbol{A}_{j}^{(2,n^{(1)})} \\ 
\boldsymbol{z}_{j}^{(2)}  
\end{array} \right)^T  \right].
\end{aligned}
}
\end{equation*}
Let 
\begin{equation*}
\footnotesize{
\begin{aligned}
G_{6,n}^*(\boldsymbol{\beta}) &= \frac{1}{n} \sum_{j=1}^{J^{(1)}}  \sum_{i=1}^{n_{j}^{(1)}}  \left[ exp \left( \left( \begin{array}{c}
1 \\
\boldsymbol{a}_{j}^{(1)}\\
\boldsymbol{z}_{j}^{(1)} 
\end{array} \right)^T \boldsymbol{\beta}   \right) \left( \begin{array}{c}
1 \\
\boldsymbol{a}_{j}^{(1)}\\
\boldsymbol{z}_{j}^{(1)} 
\end{array} \right) \left( \begin{array}{c}
1 \\
\boldsymbol{a}_{j}^{(1)}\\
\boldsymbol{z}_{j}^{(1)} 
\end{array} \right)^T exp \left( \left( \begin{array}{c}
1 \\
\boldsymbol{a}_{j}^{(1)}\\
\boldsymbol{z}_{j}^{(1)} 
\end{array} \right)^T \boldsymbol{\beta}^*   \right) \right.\\   
&\qquad\qquad\qquad\qquad\qquad\qquad \left. - 2 \; \left( exp \left( 2 \left( \begin{array}{c}
1 \\
\boldsymbol{a}_{j}^{(1)}\\
\boldsymbol{z}_{j}^{(1)} 
\end{array} \right)^T \boldsymbol{\beta}   \right) \right) \left( \begin{array}{c}
1 \\
\boldsymbol{a}_{j}^{(1)}\\
\boldsymbol{z}_{j}^{(1)} 
\end{array} \right) \left( \begin{array}{c}
1 \\
\boldsymbol{a}_{j}^{(1)}\\
\boldsymbol{z}_{j}^{(1)} 
\end{array} \right)^T  \right],
\end{aligned}
}
\end{equation*}
\begin{equation*}
\footnotesize{
\begin{aligned}
G_{7,n}^*(\boldsymbol{\beta}) &= \frac{1}{n} \sum_{j=1}^{J^{(2)}}  \sum_{i=1}^{n_{j}^{(2)}}  \left[ exp \left( \left( \begin{array}{c}
1 \\
\boldsymbol{a}_{j}^{(2)} \\
\boldsymbol{z}_{j}^{(2)}  
\end{array} \right)^T \boldsymbol{\beta}   \right) \left( \begin{array}{c}
1 \\
\boldsymbol{a}_{j}^{(2)} \\
\boldsymbol{z}_{j}^{(2)}  
\end{array} \right) \left( \begin{array}{c}
1 \\
\boldsymbol{a}_{j}^{(2)} \\
\boldsymbol{z}_{j}^{(2)}  
\end{array} \right)^T exp \left( \left( \begin{array}{c}
1 \\
\boldsymbol{a}_{j}^{(2)} \\
\boldsymbol{z}_{j}^{(2)}  
\end{array} \right)^T \boldsymbol{\beta}^*   \right) \right.\\
&\left. \qquad \qquad \qquad \qquad \qquad \qquad \qquad- 2 \; \left( exp \left( 2 \left( \begin{array}{c}
1 \\
\boldsymbol{a}_{j}^{(2)} \\
\boldsymbol{z}_{j}^{(2)}  
\end{array} \right)^T \boldsymbol{\beta}   \right) \right) \left( \begin{array}{c}
1 \\
\boldsymbol{a}_{j}^{(2)} \\
\boldsymbol{z}_{j}^{(2)}  
\end{array} \right) \left( \begin{array}{c}
1 \\
\boldsymbol{a}_{j}^{(2)} \\ 
\boldsymbol{z}_{j}^{(2)}  
\end{array} \right)^T  \right].
\end{aligned}
}
\end{equation*}
We show that $\frac{\partial}{\partial \boldsymbol{\beta}} \boldsymbol{U}({\boldsymbol{\tilde{\beta}}}) - \left( G_{6,n}^*(\boldsymbol{\beta}^*)+G_{7,n}^*(\boldsymbol{\beta}^*) \right)$ converges in probability to 0. Notice that
\begin{equation*}
\begin{aligned}
\frac{\partial}{\partial \boldsymbol{\beta}} \boldsymbol{U}({\boldsymbol{\tilde{\beta}}})  - \left( G_{6,n}^*(\boldsymbol{\beta}^*)+G_{7,n}^*(\boldsymbol{\beta}^*) \right)
&= \left(G_{6,n}(\boldsymbol{\tilde{\beta}}) - G_{6,n}^*(\boldsymbol{\beta}^*)\right) + \left(G_{7,n}(\boldsymbol{\tilde{\beta}}) -  G_{7,n}^*(\boldsymbol{\beta}^*)\right).\\
\end{aligned}
\end{equation*}
We will show that both terms converge in probability to 0. For the first term, using 
{\footnotesize
$$Y_{ij}^{(1)} = exp \left( \left( \begin{array}{c}
1 \\
\boldsymbol{a}_{j}^{(1)} \\
\boldsymbol{z}_{j}^{(1)}  
\end{array} \right)^T \boldsymbol{\beta}^*   \right) + \epsilon_{ij}^{(1)} ,$$
}
\begin{equation}
\footnotesize{
\begin{aligned}
&G_{6,n}(\boldsymbol{\tilde{\beta}}) - G_{6,n}^*(\boldsymbol{\beta}^*) \\
&= \frac{1}{n} \sum_{j=1}^{J^{(1)}}  \sum_{i=1}^{n_{j}^{(1)}}  \left[ exp \left( \left( \begin{array}{c}
1 \\
\boldsymbol{a}_{j}^{(1)}\\
\boldsymbol{z}_{j}^{(1)} 
\end{array} \right)^T \boldsymbol{\tilde{\beta}}   \right) \left( \begin{array}{c}
1 \\
\boldsymbol{a}_{j}^{(1)}\\
\boldsymbol{z}_{j}^{(1)} 
\end{array} \right) \left( \begin{array}{c}
1 \\
\boldsymbol{a}_{j}^{(1)}\\
\boldsymbol{z}_{j}^{(1)} 
\end{array} \right)^T \left( exp \left( \left( \begin{array}{c}
1 \\
\boldsymbol{a}_{j}^{(1)} \\
\boldsymbol{z}_{j}^{(1)}  
\end{array} \right)^T \boldsymbol{\beta}^*   \right) + \epsilon_{ij}^{(1)} \right) \right.\\   
&\qquad\qquad\qquad\qquad\qquad\qquad  - 2 \; \left( exp \left( 2 \left( \begin{array}{c}
1 \\
\boldsymbol{a}_{j}^{(1)}\\
\boldsymbol{z}_{j}^{(1)} 
\end{array} \right)^T \boldsymbol{\tilde{\beta}}   \right) \right) \left( \begin{array}{c}
1 \\
\boldsymbol{a}_{j}^{(1)}\\
\boldsymbol{z}_{j}^{(1)} 
\end{array} \right) \left( \begin{array}{c}
1 \\
\boldsymbol{a}_{j}^{(1)}\\
\boldsymbol{z}_{j}^{(1)} 
\end{array} \right)^T \\
&\qquad \qquad \qquad \qquad \qquad - exp \left( \left( \begin{array}{c}
1 \\
\boldsymbol{a}_{j}^{(1)}\\
\boldsymbol{z}_{j}^{(1)} 
\end{array} \right)^T \boldsymbol{\beta}^*   \right) \left( \begin{array}{c}
1 \\
\boldsymbol{a}_{j}^{(1)}\\
\boldsymbol{z}_{j}^{(1)} 
\end{array} \right) \left( \begin{array}{c}
1 \\
\boldsymbol{a}_{j}^{(1)}\\
\boldsymbol{z}_{j}^{(1)} 
\end{array} \right)^T exp \left( \left( \begin{array}{c}
1 \\
\boldsymbol{a}_{j}^{(1)}\\
\boldsymbol{z}_{j}^{(1)} 
\end{array} \right)^T \boldsymbol{\beta}^*   \right) \\   
&\qquad\qquad\qquad\qquad\qquad\qquad \left. + 2 \; \left( exp \left( 2 \left( \begin{array}{c}
1 \\
\boldsymbol{a}_{j}^{(1)}\\
\boldsymbol{z}_{j}^{(1)} 
\end{array} \right)^T \boldsymbol{\beta}^*   \right) \right) \left( \begin{array}{c}
1 \\
\boldsymbol{a}_{j}^{(1)}\\
\boldsymbol{z}_{j}^{(1)} 
\end{array} \right) \left( \begin{array}{c}
1 \\
\boldsymbol{a}_{j}^{(1)}\\
\boldsymbol{z}_{j}^{(1)} 
\end{array} \right)^T  \right]. 
\end{aligned}
}
\label{G5betatilda}
\end{equation}
Notice that the most complicated entry of 
\begin{equation}
\footnotesize{
\begin{aligned}
&\frac{1}{n} \sum_{j=1}^{J^{(1)}}  \sum_{i=1}^{n_{j}^{(1)}} \left[ exp \left( \left( \begin{array}{c}
1 \\
\boldsymbol{a}_{j}^{(1)}\\
\boldsymbol{z}_{j}^{(1)} 
\end{array} \right)^T \boldsymbol{\tilde{\beta}}   \right) \left( \begin{array}{c}
1 \\
\boldsymbol{a}_{j}^{(1)}\\
\boldsymbol{z}_{j}^{(1)} 
\end{array} \right) \left( \begin{array}{c}
1 \\
\boldsymbol{a}_{j}^{(1)}\\
\boldsymbol{z}_{j}^{(1)} 
\end{array} \right)^T \left( exp \left( \left( \begin{array}{c}
1 \\
\boldsymbol{a}_{j}^{(1)} \\
\boldsymbol{z}_{j}^{(1)}  
\end{array} \right)^T \boldsymbol{\beta}^*   \right) \right) \right. \\
&\left.\qquad \qquad \qquad - exp \left( \left( \begin{array}{c}
1 \\
\boldsymbol{a}_{j}^{(1)}\\
\boldsymbol{z}_{j}^{(1)} 
\end{array} \right)^T \boldsymbol{\beta}^*   \right) \left( \begin{array}{c}
1 \\
\boldsymbol{a}_{j}^{(1)}\\
\boldsymbol{z}_{j}^{(1)} 
\end{array} \right) \left( \begin{array}{c}
1 \\
\boldsymbol{a}_{j}^{(1)}\\
\boldsymbol{z}_{j}^{(1)} 
\end{array} \right)^T exp \left( \left( \begin{array}{c}
1 \\
\boldsymbol{a}_{j}^{(1)}\\
\boldsymbol{z}_{j}^{(1)} 
\end{array} \right)^T \boldsymbol{\beta}^*   \right) \right]\\
\label{G5-G5*matrix form}
\end{aligned}
}
\end{equation}
is 
\begin{equation}
\begin{aligned}
&\frac{1}{n} \sum_{j=1}^{J^{(1)}}  \sum_{i=1}^{n_{j}^{(1)}} \left[
e^{\left(\boldsymbol{\tilde{\beta}_0} + \boldsymbol{\tilde{\beta}_1}\boldsymbol{a}_{j}^{(1)} + \boldsymbol{\tilde{\beta}_2}\boldsymbol{z}_{j}^{(1)}\right)}
\left(\boldsymbol{a}_{j}^{(1)}\right)^{\otimes 2} 
e^{\left(\boldsymbol{{\beta}_0^*} + \boldsymbol{{\beta}_1^*}\boldsymbol{a}_{j}^{(1)} + \boldsymbol{{\beta}_2^*}\boldsymbol{z}_{j}^{(1)}\right)} \right. \\
&\left.\qquad \qquad \qquad\qquad \qquad \qquad - e^{\left(\boldsymbol{{\beta}_0^*} + \boldsymbol{{\beta}_1^*}\boldsymbol{a}_{j}^{(1)} + \boldsymbol{{\beta}_2^*}\boldsymbol{z}_{j}^{(1)}\right)}  \left(\boldsymbol{a}_{j}^{(1)}\right)^{\otimes 2}
e^{\left(\boldsymbol{{\beta}_0^*} + \boldsymbol{{\beta}_1^*}\boldsymbol{a}_{j}^{(1)} + \boldsymbol{{\beta}_2^*}\boldsymbol{z}_{j}^{(1)}\right)} \right] \\
&= \frac{1}{n} \sum_{j=1}^{J^{(1)}}  \sum_{i=1}^{n_{j}^{(1)}} \left[ 
\left( e^{\left(\boldsymbol{\tilde{\beta}_0} + \boldsymbol{\tilde{\beta}_1}\boldsymbol{a}_{j}^{(1)} + \boldsymbol{\tilde{\beta}_2}\boldsymbol{z}_{j}^{(1)}\right)} - e^{\left(\boldsymbol{{\beta}_0^*} + \boldsymbol{{\beta}_1^*}\boldsymbol{a}_{j}^{(1)} + \boldsymbol{{\beta}_2^*}\boldsymbol{z}_{j}^{(1)}\right)} \right) \left(\boldsymbol{a}_{j}^{(1)}\right)^{\otimes 2}
e^{\left(\boldsymbol{{\beta}_0^*} + \boldsymbol{{\beta}_1^*}\boldsymbol{a}_{j}^{(1)} + \boldsymbol{{\beta}_2^*}\boldsymbol{z}_{j}^{(1)}\right)} \right]\\
&=  \sum_{j=1}^{J^{(1)}} \frac{n_{j}^{(1)}}{n} \left[ 
\left( e^{\left(\boldsymbol{\tilde{\beta}_0} + \boldsymbol{\tilde{\beta}_1}\boldsymbol{a}_{j}^{(1)} + \boldsymbol{\tilde{\beta}_2}\boldsymbol{z}_{j}^{(1)}\right)} - e^{\left(\boldsymbol{{\beta}_0^*} + \boldsymbol{{\beta}_1^*}\boldsymbol{a}_{j}^{(1)} + \boldsymbol{{\beta}_2^*}\boldsymbol{z}_{j}^{(1)}\right)} \right) \left(\boldsymbol{a}_{j}^{(1)}\right)^{\otimes 2}
e^{\left(\boldsymbol{{\beta}_0^*} + \boldsymbol{{\beta}_1^*}\boldsymbol{a}_{j}^{(1)} + \boldsymbol{{\beta}_2^*}\boldsymbol{z}_{j}^{(1)}\right)} \right].
\end{aligned}
\label{G5-G5* first part}
\end{equation}
Since $\hat{\boldsymbol{\beta}}$ and therefore $\tilde{\boldsymbol{\beta}}$ are consistent, $\tilde{\boldsymbol{\beta}} \xrightarrow{P} \boldsymbol{\beta}^*$. By Assumption \ref{glm_unif_bound_assu} from the main text and the Continuous Mapping Theorem, $e^{\left(\boldsymbol{\tilde{\beta}_0} + \boldsymbol{\tilde{\beta}_1}\boldsymbol{a}_{j}^{(1)} + \boldsymbol{\tilde{\beta}_2}\boldsymbol{z}_{j}^{(1)}\right)} - e^{\left(\boldsymbol{{\beta}_0^*} + \boldsymbol{{\beta}_1^*}\boldsymbol{a}_{j}^{(1)} + \boldsymbol{{\beta}_2^*}\boldsymbol{z}_{j}^{(1)}\right)} \xrightarrow{P} 0$ for each $j$. Since $J^{(1)}$ is fixed, and rest of the terms within the summation are bounded, ${n_{j}^{(1)}}/{n} < 1$, we conclude that equation (\ref{G5-G5* first part}) goes to 0 in probability. This argument can be applied to other entries of equation (\ref{G5-G5*matrix form}), so we conclude that equation (\ref{G5-G5*matrix form}) goes to 0 in probability.
Similarly,
\begin{equation*}
\footnotesize{
\begin{aligned}
&\frac{2}{n} \sum_{j=1}^{J^{(1)}}  \sum_{i=1}^{n_{j}^{(1)}} \left[- \; \left( exp \left( 2 \left( \begin{array}{c}
1 \\
\boldsymbol{a}_{j}^{(1)}\\
\boldsymbol{z}_{j}^{(1)} 
\end{array} \right)^T \boldsymbol{\tilde{\beta}}   \right) \right) \left( \begin{array}{c}
1 \\
\boldsymbol{a}_{j}^{(1)}\\
\boldsymbol{z}_{j}^{(1)} 
\end{array} \right) \left( \begin{array}{c}
1 \\
\boldsymbol{a}_{j}^{(1)}\\
\boldsymbol{z}_{j}^{(1)} 
\end{array} \right)^T \right.  \\
&\left. \qquad \qquad \qquad + \; \left( exp \left( 2 \left( \begin{array}{c}
1 \\
\boldsymbol{a}_{j}^{(1)}\\
\boldsymbol{z}_{j}^{(1)} 
\end{array} \right)^T \boldsymbol{\beta}^*   \right) \right) \left( \begin{array}{c}
1 \\
\boldsymbol{a}_{j}^{(1)}\\
\boldsymbol{z}_{j}^{(1)} 
\end{array} \right) \left( \begin{array}{c}
1 \\
\boldsymbol{a}_{j}^{(1)}\\
\boldsymbol{z}_{j}^{(1)} 
\end{array} \right)^T \right] \xrightarrow{P} 0.
\end{aligned}
}
\end{equation*}
Next, following the same arguments as for $G_{1,n}$ in the proof of consistency (Appendix \ref{consistency_appendix_loglink}),
\begin{equation*}
\footnotesize{
\begin{aligned}
&\frac{1}{n} \sum_{j=1}^{J^{(1)}}  \sum_{i=1}^{n_{j}^{(1)}} \left[exp \left( \left( \begin{array}{c}
1 \\
\boldsymbol{a}_{j}^{(1)}\\
\boldsymbol{z}_{j}^{(1)} 
\end{array} \right)^T \boldsymbol{\tilde{\beta}}   \right) \left( \begin{array}{c}
1 \\
\boldsymbol{a}_{j}^{(1)}\\
\boldsymbol{z}_{j}^{(1)} 
\end{array} \right) \left( \begin{array}{c}
1 \\
\boldsymbol{a}_{j}^{(1)}\\
\boldsymbol{z}_{j}^{(1)} 
\end{array} \right)^T \epsilon_{ij}^{(1)} \right] \xrightarrow{P} 0.
\end{aligned}
}
\end{equation*}
We conclude that indeed $\left(G_{6,n}(\boldsymbol{\tilde{\beta}}) - G_{6,n}^*(\boldsymbol{\beta}^*)\right) \xrightarrow{P} 0$.

For $\left(G_{7,n}(\boldsymbol{\tilde{\beta}}) -  G_{7,n}^*(\boldsymbol{\beta}^*)\right)$, we show $\sup_{\boldsymbol{\beta}} \| G_{7,n}(\boldsymbol{{\beta}}) -  G_{7,n}^*(\boldsymbol{\beta})\| \xrightarrow{P} 0$, from which $\left(G_{7,n}(\boldsymbol{\tilde{\beta}}) -  G_{7,n}^*(\boldsymbol{\beta}^*)\right) \xrightarrow{P} 0.$ 
Using that 
$$
\footnotesize{
Y_{ij}^{(2,n^{(1)})} = exp \left( \left( \begin{array}{c}
1 \\
\boldsymbol{A}_{j}^{(2,n^{(1)})} \\
\boldsymbol{z}_{j}^{(1)}  
\end{array} \right)^T \boldsymbol{\beta}^*   \right) + \epsilon_{ij}^{(2)}, 
}
$$
we derive 
\begin{equation*}
\footnotesize{
\begin{aligned}
&G_{7,n}(\boldsymbol{{\beta}}) -  G_{7,n}^*(\boldsymbol{\beta}) \\
&= \frac{1}{n} \sum_{j=1}^{J^{(2)}}  \sum_{i=1}^{n_{j}^{(2)}} 
\left[ exp \left( \left( \begin{array}{c}
1 \\
\boldsymbol{A}_{j}^{(2, n^{(1)})} \\
\boldsymbol{z}_{j}^{(2)}  
\end{array} \right)^T \boldsymbol{\beta}   \right) \left( \begin{array}{c}
1 \\
\boldsymbol{A}_{j}^{(2, n^{(1)})} \\
\boldsymbol{z}_{j}^{(2)}  
\end{array} \right) \left( \begin{array}{c}
1 \\
\boldsymbol{A}_{j}^{(2, n^{(1)})} \\
\boldsymbol{z}_{j}^{(2)}  
\end{array} \right)^T Y_{ij}^{(2, n^{(1)})} \right.\\
&\qquad \qquad \qquad  \pm exp \left( \left( \begin{array}{c}
1 \\
\boldsymbol{A}_{j}^{(2, n^{(1)})} \\
\boldsymbol{z}_{j}^{(2)}  
\end{array} \right)^T \boldsymbol{\beta}   \right) \left( \begin{array}{c}
1 \\
\boldsymbol{A}_{j}^{(2, n^{(1)})} \\
\boldsymbol{z}_{j}^{(2)}  
\end{array} \right) \left( \begin{array}{c}
1 \\
\boldsymbol{A}_{j}^{(2, n^{(1)})} \\
\boldsymbol{z}_{j}^{(2)}  
\end{array} \right)^T exp \left( \left( \begin{array}{c}
1 \\
\boldsymbol{A}_{j}^{(2, n^{(1)})} \\
\boldsymbol{z}_{j}^{(2)}  
\end{array} \right)^T \boldsymbol{\beta}^*   \right)\\
&\left. \qquad \qquad  \qquad \qquad  \qquad \qquad- 2 \; exp \left(2 \left( \begin{array}{c}
1 \\
\boldsymbol{A}_{j}^{(2, n^{(1)})} \\
\boldsymbol{z}_{j}^{(2)}  
\end{array} \right)^T \boldsymbol{\beta}   \right)  \left( \begin{array}{c}
1 \\
\boldsymbol{A}_{j}^{(2, n^{(1)})} \\
\boldsymbol{z}_{j}^{(2)}  
\end{array} \right) \left( \begin{array}{c}
1 \\
\boldsymbol{A}_{j}^{(2, n^{(1)})} \\
\boldsymbol{z}_{j}^{(2)}  
\end{array} \right)^T  \right.\\
&\qquad \qquad \qquad  \qquad \qquad - exp \left( \left( \begin{array}{c}
1 \\
\boldsymbol{a}_{j}^{(2)} \\
\boldsymbol{z}_{j}^{(2)}  
\end{array} \right)^T \boldsymbol{\beta}   \right) \left( \begin{array}{c}
1 \\
\boldsymbol{a}_{j}^{(2)} \\
\boldsymbol{z}_{j}^{(2)}  
\end{array} \right) \left( \begin{array}{c}
1 \\
\boldsymbol{a}_{j}^{(2)} \\
\boldsymbol{z}_{j}^{(2)}  
\end{array} \right)^T exp \left( \left( \begin{array}{c}
1 \\
\boldsymbol{a}_{j}^{(2)} \\
\boldsymbol{z}_{j}^{(2)}  
\end{array} \right)^T \boldsymbol{\beta}^*   \right) \\
&\left. \qquad \qquad \qquad \qquad \qquad \qquad  \qquad \qquad+ 2 \;  exp \left(2 \left( \begin{array}{c}
1 \\
\boldsymbol{a}_{j}^{(2)} \\
\boldsymbol{z}_{j}^{(2)}  
\end{array} \right)^T \boldsymbol{\beta}   \right)  \left( \begin{array}{c}
1 \\
\boldsymbol{a}_{j}^{(2)} \\
\boldsymbol{z}_{j}^{(2)}  
\end{array} \right) \left( \begin{array}{c}
1 \\
\boldsymbol{a}_{j}^{(2)} \\ 
\boldsymbol{z}_{j}^{(2)}  
\end{array} \right)^T  \right]\\
&= G_{7\_1,n}(\boldsymbol{\beta}) + G_{7\_2,n}(\boldsymbol{\beta}) + G_{7\_3,n}(\boldsymbol{\beta})
\end{aligned}
}
\end{equation*}
where 
\begin{equation*}
\footnotesize{
\begin{aligned}
G_{7\_1,n}(\boldsymbol{\beta}) &= \frac{1}{n} \sum_{j=1}^{J^{(2)}}  \sum_{i=1}^{n_{j}^{(2)}} 
\left[ exp \left( \left( \begin{array}{c}
1 \\
\boldsymbol{A}_{j}^{(2, n^{(1)})} \\
\boldsymbol{z}_{j}^{(2)}  
\end{array} \right)^T \boldsymbol{\beta}   \right) \left( \begin{array}{c}
1 \\
\boldsymbol{A}_{j}^{(2, n^{(1)})} \\
\boldsymbol{z}_{j}^{(2)}  
\end{array} \right) \left( \begin{array}{c}
1 \\
\boldsymbol{A}_{j}^{(2, n^{(1)})} \\
\boldsymbol{z}_{j}^{(2)}  
\end{array} \right)^T  \epsilon_{ij}^{(2)}  \right]
\end{aligned}
}
\end{equation*}
\begin{equation*}
\footnotesize{
\begin{aligned}
&G_{7\_2,n}(\boldsymbol{\beta})\\
&= \frac{1}{n} \sum_{j=1}^{J^{(2)}}  \sum_{i=1}^{n_{j}^{(2)}} 
\left[ exp \left( \left( \begin{array}{c}
1 \\
\boldsymbol{A}_{j}^{(2, n^{(1)})} \\
\boldsymbol{z}_{j}^{(2)}  
\end{array} \right)^T \boldsymbol{\beta}   \right) \left( \begin{array}{c}
1 \\
\boldsymbol{A}_{j}^{(2, n^{(1)})} \\
\boldsymbol{z}_{j}^{(2)}  
\end{array} \right) \left( \begin{array}{c}
1 \\
\boldsymbol{A}_{j}^{(2, n^{(1)})} \\
\boldsymbol{z}_{j}^{(2)}  
\end{array} \right)^T exp \left( \left( \begin{array}{c}
1 \\
\boldsymbol{A}_{j}^{(2, n^{(1)})} \\
\boldsymbol{z}_{j}^{(2)}  
\end{array} \right)^T \boldsymbol{\beta}^*   \right) \right.\\
&\left. \qquad \qquad \qquad \qquad \qquad - exp \left( \left( \begin{array}{c}
1 \\
\boldsymbol{a}_{j}^{(2)} \\
\boldsymbol{z}_{j}^{(2)}  
\end{array} \right)^T \boldsymbol{\beta}   \right) \left( \begin{array}{c}
1 \\
\boldsymbol{a}_{j}^{(2)} \\
\boldsymbol{z}_{j}^{(2)}  
\end{array} \right) \left( \begin{array}{c}
1 \\
\boldsymbol{a}_{j}^{(2)} \\
\boldsymbol{z}_{j}^{(2)}  
\end{array} \right)^T exp \left( \left( \begin{array}{c}
1 \\
\boldsymbol{a}_{j}^{(2)} \\
\boldsymbol{z}_{j}^{(2)}  
\end{array} \right)^T \boldsymbol{\beta}^*   \right) \right] 
\end{aligned}
}
\end{equation*}
\begin{equation*}
\footnotesize{
\begin{aligned}
G_{7\_3,n}(\boldsymbol{\beta}) &= \frac{2}{n} \sum_{j=1}^{J^{(2)}}  \sum_{i=1}^{n_{j}^{(2)}}  \left[ \left( exp \left( 2 \left( \begin{array}{c}
1 \\
\boldsymbol{a}_{j}^{(2)} \\
\boldsymbol{z}_{j}^{(2)}  
\end{array} \right)^T \boldsymbol{\beta}   \right) \right) \left( \begin{array}{c}
1 \\
\boldsymbol{a}_{j}^{(2)} \\
\boldsymbol{z}_{j}^{(2)}  
\end{array} \right) \left( \begin{array}{c}
1 \\
\boldsymbol{a}_{j}^{(2)} \\ 
\boldsymbol{z}_{j}^{(2)}  
\end{array} \right)^T \right.\\
&\left. \qquad \qquad \qquad \qquad \qquad - \left( exp \left( 2 \left( \begin{array}{c}
1 \\
\boldsymbol{A}_{j}^{(2, n^{(1)})} \\
\boldsymbol{z}_{j}^{(2)}  
\end{array} \right)^T \boldsymbol{\beta}   \right) \right) \left( \begin{array}{c}
1 \\
\boldsymbol{A}_{j}^{(2, n^{(1)})} \\
\boldsymbol{z}_{j}^{(2)}  
\end{array} \right) \left( \begin{array}{c}
1 \\
\boldsymbol{A}_{j}^{(2, n^{(1)})} \\
\boldsymbol{z}_{j}^{(2)}  
\end{array} \right)^T   \right].
\end{aligned}
}
\end{equation*}
Following the same arguments as for $G_{2,n}$, $G_{3,n}$, and $G_{4,n}$ in the proof of consistency (Appendix \ref{consistency_appendix_loglink}), respectively, we conclude that suprema over $\boldsymbol{\beta}$ of $G_{7\_1,n}(\boldsymbol{\beta})$, $G_{7\_2,n}(\boldsymbol{\beta})$, and $G_{7\_3,n}(\boldsymbol{\beta})$ converge to 0 in probability. Next, notice that
\begin{equation*}
\begin{aligned}
G_{7,n}(\boldsymbol{\tilde{\beta}}) - G_{7,n}^*(\boldsymbol{\beta}^*) &= G_{7,n}(\boldsymbol{\tilde{\beta}}) \pm G_{7,n}^*(\boldsymbol{\tilde{\beta}}) - G_{7,n}^*(\boldsymbol{\beta}^*) \\
&= \left(G_{7,n}(\boldsymbol{\tilde{\beta}}) - G_{7,n}^*(\boldsymbol{\tilde{\beta}}) \right) + \left( G_{7,n}^*(\boldsymbol{\tilde{\beta}}) - G_{7,n}^*(\boldsymbol{\beta}^*) \right).
\end{aligned}
\end{equation*}
Since $\sup_{\boldsymbol{\beta}} \| G_{7,n}(\boldsymbol{{\beta}}) -  G_{7,n}^*(\boldsymbol{\beta})\| \xrightarrow{P} 0$, it follows that $\left(G_{7,n}(\boldsymbol{\tilde{\beta}}) - G_{7,n}^*(\boldsymbol{\tilde{\beta}}) \right) \xrightarrow{P} 0$. Similar to equation (\ref{G5betatilda}), we apply the Mean Value Theorem 
and conclude that 
$$\left(G_{7,n}^*(\boldsymbol{\tilde{\beta}}) - G_{7,n}^*(\boldsymbol{\beta}^*)\right) \xrightarrow{P} 0.$$
Combining, we conclude $\left(G_{7,n}(\boldsymbol{\tilde{\beta}}) -  G_{7,n}^*(\boldsymbol{\beta}^*)\right) \xrightarrow{P} 0$.

Since $\frac{\partial}{\partial \boldsymbol{\beta}} \boldsymbol{U}({\boldsymbol{\tilde{\beta}}}) - \left( G_{6,n}^*(\boldsymbol{\beta}^*)+G_{7,n}^*(\boldsymbol{\beta}^*) \right) \xrightarrow{P} 0$, calculating the limit of $-\frac{\partial}{\partial \boldsymbol{\beta}} \boldsymbol{U}({\boldsymbol{\tilde{\beta}}})$ is equivalent to calculating the limit of $- \left( G_{6,n}^*(\boldsymbol{\beta}^*)+G_{7,n}^*(\boldsymbol{\beta}^*) \right)$. We therefore show that 
$
\left(-\left.\frac{\partial}{\partial \boldsymbol{\beta}}\right|_{\tilde{\boldsymbol{\beta}}} \boldsymbol{U}\left(\tilde{\boldsymbol{\beta}}\right)\right) 
$
converges in probability to $J\left(\boldsymbol{\beta}^{*}\right)$ from equation (\ref{ddbeta_u_beta_tilde_converge}).

\begin{equation*}
\footnotesize{
\begin{aligned}
&G_{6,n}^*(\boldsymbol{\beta}^*) - \sum_{j=1}^{J^{(1)}} \alpha_{j1}    \left\{  - \; exp \left( 2 \left( \begin{array}{c}
1 \\
\boldsymbol{a}_{j}^{(1)}\\
\boldsymbol{z}_{j}^{(1)} 
\end{array} \right)^T \boldsymbol{\beta}^*   \right)   \left( \begin{array}{c}
1 \\
\boldsymbol{a}_{j}^{(1)}\\
\boldsymbol{z}_{j}^{(1)} 
\end{array} \right) \left( \begin{array}{c}
1 \\
\boldsymbol{a}_{j}^{(1)}\\
\boldsymbol{z}_{j}^{(1)} 
\end{array} \right)^T  \right\}\\
&= \frac{1}{n} \sum_{j=1}^{J^{(1)}}  \sum_{i=1}^{n_{j}^{(1)}}  \left[ exp \left( \left( \begin{array}{c}
1 \\
\boldsymbol{a}_{j}^{(1)}\\
\boldsymbol{z}_{j}^{(1)} 
\end{array} \right)^T \boldsymbol{\beta}^*   \right) \left( \begin{array}{c}
1 \\
\boldsymbol{a}_{j}^{(1)}\\
\boldsymbol{z}_{j}^{(1)} 
\end{array} \right) \left( \begin{array}{c}
1 \\
\boldsymbol{a}_{j}^{(1)}\\
\boldsymbol{z}_{j}^{(1)} 
\end{array} \right)^T exp \left( \left( \begin{array}{c}
1 \\
\boldsymbol{a}_{j}^{(1)}\\
\boldsymbol{z}_{j}^{(1)} 
\end{array} \right)^T \boldsymbol{\beta}^*   \right) \right.\\   
&\qquad\qquad\qquad\qquad\qquad\qquad \left. - 2 \; \left( exp \left( 2 \left( \begin{array}{c}
1 \\
\boldsymbol{a}_{j}^{(1)}\\
\boldsymbol{z}_{j}^{(1)} 
\end{array} \right)^T \boldsymbol{\beta}^*   \right) \right) \left( \begin{array}{c}
1 \\
\boldsymbol{a}_{j}^{(1)}\\
\boldsymbol{z}_{j}^{(1)} 
\end{array} \right) \left( \begin{array}{c}
1 \\
\boldsymbol{a}_{j}^{(1)}\\
\boldsymbol{z}_{j}^{(1)} 
\end{array} \right)^T  \right]\\
&\qquad \qquad \qquad \qquad -\sum_{j=1}^{J^{(1)}} \alpha_{j1}    \left\{  - \; exp \left( 2 \left( \begin{array}{c}
1 \\
\boldsymbol{a}_{j}^{(1)}\\
\boldsymbol{z}_{j}^{(1)} 
\end{array} \right)^T \boldsymbol{\beta}^*   \right)   \left( \begin{array}{c}
1 \\
\boldsymbol{a}_{j}^{(1)}\\
\boldsymbol{z}_{j}^{(1)} 
\end{array} \right) \left( \begin{array}{c}
1 \\
\boldsymbol{a}_{j}^{(1)}\\
\boldsymbol{z}_{j}^{(1)} 
\end{array} \right)^T  \right\}\\
&=\sum_{j=1}^{J^{(1)}} \frac{n_j^{(1)}}{n}    \left\{  - \; exp \left( 2 \left( \begin{array}{c}
1 \\
\boldsymbol{a}_{j}^{(1)}\\
\boldsymbol{z}_{j}^{(1)} 
\end{array} \right)^T \boldsymbol{\beta}^*   \right)   \left( \begin{array}{c}
1 \\
\boldsymbol{a}_{j}^{(1)}\\
\boldsymbol{z}_{j}^{(1)} 
\end{array} \right) \left( \begin{array}{c}
1 \\
\boldsymbol{a}_{j}^{(1)}\\
\boldsymbol{z}_{j}^{(1)} 
\end{array} \right)^T  \right\} \\
&\qquad \qquad \qquad \qquad -\sum_{j=1}^{J^{(1)}} \alpha_{j1}    \left\{  - \; exp \left( 2 \left( \begin{array}{c}
1 \\
\boldsymbol{a}_{j}^{(1)}\\
\boldsymbol{z}_{j}^{(1)} 
\end{array} \right)^T \boldsymbol{\beta}^*   \right)   \left( \begin{array}{c}
1 \\
\boldsymbol{a}_{j}^{(1)}\\
\boldsymbol{z}_{j}^{(1)} 
\end{array} \right) \left( \begin{array}{c}
1 \\
\boldsymbol{a}_{j}^{(1)}\\
\boldsymbol{z}_{j}^{(1)} 
\end{array} \right)^T  \right\}\\
&\xrightarrow{P} 0.
\end{aligned}
}
\end{equation*}
Similarly, 
\begin{equation*} 
\footnotesize{
\begin{aligned}
&G_{7,n}^*(\boldsymbol{\beta}^*) - \sum_{j=1}^{J^{(2)}} \alpha_{j2} \left\{- \;  exp \left( 2 \left( \begin{array}{c}
1 \\
\boldsymbol{a}_{j}^{(2)} \\
\boldsymbol{z}_{j}^{(2)}  
\end{array} \right)^T \boldsymbol{\beta}^*   \right)  \left( \begin{array}{c}
1 \\
\boldsymbol{a}_{j}^{(2)} \\
\boldsymbol{z}_{j}^{(2)}  
\end{array} \right) \left( \begin{array}{c}
1 \\
\boldsymbol{a}_{j}^{(2)} \\ 
\boldsymbol{z}_{j}^{(2)}  
\end{array} \right)^T  \right\} \xrightarrow{P} 0.
\end{aligned}
}
\end{equation*}
Equation (\ref{ddbeta_u_beta_tilde_converge}) follows, and we conclude that 
$$
\left.\frac{\partial}{\partial \boldsymbol{\beta}}\right|_{\tilde{\boldsymbol{\beta}}} \boldsymbol{U}(\boldsymbol{\beta}) \stackrel{P}{\longrightarrow} J\left(\boldsymbol{\beta}^{*}\right) .
$$

To show that $\sqrt{n} \:\: \boldsymbol{U}\left(\boldsymbol{\beta}^{*}\right)$ from equation (\ref{AN_break_down}) converges in distribution to $N\left(0, V\left(\boldsymbol{\beta}^*\right)\right)$, we derive
\begin{equation*}
\footnotesize{
\begin{aligned}
&\sqrt{n} \ \boldsymbol{U}(\boldsymbol{\beta}^*)\\
&=  \frac{1}{\sqrt{n}} \sum_{j=1}^{J^{(1)}}  \sum_{i=1}^{n_{j}^{(1)}}  exp \left( \left( \begin{array}{c}
1 \\
\boldsymbol{a}_{j}^{(1)}\\
\boldsymbol{z}_{j}^{(1)} 
\end{array} \right)^T \boldsymbol{\beta}^*   \right) \left( \begin{array}{c}
1 \\
\boldsymbol{a}_{j}^{(1)}\\
\boldsymbol{z}_{j}^{(1)} 
\end{array} \right) \left( Y_{ij}^{(1)} - exp\left( \left( \begin{array}{c}
1 \\
\boldsymbol{a}_{j}^{(1)}\\
\boldsymbol{z}_{j}^{(1)} 
\end{array} \right)^T \boldsymbol{\beta}^*  \right)   \right) \\
&+  \frac{1}{\sqrt{n}} \sum_{j=1}^{J^{(2)}}  \sum_{i=1}^{n_{j}^{(2)}}  exp \left( \left( \begin{array}{c}
1 \\
\boldsymbol{A}_{j}^{(2,n^{(1)})} \\
\boldsymbol{z}_{j}^{(2)}  
\end{array} \right)^T \boldsymbol{\beta}^*   \right)\left( \begin{array}{c}
1 \\
\boldsymbol{A}_{j}^{(2,n^{(1)})} \\
\boldsymbol{z}_{j}^{(2)}  
\end{array} \right) \left(Y_{ij}^{(2,n^{(1)})} -exp\left( \left( \begin{array}{c}
1 \\
\boldsymbol{A}_{j}^{(2,n^{(1)})} \\
\boldsymbol{z}_{j}^{(2)}  
\end{array} \right)^T \boldsymbol{\beta}^*  \right)   \right) \\
&= \frac{1}{\sqrt{n}} \sum_{j=1}^{J^{(1)}}  \sum_{i=1}^{n_{j}^{(1)}}  exp \left( \left( \begin{array}{c}
1 \\
\boldsymbol{a}_{j}^{(1)}\\
\boldsymbol{z}_{j}^{(1)} 
\end{array} \right)^T \boldsymbol{\beta}^*   \right) \left( \begin{array}{c}
1 \\
\boldsymbol{a}_{j}^{(1)}\\
\boldsymbol{z}_{j}^{(1)} 
\end{array} \right) \epsilon_{ij}^{(1)} \\
&\qquad \qquad + \frac{1}{\sqrt{n}} \sum_{j=1}^{J^{(2)}}  \sum_{i=1}^{n_{j}^{(2)}} exp \left( \left( \begin{array}{c}
1 \\
\boldsymbol{A}_{j}^{(2,n^{(1)})} \\
\boldsymbol{z}_{j}^{(2)}  
\end{array} \right)^T \boldsymbol{\beta}^*   \right)\left( \begin{array}{c}
1 \\
\boldsymbol{A}_{j}^{(2,n^{(1)})} \\
\boldsymbol{z}_{j}^{(2)}  
\end{array} \right) \epsilon_{ij}^{(2,n^{(1)})}\\
&\qquad \qquad \pm 
\frac{1}{\sqrt{n}} \sum_{j=1}^{J^{(2)}}  \sum_{i=1}^{n_{j}^{(2)}}  exp \left( \left( \begin{array}{c}
1 \\
\boldsymbol{a}_{j}^{(2)} \\
\boldsymbol{z}_{j}^{(2)}  
\end{array} \right)^T \boldsymbol{\beta}^*   \right)\left( \begin{array}{c}
1 \\
\boldsymbol{a}_{j}^{(2)} \\
\boldsymbol{z}_{j}^{(2)}  
\end{array} \right) \epsilon_{ij}^{(2,n^{(1)})} \\
&= \sum_{j=1}^{J^{(1)}} \frac{\sqrt{n_j^{(1)}}}{\sqrt{n}} \frac{1}{\sqrt{n_j^{(1)}}} \sum_{i=1}^{n_{j}^{(1)}} exp \left( \left( \begin{array}{c}
1 \\
\boldsymbol{a}_{j}^{(1)}\\
\boldsymbol{z}_{j}^{(1)} 
\end{array} \right)^T \boldsymbol{\beta}^*   \right) \left( \begin{array}{c}
1 \\
\boldsymbol{a}_{j}^{(1)}\\
\boldsymbol{z}_{j}^{(1)} 
\end{array} \right) \epsilon_{ij}^{(1)}\\
&\qquad \qquad  + \frac{1}{\sqrt{n}} \sum_{j=1}^{J^{(2)}}  \sum_{i=1}^{n_{j}^{(2)}} exp \left( \left( \begin{array}{c}
1 \\
\boldsymbol{A}_{j}^{(2,n^{(1)})} \\
\boldsymbol{z}_{j}^{(2)}  
\end{array} \right)^T \boldsymbol{\beta}^*   \right)\left( \begin{array}{c}
1 \\
\boldsymbol{A}_{j}^{(2,n^{(1)})} \\
\boldsymbol{z}_{j}^{(2)}  
\end{array} \right) \epsilon_{ij}^{(2,n^{(1)})}\\ 
&\qquad \qquad -  \frac{1}{\sqrt{n}} \sum_{j=1}^{J^{(2)}}  \sum_{i=1}^{n_{j}^{(2)}}  exp \left( \left( \begin{array}{c}
1 \\
\boldsymbol{a}_{j}^{(2)} \\
\boldsymbol{z}_{j}^{(2)}  
\end{array} \right)^T \boldsymbol{\beta}^*   \right)\left( \begin{array}{c}
1 \\
\boldsymbol{a}_{j}^{(2)} \\
\boldsymbol{z}_{j}^{(2)}  
\end{array} \right) \epsilon_{ij}^{(2,n^{(1)})}\\ 
&\qquad \qquad + \sum_{j=1}^{J^{(2)}} \frac{\sqrt{n_j^{(2)}}}{\sqrt{n}} \frac{1}{\sqrt{n_j^{(2)}}} \sum_{i=1}^{n_{j}^{(2)}}  exp \left( \left( \begin{array}{c}
1 \\
\boldsymbol{a}_{j}^{(2)} \\
\boldsymbol{z}_{j}^{(2)}  
\end{array} \right)^T \boldsymbol{\beta}^*   \right)\left( \begin{array}{c}
1 \\
\boldsymbol{a}_{j}^{(2)} \\
\boldsymbol{z}_{j}^{(2)}  
\end{array} \right) \epsilon_{ij}^{(2,n^{(1)})}\\
&= \boldsymbol{U}_{1,n} + \boldsymbol{U}_{2\_1, n} + \boldsymbol{U}_{2\_2, n},
\label{u_beta_*_log}
\end{aligned}
}
\end{equation*}
where 
\begin{equation*}
\footnotesize{
\begin{aligned}
\boldsymbol{U}_{1,n} &=  \sum_{j=1}^{J^{(1)}} \frac{\sqrt{n_j^{(1)}}}{\sqrt{n}} \frac{1}{\sqrt{n_j^{(1)}}} \sum_{i=1}^{n_{j}^{(1)}} exp \left( \left( \begin{array}{c}
1 \\
\boldsymbol{a}_{j}^{(1)}\\
\boldsymbol{z}_{j}^{(1)} 
\end{array} \right)^T \boldsymbol{\beta}^*   \right) \left( \begin{array}{c}
1 \\
\boldsymbol{a}_{j}^{(1)}\\
\boldsymbol{z}_{j}^{(1)} 
\end{array} \right) \epsilon_{ij}^{(1)},\\ 
    \boldsymbol{U}_{2\_1, n} &=  \sum_{j=1}^{J^{(2)}} \frac{\sqrt{n_j^{(2)}}}{\sqrt{n}} \frac{1}{\sqrt{n_j^{(2)}}} \sum_{i=1}^{n_{j}^{(2)}}  \left[ exp \left( \left( \begin{array}{c}
1 \\
\boldsymbol{A}_{j}^{(2,n^{(1)})} \\
\boldsymbol{z}_{j}^{(2)}  
\end{array} \right)^T \boldsymbol{\beta}^*   \right)\left( \begin{array}{c}
1 \\
\boldsymbol{A}_{j}^{(2,n^{(1)})} \\
\boldsymbol{z}_{j}^{(2)}  
\end{array} \right) \epsilon_{ij}^{(2,n^{(1)})} \right.\\
&\qquad \qquad \qquad \qquad \qquad  \qquad \qquad \qquad  \left.- exp \left( \left( \begin{array}{c}
1 \\
\boldsymbol{a}_{j}^{(2)} \\
\boldsymbol{z}_{j}^{(2)}  
\end{array} \right)^T \boldsymbol{\beta}^*   \right)\left( \begin{array}{c}
1 \\
\boldsymbol{a}_{j}^{(2)} \\
\boldsymbol{z}_{j}^{(2)}  
\end{array} \right) \epsilon_{ij}^{(2,n^{(1)})} \right],\\
\boldsymbol{U}_{2\_2, n} &= \sum_{j=1}^{J^{(2)}} \frac{\sqrt{n_j^{(2)}}}{\sqrt{n}} \frac{1}{\sqrt{n_j^{(2)}}} \sum_{i=1}^{n_{j}^{(2)}}  exp \left( \left( \begin{array}{c}
1 \\
\boldsymbol{a}_{j}^{(2)} \\
\boldsymbol{z}_{j}^{(2)}  
\end{array} \right)^T \boldsymbol{\beta}^*   \right)\left( \begin{array}{c}
1 \\
\boldsymbol{a}_{j}^{(2)} \\
\boldsymbol{z}_{j}^{(2)}  
\end{array} \right) \epsilon_{ij}^{(2,n^{(1)})}.
\end{aligned}
}
\end{equation*}
To show that $\sqrt{n} \ \boldsymbol{U}(\boldsymbol{\beta}^*)$ converges in distribution to a normal distribution, we show that the joint distribution of $\boldsymbol{U}_{1,n}$, $\boldsymbol{U}_{2\_1, n}$, $\boldsymbol{U}_{2\_2, n}$ converges in distribution to a normal distribution. By Assumption \ref{ind_error} of the main text, 
for each fixed value of $j$, the central limit theorem for i.i.d. observations implies that 
\begin{equation*}
\footnotesize{
    \frac{1}{\sqrt{n_j^{(1)}}}  \sum_{i=1}^{n_{j}^{(1)}} exp \left( \left( \begin{array}{c}
1 \\
\boldsymbol{a}_{j}^{(1)}\\
\boldsymbol{z}_{j}^{(1)} 
\end{array} \right)^T \boldsymbol{\beta}^*   \right) \left( \begin{array}{c}
1 \\
\boldsymbol{a}_{j}^{(1)}\\
\boldsymbol{z}_{j}^{(1)} 
\end{array} \right) \epsilon_{ij}^{(1)}
}
\end{equation*}
converges in distribution to a normal distribution with mean 0 and variance 
$$\footnotesize{exp \left(2 \left( \begin{array}{c}
1 \\
\boldsymbol{a}_{j}^{(1)}\\
\boldsymbol{z}_{j}^{(1)} 
\end{array} \right)^T \boldsymbol{\beta}^*   \right) \left( \begin{array}{c}
1 \\
\boldsymbol{a}_{j}^{(1)}\\
\boldsymbol{z}_{j}^{(1)} 
\end{array} \right)\left( \begin{array}{c}
1 \\
\boldsymbol{a}_{j}^{(1)}\\
\boldsymbol{z}_{j}^{(1)} 
\end{array} \right)^T \sigma^2(\boldsymbol{z}_{j}^{(1)}).}$$From e.g. Lévy's Continuity Theorem for characteristic functions \citep{eisenberg1983uniform} and Slutsky's Theorem, 
it follows that $\boldsymbol{U}_{1,n}$ converges to a normal distribution with mean 0 and variance 
$$ \footnotesize{ \sum_{j=1}^{J^{(1)}} \alpha_{j1} \; exp \left(2 \left( \begin{array}{c}
1 \\
\boldsymbol{a}_{j}^{(1)}\\
\boldsymbol{z}_{j}^{(1)} 
\end{array} \right)^T \boldsymbol{\beta}^*   \right) \left( \begin{array}{c}
1 \\
\boldsymbol{a}_{j}^{(1)}\\
\boldsymbol{z}_{j}^{(1)} 
\end{array} \right)\left( \begin{array}{c}
1 \\
\boldsymbol{a}_{j}^{(1)}\\
\boldsymbol{z}_{j}^{(1)} 
\end{array} \right)^T \sigma^2(\boldsymbol{z}_{j}^{(1)}).}$$ 
We show $\boldsymbol{U}_{2\_1, n} \xrightarrow{P} 0$ by showing that $E(\boldsymbol{U}_{2\_1, n}) = 0$ and $E(\boldsymbol{U}_{2\_1, n}^{\otimes2}) \xrightarrow{} 0$, so that $\boldsymbol{U}_{2\_1, n} \xrightarrow{P} 0$ by Chebyshev's Inequality. We show $E(\boldsymbol{U}_{2\_1, n}) = 0$ first. For each $j$, we have
\begin{equation*}
\footnotesize{
\begin{aligned}
&E\left\{\frac{1}{\sqrt{n_j^{(2)}}}  \sum_{i=1}^{n_{j}^{(2)}}  \left[ exp \left( \left( \begin{array}{c}
1 \\
\boldsymbol{A}_{j}^{(2,n^{(1)})} \\
\boldsymbol{z}_{j}^{(2)}  
\end{array} \right)^T \boldsymbol{\beta}^*   \right)\left( \begin{array}{c}
1 \\
\boldsymbol{A}_{j}^{(2,n^{(1)})} \\
\boldsymbol{z}_{j}^{(2)}  
\end{array} \right) \epsilon_{ij}^{(2,n^{(1)})} \right. \right. \\
&\left. \left. \qquad \qquad \qquad \qquad \qquad \qquad \qquad \qquad \qquad - exp \left( \left( \begin{array}{c}
1 \\
\boldsymbol{a}_{j}^{(2)} \\
\boldsymbol{z}_{j}^{(2)}  
\end{array} \right)^T \boldsymbol{\beta}^*   \right)\left( \begin{array}{c}
1 \\
\boldsymbol{a}_{j}^{(2)} \\
\boldsymbol{z}_{j}^{(2)}  
\end{array} \right) \epsilon_{ij}^{(2)} \right]\right\}=0
\end{aligned}
}
\end{equation*}
by conditioning on $\boldsymbol{A}_{j}^{(2,n^{(1)})}$. so $E(\boldsymbol{U}_{2\_1, n}) = 0$. To see $E(\boldsymbol{U}_{2\_1, n}^{\otimes2}) \xrightarrow{} 0$, 
\begin{equation*}
\footnotesize{
\begin{aligned}
&E\left\{\frac{1}{\sqrt{n_{j}^{(2)}}}  \sum_{i=1}^{n_{j}^{(2)}}  \left[ exp \left( \left( \begin{array}{c}
1 \\
\boldsymbol{A}_{j}^{(2,n^{(1)})} \\
\boldsymbol{z}_{j}^{(2)}  
\end{array} \right)^T \boldsymbol{\beta}^*   \right)\left( \begin{array}{c}
1 \\
\boldsymbol{A}_{j}^{(2,n^{(1)})} \\
\boldsymbol{z}_{j}^{(2)}  
\end{array} \right) \epsilon_{ij}^{(2,n^{(1)})} \right. \right.\\
&\qquad \qquad \qquad \qquad \qquad \qquad \qquad \qquad  \left. \left. - exp \left( \left( \begin{array}{c}
1 \\
\boldsymbol{a}_{j}^{(2)} \\
\boldsymbol{z}_{j}^{(2)}  
\end{array} \right)^T \boldsymbol{\beta}^*   \right)\left( \begin{array}{c}
1 \\
\boldsymbol{a}_{j}^{(2)} \\
\boldsymbol{z}_{j}^{(2)}  
\end{array} \right) \epsilon_{ij}^{(2,n^{(1)})} \right]\right\}^{\otimes 2}\\
&= \frac{1}{{n_{j}^{(2)}}}  \sum_{i=1}^{n_{j}^{(2)}}  E\left\{\left[ exp \left( \left( \begin{array}{c}
1 \\
\boldsymbol{A}_{j}^{(2,n^{(1)})} \\
\boldsymbol{z}_{j}^{(2)}  
\end{array} \right)^T \boldsymbol{\beta}^*   \right)\left( \begin{array}{c}
1 \\
\boldsymbol{A}_{j}^{(2,n^{(1)})} \\
\boldsymbol{z}_{j}^{(2)}  
\end{array} \right) \epsilon_{ij}^{(2,n^{(1)})} \right.\right. \\
&\qquad \qquad \qquad \qquad \qquad \qquad \qquad \qquad  \left. \left. - exp \left( \left( \begin{array}{c}
1 \\
\boldsymbol{a}_{j}^{(2)} \\
\boldsymbol{z}_{j}^{(2)}  
\end{array} \right)^T \boldsymbol{\beta}^*   \right)\left( \begin{array}{c}
1 \\
\boldsymbol{a}_{j}^{(2)} \\
\boldsymbol{z}_{j}^{(2)}  
\end{array} \right) \epsilon_{ij}^{(2,n^{(1)})} \right]^{\otimes 2}\right\}\\
&= \frac{1}{n_{j}^{(2)}} \sum_{i=1}^{n_{j}^{(2)}}  E\left\{\left[ exp \left( \left( \begin{array}{c}
1 \\
\boldsymbol{A}_{j}^{(2,n^{(1)})} \\
\boldsymbol{z}_{j}^{(2)}  
\end{array} \right)^T \boldsymbol{\beta}^*   \right)\left( \begin{array}{c}
1 \\
\boldsymbol{A}_{j}^{(2,n^{(1)})} \\
\boldsymbol{z}_{j}^{(2)}  
\end{array} \right) \right.\right.\\
&\qquad \qquad \qquad \qquad \qquad \qquad \qquad \qquad  \left. \left.- exp \left( \left( \begin{array}{c}
1 \\
\boldsymbol{a}_{j}^{(2)} \\
\boldsymbol{z}_{j}^{(2)}  
\end{array} \right)^T \boldsymbol{\beta}^*   \right)\left( \begin{array}{c}
1 \\
\boldsymbol{a}_{j}^{(2)} \\
\boldsymbol{z}_{j}^{(2)}  
\end{array} \right)  \right]^{\otimes 2} \left(\epsilon_{ij}^{(2,n^{(1)})}\right)^2 \right\}.
\end{aligned}
}
\end{equation*}
where the second line follows since all terms belonging to different $i$ are uncorrelated, which can be seen by conditioning on $\boldsymbol{A}_{j}^{(2,n^{(1)})}$.
By Assumption \ref{ind_error} from the main text, we have $E\left(\left( \epsilon_{ij}^{(2,n^{(1)})}\right)^2 | \boldsymbol{A}_{j}^{(2,n^{(1)})} \right)=\sigma^2(\boldsymbol{z}_{j}^{(2)})$. Then, by conditioning on $\boldsymbol{A}_{j}^{(2,n^{(1)})}$,
\begin{equation*}
\footnotesize{
\begin{aligned}
&E(\boldsymbol{U}_{2\_1, n}^{\otimes 2}) \\
&= \sum_{j=1}^{J^{(2)}} \sigma^2(\boldsymbol{z}_{j}^{(2)}) \frac{n_{j}^{(2)}}{n}  E\left(\left[ exp \left( \left( \begin{array}{c}
1 \\
\boldsymbol{A}_{j}^{(2,n^{(1)})} \\
\boldsymbol{z}_{j}^{(2)}  
\end{array} \right)^T \boldsymbol{\beta}^*   \right)\left( \begin{array}{c}
1 \\
\boldsymbol{A}_{j}^{(2,n^{(1)})} \\
\boldsymbol{z}_{j}^{(2)}  
\end{array} \right) \right.\right.\\
&\qquad \qquad \qquad \qquad \qquad \qquad \qquad \qquad \qquad \qquad  \left. \left.- exp \left( \left( \begin{array}{c}
1 \\
\boldsymbol{a}_{j}^{(2)} \\
\boldsymbol{z}_{j}^{(2)}  
\end{array} \right)^T \boldsymbol{\beta}^*   \right)\left( \begin{array}{c}
1 \\
\boldsymbol{a}_{j}^{(2)} \\
\boldsymbol{z}_{j}^{(2)}  
\end{array} \right)  \right]^{\otimes 2}  \right).
\end{aligned}
}
\end{equation*}
By Lemma \ref{sup_glm_condition} from the Appendix, the Continuous Mapping Theorem and Lebesgue's Dominated Convergence Theorem, $E(\boldsymbol{U}_{2\_1, n}^{\otimes 2}) \xrightarrow{P} 0$. By Chebyshev's Inequality, $\boldsymbol{U}_{2\_1, n} \xrightarrow{P} 0$.

For $\boldsymbol{U}_{2\_2, n}$,  
by Assumption \ref{ind_error} from the main text, replacing $\epsilon_{ij}^{(2,n^{(1)})}$ by the error terms $\epsilon_{ij}^{(2)}$ under the intervention $\boldsymbol{a}_{j}^{(2)}$ does not change the distribution of $\boldsymbol{U}_{2\_2, n}$, so it suffices to show that
\begin{equation*}
\footnotesize{
\begin{aligned}
\boldsymbol{U}_{2\_2, n}^* &:=  \sum_{j=1}^{J^{(2)}} \frac{\sqrt{n_j^{(2)}}}{\sqrt{n}} \frac{1}{\sqrt{n_j^{(2)}}} \sum_{i=1}^{n_{j}^{(2)}}  exp \left( \left( \begin{array}{c}
1 \\
\boldsymbol{a}_{j}^{(2)} \\
\boldsymbol{z}_{j}^{(2)}  
\end{array} \right)^T \boldsymbol{\beta}^*   \right)\left( \begin{array}{c}
1 \\
\boldsymbol{a}_{j}^{(2)} \\
\boldsymbol{z}_{j}^{(2)}  
\end{array} \right) \epsilon_{ij}^{(2)}
\end{aligned}
}
\end{equation*}
converges to a normal distribution. 
Following the same argument as for $\boldsymbol{U}_{1,n}$, $\boldsymbol{U}_{2\_2,n}$ converges to a normal distribution with mean 0 and variance 
$$  
\footnotesize{
\sum_{j=1}^{J^{(2)}} \alpha_{j2} \; exp \left(2 \left( \begin{array}{c}
1 \\
\boldsymbol{a}_{j}^{(2)}\\
\boldsymbol{z}_{j}^{(2)} 
\end{array} \right)^T \boldsymbol{\beta}^*   \right) \left( \begin{array}{c}
1 \\
\boldsymbol{a}_{j}^{(2)}\\
\boldsymbol{z}_{j}^{(2)} 
\end{array} \right)\left( \begin{array}{c}
1 \\
\boldsymbol{a}_{j}^{(2)}\\
\boldsymbol{z}_{j}^{(2)} 
\end{array} \right)^T \sigma^2(\boldsymbol{z}_{j}^{(2)}).
}
$$
Hence,  $\sqrt{n} \ \boldsymbol{U}(\boldsymbol{\beta}^*)$ has the same limiting distribution as 
\begin{equation}
\begin{aligned}
&\boldsymbol{U}_{1, n} +  \boldsymbol{U}_{2\_2, n}^* =\frac{1}{\sqrt{n}} \sum_{j=1}^{J^{(1)}}  \sum_{i=1}^{n_{j}^{(1)}}  exp \left( \left( \begin{array}{c}
1 \\
\boldsymbol{a}_{j}^{(1)}\\
\boldsymbol{z}_{j}^{(1)} 
\end{array} \right)^T \boldsymbol{\beta}^*   \right) \left( \begin{array}{c}
1 \\
\boldsymbol{a}_{j}^{(1)}\\
\boldsymbol{z}_{j}^{(1)} 
\end{array} \right) \epsilon_{ij}^{(1)} \\
& \qquad \qquad \qquad \qquad + \frac{1}{\sqrt{n}} \sum_{j=1}^{J^{(2)}}  \sum_{i=1}^{n_{j}^{(2)}}  exp \left( \left( \begin{array}{c}
1 \\
\boldsymbol{a}_{j}^{(2)} \\
\boldsymbol{z}_{j}^{(2)}  
\end{array} \right)^T \boldsymbol{\beta}^*   \right)\left( \begin{array}{c}
1 \\
\boldsymbol{a}_{j}^{(2)} \\
\boldsymbol{z}_{j}^{(2)}  
\end{array} \right) \epsilon_{ij}^{(2)}.
\end{aligned}
\end{equation}
By the definition of $\boldsymbol{a}_{j}^{(2)}$, $\boldsymbol{U}_{1, n}$ and $ \boldsymbol{U}_{2\_2, n}^*$ are independent of one another. Thus, equation (\ref{u_beta_*_log}) has the same asymptotic distribution as a fixed two-stage design with $a_1^{(1)},\cdots, a_{J^{(1)}}^{(1)}$, $a_1^{(2)},\cdots, a_{J^{(2)}}^{(2)}$ as interventions decided on before the trial (\cite{liang1986longitudinal}). The limiting distribution of $\sqrt{n} \; \boldsymbol{U}\left(\boldsymbol{\beta}^{*}\right)$ is $N(0, V\left(\boldsymbol{\beta}^*\right))$.

Combining equations (\ref{nominator_variance}) and (\ref{ddbeta_u_beta_tilde_converge}) implies that equation (\ref{aympstotic_normality}) holds.

\newpage 
\section{Extension to number of stages: \texorpdfstring{$K > 2$}{Lg}} \label{K>2}
In the case where $K > 2$, we first modify the notations from the main text. 
In stage 1, the notations remain the same.
Let $\boldsymbol{X}_{j}^{\left(k, \bar{n}_{k-}\right)}$ be the recommended intervention package for center $j$ in stage $k$. 
The superscript $\left(k, \bar{n}_{k-}\right)$ indicates that $\boldsymbol{X}_{j}^{\left(k, \bar{n}_{k-}\right)}$ depends on the data of patients from all previous k-1 stages.
Similar to Remark \ref{remark} of the main text, $\boldsymbol{X}_{j}^{\left(k, \bar{n}_{k-}\right)}$ is solved using the data from all previous stages based on function $f$. Let $\boldsymbol{A}_{j}^{\left(k, \bar{n}_{k-}\right)}$ be the actual intervention package implemented for center $j$ in stage $k$, where $\boldsymbol{A}_{j}^{\left(k, \bar{n}_{k-}\right)} =  h_j^{(k)} \left(\boldsymbol{X}_{j}^{\left(k, \bar{n}_{k-}\right)}\right)$, and $h_j^{(k)}$ is a continuous deterministic function for each center $j$ in stage $k$. 

Let $\boldsymbol{Y}_{j}^{\left(k, \bar{n}_{k-}\right)}= \left(Y_{1 j}^{\left(k, \bar{n}_{k-}\right)}, \ldots, Y_{n^{(k)}_j j}^{\left(k, \bar{n}_{k-}\right)}\right) $ be the outcomes for center $j$ in stage $k$. 
Let $\overline{\boldsymbol{A}}^{\left(k, \bar{n}_{k-}\right)}=\left(\boldsymbol{A}_{1}^{\left(k, \bar{n}_{k-}\right)}, \ldots, \boldsymbol{A}_{J^{(k)}}^{\left(k, \bar{n}_{k-}\right)}\right)$, $\overline{\boldsymbol{z}}^{(k)}=\left(\boldsymbol{z}_{1}^{(k)}, \ldots, \boldsymbol{z}_{J^{(k)}}^{(k)}\right)$, and \\ $\overline{\boldsymbol{Y}}^{\left(k, \bar{n}_{k-}\right)}=\left(\boldsymbol{Y}_{1}^{\left(k, \bar{n}_{k-}\right)}, \ldots, \boldsymbol{Y}_{J^{(k)}}^{\left(k, \bar{n}_{k-}\right)}\right)$ be the actual interventions, center-specific characteristics and outcomes for stage $k$, respectively. Additionally, let\\ 
$\overline{\hat{\boldsymbol{x}}}^{o p t,\left(k, \bar{n}_{k-}\right)}=\left(\hat{\boldsymbol{x}}_{1}^{o p t,\left(k, \bar{n}_{k-}\right)}, \ldots, \hat{\boldsymbol{x}}_{J^{(k)}}^{o p t,\left(k, \bar{n}_{k-}\right)}\right)$ be the recommended interventions for the $J^{(k)}$ centers at stage $k$.

Let $\widetilde{\boldsymbol{X}}^{\left(k, \bar{n}_{k-}\right)}=\left(\overline{\boldsymbol{x}}^{(1)}, \overline{\boldsymbol{X}}^{\left(2, n^{(1)}\right)}, \ldots, \overline{\boldsymbol{X}}^{\left(k, \bar{n}_{k-}\right)}\right)$, \\
$\widetilde{\boldsymbol{A}}^{\left(k, \bar{n}_{k-}\right)}=\left(\overline{\boldsymbol{a}}^{(1)}, \overline{\boldsymbol{A}}^{\left(2, n^{(1)}\right)}, \ldots, \overline{\boldsymbol{A}}^{\left(k, \bar{n}_{k-}\right)}\right)$, and
$\widetilde{\boldsymbol{Y}}^{\left(k, \bar{n}_{k-}\right)}=\left(\overline{\boldsymbol{Y}}^{(1)}, \ldots, \overline{\boldsymbol{Y}}^{\left(k, \bar{n}_{k-}\right)}\right)$ be the recommended intervention package, actual intervention package and actual outcomes until stage $k$ (including all previous $k-1$ stages), respectively. The following additional assumptions are needed for the case where $K>2$. 

\begin{assu}\label{k>3_glm_model_assumption}
The outcome of interest $Y_{i j}^{\left(k, \bar{n}_{k-}\right)}$ in stage $k$ with center specific characteristics $\boldsymbol{z}_{j}^{(k)}$ under treatment $\boldsymbol{A}_{j}^{\left(k, \bar{n}_{k-}\right)} = \boldsymbol{a}_{j}^{\left(k, \bar{n}_{k-}\right)}$, follows a GLM
\begin{equation}
\begin{aligned}
&g\left(E\left(Y_{i j}^{\left(k, \bar{n}_{k-}\right)} | \boldsymbol{A}_{j}^{\left(k, \bar{n}_{k-}\right)} = \boldsymbol{a}_{j}^{\left(k, \bar{n}_{k-}\right)}, \boldsymbol{X}_{j}^{\left(k, \bar{n}_{k-}\right)} =\boldsymbol{x}_j^{\left(k, \bar{n}_{k-}\right)}, \boldsymbol{z}_{j}^{(k)}; \boldsymbol{\beta} \right)\right) \\
&\qquad \qquad \qquad \qquad = \beta_{0}+\boldsymbol{\beta}_{1}^{T} \boldsymbol{a}_{j}^{\left(k, \bar{n}_{k-}\right)}+\boldsymbol{\beta}_{2}^{T} \boldsymbol{z}_{j}^{(k)},
\end{aligned}
\end{equation}
where $g()$ is a general link function and $\boldsymbol{\beta}^T=\left({\beta}_{0}, \boldsymbol{\beta}_{1}^T,\boldsymbol{\beta}_{2}^T\right)$ are unknown parameters to be estimated from the data.
\end{assu}

\begin{assu} \label{k>3_main_assum_1_1}
Conditionally on $\overline{\hat{\boldsymbol{x}}}^{o p t,\left(k, \bar{n}_{k-}\right)}$, 
$\left(\overline{\boldsymbol{A}}^{\left(k, \bar{n}_{k-}\right)},\overline{\boldsymbol{Y}}^{\left(k, \bar{n}_{k-}\right)}\right)$ are independent of previous stages $\left(\widetilde{\boldsymbol{A}}^{\left(k-1, \bar{n}_{\left(k-1\right)-}\right)},\widetilde{\boldsymbol{Y}}^{\left(k-1, \bar{n}_{\left(k-1\right)-}\right)} \right)$. 
\end{assu}

\begin{assu} For each $j=1,\ldots,J^{(k)}$, the stage $k$ recommended intervention  $\hat{\boldsymbol{x}}^{opt,\left(k, \bar{n}_{k-}\right)}_{j}$ converges in probability to a center-specific limit $\boldsymbol{x}_{j}^{(k)}$.
\label{k>3_intervention_cvg_assumption}
\end{assu}
Assumption \ref{intervention_cvg_assumption} from the main text showed that this Assumption holds for the case where $K=2$. 
If the recommended intervention suggested by the LAGO method is the estimated optimal intervention, then the results of the two-stage LAGO imply that, 
for stage $k=3$,  $\hat{\boldsymbol{x}}^{opt,\left(3, \left(n^{(1)}, n^{(2)}\right)\right)}_{j} \xrightarrow{P} \boldsymbol{x}_j^{opt,(3)}$. 
By forward induction, Assumption \ref{k>3_intervention_cvg_assumption} holds.
Under Assumption \ref{k>3_intervention_cvg_assumption}, the definition of $h_j^{(k)}$, and the Continuous Mapping Theorem, we conclude that $\boldsymbol{A}_{j}^{\left(k, \bar{n}_{k-}\right)} =  h_j^{(k)} \left(\boldsymbol{X}_{j}^{\left(k, \bar{n}_{k-}\right)}\right)$ converges in probability to $\boldsymbol{a}_j^{(k)} = h_j^{(k)}\left(\boldsymbol{x}^{(k)}_{j}  \right)$.

\begin{lem}
Under Assumption \ref{k>3_intervention_cvg_assumption},
there exist $\boldsymbol{a}_{j}^{(k)}$, which is the probability limit of $\boldsymbol{A}_{j}^{\left(k, \bar{n}_{k-}\right)}$ as each of $n^{(1)}, \cdots, n_{k-1} \rightarrow \infty$ separately, such that for stage $k$, 
$$\max_{j = 1, \cdots , J^{(k)}} \left\| \boldsymbol{A}_{j}^{\left(k, \bar{n}_{k-}\right)} - \boldsymbol{a}_{j}^{(k)} \right\| \xrightarrow{P} 0.$$
\label{K>3_sup_glm_condition}
\end{lem}
Proof of Lemma \ref{K>3_sup_glm_condition} is the same as the proof of Lemma \ref{sup_glm_condition} from Appendix. 

\begin{thm}\label{K>2 thm}
Under Assumptions \ref{k>3_glm_model_assumption} -- \ref{k>3_intervention_cvg_assumption} and Lemma \ref{K>3_sup_glm_condition},
Theorem \ref{GLM_consistency_general} and Theorem \ref{AN_thm} (from the main text) both hold for the case of $K>2$.
\end{thm}

\begin{proof}
The estimating equations from the main text (equation (\ref{glm_EE})) become 
\begin{equation}
\small{
\begin{aligned}
0&=\boldsymbol{U}^{(g)}_K(\boldsymbol{\beta})=\frac{1}{n}\Biggl\{ 
\sum_{j=1}^{J^{(1)}} \sum_{i=1}^{n_j^{(1)}}
\left(\frac{\partial}{\partial \boldsymbol{\beta}} E\left(Y_{ij}^{(1)} | \boldsymbol{a}_{j}^{(1)}, \boldsymbol{z}_{j}^{(1)} ; \boldsymbol{\beta}\right)\right) \left(Y_{ij}^{(1)}-E\left(Y_{ij}^{(1)} | \boldsymbol{a}_{j}^{(1)}, \boldsymbol{z}_{j}^{(1)} ; \boldsymbol{\beta}\right)\right) \\
& \qquad \qquad \qquad + \left. 
\sum_{k=2}^{K}
\sum_{j=1}^{J^{(k)}} 
\sum_{i=1}^{n_j^{(k)}}
\left(\frac{\partial}{\partial \boldsymbol{\beta}} E\left(Y_{ij}^{\left(k, \bar{n}_{k-}\right)} | \boldsymbol{A}_{j}^{\left(k, \bar{n}_{k-}\right)}, \boldsymbol{z}_{j}^{(k)} ; \boldsymbol{\beta}\right)\right) \right.\\
&\qquad \qquad \qquad \qquad \qquad \qquad  \left(Y_{ij}^{\left(k, \bar{n}_{k-}\right)}-E\left(Y_{ij}^{\left(k, \bar{n}_{k-}\right)} | \boldsymbol{A}_{j}^{\left(k, \bar{n}_{k-}\right)}, \boldsymbol{z}_{j}^{(k)} ; \boldsymbol{\beta}\right)\right)
\Biggr\},
\end{aligned}
}
\label{k>3_glm_EE}
\end{equation}
where the second term now includes a summation from $k=2,\ldots,K$. Similar to equation (\ref{little_u_main_body}) of the main text, we define
\begin{equation}
\small{
\begin{aligned}
\boldsymbol{u}^{(g)}_K(\boldsymbol{\beta})&= \\
&\sum_{k=1}^{K} 
\sum_{j=1}^{J^{(k)}} 
\alpha_{jk}
\left(\frac{\partial}{\partial \boldsymbol{\beta}} g^{-1}\left( \boldsymbol{a}_{j}^{(k)}, \boldsymbol{z}_{j}^{(k)} ; \boldsymbol{\beta}\right)
\right) 
\left(g^{-1}\left( \boldsymbol{a}_{j}^{(k)}, \boldsymbol{z}_{j}^{(k)} ; \boldsymbol{\beta}^*\right) - g^{-1}\left( \boldsymbol{a}_{j}^{(k)}, \boldsymbol{z}_{j}^{(k)} ; \boldsymbol{\beta}\right)
\right).
\end{aligned}
}
\label{k>3_little_u}
\end{equation}
To prove Theorem \ref{GLM_consistency_general} and \ref{AN_thm} (from the main text) for the case where $K>2$, we can use backward inductions. 

For Theorem 1, if we only focus on the last stage $k=K$, then $\boldsymbol{U}^{(g)}_K(\boldsymbol{\beta})-\boldsymbol{u}^{(g)}_K(\boldsymbol{\beta})$ can be separated into five terms similar to equations (\ref{G1_term_general}), (\ref{G2_term_general}), (\ref{G3_term_general}), (\ref{G4_term_general}), and (\ref{G5_term_general}) (see Appendix \ref{consistency_appendix_generallink}). 
The only difficult term to generalize to the case of $K>2$ is the term similar to equation (\ref{G2_term_general}) as the other four terms follow the same proofs as described in Appendix \ref{consistency_appendix_generallink}. Let $Y_{ij}^{(K)}$ be the (counterfactual) outcomes under $\boldsymbol{a}_{j}^{(K)}$ and let $\epsilon_{ij}^{(K)}$ be the corresponding errors that patient $i$ in center $j$ would have experienced under intervention $\boldsymbol{a}_{j}^{(K)}$. 
Given $\boldsymbol{A}_j^{\left(K, \bar{n}_{K-}\right)}$ and all data from previous $K-1$ stages, by Assumption \ref{ind_error} from the main text, replacing $\epsilon_{ij}^{\left(K, \bar{n}_{K-}\right)}$ by the error terms $\epsilon_{ij}^{(K)}$ under the intervention $\boldsymbol{a}_{j}^{(K)}$ does not change the distribution of the term similar to equation (\ref{G2_term_general}). Thus, for the last stage $k=K$, $\sup _{\boldsymbol{\beta}}\left\|\boldsymbol{U}_K^{(g)}(\boldsymbol{\beta})-\boldsymbol{u}_K^{(g)}(\boldsymbol{\beta})\right\| \stackrel{P}{\rightarrow} 0$ is satisfied. Then, we can continue with stage $k=K-1$.
and use backward induction to show a similar strategy of replacing the error term can be achieved for all previous stages. 
Given $\boldsymbol{A}_j^{\left(K-1, \bar{n}_{(K-1)-}\right)}$ and all data from previous $K-2$ stages, by Assumption 5, replacing $\epsilon_{ij}^{\left(K-1, \bar{n}_{(K-1)-}\right)}$ by the error terms $\epsilon_{ij}^{(K-1)}$ under the intervention $\boldsymbol{a}_{j}^{(K-1)}$ does not change the distribution of the term similar to equation (\ref{G2_term_general}). For the first stage $k=1$, since $\boldsymbol{a}_{j}^{(1)}$ is determined before the study starts, it is not necessary to implement the error replacing approach for the first stage.  
In the end, we conclude that Theorem 1 holds for the case of $K>2$. 
The same reasoning for the replacing error approach can be applied to the proof of Theorem 2 of $K>2$. It follows that $\sqrt{n}\; \boldsymbol{U}^{(g)}_K\left(\boldsymbol{\beta}^{*}\right)$ has the same asymptotic distribution as sum of $K$ independent terms, and Theorem 2 holds. Thus, Theorem \ref{K>2 thm} follows.
\end{proof}

\newpage
\section{Additional simulation results}\label{additional simulation}
We present additional simulation results in this section. 
Sections \ref{table1ctn}, \ref{table2ctn}, and \ref{table3ctn} provide additional simulation results for Tables \ref{component_effects_table}, \ref{bias_table}, and \ref{cp table} from the main text, respectively. Section \ref{mimicbbstudy} describes how we mimicked the BetterBirth study as if a c-LAGO design was used, and provides the corresponding simulation results. Section \ref{additional cost} describes the motivations behind using the cubic cost function in healthcare settings and provide additional simulation results corresponding to the material covered in Tables \ref{component_effects_table cubic cost}, \ref{bias_table cubic cost}, and \ref{meanopt table cubic cost} from the main text.

\subsection{Continuation of Table \ref{component_effects_table} from the main text}\label{table1ctn}
\newpage
\begin{table}[ht]
\caption{Simulation study results for individual package component effects with a linear cost function}
\centering
\scriptsize{
\begin{tabular}{llllllllllll}
\tiny$\boldsymbol{\beta}^*=(\beta_{11}^*, \beta_{12}^*)$ & \tiny$n_j^{(1)}$ & \tiny$n_j^{(2)}$ & \tiny$J$ &  &  & \tiny$\hat{\boldsymbol{\beta}}_{11}$ &  &  &  & \tiny$\hat{\boldsymbol{\beta}}_{12}$ &  \\
 &  &  &  &  & $\%$RelBias & \begin{tabular}[c]{@{}l@{}}$\frac{SE}{EMP.SD}$\\ ($\times 100$)\end{tabular}  & CP95 &  & $\%$RelBias & \begin{tabular}[c]{@{}l@{}}$\frac{SE}{EMP.SD}$\\ ($\times 100$)\end{tabular}  & CP95 \\
\multicolumn{4}{l}{Scenario 1 ($J_1 = J_2 = J$)} &  &  &  &  &  &  &  &  \\
(0.1863, 0.15) & 50 & 100 & 6 &  & 1.64 & 84.8 & 96.1 &  & -0.46 & 76.1$^*$ & 95.0 \\
 &  &  & 10 &  & 1.78 & 90.6 & 94.9 &  & -0.55 & 79.3 & 94.9 \\
 &  &  & 20 &  & 1.54 & 97.6 & 95.1 &  & -0.37 & 91.2 & 95.2 \\
 &  & 200 & 6 &  & 1.44 & 77.9 & 95.2 &  & -0.42 & 70.9 & 95.5 \\
 &  &  & 10 &  & 1.91 & 86.1 & 95.4 &  & -0.49 & 73.6 & 95.5 \\
 &  &  & 20 &  & 1.14 & 97.8 & 95.4 &  & -0.26 & 88.2 & 94.8 \\
 & 100 & 100 & 6 &  & 1.10 & 87.4 & 95.4 &  & -0.20 & 81.0 & 95.3 \\
 &  &  & 10 &  & 1.63 & 97.0 & 96.0 &  & -0.39 & 88.7 & 95.1 \\
 &  &  & 20 &  & 1.25 & 93.1 & 94.4 &  & -0.33 & 84.9 & 94.2 \\
 &  & 200 & 6 &  & 0.47 & 84.6 & 95.8 &  & 0.02 & 74.3 & 95.2 \\
 &  &  & 10 &  & 1.10 & 91.4 & 95.2 &  & -0.17 & 81.2 & 94.8 \\
 &  &  & 20 &  & 1.06 & 98.2 & 95.2 &  & -0.29 & 87.1 & 95.3 \\
(0.1, 0.2133) & 50 & 100 & 6 &  & -0.45 & 96.4 & 95.8 &  & 0.17 & 87.9 & 94.4 \\
 &  &  & 10 &  & -0.49 & 99.1 & 95.6 &  & 0.20 & 92.5 & 94.7 \\
 &  &  & 20 &  & 0.90 & 98.1 & 95.0 &  & 0.04 & 93.5 & 95.1 \\
 &  & 200 & 6 &  & 0.18 & 85.1 & 95.6 &  & 0.13 & 73.8 & 96.1 \\
 &  &  & 10 &  & 0.26 & 97.7 & 95.1 &  & 0.14 & 90.6 & 94.3 \\
 &  &  & 20 &  & 1.33 & 99.8 & 94.8 &  & -0.01 & 96.9 & 95.0 \\
 & 100 & 100 & 6 &  & -0.19 & 95.3 & 95.5 &  & 0.12 & 87.9 & 94.7 \\
 &  &  & 10 &  & 0.07 & 100.1 & 95.3 &  & 0.13 & 95.0 & 94.6 \\
 &  &  & 20 &  & -0.13 & 95.7 & 94.8 &  & 0.05 & 97.0 & 94.8 \\
 &  & 200 & 6 &  & -0.67 & 97.5 & 95.8 &  & 0.22 & 87.4 & 95.1 \\
 &  &  & 10 &  & 0.91 & 99.3 & 95.0 &  & 0.07 & 94.8 & 95.4 \\
 &  &  & 20 &  & 0.07 & 98.3 & 95.0 &  & 0.01 & 96.1 & 95.1 \\
\multicolumn{4}{l}{Scenario 2a ($J_1 = 6, J_2 = 12$)} &  &  &  &  &  &  &  &  \\
(0.1863, 0.15) & 50 & 200 &  &  & 0.04 & 85.9 & 95.0 &  & -0.09 & 75.0 & 94.8 \\
(0.1, 0.2133) & 50 & 200 &  &  & 0.56 & 93.7 & 95.8 &  & 0.14 & 84.3 & 95.7 \\
\multicolumn{4}{l}{} &  &  &  &  &  &  &  &  \\
\multicolumn{4}{l}{Scenario 2b ($J_1 = 10, J_2 = 20$)} &  &  &  &  &  &  &  &  \\
(0.1863, 0.15) & 50 & 200 &  &  & 0.44 & 93.4 & 95.9 &  & -0.13 & 86.8 & 95.9 \\
(0.1, 0.2133) & 50 & 200 &  &  & 0.56 & 96.5 & 95.3 &  & 0.09 & 86.8 & 94.9
\end{tabular}
}
\scriptsize	\\
\raggedright{
$n_j^{(1)}$: number of patients in center $j$ at stage 1, 
$n_j^{(2)}$: number of patients in center $j$ at stage 2.\\
$J$: number of centers for each stage. \\
\%RelBias: percent relative bias $100(\hat{\beta}-\beta^\star)/\beta^\star$.\\
SE: mean estimated standard error, 
EMP.SD: empirical standard deviation.\\
CP95: empirical coverage rate of 95\% confidence intervals.\\
$^*:$ The mean SE for covered iterations was found to be higher than that for non-covered iterations. Despite having thoroughly verified the code's accuracy, we cannot explain why coverage remained satisfactory when SE/EMP.SD was small. However, our analysis confirmed that this phenomenon persisted.
}
\end{table}

\begin{table}[ht]
\caption{Simulation study results for individual package component effects with a linear cost function}
\centering
\scriptsize{
\begin{tabular}{llllllllllll}
\tiny$\boldsymbol{\beta}^*=(\beta_{11}^*, \beta_{12}^*)$ & \tiny$n_j^{(1)}$ & \tiny$n_j^{(2)}$ & \tiny$J$ &  &  & \tiny$\hat{\boldsymbol{\beta}}_{11}$ &  &  &  & \tiny$\hat{\boldsymbol{\beta}}_{12}$ &  \\
 &  &  &  &  & $\%$RelBias & \begin{tabular}[c]{@{}l@{}}$\frac{SE}{EMP.SD}$\\ ($\times 100$)\end{tabular}  & CP95 &  & $\%$RelBias & \begin{tabular}[c]{@{}l@{}}$\frac{SE}{EMP.SD}$\\ ($\times 100$)\end{tabular}  & CP95 \\
\multicolumn{4}{l}{Scenario 1 ($J_1 = J_2 = J$)} &  &  &  &  &  &  &  &  \\
(0.0438, 0.17) & 50 & 100 & 6 &  & 30.28 & 72.0 & 96.0 &  & -1.89 & 64.1$^*$ & 95.2 \\
 &  &  & 10 &  & 18.25 & 81.6 & 95.1 &  & -1.16 & 68.2 & 95.3 \\
 &  &  & 20 &  & 9.24 & 95.1 & 94.8 &  & -0.54 & 80.5 & 95.4 \\
 &  & 200 & 6 &  & 24.19 & 66.8 & 96.1 &  & -1.52 & 57.7 & 95.7 \\
 &  &  & 10 &  & 14.54 & 79.3 & 95.7 &  & -0.94 & 62.7 & 94.9 \\
 &  &  & 20 &  & 5.61 & 96.5 & 95.6 &  & -0.37 & 79.2 & 95.2 \\
 & 100 & 100 & 6 &  & 29.95 & 77.2 & 95.7 &  & -1.91 & 68.1 & 94.9 \\
 &  &  & 10 &  & 19.53 & 84.2 & 95.2 &  & -1.12 & 71.5 & 94.4 \\
 &  &  & 20 &  & 8.62 & 86.9 & 93.9 &  & -0.51 & 75.9 & 93.9 \\
 &  & 200 & 6 &  & 24.25 & 68.6 & 95.7 &  & -1.47 & 58.3 & 95.0 \\
 &  &  & 10 &  & 12.89 & 84.4 & 95.2 &  & -0.72 & 68.7 & 94.9 \\
 &  &  & 20 &  & 4.43 & 94.9 & 95.7 &  & -0.27 & 80.0 & 95.7 \\
(0.1062, 0.16) & 50 & 100 & 6 &  & 12.02 & 77.0 & 95.8 &  & -1.86 & 68.2 & 95.0 \\
 &  &  & 10 &  & 4.83 & 84.2 & 95.2 &  & -0.80 & 72.3 & 94.9 \\
 &  &  & 20 &  & 3.21 & 97.9 & 94.7 &  & -0.44 & 85.9 & 95.2 \\
 &  & 200 & 6 &  & 9.00 & 68.1 & 95.8 &  & -1.41 & 58.3 & 95.2 \\
 &  &  & 10 &  & 2.98 & 84.1 & 95.6 &  & -0.47 & 69.7 & 95.1 \\
 &  &  & 20 &  & 2.36 & 96.0 & 95.6 &  & -0.36 & 81.3 & 94.8 \\
 & 100 & 100 & 6 &  & 9.93 & 79.5 & 95.7 &  & -1.57 & 71.0 & 95.2 \\
 &  &  & 10 &  & 6.39 & 93.6 & 95.0 &  & -0.91 & 82.3 & 94.7 \\
 &  &  & 20 &  & 3.53 & 89.7 & 94.1 &  & -0.54 & 80.2 & 94.2 \\
 &  & 200 & 6 &  & 7.03 & 74.9 & 95.6 &  & -1.04 & 64.6 & 94.9 \\
 &  &  & 10 &  & 3.73 & 90.5 & 95.1 &  & -0.45 & 77.5 & 94.9 \\
 &  &  & 20 &  & 2.15 & 96.7 & 95.5 &  & -0.34 & 83.7 & 95.6 \\
\multicolumn{4}{l}{Scenario 2a ($J_1 = 6, J_2 = 12$)} &  &  &  &  &  &  &  &  \\
(0.0438, 0.17) & 50 & 200 &  &  & 8.32 & 67.4 & 95.5 &  & -0.63 & 53.7 & 95.4 \\
(0.1062, 0.16) & 50 & 200 &  &  & 3.66 & 69.4 & 95.9 &  & -0.65 & 55.5 & 95.5 \\
\multicolumn{4}{l}{} &  &  &  &  &  &  &  &  \\
\multicolumn{4}{l}{Scenario 2b ($J_1 = 10, J_2 = 20$)} &  &  &  &  &  &  &  &  \\
(0.0438, 0.17) & 50 & 200 &  &  & 3.49 & 90.1 & 95.2 &  & -0.31 & 71.2 & 95.6 \\
(0.1062, 0.16) & 50 & 200 &  &  & 0.60 & 95.5 & 95.9 &  & -0.19 & 83.6 & 95.3
\end{tabular}
}
\label{component_effects_table_ctn}
\footnotesize	\\
\raggedright{
$n_j^{(1)}$: number of patients in center $j$ at stage 1, 
$n_j^{(2)}$: number of patients in center $j$ at stage 2.\\
$J$: number of centers for each stage. \\
\%RelBias: percent relative bias $100(\hat{\beta}-\beta^\star)/\beta^\star$.\\
SE: mean estimated standard error, 
EMP.SD: empirical standard deviation.\\
CP95: empirical coverage rate of 95\% confidence intervals.\\
$^*:$ $^*:$ The mean SE for covered iterations was found to be higher than that for non-covered iterations. Despite having thoroughly verified the code's accuracy, we cannot explain why coverage remained satisfactory when SE/EMP.SD was small. However, our analysis confirmed that this phenomenon persisted.
}
\end{table}

\newpage
\textcolor{white}{nothing}

\newpage
\subsection{Continuation of Table \ref{bias_table} from the main text}\label{table2ctn}
\textcolor{white}{nothing}
\begin{table}[ht]
\caption{Simulation study results for estimated optimal intervention with a linear cost function}
\centering
\scriptsize{
\begin{tabular}{llllllllllll}
$\boldsymbol{\beta}^*=(\beta_{11}^*,\beta_{12}^*)$ & $\boldsymbol{x}^{opt}$ & $n_j^{(1)}$ & $n_j^{(2)}$ &  & \multicolumn{3}{c}{Stage 1} &  & \multicolumn{3}{c}{Stage 2/LAGO optimized} \\
 &  &  &  &  & \begin{tabular}[c]{@{}l@{}}Bias of\\ $\hat{x}_{1}^{o p t}$\\ ($\times$100)\end{tabular} & \begin{tabular}[c]{@{}l@{}}Bias of\\ $\hat{x}_{2}^{o p t}$\\ ($\times$100)\end{tabular} & \begin{tabular}[c]{@{}l@{}}rMSE\\ ($\times$100)\end{tabular} &  & \begin{tabular}[c]{@{}l@{}}Bias of\\ $\hat{x}_{1}^{o p t}$\\ ($\times$100)\end{tabular} & \begin{tabular}[c]{@{}l@{}}Bias of\\ $\hat{x}_{2}^{o p t}$\\ ($\times$100)\end{tabular} & \begin{tabular}[c]{@{}l@{}}rMSE\\ ($\times$100)\end{tabular} \\
\multicolumn{4}{l}{Scenario 1 ($J_1 = J_2 = 20$)} &  &  &  &  &  &  &  &  \\
(0.1863, 0.15) & (1,8) & 50 & 100 &  & 15.2 & 96.9 & 115.0 &  & -0.27 & 0.01 & 35.3 \\
 &  &  & 500 &  &  &  &  &  & -0.70 & 0.00 & 28.2 \\
 &  & 100 & 100 &  & 14.0 & 51.6 & 94.1 &  & 0.11 & 0.01 & 32.7 \\
 &  &  & 500 &  &  &  &  &  & -0.25 & 0.00 & 24.5 \\
(0.0438, 0.17) & (0.6, 8) & 50 & 100 &  & -19.3 & 83.8 & 117.7 &  & -4.56 & 2.98 & 67.8 \\
 &  &  & 500 &  &  &  &  &  & -6.04 & 0.91 & 58.6 \\
 &  & 100 & 100 &  & -18.9 & 55.7 & 106.2 &  & -3.08 & 2.69 & 64.4 \\
 &  &  & 500 &  &  &  &  &  & -4.34 & 0.69 & 53.1 \\
(0.1, 0.2133) & (0, 6.5) & 50 & 100 &  & -17.7 & 34.9 & 111.3 &  & 0.00 & 0.02 & 30.0 \\
 &  &  & 500 &  &  &  &  &  & 0.00 & -0.19 & 20.9 \\
 &  & 100 & 100 &  & -8.7 & 6.4 & 95.9 &  & 0.00 & 0.13 & 28.5 \\
 &  &  & 500 &  &  &  &  &  & 0.00 & -0.13 & 19.2 \\
\multicolumn{4}{l}{} &  &  &  &  &  &  &  &  \\
\multicolumn{4}{l}{Scenario 2a ($J_1 = 6, J_2 = 12$)} &  &  &  &  &  &  &  &  \\
(0.1863, 0.15) & (1,8) & 50 & 200 &  & 20.7 & 341.5 & 191.1 &  & -1.5 & 1.5 & 52.3 \\
(0.0438, 0.17) & (0.6, 8) & 50 & 200 &  & -18.0 & 297.1 & 182.8 &  & -2.8 & 10.6 & 78.9 \\
(0.1, 0.2133) & (0, 6.5) & 50 & 200 &  & -46.2 & 239.2 & 172.9 &  & -0.1 & 0.1 & 39.2 \\
\multicolumn{4}{l}{} &  &  &  &  &  &  &  &  \\
\multicolumn{4}{l}{Scenario 2b ($J_1 = 10, J_2 = 20$)} &  &  &  &  &  &  &  &  \\
(0.1863, 0.15) & (1,8) & 50 & 200 &  & 13.9 & 214.0 & 156.0 &  & -1.2 & 0.1 & 39.9 \\
(0.0438, 0.17) & (0.6, 8) & 50 & 200 &  & -23.0 & 174.0 & 149.0 &  & -3.7 & 3.5 & 70.4 \\
(0.1, 0.2133) & (0, 6.5) & 50 & 200 &  & -35.5 & 126.7 & 143.8 &  & 0.0 & 0.4 & 29.6 \\
 &  &  &  &  &  &  &  &  &  &  & 
\end{tabular}
}
\scriptsize \\
\raggedright {
$n_j^{(1)}$: number of patients in center $j$ at stage 1,
$n_j^{(2)}$: number of patients in center $j$ at stage 2.\\
Bias of $\hat{x}_{1}^{o p t}$: bias of the first component of the estimated optimal intervention,  
Bias of $\hat{x}_{2}^{o p t}$: bias of the second component of the estimated optimal intervention. \\
rMSE: root of mean squared errors, $\left\{\operatorname{mean}\left(\left\|\hat{\boldsymbol{x}}^{o p t}-\boldsymbol{x}^{o p t}\right\|^{2}\right)\right\}^{1 / 2}$, mean is taken over simulation iterations.
}
\end{table}

\begin{table}[ht]
\caption{Simulation study results for estimated optimal intervention with a linear cost function}
\centering
\scriptsize{
\begin{tabular}{llllllllllll}
$\boldsymbol{\beta}^*=(\beta_{11}^*,\beta_{12}^*)$ & $\boldsymbol{x}^{opt}$ & $n_j^{(1)}$ & $n_j^{(2)}$ &  & \multicolumn{3}{c}{Stage 1} &  & \multicolumn{3}{c}{Stage 2/LAGO optimized} \\
 &  &  &  &  & \begin{tabular}[c]{@{}l@{}}Bias of\\ $\hat{x}_{1}^{o p t}$\\ ($\times$100)\end{tabular} & \begin{tabular}[c]{@{}l@{}}Bias of\\ $\hat{x}_{2}^{o p t}$\\ ($\times$100)\end{tabular} & \begin{tabular}[c]{@{}l@{}}rMSE\\ ($\times$100)\end{tabular} &  & \begin{tabular}[c]{@{}l@{}}Bias of\\ $\hat{x}_{1}^{o p t}$\\ ($\times$100)\end{tabular} & \begin{tabular}[c]{@{}l@{}}Bias of\\ $\hat{x}_{2}^{o p t}$\\ ($\times$100)\end{tabular} & \begin{tabular}[c]{@{}l@{}}rMSE\\ ($\times$100)\end{tabular} \\
\multicolumn{4}{l}{Scenario 1 ($J_1 = J_2 = 20$)} &  &  &  &  &  &  &  &  \\
(0.1062, 0.16) & (1,8) & 50 & 100 &  & 16.3 & 80.5 & 111.9 &  & -0.81 & 0.28 & 48.7 \\
 &  &  & 500 &  &  &  &  &  & -2.75 & 0.02 & 39.7 \\
 &  & 100 & 100 &  & 15.6 & 46.9 & 96.2 &  & -0.28 & 0.18 & 43.8 \\
 &  &  & 500 &  &  &  &  &  & -1.65 & 0.08 & 34.2 \\
\multicolumn{4}{l}{} &  &  &  &  &  &  &  &  \\
\multicolumn{4}{l}{Scenario 2a ($J_1 = 6, J_2 = 12$)} &  &  &  &  &  &  &  &  \\
(0.1062, 0.16) & (1,8) & 50 & 200 &  & 20.0 & 304.4 & 182.8 &  & -2.2 & 5.3 & 66.6 \\
\multicolumn{4}{l}{} &  &  &  &  &  &  &  &  \\
\multicolumn{4}{l}{Scenario 2b ($J_1 = 10, J_2 = 20$)} &  &  &  &  &  &  &  &  \\
(0.1062, 0.16) & (1,8) & 50 & 200 &  & 14.2 & 182.7 & 148.4 &  & -4.0 & 0.5 & 54.0 \\
 &  &  &  &  &  &  &  &  &  &  & 
\end{tabular}
}
\label{bias_table_ctn}
\footnotesize	\\
\raggedright {
$n_j^{(1)}$: number of patients in center $j$ at stage 1,
$n_j^{(2)}$: number of patients in center $j$ at stage 2.\\
Bias of $\hat{x}_{1}^{o p t}$: bias of the first component of the estimated optimal intervention,  
Bias of $\hat{x}_{2}^{o p t}$: bias of the second component of the estimated optimal intervention. \\
rMSE: root of mean squared errors, $\left\{\operatorname{mean}\left(\left\|\hat{\boldsymbol{x}}^{o p t}-\boldsymbol{x}^{o p t}\right\|^{2}\right)\right\}^{1 / 2}$, mean is taken over simulation iterations.
}
\end{table}

\newpage 
\textcolor{white}{nothing}

\subsection{Continuation of Table \ref{cp table} from the main text}\label{table3ctn}
\begin{table}[ht] 
\caption{Simulation study results for estimated optimal intervention, confidence set and confidence band with a linear cost function}
\centering
\scriptsize{
\begin{tabular}{lllllllll}
\tiny$\boldsymbol{\beta}^*=(\beta_{11}^*,\beta_{12}^*)$ & $\boldsymbol{x}^{opt}$ & $n_j^{(1)}$ & $n_j^{(2)}$ &  &  &  &  &  \\
 &  &  &  & \begin{tabular}[c]{@{}l@{}}True Mean\\ under\\ $\hat{x}^{opt,(2,n^{(1)})}$\\ (Q2.5,Q97.5)\end{tabular} & \begin{tabular}[c]{@{}l@{}}True Mean\\ under\\ $\hat{x}^{opt}$\\ (Q2.5,Q97.5)\end{tabular} & \begin{tabular}[c]{@{}l@{}}SetCP95\\ $\%$\}\end{tabular} & \begin{tabular}[c]{@{}l@{}}SetPerc\\ $\%$\end{tabular} & \begin{tabular}[c]{@{}l@{}}BandsCP95\\ $\%$\end{tabular} \\
\multicolumn{4}{l}{Scenario 1 ($J_1 = J_2 = 20$)} &  &  &  &  &  \\
(0.1863, 0.15) & (1,8) & 50 & 100 & (0.655, 0.818) & (0.789, 0.811) & 95.3 & 5.5 & 95.2 \\
 &  &  & 500 &  & (0.794, 0.808) & 95.1 & 3.7 & 95.7 \\
 &  & 100 & 100 & (0.708, 0.816) & (0.791, 0.808) & 94.4 & 4.4 & 95.1 \\
 &  &  & 500 &  & (0.795, 0.805) & 95.2 & 2.8 & 95.8 \\
(0.0438, 0.17) & (0.6, 8) & 50 & 100 & (0.685, 0.810) & (0.787, 0.810) & 95.6 & 6.3 & 95.7 \\
 &  &  & 500 &  & (0.793, 0.810) & 95.1 & 4.5 & 95.5 \\
 &  & 100 & 100 & (0.719, 0.810) & (0.788, 0.808) & 94.2 & 5.3 & 94.7 \\
 &  &  & 500 &  & (0.794, 0.808) & 94.9 & 3.7 & 95.7 \\
(0.1, 0.2133) & (0, 6.5) & 50 & 100 & (0.672, 0.856) & (0.792, 0.808) & 95.2 & 7.1 & 96.0 \\
 &  &  & 500 &  & (0.796, 0.804) & 95.2 & 5.8 & 95.8 \\
 &  & 100 & 100 & (0.719, 0.854) & (0.793, 0.807) & 94.8 & 5.5 & 94.9 \\
 &  &  & 500 &  & (0.797, 0.803) & 95.5 & 4.4 & 95.7 \\
\multicolumn{4}{l}{} &  &  &  & \multicolumn{1}{r}{} &  \\
\multicolumn{4}{l}{Scenario 2a ($J_1 = 6, J_2 = 12$)} &  &  &  & \multicolumn{1}{r}{} &  \\
(0.1863, 0.15) & (1,8) & 50 & 200 & (0.529, 0.826) & (0.773, 0.826) & 94.5 & 9.9 & 95.4 \\
(0.0438, 0.17) & (0.6, 8) & 50 & 200 & (0.508, 0.810) & (0.772, 0.810) & 94.7 & 10.1 & 96.1 \\
(0.1, 0.2133) & (0, 6.5) & 50 & 200 & (0.515, 0.857) & (0.783, 0.818) & 95.7 & 12.6 & 95.7 \\
\multicolumn{4}{l}{} &  &  & \multicolumn{1}{r}{} & \multicolumn{1}{r}{} & \multicolumn{1}{r}{} \\
\multicolumn{4}{l}{Scenario 2b ($J_1 = 10, J_2 = 20$)} &  &  & \multicolumn{1}{r}{} & \multicolumn{1}{r}{} & \multicolumn{1}{r}{} \\
(0.1863, 0.15) & (1,8) & 50 & 200 & (0.558, 0.821) & (0.787, 0.815) & 96.2 & 6.6 & 95.8 \\
(0.0438, 0.17) & (0.6, 8) & 50 & 200 & (0.516, 0.810) & (0.785, 0.810) & 95.2 & 6.9 & 95.8 \\
(0.1, 0.2133) & (0, 6.5) & 50 & 200 & (0.534, 0.858) & (0.790, 0.808) & 94.9 & 8.8 & 95.9 \\
 &  &  &  &  &  &  &  & 
\end{tabular}
}
\scriptsize	\\
\raggedright{
$n_j^{(1)}$: number of patients in center $j$ at stage 1,
$n_j^{(2)}$: number of patients in center $j$ at stage 2. \\
True Mean under $\hat{x}^{opt,(2,n^{(1)})}$: true mean outcome under the stage 2 recommended intervention. \\
True Mean under $\hat{x}^{opt}$: true mean outcome under the final estimated optimal intervention. \\
$Q$2.5 and $Q$97.5: 2.5$\%$ and 97.5$\%$ quantiles.\\ 
SetCP95$\%$: empirical coverage percentage of confidence set for the optimal intervention. \\
SetPerc$\%$: mean percentage of the size of the confidence set as a percent of the total sample space.\\
BandsCP95$\%$: empirical coverage of 95$\%$ confidence band. 
}
\end{table}

\begin{table}[ht] 
\caption{Simulation study results for estimated optimal intervention, confidence set and confidence band with a linear cost function}
\centering
\scriptsize{
\begin{tabular}{lllllllll}
\tiny$\boldsymbol{\beta}^*=(\beta_{11}^*,\beta_{12}^*)$ & $\boldsymbol{x}^{opt}$ & $n_j^{(1)}$ & $n_j^{(2)}$ &  &  &  &  &  \\
 &  &  &  & \begin{tabular}[c]{@{}l@{}}True Mean\\ under\\ $\hat{x}^{opt,(2,n^{(1)})}$\\ (Q2.5,Q97.5)\end{tabular} & \begin{tabular}[c]{@{}l@{}}True Mean\\ under\\ $\hat{x}^{opt,(2)}$\\ (Q2.5,Q97.5)\end{tabular} & \begin{tabular}[c]{@{}l@{}}SetCP95\\ $\%$\end{tabular} & \begin{tabular}[c]{@{}l@{}}SetPerc\\ $\%$\end{tabular} & \begin{tabular}[c]{@{}l@{}}BandsCP95\\ $\%$\end{tabular} \\
\multicolumn{4}{l}{Scenario 1 ($J_1 = J_2 = 20$)} &  &  &  &  &  \\
(0.1062, 0.16) & (1,8) & 50 & 100 & (0.675, 0.815) & (0.785, 0.812) & 95.4 & 5.4 & 95.5 \\
 &  &  & 500 &  & (0.794, 0.811) & 95.2 & 3.6 & 95.5 \\
 &  & 100 & 100 & (0.714, 0.813) & (0.789, 0.809) & 94.5 & 4.3 & 95.3 \\
 &  &  & 500 &  & (0.796, 0.807) & 95.1 & 2.9 & 95.9 \\
\multicolumn{4}{l}{} &  &  &  & \multicolumn{1}{r}{} &  \\
\multicolumn{4}{l}{Scenario 2a ($J_1 = 6, J_2 = 12$)} &  &  &  & \multicolumn{1}{r}{} &  \\
(0.1062, 0.16) & (1,8) & 50 & 200 & (0.517, 0.816) & (0.771, 0.816) & 94.5 & 9.4 & 96.2 \\
\multicolumn{4}{l}{} &  &  & \multicolumn{1}{r}{} & \multicolumn{1}{r}{} & \multicolumn{1}{r}{} \\
\multicolumn{4}{l}{Scenario 2b ($J_1 = 10, J_2 = 20$)} &  &  & \multicolumn{1}{r}{} & \multicolumn{1}{r}{} & \multicolumn{1}{r}{} \\
(0.1062, 0.16) & (1,8) & 50 & 200 & (0.536, 0.816) & (0.786, 0.816) & 95.8 & 6.2 & 95.7 \\
 &  &  &  &  &  &  &  & 
\end{tabular}
}
\label{meanopt table ctn}
\footnotesize	\\
\raggedright{
$n_j^{(1)}$: number of patients in center $j$ at stage 1,
$n_j^{(2)}$: number of patients in center $j$ at stage 2. \\
True Mean under $\hat{x}^{opt,(2,n^{(1)})}$: true mean outcome under the stage 2 recommended intervention. \\
True Mean under $\hat{x}^{opt}$: true mean outcome under the final estimated optimal intervention. \\
$Q$2.5 and $Q$97.5: 2.5$\%$ and 97.5$\%$ quantiles.\\ 
SetCP95$\%$: empirical coverage percentage of confidence set for the optimal intervention. \\
SetPerc$\%$: mean percentage of the size of the confidence set as a percent of the total sample space.\\
BandsCP95$\%$: empirical coverage of 95$\%$ confidence band. 
}
\end{table}

\newpage
\textcolor{white}{nothing}

\newpage
\subsection{Mimic the BetterBirth study}\label{mimicbbstudy}
As described in the main text, we mimicked the BetterBirth study as if a c-LAGO design was used. The true coefficient values for the simulation were based on the analysis of the actual BetterBirth study (see Section \ref{application} of the main text for details of the BetterBirth study). In particular, we used the same value for the number of centers in each stage, the number of patients in each center for each stage, the number of centers in the control/intervention arm for each stage, the values of the center baseline characteristics $z$ for each stage, the cost of the two intervention components, the lower and upper limits of the two intervention components, and the stage 1 interventions as those used in the BetterBirth study. 

Using the last column of Table \ref{estimated effects bb} from the main text, we set the true values of the intercept $\beta_{0}^*=-0.138$, the coefficient of the launch duration $\beta_{11}^*=0.17$, the coefficient of the number of coaching visits $\beta_{12}^*=0.172$, and the coefficient value of the center baseline characteristic: birth volume $\beta_z^*=-0.202$. 
In each simulation iteration,
we first used the stage 1 interventions to generate stage 1 outcomes. Then we calculated the stage 2 interventions based on the stage 1 model and stage 2 center baseline characteristics, adjusting the variance of the two intervention components to match the actual variance of the intervention components from the corresponding stage of the BetterBirth study. We generated stage 2 outcomes using the true model, and calculated stage 3 interventions based on the stage 2 model and stage 3 center baseline characteristics, using data from both stages 1 and 2.  
We followed the same steps to obtain the final model using data from all three stages, and we calculated the final predicted optimal intervention using this final model. 

The percent relative bias for $\beta_{11}$ and $\beta_{12}$ is -0.45 and -1.39, respectively. The ratio between mean estimated standard error and empirical standard deviation, and the 95\% coverage probability for $\beta_{11}$ and $\beta_{12}$ are 95.60 and 94.70, and 32.41 and 95.15, respectively. For an average center with a baseline center characteristic $z=1.75$, the true optimal intervention package is ${\boldsymbol{x}}^{opt}=(5, 29.9)$. In 2000 iterations of this simulation, the mean recommended intervention for the same average center after stage 1 is $(3.3, 29.3)$, and using data from all three stages, the mean recommended intervention is $(4.4, 32.5)$. These results suggest that LAGO would be effective if it was used in the BetterBirth study.

\subsection{Additional results of simulation 1 using cubic cost function}\label{additional cost}
As described in Section \ref{setting} of the main text, we also considered a cubic cost function 
for the LAGO design. 
The cubic cost function we considered has an initial economy of scale, followed by increasing marginal costs when the component levels exceed a certain threshold (as shown in Figure 1 of \cite{greer2010electricity}). 
This type of cubic cost function is particularly applicable in healthcare settings, where the marginal cost of a product may rise prohibitively as local supplies are depleted. For example, if a large hospital in a rural area needs to hire many nurses, the marginal cost for training and monitoring would initially decrease as more nurses are hired, benefiting from economies of scale. However, once all of the local nurses have been hired, the hospital may need to hire people from surrounding cities, which would also incur relocation costs and potentially higher wages to attract talent from afar, causing the marginal cost to eventually increase. 

Complete Tables of \ref{component_effects_table cubic cost} -- \ref{meanopt table cubic cost} from the main text are shown below.

\begin{table}[ht]
\caption{Simulation study results for individual package component effects with a cubic cost function}
\centering
\scriptsize{
\begin{tabular}{llllllllllll}
\tiny$\boldsymbol{\beta}^*=(\beta_{11}^*, \beta_{12}^*)$ & \tiny$n_j^{(1)}$ & \tiny$n_j^{(2)}$ & \tiny$J$ &  &  & \tiny$\hat{\boldsymbol{\beta}}_{11}$ &  &  &  & \tiny$\hat{\boldsymbol{\beta}}_{12}$ &  \\
 &  &  &  &  & $\%$RelBias & \begin{tabular}[c]{@{}l@{}}$\frac{SE}{EMP.SD}$\\ ($\times 100$)\end{tabular}  & CP95 &  & $\%$RelBias & \begin{tabular}[c]{@{}l@{}}$\frac{SE}{EMP.SD}$\\ ($\times 100$)\end{tabular}  & CP95 \\
\multicolumn{4}{l}{Scenario 1 ($J_1 = J_2 = J$)} &  &  &  &  &  &  &  &  \\
(0.1863, 0.15) & 50 & 100 & 6 &  & 2.31 & 81.9 & 95.8 &  & -0.79 & 78.1$^*$ & 95.5 \\
 &  &  & 10 &  & 4.59 & 95.8 & 95.4 &  & -1.63 & 88.3 & 95.0 \\
 &  &  & 20 &  & 4.35 & 101.3 & 95.7 &  & -1.49 & 94.1 & 96.0 \\
 &  & 200 & 6 &  & 1.36 & 83.9 & 95.5 &  & -0.52 & 78.8 & 95.5 \\
 &  &  & 10 &  & 2.68 & 94.2 & 95.9 &  & -0.89 & 85.6 & 96.0 \\
 &  &  & 20 &  & 3.90 & 96.3 & 96.4 &  & -1.29 & 88.1 & 96.1 \\
 & 100 & 100 & 6 &  & 2.67 & 93.1 & 95.9 &  & -0.86 & 88.5 & 95.8 \\
 &  &  & 10 &  & 5.74 & 103.6 & 96.4 &  & -1.92 & 97.5 & 96.4 \\
 &  &  & 20 &  & 6.05 & 103.3 & 95.6 &  & -2.10 & 98.1 & 95.8 \\
 &  & 200 & 6 &  & 1.69 & 89.8 & 96.4 &  & -0.51 & 84.2 & 95.8 \\
 &  &  & 10 &  & 5.48 & 94.2 & 95.6 &  & -1.74 & 86.8 & 95.5 \\
 &  &  & 20 &  & 4.91 & 103.8 & 96.1 &  & -1.72 & 97.5 & 96.3 \\
(0.1, 0.2133) & 50 & 100 & 6 &  & 6.73 & 91.0 & 96.3 &  & -0.83 & 83.4 & 95.1 \\
 &  &  & 10 &  & 7.55 & 88.6 & 95.9 &  & -0.96 & 76.7 & 95.1 \\
 &  &  & 20 &  & 11.54 & 87.8 & 94.4 &  & -1.38 & 72.8 & 94.9 \\
 &  & 200 & 6 &  & 7.99 & 78.1 & 95.6 &  & -0.99 & 67.7 & 95.7 \\
 &  &  & 10 &  & 6.62 & 93.9 & 95.4 &  & -0.82 & 79.2 & 95.0 \\
 &  &  & 20 &  & 8.32 & 89.3 & 94.7 &  & -1.04 & 68.3 & 94.9 \\
 & 100 & 100 & 6 &  & 9.43 & 88.6 & 95.8 &  & -1.19 & 80.8 & 95.0 \\
 &  &  & 10 &  & 13.20 & 85.5 & 95.9 &  & -1.53 & 75.2 & 95.0 \\
 &  &  & 20 &  & 10.39 & 78.0 & 94.6 &  & -1.25 & 65.1 & 94.8 \\
 &  & 200 & 6 &  & 6.58 & 89.5 & 96.4 &  & -0.81 & 79.2 & 95.7 \\
 &  &  & 10 &  & 9.35 & 84.9 & 95.2 &  & -1.11 & 70.1 & 95.2 \\
 &  &  & 20 &  & 7.09 & 84.8 & 95.1 &  & -0.87 & 66.9 & 95.2 \\
\multicolumn{4}{l}{Scenario 2a ($J_1 = 6, J_2 = 12$)} &  &  &  &  &  &  &  &  \\
(0.1863, 0.15) & 50 & 200 &  &  & 0.21 & 83.1 & 95.2 &  & -0.22 & 75.3 & 94.9 \\
(0.1, 0.2133) & 50 & 200 &  &  & 2.54 & 88.1 & 95.6 &  & -0.31 & 74.5 & 95.8 \\
\multicolumn{4}{l}{} &  &  &  &  &  &  &  &  \\
\multicolumn{4}{l}{Scenario 2b ($J_1 = 10, J_2 = 20$)} &  &  &  &  &  &  &  &  \\
(0.1863, 0.15) & 50 & 200 &  &  & 0.88 & 90.1 & 95.9 &  & -0.37 & 80.6 & 95.6 \\
(0.1, 0.2133) & 50 & 200 &  &  & 3.98 & 93.1 & 95.3 &  & -0.52 & 74.3 & 95.8
\end{tabular}
}
\label{component_effects_table_ctn_cubic_part1}
\footnotesize	\\
\raggedright{
$n_j^{(1)}$: number of patients in center $j$ at stage 1, 
$n_j^{(2)}$: number of patients in center $j$ at stage 2.\\
$J$: number of centers for each stage. \\
\%RelBias: percent relative bias $100(\hat{\beta}-\beta^\star)/\beta^\star$.\\
SE: mean estimated standard error, 
EMP.SD: empirical standard deviation.\\
CP95: empirical coverage rate of 95\% confidence intervals.\\
$^*:$ The mean SE for covered iterations was found to be higher than that for non-covered iterations. Despite having thoroughly verified the code's accuracy, we cannot explain why coverage remained satisfactory when SE/EMP.SD was small. However, our analysis confirmed that this phenomenon persisted.
}
\end{table}

\begin{table}[ht]
\caption{Simulation study results for individual package component effects with a cubic cost function}
\centering
\scriptsize{
\begin{tabular}{llllllllllll}
\tiny$\boldsymbol{\beta}^*=(\beta_{11}^*, \beta_{12}^*)$ & \tiny$n_j^{(1)}$ & \tiny$n_j^{(2)}$ & \tiny$J$ &  &  & \tiny$\hat{\boldsymbol{\beta}}_{11}$ &  &  &  & \tiny$\hat{\boldsymbol{\beta}}_{12}$ &  \\
 &  &  &  &  & $\%$RelBias & \begin{tabular}[c]{@{}l@{}}$\frac{SE}{EMP.SD}$\\ ($\times 100$)\end{tabular}  & CP95 &  & $\%$RelBias & \begin{tabular}[c]{@{}l@{}}$\frac{SE}{EMP.SD}$\\ ($\times 100$)\end{tabular}  & CP95 \\
\multicolumn{4}{l}{Scenario 1 ($J_1 = J_2 = J$)} &  &  &  &  &  &  &  &  \\
(0.1062, 0.16) & 50 & 100 & 6 &  & 14.53 & 81.5 & 95.8 &  & -2.43 & 75.4$^*$ & 94.9 \\
 &  &  & 10 &  & 11.46 & 88.6 & 95.4 &  & -2.12 & 79.7 & 94.9 \\
 &  &  & 20 &  & 11.90 & 93.0 & 95.4 &  & -2.10 & 83.6 & 96.1 \\
 &  & 200 & 6 &  & 11.77 & 75.5 & 96.1 &  & -2.00 & 68.4 & 95.5 \\
 &  &  & 10 &  & 7.73 & 89.3 & 95.9 &  & -1.43 & 78.3 & 96.5 \\
 &  &  & 20 &  & 9.75 & 90.2 & 95.5 &  & -1.73 & 78.1 & 95.4 \\
 & 100 & 100 & 6 &  & 13.29 & 86.5 & 96.1 &  & -2.30 & 80.7 & 95.6 \\
 &  &  & 10 &  & 16.15 & 95.0 & 95.7 &  & -2.76 & 86.0 & 95.7 \\
 &  &  & 20 &  & 16.43 & 88.3 & 94.8 &  & -2.86 & 80.5 & 95.3 \\
 &  & 200 & 6 &  & 9.92 & 82.5 & 95.8 &  & -1.67 & 75.2 & 95.2 \\
 &  &  & 10 &  & 11.59 & 89.2 & 95.8 &  & -1.94 & 79.4 & 95.9 \\
 &  &  & 20 &  & 13.16 & 88.3 & 95.4 &  & -2.31 & 77.3 & 95.3 \\
(0.0438, 0.17) & 50 & 100 & 6 &  & 41.46 & 79.5 & 96.1 &  & -2.74 & 72.8 & 95.1 \\
 &  &  & 10 &  & 47.36 & 81.9 & 95.6 &  & -3.23 & 72.2 & 95.5 \\
 &  &  & 20 &  & 35.35 & 86.4 & 94.9 &  & -2.33 & 73.6 & 95.4 \\
 &  & 200 & 6 &  & 37.62 & 76.6 & 96.4 &  & -2.51 & 69.1 & 96.0 \\
 &  &  & 10 &  & 31.61 & 81.9 & 95.8 &  & -2.16 & 69.7 & 95.4 \\
 &  &  & 20 &  & 26.28 & 81.6 & 95.6 &  & -1.79 & 65.4 & 95.1 \\
 & 100 & 100 & 6 &  & 47.42 & 79.9 & 96.0 &  & -3.21 & 73.7 & 95.4 \\
 &  &  & 10 &  & 44.37 & 85.4 & 95.8 &  & -2.91 & 75.9 & 95.0 \\
 &  &  & 20 &  & 40.01 & 74.0 & 93.7 &  & -2.62 & 64.0 & 94.2 \\
 &  & 200 & 6 &  & 35.16 & 76.7 & 96.2 &  & -2.32 & 68.8 & 95.1 \\
 &  &  & 10 &  & 33.17 & 81.6 & 95.7 &  & -2.17 & 69.5 & 95.1 \\
 &  &  & 20 &  & 30.70 & 76.5 & 94.9 &  & -2.03 & 62.5 & 94.8 \\
\multicolumn{4}{l}{Scenario 2a ($J_1 = 6, J_2 = 12$)} &  &  &  &  &  &  &  &  \\
(0.1062, 0.16) & 50 & 200 &  &  & 4.24 & 73.6 & 96 &  & -0.81 & 63.5 & 95.6 \\
(0.0438, 0.17) & 50 & 200 &  &  & 17.94 & 71.1 & 95.8 &  & -1.33 & 60.2 & 95.7 \\
\multicolumn{4}{l}{} &  &  &  &  &  &  &  &  \\
\multicolumn{4}{l}{Scenario 2b ($J_1 = 10, J_2 = 20$)} &  &  &  &  &  &  &  &  \\
(0.1062, 0.16) & 50 & 200 &  &  & 4.69 & 91.8 & 96.1 &  & -0.99 & 78.6 & 95.7 \\
(0.0438, 0.17) & 50 & 200 &  &  & 14.76 & 82.2 & 95.5 &  & -1.12 & 66.5 & 95
\end{tabular}
}
\label{component_effects_table_ctn_cubic_part2}
\footnotesize	\\
\raggedright{
$n_j^{(1)}$: number of patients in center $j$ at stage 1, 
$n_j^{(2)}$: number of patients in center $j$ at stage 2.\\
$J$: number of centers for each stage. \\
\%RelBias: percent relative bias $100(\hat{\beta}-\beta^\star)/\beta^\star$.\\
SE: mean estimated standard error, 
EMP.SD: empirical standard deviation.\\
CP95: empirical coverage rate of 95\% confidence intervals.\\
$^*:$ The mean SE for covered iterations was found to be higher than that for non-covered iterations. Despite having thoroughly verified the code's accuracy, we cannot explain why coverage remained satisfactory when SE/EMP.SD was small. However, our analysis confirmed that this phenomenon persisted.
}
\end{table}

\begin{table}[ht]
\caption{Simulation study results for estimated optimal intervention with a linear cost function}
\centering
\scriptsize{
\begin{tabular}{llllllllllll}
$\boldsymbol{\beta}^*=(\beta_{11}^*,\beta_{12}^*)$ & $\boldsymbol{x}^{opt}$ & $n_j^{(1)}$ & $n_j^{(2)}$ &  & \multicolumn{3}{c}{Stage 1} &  & \multicolumn{3}{c}{Stage 2/LAGO optimized} \\
 &  &  &  &  & \begin{tabular}[c]{@{}l@{}}Bias of\\ $\hat{x}_{1}^{o p t}$\\ ($\times$100)\end{tabular} & \begin{tabular}[c]{@{}l@{}}Bias of\\ $\hat{x}_{2}^{o p t}$\\ ($\times$100)\end{tabular} & \begin{tabular}[c]{@{}l@{}}rMSE\\ ($\times$100)\end{tabular} &  & \begin{tabular}[c]{@{}l@{}}Bias of\\ $\hat{x}_{1}^{o p t}$\\ ($\times$100)\end{tabular} & \begin{tabular}[c]{@{}l@{}}Bias of\\ $\hat{x}_{2}^{o p t}$\\ ($\times$100)\end{tabular} & \begin{tabular}[c]{@{}l@{}}rMSE\\ ($\times$100)\end{tabular} \\
\multicolumn{4}{l}{Scenario 1 ($J_1 = J_2 = 20$)} &  &  &  &  &  &  &  &  \\
(0.1863, 0.15) & (1.5,7.4) & 50 & 100 &  & 40.0 & 81.9 & 118.8 &  & 4.9 & 4.8 & 65.9 \\
 &  &  & 500 &  &  &  &  &  & 4.4 & -0.8 & 56.2 \\
 &  & 100 & 100 &  & 32.7 & 45.8 & 103.4 &  & 2.6 & 5.9 & 63.9 \\
 &  &  & 500 &  &  &  &  &  & 3.8 & 0.0 & 56.2 \\
(0.0438, 0.17) & (0.7, 8) & 50 & 100 &  & 0.0 & 153.2 & 139.9 &  & -3.3 & 53.5 & 99.1 \\
 &  &  & 500 &  &  &  &  &  & 1.1 & 57.5 & 95.1 \\
 &  & 100 & 100 &  & -7.2 & 90.0 & 118.6 &  & -7.0 & 23.3 & 82.3 \\
 &  &  & 500 &  &  &  &  &  & -3.6 & 22.2 & 73.6 \\
(0.1, 0.2133) & (0, 6.5) & 50 & 100 &  & -50.7 & 85.9 & 110.3 &  & -16.4 & 3.6 & 49.6 \\
 &  &  & 500 &  &  &  &  &  & -10.6 & 0.5 & 40.0 \\
 &  & 100 & 100 &  & -42.5 & 53.6 & 95.1 &  & -14.3 & 2.3 & 45.9 \\
 &  &  & 500 &  &  &  &  &  & -9.0 & -0.1 & 36.2 \\
(0.1062, 0.16) & (1.2, 7.9) & 50 & 100 &  & 27.4 & 103.5 & 122.1 &  & 3.3 & 14.0 & 68.7 \\
 &  &  & 500 &  &  &  &  &  & 6.3 & 8.5 & 60.4 \\
 &  & 100 & 100 &  & 23.4 & 57.2 & 103.0 &  & 0.1 & 5.3 & 60.6 \\
 &  &  & 500 &  &  &  &  &  & 3.9 & 0.0 & 51.4 \\
 &  &  &  &  &  &  &  &  &  &  &  \\
\multicolumn{4}{l}{Scenario 2a ($J_1 = 6, J_2 = 12$)} &  &  &  &  &  &  &  &  \\
(0.1863, 0.15) & (1.5,7.4) & 50 & 200 &  & 64.4 & 314.0 & 187.6 &  & 9.7 & 19.9 & 79.8 \\
(0.0438, 0.17) & (0.7, 8) & 50 & 200 &  & 8.7 & 388.6 & 202.5 &  & 13.2 & 142.1 & 134.8 \\
(0.1, 0.2133) & (0, 6.5) & 50 & 200 &  & -64.2 & 261.3 & 170.0 &  & -13.8 & 7.3 & 54.1 \\
(0.1062, 0.16) & (1.2, 7.9) & 50 & 200 &  & 46.6 & 356.9 & 195.5 &  & 19.8 & 87.6 & 110.1 \\
 &  &  &  &  &  &  &  &  &  &  &  \\
\multicolumn{4}{l}{Scenario 2b ($J_1 = 10, J_2 = 20$)} &  &  &  &  &  &  &  &  \\
(0.1863, 0.15) & (1.5,7.4) & 50 & 200 &  & 48.9 & 188.7 & 153.9 &  & 6.2 & 2.3 & 64.6 \\
(0.0438, 0.17) & (0.7, 8) & 50 & 200 &  & 2.3 & 268.1 & 173.1 &  & 6.6 & 103.8 & 119.7 \\
(0.1, 0.2133) & (0, 6.5) & 50 & 200 &  & -62.5 & 166.9 & 141.3 &  & -11.6 & 3.4 & 46.2 \\
(0.1062, 0.16) & (1.2, 7.9) & 50 & 200 &  & 35.8 & 231.1 & 163.0 &  & 9.1 & 37.7 & 84.4
\end{tabular}
}
\label{bias_table_cubic}
\footnotesize	\\
\raggedright {
$n_j^{(1)}$: number of patients in center $j$ at stage 1,
$n_j^{(2)}$: number of patients in center $j$ at stage 2.\\
Bias of $\hat{x}_{1}^{o p t}$: bias of the first component of the estimated optimal intervention. \\
Bias of $\hat{x}_{2}^{o p t}$: bias of the second component of the estimated optimal intervention.\\ 
rMSE: root of mean squared errors, $\left\{\operatorname{mean}\left(\left\|\hat{\boldsymbol{x}}^{o p t}-\boldsymbol{x}^{o p t}\right\|^{2}\right)\right\}^{1 / 2}$, mean is taken over simulation iterations.
}
\end{table}

\begin{table}[ht] 
\caption{Simulation study results for estimated optimal intervention, confidence set and confidence band with a cubic cost function}
\centering
\scriptsize{
\begin{tabular}{lllllllll}
\tiny$\boldsymbol{\beta}^*=(\beta_{11}^*,\beta_{12}^*)$ & $\boldsymbol{x}^{opt}$ & $n_j^{(1)}$ & $n_j^{(2)}$ &  &  &  &  &  \\
 &  &  &  & \begin{tabular}[c]{@{}l@{}}True Mean\\ under\\ $\hat{x}^{opt,(2,n^{(1)})}$\\ (Q2.5,Q97.5)\end{tabular} & \begin{tabular}[c]{@{}l@{}}True Mean\\ under\\ $\hat{x}^{opt}$\\ (Q2.5,Q97.5)\end{tabular} & \begin{tabular}[c]{@{}l@{}}SetCP95\\ $\%$\end{tabular} & \begin{tabular}[c]{@{}l@{}}SetPerc\\ $\%$\end{tabular} & \begin{tabular}[c]{@{}l@{}}BandsCP95\\ $\%$\end{tabular} \\
\multicolumn{4}{l}{Scenario 1 ($J_1 = J_2 = 20$)} &  &  &  &  &  \\
(0.1863, 0.15) & (1.5,7.4) & 50 & 100 & (0.656, 0.812) & (0.769, 0.809) & 95.8 & 8.3 & 96.1 \\
 &  &  & 500 &  & (0.786, 0.809) & 95.3 & 5.3 & 95.6 \\
 &  & 100 & 100 & (0.708, 0.816) & (0.776, 0.807) & 95.9 & 7.0 & 96.3 \\
 &  &  & 500 &  & (0.789, 0.806) & 95.7 & 4.6 & 95.9 \\
(0.0438, 0.17) & (0.7, 8) & 50 & 100 & (0.500, 0.809) & (0.500, 0.808) & 95.9 & 7.4 & 95.5 \\
 &  &  & 500 &  & (0.500, 0.808) & 95.4 & 5.5 & 95.7 \\
 &  & 100 & 100 & (0.500, 0.810) & (0.773, 0.808) & 94.1 & 6.0 & 94.9 \\
 &  &  & 500 &  & (0.500, 0.807) & 94.4 & 4.5 & 95.6 \\
(0.1, 0.2133) & (0, 6.5) & 50 & 100 & (0.688, 0.819) & (0.786, 0.818) & 94.9 & 9.7 & 96.0 \\
 &  &  & 500 &  & (0.790, 0.815) & 95.2 & 7.2 & 95.8 \\
 &  & 100 & 100 & (0.727, 0.818) & (0.791, 0.814) & 94.7 & 7.4 & 95.0 \\
 &  &  & 500 &  & (0.793, 0.813) & 95.0 & 5.4 & 95.4 \\
(0.1062, 0.16) & (1.2, 7.9) & 50 & 100 & (0.550, 0.814) & (0.766, 0.812) & 96.1 & 7.4 & 96.2 \\
 &  &  & 500 &  & (0.777, 0.810) & 95.3 & 5.0 & 96.2 \\
 &  & 100 & 100 & (0.713, 0.812) & (0.780, 0.810) & 95.3 & 5.9 & 95.8 \\
 &  &  & 500 &  & (0.789, 0.808) & 95.4 & 4.1 & 95.9 \\
 &  &  &  &  &  &  & \multicolumn{1}{r}{} &  \\
\multicolumn{4}{l}{Scenario 2a ($J_1 = 6, J_2 = 12$)} &  &  &  & \multicolumn{1}{r}{} &  \\
(0.1863, 0.15) & (1.5,7.4) & 50 & 200 & (0.509, 0.814) & (0.735, 0.821) & 94.8 & 11.5 & 95.8 \\
(0.0438, 0.17) & (0.7, 8) & 50 & 200 & (0.500, 0.807) & (0.500, 0.809) & 94.9 & 11.1 & 96.3 \\
(0.1, 0.2133) & (0, 6.5) & 50 & 200 & (0.515, 0.823) & (0.775, 0.820) & 95.9 & 14.7 & 95.8 \\
(0.1062, 0.16) & (1.2, 7.9) & 50 & 200 & (0.500, 0.812) & (0.500, 0.813) & 94.8 & 10.7 & 96.2 \\
 &  &  &  &  &  & \multicolumn{1}{r}{} & \multicolumn{1}{r}{} & \multicolumn{1}{r}{} \\
\multicolumn{4}{l}{Scenario 2b ($J_1 = 10, J_2 = 20$)} &  &  & \multicolumn{1}{r}{} & \multicolumn{1}{r}{} & \multicolumn{1}{r}{} \\
(0.1863, 0.15) & (1.5,7.4) & 50 & 200 & (0.556, 0.815) & (0.771, 0.814) & 96.2 & 8.3 & 95.7 \\
(0.0438, 0.17) & (0.7, 8) & 50 & 200 & (0.500, 0.808) & (0.500, 0.808) & 95.3 & 8.2 & 95.9 \\
(0.1, 0.2133) & (0, 6.5) & 50 & 200 & (0.535, 0.820) & (0.782, 0.816) & 95.7 & 10.9 & 96.0 \\
(0.1062, 0.16) & (1.2, 7.9) & 50 & 200 & (0.500, 0.813) & (0.500, 0.813) & 95.7 & 7.7 & 96.3
\end{tabular}
}
\label{meanopt table cubic}
\footnotesize	\\
\raggedright{
$n_j^{(1)}$: number of patients in center $j$ at stage 1,
$n_j^{(2)}$: number of patients in center $j$ at stage 2. \\
True Mean under $\hat{x}^{opt,(2,n^{(1)})}$: true mean outcome under the stage 2 recommended intervention.\\
True Mean under $\hat{x}^{opt}$: true mean outcome under the final estimated optimal intervention. \\
$Q$2.5 and $Q$97.5: 2.5$\%$ and 97.5$\%$ quantiles.\\ 
SetCP95$\%$: empirical coverage percentage of confidence set for the optimal intervention. \\
SetPerc$\%$: mean percentage of the size of the confidence set as a percent of the total sample space.\\
BandsCP95$\%$: empirical coverage of 95$\%$ confidence band. 
}
\end{table}

\end{appendix}

\end{document}